\documentclass[aps,prx,superscriptaddress,longbibliography,showpacs,onecolumn, notitlepage]{revtex4-1}

\usepackage{soul}
\setstcolor{red}
\usepackage{comment}

\usepackage{caption}
\DeclareCaptionJustification{justified}{\leftskip=0pt \rightskip=0pt \parfillskip=0pt plus 1fil}

\usepackage{graphicx}
\usepackage{svg}
\usepackage{wrapfig}
\usepackage{dcolumn}
\usepackage{bm}
\usepackage{caption}
\usepackage{amsmath,amssymb,amsthm}
\usepackage{braket}
\usepackage{empheq}
\usepackage{subfig}
\usepackage{tikz}
\usepackage{url}
\usepackage{hyperref}
\usepackage{quantikz}
\usepackage{xcolor}
\usepackage{float}
\makeatletter
\let\newfloat\newfloat@ltx
\makeatother
\usepackage{algorithm}
\usepackage[noend]{algpseudocode}

\makeatletter
\renewcommand\fnum@algorithm{\algorithmname~\thealgorithm}
\makeatother
\newcommand\norm[1]{\left\lVert#1\right\rVert}

\DeclareMathOperator{\Tr}{Tr}
\DeclareMathOperator{\poly}{poly}

\DeclareMathOperator{\D}{\mathcal{D}}
\DeclareMathOperator{\U}{\mathcal{U}}
\DeclareMathOperator{\prop}{\mathrm{prop}}
\DeclareMathOperator{\re}{\mathrm{Re}}

\DeclareMathOperator{\steady}{\mathrm{ss}}
\DeclareMathOperator{\out}{\mathrm{out}}
\DeclareMathOperator{\inp}{\mathrm{in}}
\DeclareMathOperator{\QCSAT}{\mathrm{QCSAT}}
\DeclareMathOperator{\qcirc}{\mathrm{circ}}
\DeclareMathOperator{\unary}{\mathrm{unary}}
\DeclareMathOperator{\loc}{\mathrm{loc}}
\DeclareMathOperator{\ill}{\mathrm{illegal}}
\DeclareMathOperator{\leg}{\mathrm{legal}}
\newcommand\gap[1]{\mathrm{gap}\left(#1\right)}
\newcommand\Span[1]{\mathrm{span}\left(#1\right)}
\newcommand\Sp[1]{\mathrm{Sp}\left(#1\right)}
\newcommand\LLPS[1]{\textsc{#1-Local Lindbladian Pure Steady-State}}
\newcommand\LL[1]{\textsc{#1-Local Lindbladian}}
\newcommand\SEPH[1]{\textsc{Symmetric* Separable #1-Local Hamiltonian}}
\newcommand\RLL[1]{\textsc{Hermitian #1-Local Lindbladian}}

\hypersetup{
    breaklinks=true,
    colorlinks=true,       
    linkcolor=blue,          
    citecolor=red,        
    filecolor=magenta,      
    urlcolor=blue,           
    runcolor=cyan
}

\newtheorem{theorem}{Theorem}
\newtheorem*{theorem*}{Theorem}
\newtheorem{corollary}{Corollary}
\newtheorem{lemma}{Lemma}
\newtheorem*{lemma*}{Lemma}
\newtheorem{proposition}{Proposition}
\newtheorem*{proposition*}{Proposition}
\newtheorem{definition}{Definition}
\newtheorem*{definition*}{Definition}

\begin{document}

\title{On the Complexity of Finding Decoherence Free Subspaces}
\author{Evan Borras}
\affiliation{Center for Quantum Information and Control, University of New
Mexico, Albuquerque, NM 87131, USA}
\affiliation{Department of Physics and Astronomy, University of New Mexico,
Albuquerque, NM 87131, USA}

\begin{abstract}
    Decoherence free subspaces are a steady-state structure of the open quantum system which preserves quantum coherence between the states lying with in it and thus has found a variety of applications throughout quantum information science and technology. In this paper we study the computational complexity of deciding whether an open quantum system admits a decoherence free subspace or not. More specifically we study this problem with in the context of Markovian open quantum systems, governed by the time-independent Lindblad master equation. Along the way we introduce the \LL{$k$} problem, which captures the difficulty of computing purity decay rates under Lindbladian dynamics. We show that both problems are hard for the complexity class Quantum Merlin Arthur (QMA) when the locality $k \geq 5$, with the first under perfect completeness and the second being complete for QMA. Our hardness construction generalizes Kitaev’s clock Hamiltonian construction to the open quantum system setting by encoding the execution of a quantum circuit into the steady subspace of a Lindbladian containing both pure and mixed history states. This subspace is then mixed depending on the output of the encoded circuit. Our results suggest that deciding whether a generic Markovian open quantum system admits a decoherence free subspace is intractable even for quantum computation.
\end{abstract}

\maketitle

\section{Introduction}

The study of the ground state structure of closed quantum systems from a computational perspective has proven to be a compelling direction of research culminating in the creation of the field of Hamiltonian complexity \cite{gharibian2014quantum,osborne2012hamiltonian}, which roughly speaking, is the study of quantum constraint satisfaction problems. Along with Kitaev's landmark result showing the 5-Local Hamiltonian problem is QMA-Complete \cite{kitaev1999quantum, kitaev2002classical}, a multitude of other results have been shown \cite{kempe2006thecomplexity, cubitt2014complexity, bravyi2017complexity} further cementing the comparison with the study of classical constraint satisfaction problems.

In open quantum systems, ground state structure generalizes to become the steady-state structure of the dynamics. Given the fruitfulness of studying the ground state structure of the closed quantum system from a computational perspective, one may wonder what may come of studying the steady-state structure of open quantum systems from a computational perspective. As our model of an open quantum system, we consider systems governed by the time-independent Lindblad master equation \cite{manzano2020short}
\begin{equation}
    \label{eq:lindbladian}
    \frac{d\rho}{dt} = \mathcal{L}(\rho) = -i[H, \rho] + \sum_{\mu=1}^{m}\left(L_{\mu} \rho L_{\mu}^{*} - \frac{1}{2}\left\{L_{\mu}^{*}L_{\mu}, \rho\right\}\right),
\end{equation}
where $H$ denotes the system's Hamiltonian, and the jump operators $\{L_{\mu}\}_{\mu=1}^{m}$ model the system's interaction with its environment. Typically $\mathcal{L}$ is called a  Lindbladian and the portion of $\mathcal{L}$ generated solely by the jump operator is called a Lindbladian dissipator. From a physical perspective Lindbladians generate uniformly continuous quantum Markov semigroups \cite{androulakis2015generators, lindblad1975on} which describe the dynamics of an open quantum system coupled to a Markovian environment. From a computational standpoint, the Lindbladian is akin to a generator of a quantum mechanical continuous time Markov chain.

The steady-state structure of time-independent Lindbladians has been analyzed in detail throughout \cite{baumgartner2008analysisI,baumgartner2008analysisII,baumgartner2012structures} and \cite{albert2014symmetries}. Unlike the ground state structure of a Hamiltonian, all steady-states have the same eigenvalue of zero. Thus steady-state structure is typically classified according to how the dynamics drive the system into the steady-subspace. For reference on type of scheme see the discussion on steady-state structure in \cite{albert2014symmetries}.

One of the steady-state structures unique to quantum mechanical systems is that of a decoherence free subspace. Roughly speaking, a decoherence free subspace is a subspace of pure quantum states that are steady under the dynamics generated by the Lindbladian. What makes a decoherence free subspace uniquely quantum mechanical is that in addition to the dynamics admitting pure steady-states, even the quantum mechanical coherences between these steady-states must also be steady under the dynamics. In essence the dynamics does not cause the system to decohere when initialized inside the decoherence free subspace. Decoherence free subspaces are an example of an information preserving structure of open quantum system dynamics \cite{blumekohout2010information} and thus have found many applications throughout quantum information science. Some examples include the use of  decoherence free subspaces in the design of noise free quantum computation schemes \cite{lidar1998decoherence} as well as designing autonomous quantum error error correction schemes \cite{paz1969continuous}. 

In this paper, we study the computational hardness of deciding whether a given Lindbladian has a decoherence free subspace. To make headway towards this goal, we consider the simpler problem of deciding whether a Lindbladian has a pure steady-state, or alternatively, if the Lindbladian admits a 1-dimensional decoherence free subspace. Thus, the hardness of deciding whether a Lindbladian has a decoherence free subspace should be at least as hard as it is to decide whether it admits a pure steady-state.

\subsection{Complexity Theoretic Contributions}

Throughout the rest of this work we introduce the following problems.
\begin{itemize}
    \item \LLPS{$k$} (Definition~\ref{def:llps}): A decision problem, formalizing the notion of deciding whether a Lindbladian dissipator admits a 1-dimensional decoherence free subspace.
    \item \LL{$k$} (Definition~\ref{def:ll}): A decision problem capturing the hardness of computing the purity decay rate under a given Lindbladian dissipator.
    \item \RLL{$k$}: A variant of \LL{$k$} with the constraint that the Lindbladian dissipator is generated by Hermitian jump operators.
    \item \SEPH{$k$} (Definition~\ref{def:SEPH}): A variant of the \textsc{Separable $k$-Local Hamiltonian} introduced in \cite{chailloux2012complexity} under the additional restriction that in the \emph{(accept-case)} the product state is conjugate symmetric across the system bi-partition.
\end{itemize}
Our results can be summarized in FIG~\ref{fig:mainresults} depicting the various problems we introduce as well as the reductions between them.
\begin{figure}[ht!]
    \centering
    \includesvg[width=0.8\linewidth]{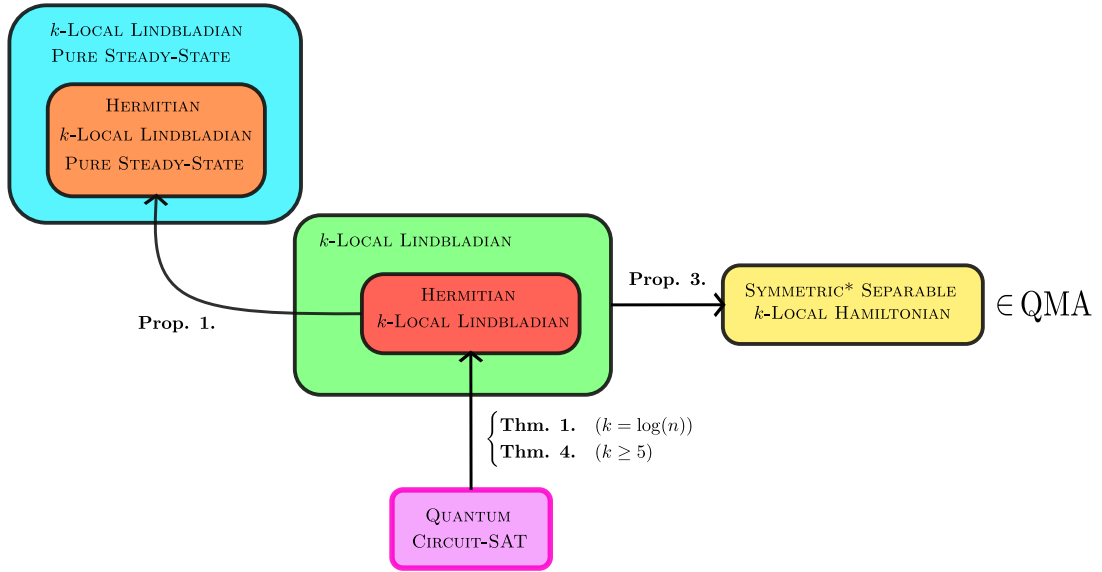}
    \caption{Diagram of the problems we analyze and their inter-relations. Arrows denote reductions between problems where if we have $A \rightarrow B$ then $A \leq B$.}
    \label{fig:mainresults}
\end{figure}

More succinctly, the highlights of what we have shown are listed below.
\begin{itemize}
    \item \LLPS{$k$} is $\mathrm{QMA}_{1}$-Hard for constant $k \geq 5$.
    \item \LL{$k$} is $\mathrm{QMA}$-Complete for constant $k \geq 5$.
    \item \SEPH{$k$} is $\mathrm{QMA}$-Complete for constant $k \geq 5$.
\end{itemize}
The complexity class Quantum Merlin Arthur (QMA) serves a similar role in quantum complexity theory that the complexity class NP does in classical complexity theory. In essence, problems in QMA are easy to verify a solution to using quantum computation, but may or may not be hard to solve using quantum computation. Under the assumption that the complexity class of quantum computationally easy to solve problems, $\mathrm{BQP}$, is not equal to $\mathrm{QMA}$, our results suggest that the problem of deciding whether a Lindbladian admits a decoherence free subspace is intractable, even for quantum computation.

\subsection{Technical Contributions}

Along with the hardness results discussed above, some of our technical results that may be of independent use. More specifically we provide the following:
\begin{itemize}
    \item \ul{A novel quantum circuit to Lindbladian dissipator mapping.} As an alternative to the construction built in \cite{verstraete2009quantum} which admits a steady-subspace of only mixed history states of the quantum cirucit, our construction drives the system into a steady-subspace that includes both mixed history states and pure history states, thus allowing the ability to interpolate between the two. We analyze both the kernel and gap of our construction, therefore giving an upper-bound on the mixing time for the dynamics generated by our construction. For the definition of our construction see Definition~\ref{def:propdiss} and for the details regarding its analysis see Appendix~\ref{app:A}.
    \item \ul{A novel unary encoding Lindbladian dissipator that is Hermitian.} This dissipator drives the system into a decoherence free subspace where the computational basis is encoded into unary. Such a construction could be built by leveraging some of the techniques used to stabilize the ground states of frustration-free Hamiltonians, see Theorem 5 of \cite{albert2018lindbladians}, but the construction would not produce a Lindbladian that is Hermitian with respect to the Hiblert-Schmidt inner product. Similar to our contribution above, we analyze the kernel as well as the gap of the Lindbladian dissipator, thus giving an upper-bound on the mixing time of its generated dynamics. For the details of our construction as well as its analysis see Appendix~\ref{app:D}.
    \item \ul{A generalization of the projection lemma} of Kempe et al. \cite{kempe2006thecomplexity}. Similar to the original tool, our pure state decay projection lemma allows one to cut out parts of operator space, projecting our analysis into an easier to manage subspace. The primary difference is that our statement is adapted to apply to the pure state decay rate under the Lindbladian's dynamics. Another major difference is that our tool allows one to project the analysis into a decoherence free subspace as opposed to the entire kernel. For the statement see Lemma~\ref{lm:psdecayproj} and for its proof see Appendix~\ref{app:E}.
\end{itemize}

The paper is divided into five sections with the first being this introduction. In the following section we introduce notation that we will use throughout the paper as well as review some tools we will be using from complexity theory. The third section we discuss decoherence free subspaces as well as introduce our definition of the \LLPS{$k$} problem. In the fourth section we introduce the \LL{$k$} and cover its relation to the \LLPS{$k$} problem as well as its inclusion in the class $\mathrm{QMA}$. Finally, in the fifth section, we show that the \LL{$k$} problem is $\mathrm{QMA}$-Hard which allows us to show $\mathrm{QMA}_{1}$-Hardness for the \LLPS{$k$} problem. We conclude with a discussion section, reviewing our results as well as providing future directions.

\section{Preliminaries}
\label{sec:preq}

We begin by fixing some notation we will use throughout the paper. We will be working with finite dimensional quantum systems, and thus the pure state space can be described by a finite dimensional Hilbert space we will commonly denote as $\mathcal{H}$. Subspaces of $\mathcal{H}$ will commonly be denoted by $\mathcal{H}$ with some subscript. The set of bounded linear operators over $\mathcal{H}$ will be denoted as $\mathcal{B}(\mathcal{H})$. If $A \in \mathcal{B}(\mathcal{H})$, then the Hermitian conjugate of $A$  with respect to the inner product on $\mathcal{H}$ will be denoted by $A^{*}$, in-addition we will use standard bra-ket notation when denoting vectors in $\mathcal{H}$ and $\mathcal{H}^{*}$ along with the inner product on $\mathcal{H}$.

The set $\mathcal{B}(\mathcal{H})$ forms a full Hilbert algebra with respect to the Hilbert-Schmidt inner product and standard matrix multiplication and addition. We denote the Hilbert-Schmidt inner product as $\braket{\bullet, \bullet}$. Similarly, if $\mathcal{R}: \mathcal{B}(\mathcal{H}) \mapsto \mathcal{B}(\mathcal{H})$ and is a bounded linear operator, then we define the Hermitian conjugate of $\mathcal{R}$ with respect to $\braket{\bullet, \bullet}$ as $\mathcal{R}^{\dagger}$. For bounded linear operators on $\mathcal{B}(\mathcal{H})$ we will typically use calligraphic font, while for bound linear operators on $\mathcal{H}$ we will typically use standard capital Latin-script. Continuing with the theme of comparison between $\mathcal{H}$ and $\mathcal{B}(\mathcal{H})$, we denote the operator norm with respect to $\mathcal{H}$ as $\norm{\bullet}$, while the operator norm with respect $\braket{\bullet, \bullet}$ over $\mathcal{B}(\mathcal{H})$ will be denoted as $\norm{\bullet}_{\infty}$. We can also define the real and imaginary parts of an operator $A \in \mathcal{B}(\mathcal{H})$,
\begin{equation}
    \begin{cases} \re A = (A + A^{*}) / 2 \\ 
                  \mathrm{Im} A = (A - A^{*}) / 2i,
    \end{cases}
\end{equation}
and an operator $\mathcal{L}$ acting on $\mathcal{B}(\mathcal{H})$,
\begin{equation}
    \begin{cases}
        \re\mathcal{L} = (\mathcal{L} + \mathcal{L}^{\dagger}) / 2 \\
        \mathrm{Im} \mathcal{L}= (\mathcal{L} - \mathcal{L}^{\dagger}) / 2i,
    \end{cases}
\end{equation}
respectively. Similarly, we define the spectrum of an $A \in \mathcal{B}(\mathcal{H})$ as $\mathrm{Sp}(A)$ and of an $\mathcal{L}$ acting on $\mathcal{B}(\mathcal{H})$ as $\mathrm{Sp}(\mathcal{L})$. Moreover, we define the restriction of an operator $A \in \mathcal{B}(\mathcal{H})$ to a subspace $\mathbb{V} \subseteq \mathcal{H}$ as
\begin{equation}
    A|_{\mathbb{V}} = P_{\mathbb{V}} A P_{\mathbb{V}} : \mathbb{V} \mapsto \mathbb{V}
\end{equation}
where $P_{\mathbb{V}}$ denotes the projector onto subspace $\mathbb{V}$. Similarly for an $\mathcal{L}$ acting on $\mathcal{B}(\mathcal{H})$ with subspace $\mathbb{V} \subseteq \mathcal{B}(\mathcal{H})$ we define
\begin{equation}
    \mathcal{L}|_{\mathbb{V}} = \mathcal{P}_{\mathbb{V}} \circ \mathcal{L} \circ \mathcal{P}_{\mathbb{V}}: \mathbb{V} \mapsto \mathbb{V}
\end{equation}
where similarly, $\mathcal{P}_{\mathbb{V}}$ denotes the projector onto subspace $\mathbb{V}$. In addition to $\norm{\bullet}$ we denote the $p$-Schatten norms of an operator $A \in \mathcal{B}(\mathcal{H})$ as $\norm{A}_{p}$ for $p \in \{1, 2, \cdots\}$. With $\norm{\bullet}_{1}$ fixed, we can define the \emph{trace distance} between $A \in \mathcal{B}(\mathcal{H})$ and $B \in \mathcal{B}(\mathcal{H})$ as $D(A, B) = \frac{1}{2}\norm{A - B}_{1}$.

Along with our notation discussed above, \emph{quantum states} or \emph{density operators} or \emph{density matrices} will be used to refer to trace normalized postive-semidefinite operators in $\mathcal{B}(\mathcal{H})$, commonly denoted using Greek-script. \emph{Pure quantum states} will be used to refer to either a normalized vector in $\mathcal{H}$ or to a rank one density operator in $\mathcal{B}(\mathcal{H})$ depending on the context. Similarly, quantum state, when used to refer to an element of $\mathcal{H}$, denotes a normalized vector of $\mathcal{H}$. Un-normalized vectors in $\mathcal{H}$ will typically be denoted with a tilde like so $\ket{\tilde{\psi}}$. 

We also define the notion of an $n$ qubit \emph{register} as an indexed set $A=\{A_{i}\}_{i=0}^{n-1}$ where each element $A_{i}$ labels a qubit in the register with the associated Hilbert space $\mathbb{C}^{2}$. Subsequently, we define the $n$ qubit Hilbert space associated to the register $A$ as $\mathcal{H}_{A} = \mathbb{C}^{2^{n}}$. For bit-string $x \in \{0,1\}^{n}$ and $n$ qubit register $A$, we denote $\ket{x} \in \mathcal{H}_{A}$ as the computational basis element associated to bit-string $x$. Given the computational basis for $\mathcal{H}_{A}$ we denote $\ket{\bar{\psi}}$ as the complex conjugate of $\ket{\psi} \in \mathcal{H}_{A}$ taken with respect to the computational basis for $\mathcal{H}_{A}$. Similarly for $G \in \mathcal{B}(\mathcal{H})$ we define $\bar{G}$ as the complex conjugate of $G$ taken with respect to the computational basis for $\mathcal{B}(\mathcal{H}_{A})$. Finally, given an index set of unitary operators $\{U_{t'}\}_{t'=0}^{t-1}$ acting on a Hilbert space $\mathcal{H}_{A}$ associated to an $n$ qubit register $A$, we can define the \emph{quantum circuit} acting on the register $A$ as the unitary $V = U_{t-1}U_{t-2}\cdots U_{0}$. The \emph{circuit size} of $V$ is simply given by $t$, in essence the number of unitary gates in set $\{U_{t'}\}_{t'=0}^{t-1}$ describing the circuit.

\subsubsection{Open Quantum Systems}
In this paper we have restricted to open quantum systems modeled by the time-independent Lindblad master equation provided in equation~\eqref{eq:lindbladian}. Before discussing the various useful properties about Lindbladian operators which will be used throughout the rest of this work, it is important to fix our definitions of locality for both Hamiltonians and Lindbladians. Let $\{H_{i}\}_{i=1}^{r}$ be a set of Hamiltonians, then $H=\sum_{i=1}^{r}$ is called \emph{$k$-local} if each $H_{i}$ acts non-trivially on at most $k$ qubits. Similarly, if $\{L_{\mu}\}_{\mu=1}^{m}$ is a set of jump operators which generate a Lindbladian dissipator $\D$, then $\D$ is called \emph{$k$-local} if each $L_{\mu}$ acts non-trivially on at most $k$ qubits. A \emph{$k$-local} Lindbladian is simply one which is generated by a $k$-local Hamiltonian $H$ and jump operators $\{L_{\mu}\}_{\mu=1}^{m}$ which generate a $k$-local dissipator $\D$. In-addition, its also useful to fix our notion of a \emph{steady-state} of a Lindbladian $\mathcal{L}$, which is simply a quantum state $\rho$ where $\rho \in \ker{\mathcal{L}}$.

One of the main features of Lindbladians is that the real part of their spectrum is negative, see proposition 1 of \cite{zhang2026direct}. Moreover Lindbladians are Hermitian preserving, i.e. any $\mathcal{L}$ satisfies $\mathcal{L}(A)^{*} = \mathcal{L}(A^{*})$. Given the characterization of a Lindbladian's spectrum we define the \emph{gap} of a Lindbladian $\mathcal{L}$ as 
\begin{equation}
    \label{eq:gapdef}
    \gap{\mathcal{L}} =  \inf_{\lambda \in \mathrm{Sp}(\mathcal{L}) \setminus \{0\}}|\re\lambda|.
\end{equation}
In addition to the properties discussed above, it is useful to have an upper-bound on $\norm{\bullet}_{\infty}$ of a Lindbladian, which we provide below.
\begin{lemma}
    \label{lm:Lupperbound}
    Let $\mathcal{L}$ be a Lindbladian, then $\norm{\mathcal{L}}_{\infty} \leq 2(\norm{H} + \sum_{\mu=1}^{m}\norm{L_{\mu}}^{2})$
\end{lemma}
\begin{proof}
    \begin{align}
        \norm{\mathcal{L}}_{\infty} &= \sup_{A \in \mathcal{B}(\mathcal{H}), \braket{A, A} = 1}\braket{\mathcal{L}(A) , \mathcal{L}(A)} =  \sup_{A \in \mathcal{B}(\mathcal{H}), \norm{A}_{2} = 1} \norm{\mathcal{L}(A)}_{2} \nonumber \\
                                    &\leq \sup_{A \in \mathcal{B}(\mathcal{H}), \norm{A}_{2} = 1} \left(\norm{HA}_{2} + \norm{AH}_{2} + \sum_{\mu=1}^{m}\left(\norm{L_{\mu}AL_{\mu}^{*}}_{2} + \frac{1}{2}\norm{L_{\mu}^{*}L_{\mu}A}_{2} + \frac{1}{2}\norm{AL_{\mu}^{*}L_{\mu}}_{2}\right)\right) \nonumber \\
                                    &\leq \left(2\norm{H}  + \sum_{\mu=1}^{m}\left(\norm{L_{\mu}}^{2} + \norm{L_{\mu}^{*}L_{\mu}}\right)\right) \leq 2\left(\norm{H} + \sum_{\mu=1}^{m}\norm{L_{\mu}}^{2}\right)
    \end{align}
\end{proof}

Throughout Section~\ref{sec:LL}, for a given Lindbladian $\mathcal{L}$ we analyze its real part. A useful property of $\re\mathcal{L}$ which we leverage throughout Section~\ref{sec:LL} as well as the Appendices is proved below.
\begin{lemma}
    \label{lm:repartL}
    Let $\mathcal{L}$ be a Lindbladian dissipator, then for all $\ket{\psi} , \ket{\phi} \in \mathcal{H}$, $\braket{\ket{\psi}\bra{\phi}, \re\mathcal{L}(\ket{\psi}\bra{\phi})} \leq 0$.
\end{lemma}
\begin{proof}
    Let $\ket{\psi}, \ket{\phi} \in \mathcal{H}$, then
    \begin{align}
        \braket{\ket{\psi}\bra{\phi}, \re\D(\ket{\psi}\bra{\phi})} &= \frac{1}{2}\sum_{\mu=1}^{m}\left(\bra{\psi}L_{\mu}\ket{\psi}\bra{\phi}L_{\mu}^{*}\ket{\phi} - \bra{\psi}L_{\mu}^{*}L_{\mu}\ket{\psi} - \bra{\phi}L_{\mu}^{*}L_{\mu}\ket{\phi} + \bra{\psi}L_{\mu}^{*}\ket{\psi}\bra{\phi}L_{\mu}\ket{\phi}\right) \nonumber \\
                                                                               &= \sum_{\mu=1}^{m}\left(\re(\bra{\psi}L_{\mu}\ket{\psi}\bra{\phi}L_{\mu}^{*}\ket{\phi}) - \frac{1}{2}\bra{\psi}L_{\mu}^{*}L_{\mu}\ket{\psi} - \frac{1}{2}\bra{\phi}L_{\mu}^{*}L_{\mu}\ket{\phi}\right) \nonumber \\
                                                                               &= \sum_{\mu=1}^{m}\Big(- \frac{1}{2}(|\bra{\psi}L_{\mu}\ket{\psi}|^{2} -2 \re(\bra{\psi}L_{\mu}\ket{\psi}\bra{\phi}L_{\mu}^{*}\ket{\phi}) + |\bra{\phi}L_{\mu}\ket{\phi}|^{2} )\nonumber \\
                                                                               &+ \frac{1}{2}(|\bra{\psi}L_{\mu}\ket{\psi}|^{2} - \bra{\psi}L_{\mu}^{*}L_{\mu}\ket{\psi}) + \frac{1}{2}(|\bra{\phi}L_{\mu}\ket{\phi}|^{2}) - \bra{\phi}L_{\mu}^{*}L_{\mu}\ket{\phi})\Big) \nonumber \\
                                                                               &= \frac{1}{2}\sum_{\mu=1}^{m}\Big( |\bra{\psi}L_{\mu}\ket{\psi}|^{2} - \bra{\psi}L_{\mu}^{*}L_{\mu}\ket{\psi} + |\bra{\phi}L_{\mu}\ket{\phi}|^{2} - \bra{\phi}L_{\mu}^{*}L_{\mu}\ket{\phi} - |\bra{\psi}L_{\mu}\ket{\psi} - \bra{\phi}L_{\mu}\ket{\phi}|^{2}\Big) \nonumber \\
                                                                               &\leq \frac{1}{2}\sum_{\mu=1}^{m}\Big(\bra{\psi}L_{\mu}^{*}L_{\mu}\ket{\psi}\braket{\psi|\psi} - \bra{\psi}L_{\mu}^{*}L_{\mu}\ket{\psi} + \bra{\phi}L_{\mu}^{*}L_{\mu}\ket{\phi}\braket{\phi|\phi} - \bra{\phi}L_{\mu}^{*}L_{\mu}\ket{\phi} \Big) = 0.
    \end{align}
\end{proof}

Along with the various useful properties of Lindbladians discussed above, a technique useful in the analysis of open quantum systems is the operator-vector correspondence \cite{watrous_2018} between $\mathcal{B}(\mathcal{H})$ and $\mathcal{H}^{\otimes 2}$. The operator-vector correspondence is the isomorphism given by
\begin{equation}
    \mathcal{B}(\mathcal{H}) \ni \ket{i}\bra{j} \mapsto \mathrm{vec}(\ket{i}\bra{j}) = \ket{i}\otimes\ket{j} \in \mathcal{H}^{\otimes 2},
\end{equation}
where $\ket{i}$ denotes a computational basis element for $\mathcal{H}$. The $\mathrm{vec}$ described above roughly speaking maps pure states in $\mathcal{B}(\mathcal{H})$ to product states in $\mathcal{H}^{\otimes 2}$ and mixed states in $\mathcal{B}(\mathcal{H})$ to entangled states in $\mathcal{H}^{\otimes 2}$. To see how this works consider the pure state $\ket{\psi}\bra{\psi} \in \mathcal{B}(\mathcal{H})$, applying our mapping we find
\begin{equation}
    \mathrm{vec}(\ket{\psi}\bra{\psi}) = \sum_{ij}c_{i}c_{j}^{*}\mathrm{vec}(\ket{i}\bra{j}) = \sum_{ij}c_{i}c_{j}^{*}\ket{i}\otimes\ket{j} = \ket{\psi}\otimes\ket{\bar{\psi}},
\end{equation}
thus our pure state in $\mathcal{B}(\mathcal{H})$ maps to product state in $\mathcal{H}^{\otimes 2}$. For mixed states $\rho \in \mathcal{B}(\mathcal{H})$ we find 
\begin{equation}
    \mathrm{vec}(\rho) = \sum_{i}p_{i}\mathrm{vec}(\ket{\psi_{i}}\bra{\psi_{i}}) = \sum_{i}p_{i}\ket{\psi_{i}}\otimes \ket{\bar{\psi}_{i}},
\end{equation}
which is an un-normalized entangled state in $\mathcal{H}^{\otimes 2}$. The $\mathrm{vec}$ map induces a similar isomorphism on the bounded linear operators acting on $\mathcal{B}(\mathcal{H})$ given by 
\begin{equation}
    \mathrm{vec}(\mathcal{L}(A)) = \mathrm{Op}(\mathcal{L})\mathrm{vec}(A).
\end{equation}
From \cite{watrous_2018} we have the following statement regarding the $\mathrm{Op}$ map.
\begin{lemma}[From Proposition 2.20 of \cite{watrous_2018}]
    \label{lm:opmap}
    Let $\mathcal{M}$ be a bounded linear operator on $\mathcal{B}(\mathcal{H})$ defined by $\mathcal{M}(\bullet) = \sum_{\mu}A_{\mu} \bullet B_{\mu}$ for $A_{\mu}, B_{\mu} \in \mathcal{B}(\mathcal{H})$, then $\mathrm{Op}(\mathcal{M}) = \sum_{\mu}A_{\mu}\otimes \bar{B}_{\mu}.$ 
\end{lemma}

\subsubsection{Complexity Theory}

Given that the topic of this paper involves computational hardness we will be using some notions from complexity theory \cite{arora2009computational, papadimitriou1994computational, moore2011nature} throughout this work. Roughly speaking, complexity theory is concerned with how the computational resources required to solve a problem scale, as the size the problem increases. Typically the computational resources considered are time and space, which roughly correspond to the number of steps and size of memory an algorithm needs in-order to solve an instance of the problem. Usually an algorithm is called \emph{efficient} if its computational resource costs scale as a polynomial of the size of the instance of the problem.

An \emph{instance of} a problem is formalized by a string $x \in \{0, 1\}^{*}$ which encodes the data needed to solve the problem. The size of the instance of the problem $x$ is then given by the length of the string $x$ denoted $|x|$. The types of problems we will concern ourselves with are called decision problems where the output to the problem is to either accept or reject the instance of the problem. More specifically we consider the notion of a promise problem, where a \emph{promise problem} \textsc{L} is a pair $\textsc{L} = (\textsc{L}_{A}, \textsc{L}_{R})$ where $\textsc{L}_{A}, \textsc{L}_{R} \subseteq \{0, 1\}^{*}$ and $\textsc{L}_{A} \cap \textsc{L}_{R} = \emptyset$. The set $\textsc{L}_{A}$ contains the \emph{(accept-case)} instances of the problem \textsc{L} while $\textsc{L}_{R}$ contains the \emph{(reject-case)} instances and every instance $x$ of a promise problem \textsc{L} is promised to be contained in either $\textsc{L}_{A}$ or $\textsc{L}_{R}$.

The class of promise problems we consider is the complexity class Quantum Merlin Arthur or $\mathrm{QMA}$. Roughly speaking the class $\mathrm{QMA}$ formalizes the notion of problems that are easy for quantum computation to verify a solution to but may or may not be hard to find a solution for. In essence if a verifier, Arthur, has access to a prover, Merlin, which can provide a quantum state as proof of a solution for the problem, then Arthur can efficiently verify using quantum computation if the quantum state given by Merlin is a valid proof for the problem or not. Formally we define the class $\mathrm{QMA}$ below.
\begin{definition}[$\mathrm{QMA}$]
    \label{def:QMA}
    A promise promblem \textsc{L} is in $\mathrm{QMA}$ if there exists a polynomial-time uniform family of quantum circuits $V_{x}$ acting on register $P$ of $\poly(|x|)$ qubits such that for every instance $x$ of problem \textsc{L},
    \begin{itemize}
        \item \emph{(accept-case)} if $x \in \textsc{L}_{A}$ then there exists a quantum state $\ket{\psi_{x}} \in \mathcal{H}_{P}$ such that $\mathrm{Pr}\left[\mathrm{accept} \; | \; V_{x}\ket{\psi_{x}}\right] \geq b$.
        \item \emph{(reject-case)} if $x \in \textsc{L}_{R}$ then for all states $\ket{\psi} \in \mathcal{H}_{P}$, $\mathrm{Pr}\left[\mathrm{accept} \; | \; V_{x}\ket{\psi}\right] \leq a$.
    \end{itemize}
    where $b - a \geq 1 / \poly(n)$  and $0 < \epsilon < a < b < 1 - \epsilon$ with $\epsilon \geq 1 / 2^{-\poly(n)}$. For $b=1$ the class is denoted $\mathrm{QMA}_{1}$.
\end{definition}
Note, by polynomial-time uniform family of quantum circuits $V_{x}$ we mean a family of quantum circuits $V_{x}$ that can be constructed from an instance $x$ using a classical algorithm in $\poly(|x|)$ time. Such a construction ensures that each $V_{x}$ has a circuit size of $\poly(|x|)$ as well. In essence each circuit $V_{x}$ can be compiled efficiently and has a circuit size that also scales efficiently with the size of the problem. In-addition, when the instance of the problem is in the \emph{(accept-case)} then Arthur's efficiently compiled and run quantum circuit can verify with high probability that the instance of the problem is to be accepted assuming Merlin provides the correct quantum proof. This is to be contrasted with the \emph{(reject-case)} of the problem where Arthur is guaranteed to be able to verify that the instance of the problem is to be rejected with a low probability of error. Typically the \emph{(accept-case)} condition is called the \emph{completeness} condition of the class and the \emph{(reject-case)} condition is called the \emph{soundness} condition of the class. This naming convention follows from logic where we want true statements to have proofs while false statements to not have valid proofs.

As an important side note, the class $\mathrm{QMA}_{1}$ is included in $\mathrm{QMA}$ via the above definition, but it is unknown whether it equals $\mathrm{QMA}$. Thus we must treat these two classes as potentially different complexity classes. For a good reference on quantum complexity theory see \cite{watrous2008quantumcomputationalcomplexity}.

The final notion from complexity theory that we will need to review is the definition of a \emph{reduction} or what is often called a polynomial-time many-one reduction. A problem $\textsc{L}^{(1)}$ can be \emph{reduced to} a problem $\textsc{L}^{(2)}$, denoted $\textsc{L}^{(1)} \leq \textsc{L}^{(2)}$ if there exists a polynomial-time classical algorithm that transforms every instance $x$ of problem $\textsc{L}^{(1)}$ to an instance $y$ of problem $\textsc{L}^{(2)}$ where $x \in \textsc{L}^{(1)}_{A}$ transforms into an instance $y \in \textsc{L}^{(2)}_{A}$ and $x \in \textsc{L}^{(1)}_{R}$ transforms into an instance $y \in \textsc{L}^{(2)}_{R}$. In essence, instances of $\textsc{L}^{(1)}$ in the  \emph{(accept-case)} transform into instances of $\textsc{L}^{(2)}$ in the \emph{(accept-case)} and instances of $\textsc{L}^{(1)}$ in the  \emph{(reject-case)} transform into instances of $\textsc{L}^{(2)}$ in the \emph{(reject-case)} under a reduction from $\textsc{L}^{(1)}$ to $\textsc{L}^{(2)}$. Typically the condition that \emph{(accept-case)} of $\textsc{L}^{(1)}$ transforms into an \emph{(accept-case)} of $\textsc{L}^{(2)}$ is called the \emph{completeness} of the reduction. Similarly, for the condition that an \emph{(reject-case)} of $\textsc{L}^{(1)}$ transforms into a \emph{(reject-case)} of $\textsc{L}^{(2)}$ is called the \emph{soundness} of the reduction. Intuitively if $\textsc{L}^{(1)} \leq \textsc{L}^{(2)}$ then $\textsc{L}^{(1)}$ is as hard as $\textsc{L}^{(2)}$ since if we could solve $\textsc{L}^{(2)}$ efficiently, we could convert $\textsc{L}^{(1)}$ into $\textsc{L}^{(2)}$ efficiently while also preserving its \emph{(accept-case)} and \emph{(reject-case)} instances, then solve $\textsc{L}^{(1)}$ by apply our method to solve $\textsc{L}^{(2)}$ efficiently. Thus the whole procedure to solve $\textsc{L}^{(1)}$ via $\textsc{L}^{(2)}$ is efficient. This intuition is where we define our notion of $\mathrm{QMA}$-Hardness from. A problem \textsc{M} is called $\mathrm{QMA}$-Hard if for all $\textsc{L} \in \mathrm{QMA}$ then $\textsc{L} \leq \textsc{M}$. Essentially the whole complexity class $\mathrm{QMA}$ is as hard as problem \textsc{M}. In-addition, \textsc{M} is called $\mathrm{QMA}$-Complete if $\textsc{M} \in \mathrm{QMA}$ as well. A similar notion for $\mathrm{QMA}_{1}$-Hard and $\mathrm{QMA}_{1}$-Complete can be defined as well. For references on complexity

\section{Decoherence Free Subspaces}
\label{sec:dfs}

Intuitively a decoherence free subspace is a pure state subspace of the system's Hilbert space $\mathcal{H}$ that does not feel the noise the environment induces on the system. Thus, if the environment induced dynamics is modeled by a Lindbladian $\mathcal{L}$ a decoherence free subspace $\mathcal{H}_{\mathrm{DFS}}$ is a subspace of $\mathcal{H}$ where the space $\mathcal{B}(\mathcal{H}_{DFS})$ is contained in the kernel of $\mathcal{L}$. Decoherence free subspaces are a uniquely quantum mechanical steady-state structure of $\mathcal{L}$ since even the coherences between pure states that lie in the decoherence free subspace are also preserved under the dynamics of $\mathcal{L}$. If the Lindbladian $\mathcal{L}$ can be cleanly written as a sum between a portion generated solely by a Hamiltonian $H$ and a dissipator $\mathcal{D}$, then we typically ignore the portion of $\mathcal{L}$ generated by $H$ as long as the Hamiltonian maps the decoherence free subspace of $\mathcal{D}$ onto itself \cite{karasik2008criteria}. Although the portion of $\mathcal{L}$ generated by the Hamiltonian $H$ does not annihilate every operator in $\mathcal{B}(\mathcal{H}_{\mathrm{DFS}})$, the dynamics does reduce to unitary dynamics on $\mathcal{B}(\mathcal{H}_{\mathrm{DFS}})$ and thus can be accounted for when engineering control over the decoherence free subspace since no loss of coherence occurs. For our purposes we will consider $\mathcal{B}(\mathcal{H}_{\mathrm{DFS}})$ to be contained inside the kernel of the entire Lindbladian, with the  Lindlbadian's we will consider being Lindbladian dissipators.

If the Lindbladian $\mathcal{L}$ is given solely by a dissipator then we have the following useful Proposition from \cite{karasik2008criteria} providing a necessary and sufficient conditions for $\mathcal{L}$ to admits a pure steady-state.
\begin{lemma}[Proposition 2 in \cite{karasik2008criteria}]
    \label{lm:steadystatecond}
    Let $\D$ be a Lindbladian dissipator generated by jump operators $\{L_{\mu}\}_{\mu=1}^{m}$, then $\D(\ket{\psi}\bra{\psi}) = 0$ if and only if for all $\mu \in [m]$, $L_{\mu}\ket{\psi} = c_{\mu}\ket{\psi}$ and $\sum_{\mu=1}^{m}L_{\mu}^{*}L_{\mu}\ket{\psi} = \sum_{\mu=1}^{m}|c_{\mu}|^{2} \ket{\psi}$
\end{lemma}
One can see that requiring this statement to hold for all pure states with-in a subspace of $\mathcal{H}$ gives us a necessary and sufficient condition for when a Lindbladian dissipator admits a decoherence free subspace. In-addition, if we choose to restrict the jump operators used to generate the dissipator, we can further simplify the above condition, giving us the following Lemma.
\begin{lemma}
    \label{lm:restrictdfscond}
    Let $\D$ be a Lindbladian dissipator generated by jump operators $\{L_{\mu}\}_{\mu=1}^{m}$ where $L_{\mu} = \L_{\mu}^{*}$, then $\mathcal{D}(\ket{\psi}\bra{\psi}) = 0$ if and only if for all $\mu \in [m]$, $L_{\mu}\ket{\psi} = c_{\mu}\ket{\psi}$.
\end{lemma}
\begin{proof}
    Consider the backwards direction. We have for all $\mu\in [m]$, $L_{\mu}\ket{\psi} = c_{\mu}\ket{\psi}$ from our assumption, in-addtion, we have
    \begin{equation}
        \sum_{\gamma=1}^{m}L_{\mu}^{*}L_{\mu}\ket{\psi} = \sum_{\gamma=1}^{m}L_{\mu}^{2}\ket{\psi} =  \sum_{\mu=1}^{m}c_{\mu}^{2}\ket{\psi}  = \sum_{\mu=1}^{m}|c_{\mu}|^{2}\ket{\psi}.
    \end{equation}
    Thus by Lemma~\ref{lm:steadystatecond} we have $\D(\ket{\psi}\bra{\psi}) = 0$ implying the statement of the backwards direction. Since the forward direction of the statement is implied by Lemma~\ref{lm:steadystatecond}, the statement of Lemma~\ref{lm:restrictdfscond} holds.
\end{proof}

Given that all decoherence free subspaces contain at least a single pure state, we can define a simplified version of the decoherence free subspace problem, described below.
\begin{definition}[\LLPS{$k$} Problem]
    \label{def:llps}
    Let $\D$ be a $k$-local Lindbladian dissipator acting on $n$ qubits with associated Hilbert space $\mathcal{H}$, given by jump operators $\{L_{\mu}\}_{\mu=1}^{m=\poly(n)}$, with each jump operator having bounded operator norm $\norm{L_{\mu}} \leq \poly(n)$ decide,
    \begin{itemize}
        \item \emph{(accept-case)} there exists a pure quantum state $\ket{\psi} \in \mathcal{H}$ such that $\D(\ket{\psi}\bra{\psi}) = 0$,
        \item \emph{(reject-case)} for all pure  quantum states $\ket{\psi} \in \mathcal{H}$ and all steady-states $\rho_{\steady} \in \mathcal{B}(\mathcal{H})$, $D(\ket{\psi}\bra{\psi}, \rho_{\steady}) \geq 1 /\poly(n)$,
    \end{itemize}
    promised one of these to be the case.
\end{definition}
We define an \emph{instance of} \LLPS{$k$} as the set of jump operators $J$ generating $\D$ as specified in Definition~\ref{def:llps}

The \LLPS{$k$} problem asks one to decide whether a given Lindbladian has a pure steady-state or not. In the \emph{(accept-case)} of the problem the Lindbladian has at least a single pure steady-state, and thus a decoherence free subspace of dimension one or greater. In the \emph{(reject-case)} of the problem the Lindbladian cannot have a decoherence free subspace since all of the steady-states are guranteed to be mixed. Thus deciding whether a given Lindbladian admits a decoherence free subspace should be at least as hard as deciding whether it admits a single pure steady-state, ie deciding the \LLPS{$k$} problem.

\section{The Local Lindbladian Problem} 
\label{sec:LL}

The \LLPS{$k$} problem introduced in Definition~\ref{def:llps} provides us a simpler problem that encompasses the difficulty of deciding whether a Lindbladian dissipator has a decoherence free subspace. In-order to show that the \LLPS{$k$} problem is hard, we find it useful to introduce a new problem, for which it is easier to show hardness for, and can be reduced to the \LLPS{$k$} problem. We call this problem the \LL{$k$}, for which a restricted form of this problem can be reduced to an instance of the \LLPS{$k$} problem. Thus if we can show this restriced form of the \LL{$k$} problem is hard then we also can show that the general \LLPS{$k$} problem is hard as well.

Not only do find the introduction of the \LL{$k$} problem useful in our hardness argument for the \LLPS{$k$} problem, but we also find that the general \LL{$k$} problem can be reduced to a type of Local Hamiltonian problem. Thus building an interesting bridge between problems regarding the steady-state structure of open quantum systems and that of the ground state structure of local Hamiltonians. In addition, this reduction allows us to show that the \LL{$k$} is $\mathrm{QMA}$-Complete. We define our notion of the \LL{$k$} problem below.
\begin{definition}[\LL{$k$} Problem]
    \label{def:ll}
    Let $\D$ be a $k$-local Lindbladian dissipator acting on $n$ qubits with associated Hilbert space $\mathcal{H}$, given by jump operators $\{L_{\mu}\}_{\mu=1}^{m=\poly(n)}$, with each jump operator having bounded operator norm $\norm{L_{\mu}} \leq \poly(n)$ decide,
    \begin{itemize}
        \item \emph{(accept-case)} there exists a pure quantum state $\ket{\psi} \in \mathcal{H}$ such that $|\braket{\ket{\psi}\bra{\psi}, \re\D(\ket{\psi}\bra{\psi})}| \leq a$,
        \item \emph{(reject-case)} for all pure states $\ket{\psi} \in \mathcal{H}$, $|\braket{\ket{\psi}\bra{\psi}, \re\D(\ket{\psi}\bra{\psi})}| \geq b$,
    \end{itemize}
    promised one of these to be the case, where $b-a \geq 1 / \poly(n)$.
\end{definition}
We define an \emph{instance of} \LL{$k$} as a tuple $(J, a, b)$ where  $J$ denotes the set of jump operators generating $\D$ and $a$ and $b$ are the thresholds as specified in Definition~\ref{def:ll}. 

At first glance, the \LL{$k$} problem appears to resemble a local Hamiltonian problem more so than the \LLPS{$k$} problem does. 
From a physics perspective, the \LL{$k$} problem is concerned with the decay of the pure state manifold under the dynamics generated by $\D$. To see this, consider calculating the decay rate of the purity at time $t$. We find,
\begin{equation}
    \left|\frac{d}{dt}\left[\mathrm{purity}(\rho(t))\right]\right| = \left|\frac{d}{dt}\braket{\rho(t), \rho(t)} =  \braket{\rho(t), (\D + \D^{\dagger})[\rho(t)]}\right| = 2\left|\braket{\rho(t), \re\D(\rho(t))}\right|.
\end{equation}
Thus, the \LL{$k$} problem asks us to decide if there exists a pure state with a purity decay rate slower than some threshold $a$, or if the entire pure state manifold decays faster than some threshold $b$. 

We also introduce our restricted version of the \LL{$k$} problem which is reducible to our pure steady-state problem. Formally we define the \RLL{$k$} problem as the \LL{$k$} problem with the additional constraint that each jump operator $L_{\mu}$ in Definition~\ref{def:ll} is Hermitian. An \emph{instance of} \RLL{$k$} is defined similarly to that of \LL{$k$}, where the set of jump operators $J$ in the instance $(J, a, b)$ must satisfy the additional constraint.

In-addition to Lemma~\ref{lm:restrictdfscond} of Section~\ref{sec:dfs}, the main tool used in our reduction of the \RLL{$k$} problem to the \LLPS{$k$} problem is what we call our distance to decay Lemma, which allows us to relate the trace distance from the steady-subspace to the decay under our Lindbladian.
\begin{lemma}[Distance to Decay Lemma]
    \label{lm:dtod}
    Let $\mathcal{L}$ be a Lindbladian, for all $S \in  \ker{\mathcal{L}}$ and $A \in \mathcal{B}(\mathcal{H})$,
    \begin{equation}
        D(A, S) \geq \frac{1}{2\norm{\mathcal{L}^{\dagger}}_{\infty}}|\braket{A, \re\mathcal{L}(A)}|.
    \end{equation}
\end{lemma}
\begin{proof}
    Let $A \in \mathcal{B}(\mathcal{H})$ and $S \in \ker{\mathcal{L}}$, we find
    \begin{align}
        |\braket{A, \re\D(A)}| &= \frac{1}{2}|\braket{\mathcal{L}(A), A} + \braket{A, \mathcal{L}(A)}| \leq  |\braket{A, \mathcal{L}(A - S)}| = |\braket{\mathcal{L}^{\dagger}(A), A - S}| \nonumber \\
                               &\leq \norm{\mathcal{L}^{\dagger}(A)}_{2}\norm{A - S}_{2} \leq \norm{\mathcal{L}^{\dagger}}_{\infty}\norm{A - S}_{2} \leq \norm{\mathcal{L}^{\dagger}}_{\infty}\norm{A - S}_{1} = 2\norm{\mathcal{L}^{\dagger}}_{\infty}D(A, S)
    \end{align}
    which implies the statement of Lemma~\ref{lm:dtod}.
\end{proof}
Using the above statement along with Lemma~\ref{lm:restrictdfscond} we can show that the \LLPS{$k$} is as hard as \RLL{$k$} when $a = 0$. To carry this out, we show that any instance of \RLL{$k$} with $a=0$ can be transformed efficiently into an instance of \LLPS{$k$} that maps an \emph{(accept-case)} or \emph{(reject-case)} instance of \RLL{$k$} to the corresponding instance  of \LLPS{$k$}. 
\begin{proposition}
    \label{prop:rlltollps}
    $\RLL{$k$} \text{ with } a=0 \leq \LLPS{$k$}$.
\end{proposition}
\begin{proof}
    Let $(J, 0, b)$ be an instance of \RLL{$k$} with $a=0$. Since $(J, 0, b)$ is a valid instance of \RLL{$k$}, we have that the set $J$ generates a $k$-local Lindbladian dissipator, in-addition  $\norm{L} \leq \poly(n)$, therefore the set $J$ forms a valid instance of \LLPS{$k$}. Transforming $(J, 0, b)$ to $J$ can be done in constant time. \\

    \ul{\textbf{Completeness:}} \\

    Let $(J, 0, b)$ be an \emph{(accept-case)} instance of \RLL{$k$} with $a=0$ and $J$ be the corresponding instance of \LLPS{$k$} generated from $(J, 0, b)$ according to the description above. For $J$ to be an \emph{(accept-case)} instance of \LLPS{$k$} we require there to exist a $\ket{\psi}$ such that
    \begin{equation}
        \D(\ket{\psi}\bra{\psi}) =0,
    \end{equation}
    where $\D$ is the Lindbladian dissipator generated by the jump operators $J$. Since $(J, 0, b)$ is an \emph{(accept-case)} instance of \RLL{$k$}, choose $\ket{\psi}$ such that 
    \begin{equation}
        \label{eq:LLPStoLL1}
        |\braket{\ket{\psi}\bra{\psi}, \re\D(\ket{\psi}\bra{\psi})}| = 0.
    \end{equation}
    Equation~\eqref{eq:LLPStoLL1} implies
    \begin{equation}
        \braket{\ket{\psi}\bra{\psi}, \D(\ket{\psi}\bra{\psi})} = 0.
    \end{equation} 
    which gives us
    \begin{equation}
        \label{eq:LLPStoLL2}
        \sum_{\mu=1}^{m}\left(|\bra{\psi}L_{\mu}\ket{\psi}|^{2} -\bra{\psi}L_{\mu}^{*}L_{\mu}\ket{\psi} \right) = 0.
    \end{equation}
    Note, from Cauchy-Schwartz inequality we have $|\bra{\psi}L_{\mu}\ket{\psi}|^{2} \leq \bra{\psi}L_{\mu}^{*}L_{\mu}\ket{\psi}$ which implies each of the terms in the sum over $\mu$ in equation~\eqref{eq:LLPStoLL2}  are less than or equal to 0. Therefore equation~\eqref{eq:LLPStoLL2} implies for all $\mu \in [m]$
    \begin{equation}
        \label{eq:LLPStoLL3}
        |\bra{\psi}L_{\mu}\ket{\psi}|^{2} -\bra{\psi}L_{\mu}^{*}L_{\mu}\ket{\psi}  = 0.
    \end{equation}
    Let $L_{\mu}\ket{\psi} = c_{\mu} \ket{\psi} + d_{\mu} \ket{\psi^{\perp}}$, applying equation~\eqref{eq:LLPStoLL3} we find
    \begin{equation}
        |c_{\mu}|^{2} - (|c_{\mu}|^{2} + |d_{\mu}|^{2}) = -|d_{\mu}|^{2} = 0
    \end{equation}
    which implies for all $\mu \in [m]$
    \begin{equation}
        \label{eq:Lmuis0}
        L_{\mu}\ket{\psi} = c_{\mu} \ket{\psi}.
    \end{equation}
    Since $(J, 0, b)$ forms a valid instance of \RLL{$k$}, $L_{\mu} = L_{\mu}^{*}$, thus we can apply Lemma~\ref{lm:restrictdfscond} of Section~\ref{sec:dfs} implying
    \begin{equation}
        \D(\ket{\psi}\bra{\psi}) = 0
    \end{equation}
    as required. \\

    \ul{\textbf{Soundness:}} \\

    Let $(J, 0, b)$ be an \emph{(reject-case)} instance of \RLL{$k$} with $a=0$ and $J$ be the corresponding instance of \LLPS{$k$} generated from $(J, 0, b)$ according to the description above. For $J$ to be a \emph{(reject-case)} instance of \LLPS{$k$} we require for all $\ket{\psi}$ and for all steady-states $\rho_{ss}$ 
    \begin{equation}
        D(\ket{\psi}\bra{\psi}, \rho_{ss}) \geq 1 / \poly(n)
    \end{equation}
    Since $(J, 0, b)$ is a \emph{(reject-case)} instance of \RLL{$k$} we have for all $\ket{\psi}$,
    \begin{equation}
        |\braket{\ket{\psi}\bra{\psi}, \re\D(\ket{\psi}\bra{\psi})}| \geq b,
    \end{equation}
    where $\D$ denotes the Lindbladian dissipator generated by jump operators $J$.
    Applying our distance to decay Lemma, Lemma~\ref{lm:dtod}, we find for all $\ket{\psi}$ and for all steady-states $\rho_{ss}$,
    \begin{equation}
        \label{eq:tracedist}
        D(\ket{\psi}\bra{\psi}, \rho_{ss}) \geq  \frac{1}{2\lVert\D^{\dagger}\rVert_{\infty}} |\braket{\ket{\psi}\bra{\psi}, \re\D(\ket{\psi}\bra{\psi})}| \geq \frac{b}{2\lVert\D^{\dagger}\rVert_{\infty}}.
    \end{equation}

    Bounding $\lVert \D^{\dagger} \rVert_{\infty}$ through Lemma~\ref{lm:Lupperbound} we find,
    \begin{equation}
        \lVert \D^{\dagger} \rVert_{\infty} \leq  2\sum_{\mu=1}^{m}\norm{L_{\mu}}^{2}.
    \end{equation}
    Since $(J, 0, b)$ is a valid instance of \RLL{$k$}, we have $m=\poly(n)$ and for all $\mu \in [m]$, $\norm{L_{\mu}} = \poly(n)$, thus
    \begin{equation}
        \lVert \D^{\dagger} \rVert_{\infty} \leq \poly(n).
    \end{equation}
    In-addition, since $b - a \geq 1 /\poly(n)$ and $a=0$, we have 
    \begin{equation}
        b \geq 1 / \poly(n).
    \end{equation}
    Thus equation~\ref{eq:tracedist} becomes 
    \begin{equation}
        D(\ket{\psi}\bra{\psi}, \rho_{ss}) \geq 1 / \poly(n)
    \end{equation}
    as required.
\end{proof}

In-addition to the \LL{$k$} problem being connected to the \LLPS{$k$} problem, we can also connect it to a type of local Hamiltonian problem. More specifically, we show that we can reduce \LL{$k$} into a problem about finding product ground states of a local Hamiltonian that are symmetric, which we call the \SEPH{$k$} problem. From a complexity point of view, this reduction allows us to include the \LL{$k$} problem in $\mathrm{QMA}$ as the \SEPH{$k$} problem is a slight modification of the \textsc{Separable $k$-Local Hamiltonian} problem studied by Chaillox et al. in \cite{chailloux2012complexity}. Our definition of \SEPH{$k$} is provided below.
\begin{definition}[\SEPH{$k$} Problem]
    \label{def:SEPH}
    Let $H$ be a $2k$-local Hamiltonian acting on $2n$ qubits partitioned into $n$ qubit registers $A$ and $B$ with $H$ given by $\{H_{i}\}_{i=1}^{r=\poly(n)}$ where $\mathbb{I} \succeq H_{i} \succeq 0$ and  if $H_{i}$ acts non-trivially on $k'$ qubits in register $A$ then it also acts non-trivially on $k'$ qubits in register $B$, decide
    \begin{itemize}
        \item \emph{(accept-case)} there exists a pure quantum state $\ket{\psi} \in \mathcal{H}_{A}$ such that $(\bra{\psi}_{A} \otimes \bra{\bar{\psi}}_{B}) H(\ket{\psi}_{A} \otimes \ket{\bar{\psi}}_{B}) \leq a$,
        \item \emph{(reject-case)} For all pure quantum states $\ket{\psi} \in \mathcal{H}_{A}$ and $\ket{\phi} \in \mathcal{H}_{B}$, $(\bra{\psi}_{A} \otimes \bra{\phi}_{B}) H (\ket{\psi}_{A} \otimes \ket{\phi}_{B}) \geq b$
    \end{itemize}
    promised one of these to be the case, where $b-a \geq 1 /\poly(n)$.
\end{definition}
We define an instance of \SEPH{$k$} as the tuple $(F, a, b, A, B)$ where $F$ denotes the set of Hamiltonians which sum to $H$ and $a$, $b$, $A$, and $B$ are specified as in Definition~\ref{def:SEPH}.

From a physics point of view, the \SEPH{$k$} problem captures the difficulty of computing ground state energies of a local Hamiltonian under the additional constraint that the ground state must be both a product state and conjugate symmetric across a bi-partition of the system. We have the following proposition regarding the \SEPH{$k$} problem.
\begin{proposition}
    \label{prop:SEPHinQMA}
    For $k \leq \log(n)$, $\SEPH{$k$} \in \mathrm{QMA}$
\end{proposition}
\begin{proof}
    See Appendix~\ref{app:H}
\end{proof}
Intuitively the connection between \LL{$k$} and \SEPH{$k$} can be made by viewing the \LL{$k$} problem from the perspective of the operator-vector correspondence of $\mathcal{B}(\mathcal{H})$. In the doubled Hilbert space representation of $\mathcal{B}(\mathcal{H})$ given by the operator-vector correspondence, mixed states $\rho \in \mathcal{B}(\mathcal{H})$ become un-normalized entangled states while pure states $\ket{\psi}\bra{\psi} \in \mathcal{B}(\mathcal{H})$ become conjugate symmetric product states. In-addition, $\re\D$ becomes a simple bounded linear operator acting on the doubled Hilbert space. Thus, the problem of finding a pure state with a decay rate below some threshold becomes equivalent to finding a product state on the doubled Hilbert space with an ``energy" below the same threshold. This argument is formalized in our proof of the following Lemma.
\begin{proposition}
    \label{prop:LLtoSEPH}
    For $k \leq \log(n)$, $\LL{$k$} \leq \SEPH{$k$}$.
\end{proposition}
\begin{proof}
    Let $(J, a, b)$ be an instance of \LL{$k$} where $J=\{L_{\mu}\}_{\mu=1}^{m}$ and $\D$ the Lindbladian dissipator generated by jump operators $J$. In-addition let  $\D_{\mu}$ be the Lindbladian dissipator generated by a single jump operator $L_{\mu}$. From Lemma~\ref{lm:opmap}, we have
    \begin{equation}
        \mathrm{Op}(\re\D) = \sum_{\mu=1}^{m}\mathrm{Op}(\mathrm{Re}\D_{\mu}) = \sum_{\mu=1}^{m}\frac{1}{2}\left([L_{\mu}^{*}]_{A} \otimes [L_{\mu}]^{\intercal}_{B} + [L_{\mu}]_{A} \otimes [L_{\mu}^{* \intercal}]_{B} - [L^{*}_{\mu}L_{\mu}]_{A}\otimes \mathbb{I}_{B} - \mathbb{I}_{A} \otimes [L_{\mu}^{*}L_{\mu}]^{\intercal}_{B} \right),
    \end{equation}
    where each register $A$ and $B$ consist of $n$ qubits each. Define the unnormalized 
    \begin{equation}
        \label{eq:unnormH}
        \tilde{H}_{\mu} = -\frac{1}{2}\left([L_{\mu}^{*}]_{A} \otimes [L_{\mu}]^{\intercal}_{B} + [L_{\mu}]_{A} \otimes [L_{\mu}^{* \intercal}]_{B} - [L^{*}_{\mu}L_{\mu}]_{A}\otimes \mathbb{I}_{B} - \mathbb{I}_{A} \otimes [L_{\mu}^{*}L_{\mu}]^{\intercal}_{B} \right),
    \end{equation}
    and the normalized-shifted version of $\tilde{H}_{\mu}$ as
    \begin{equation}
        \label{eq:normshiftedH}
        H_{\mu} = \left(\tilde{H}_{\mu}  - \lambda_{0}(\tilde{H}_{\mu}  )\mathbb{I}\right) /\sup_{\mu \in [m]}\left(\lVert \tilde{H}_{\mu} \rVert + |\lambda_{0}(\tilde{H}_{\mu})|\right)
    \end{equation}
    where $\lambda_{0}(A) = \inf_{\lambda \in \Sp{A}} \lambda $. In-addition define the set $F = \{H_{\mu}\}_{\mu=1}^{m}$ and 
    \begin{equation}
        \begin{cases}
            a' = \left(a - \sum_{\mu=1}^{m}\lambda_{0}(\tilde{H}_{\mu}) \right) / \sup_{\mu \in [m]}(\lVert \tilde{H}_{\mu}\rVert + |\lambda_{0}(\tilde{H}_{\mu})|), \\
            b' = \left(b - \sum_{\mu=1}^{m}\lambda_{0}(\tilde{H}_{\mu}) \right) / \sup_{\mu \in [m]}(\lVert \tilde{H}_{\mu}\rVert + |\lambda_{0}(\tilde{H}_{\mu})|). \\
        \end{cases}
    \end{equation}
    Note, since each $L_{\mu}$ acts non-trivially on at most $k$ qubits, each $H_{\mu}$ acts non-trivially on at most $2k$ qubits, thus $H=\sum_{\mu=1}^{m}H_{\mu}$ is $2k$-local. Furthermore, equation~\eqref{eq:unnormH} implies each $H_{\mu}$ acts on non-trivially on the same number of qubits in register $A$ as it does in register $B$. Moreover, equation~\eqref{eq:normshiftedH} implies $0 \preceq H_{\mu} \preceq \mathbb{I}$. Thus $(F, a', b', A, B)$ forms a valid instance of \SEPH{$k$}. For $k \leq \log(n)$ we can compute each $\tilde{H}_{\mu}$ in $\poly(n)$ time using standard matrix manipulation algorithms since each $L_{\mu}$ is a $\poly(n) \times \poly(n)$ dimension matrix. In-addition since each $\tilde{H}_{\mu}$ is a $\poly(n) \times \poly(n)$ dimension matrix, computing $\lambda_{0}(\tilde{H}_{\mu})$ and $\lVert \tilde{H}_{\mu} \rVert$ can also be done in $\poly(n)$ time. Therefore, since $m  = \poly(n)$, we can compute $\sup_{\mu \in [m]}\lVert \tilde{H}_{\mu}\rVert$, $\sum_{\mu=1}^{m}\lambda_{0}(\tilde{H}_{\mu})$, and $\sup_{\mu \in [m]}|\lambda_{0}(\tilde{H}_{\mu})|$ in $\poly(n)$ time as well, implying we can compute $a'$, $b'$ and each $H_{\mu}$ in a total of $\poly(n)$ time. Since $|F| = m =\poly(n)$ and we can compute each $H_{\mu} \in F$ in $\poly(n)$ time, we also can compute the whole set $F$ in $\poly(n)$ time as well. Thus we can efficiently construct an instance of \SEPH{$k$} from an instance of \LL{$k$} for $k \leq \log(n)$. \\
    
    \ul{\textbf{Completeness:}} \\
    
    Let $(J, a, b)$ be an \emph{(accept-case)} instance of \LL{$k$} and $(F, a, b, A, B)$ be the instance of \SEPH{$k$} constructed according to the description provided above. We require $(F, a, b, A, B)$ to be an \emph{(accept-case)} instance of \SEPH{$k$} implying there exists pure quantum state $\ket{\psi}\in \mathcal{H}_{A}$ such that
    \begin{equation}
        (\bra{\psi}_{A} \otimes \bra{\bar{\psi}}_{B}) H (\ket{\psi}_{A} \otimes \ket{\bar{\psi}}_{B}) \leq a',
    \end{equation}
    where $H = \sum_{\mu=1}^{m} H_{\mu}$. 

    Since $(J, a, b)$ is an \emph{(accept-case)} instance of \LL{$k$} there exists an $n$ qubit state $\ket{\psi}$ such that 
    \begin{equation}
        |\braket{\ket{\psi}\bra{\psi}, \re\D(\ket{\psi}\bra{\psi})}| \leq a.
    \end{equation}
    Choose on our doubled Hilbert space the pure quantum state $\ket{\psi}_{A}\otimes \ket{\bar{\psi}}_{B}$ and thus we choose $\ket{\psi} \in \mathcal{H}_{A}$, we find
    \begin{align}
        (\bra{\psi}_{A} \otimes \bra{\bar{\psi}}_{B}) H (\ket{\psi}_{A} \otimes \ket{\bar{\psi}}_{B}) &= \left(-(\bra{\psi}_{A} \otimes \bra{\bar{\psi}}_{B}) \mathrm{Op}(\re\D) (\ket{\psi}_{A}\otimes\ket{\bar{\psi}}_{B}) - \sum_{\mu=1}^{m}\lambda_{0}(\tilde{H}_{\mu})\right) / \sup_{\mu\in[m]}\left(\lVert\tilde{H}_{\mu}\rVert + |\lambda_{0}(\tilde{H}_{\mu})|\right) \nonumber \\
                                                                                                      &= \left(- \braket{\mathrm{vec}^{-1}(\ket{\psi}_{A} \otimes \ket{\bar{\psi}}_{B}), \re\D(\mathrm{vec}^{-1}(\ket{\psi}_{A} \otimes \ket{\bar{\psi}}_{B}))} - \sum_{\mu=1}^{m}\lambda_{0}(\tilde{H}_{\mu})\right) / \sup_{\mu\in[m]}\left(\lVert\tilde{H}_{\mu}\rVert + |\lambda_{0}(\tilde{H}_{\mu})|\right) \nonumber \\
                                                                                                      &= \left(- \braket{\ket{\psi}\bra{\psi}, \re\D(\ket{\psi}\bra{\psi})} - \sum_{\mu=1}^{m}\lambda_{0}(\tilde{H}_{\mu})\right) / \sup_{\mu\in[m]}\left(\lVert\tilde{H}_{\mu}\rVert + |\lambda_{0}(\tilde{H}_{\mu})| \right) \nonumber \\
                                                                                                      &= \left(|\braket{\ket{\psi}\bra{\psi}, \re\D(\ket{\psi}\bra{\psi})}| - \sum_{\mu=1}^{m}\lambda_{0}(\tilde{H}_{\mu})\right) / \sup_{\mu\in[m]}\left(\lVert\tilde{H}_{\mu}\rVert + |\lambda_{0}(\tilde{H}_{\mu})| \right) \nonumber \\
                                                                                                      &\leq \left(a - \sum_{\mu=1}^{m}\lambda_{0}(\tilde{H}_{\mu})\right) / \sup_{\mu\in[m]}\left(\lVert \tilde{H}_{\mu}\rVert + |\lambda_{0}(\tilde{H}_{\mu}) |\right)
    \end{align}
    where the last equality follows from Lemma~\ref{lm:repartL}. Thus 
    \begin{equation}
        (\bra{\psi}_{A} \otimes \bra{\bar{\psi}}_{B}) H (\ket{\psi}_{A} \otimes \ket{\bar{\psi}}_{B}) \leq \left(a - \sum_{\mu=1}^{m}\lambda_{0}(\tilde{H}_{\mu})\right) / \sup_{\mu\in[m]}\left(\lVert \tilde{H}_{\mu}\rVert + |\lambda_{0}(\tilde{H}_{\mu})| \right) = a'
    \end{equation}
    as required. \\

    \ul{\textbf{Soundness:}} \\

    Let $(J, a, b)$ be a \emph{(reject-case)} instance of \LL{$k$} and $(F, a, b, A, B)$ be the instance of \SEPH{$k$} constructed according to the description provided above. We require $(F, a, b, A, B)$ to be an \emph{(reject-case)} instance of \SEPH{$k$} implying for all pure quantum states $\ket{\psi} \in \mathcal{H}_{A}$ and $\ket{\phi} \in \mathcal{H}_{B}$,
    \begin{equation}
        (\bra{\psi}_{A} \otimes \bra{\phi}_{B}) H (\ket{\psi}_{A} \otimes \ket{\phi}_{B}) \geq b',
    \end{equation}
    where $H = \sum_{\mu=1}^{m}H_{\mu}$.

    Let $\ket{\psi} \in \mathcal{H}_{A}$ and $\ket{\phi}\in \mathcal{H}_{B}$ be arbitrary pure quantum states, calculating $(\bra{\psi}_{A} \otimes \bra{\phi}_{B}) H(\ket{\psi}_{A} \otimes \ket{\phi}_{B})$ we find,
    \begin{align}
        \label{eq:hamtodecay}
        (\bra{\psi}_{A} \otimes \bra{\phi}_{B}) H (\ket{\psi}_{A} \otimes \ket{\phi}_{B}) &= \left(-(\bra{\psi}_{A} \otimes \bra{\phi}_{B}) \mathrm{Op}(\re\D) (\ket{\psi}_{A}\otimes\ket{\phi}_{B}) - \sum_{\mu=1}^{m}\lambda_{0}(\tilde{H}_{\mu})\right) / \sup_{\mu\in[m]}\left(\lVert\tilde{H}_{\mu}\rVert + |\lambda_{0}(\tilde{H}_{\mu})|\right) \nonumber \\
                                                                                                      &= \left(- \braket{\mathrm{vec}^{-1}(\ket{\psi}_{A} \otimes \ket{\phi}_{B}), \re\D(\mathrm{vec}^{-1}(\ket{\psi}_{A} \otimes \ket{\phi}_{B}))} - \sum_{\mu=1}^{m}\lambda_{0}(\tilde{H}_{\mu})\right) / \sup_{\mu\in[m]}\left(\lVert\tilde{H}_{\mu}\rVert + |\lambda_{0}(\tilde{H}_{\mu})|\right) \nonumber \\
                                                                                                      &= \left(- \braket{\ket{\psi}\bra{\bar{\phi}}, \re\D(\ket{\psi}\bra{\bar{\phi}})} - \sum_{\mu=1}^{m}\lambda_{0}(\tilde{H}_{\mu})\right) / \sup_{\mu\in[m]}\left(\lVert\tilde{H}_{\mu}\rVert + |\lambda_{0}(\tilde{H}_{\mu})| \right)
    \end{align}
    Evaluating $\braket{\ket{\psi}\bra{\bar{\phi}}, \re\D(\ket{\psi}\bra{\bar{\phi}})}$ we find,
    \begin{align}
        \label{eq:uuuper}
        \braket{\ket{\psi}\bra{\bar{\phi}}, \re\D(\ket{\psi}\bra{\bar{\phi}})} &= \frac{1}{2}\sum_{\mu=1}^{m}\left(\braket{\ket{\psi}\bra{\bar{\phi}}, \D_{\mu}(\ket{\psi}\bra{\bar{\phi}})} + \braket{\ket{\psi}\bra{\bar{\phi}}, \D_{\mu}^{\dagger}(\ket{\psi}\bra{\bar{\phi}})}\right) \nonumber \\
                                                                               &= \frac{1}{2}\sum_{\mu=1}^{m}\left(\bra{\psi}L_{\mu}\ket{\psi}\bra{\bar{\phi}}L_{\mu}^{*}\ket{\bar{\phi}} - \bra{\psi}L_{\mu}^{*}L_{\mu}\ket{\psi} - \bra{\bar{\phi}}L_{\mu}^{*}L_{\mu}\ket{\bar{\phi}} + \bra{\psi}L_{\mu}^{*}\ket{\psi}\bra{\bar{\phi}}L_{\mu}\ket{\bar{\phi}}\right) \nonumber \\
                                                                               &= \sum_{\mu=1}^{m}\left(\re(\bra{\psi}L_{\mu}\ket{\psi}\bra{\bar{\phi}}L_{\mu}^{*}\ket{\bar{\phi}}) - \frac{1}{2}\bra{\psi}L_{\mu}^{*}L_{\mu}\ket{\psi} - \frac{1}{2}\bra{\bar{\phi}}L_{\mu}^{*}L_{\mu}\ket{\bar{\phi}}\right) \nonumber \\
                                                                               &= \sum_{\mu=1}^{m}\Big(- \frac{1}{2}(|\bra{\psi}L_{\mu}\ket{\psi}|^{2} -2 \re(\bra{\psi}L_{\mu}\ket{\psi}\bra{\bar{\phi}}L_{\mu}^{*}\ket{\bar{\phi}}) + |\bra{\bar{\phi}}L_{\mu}\ket{\bar{\phi}}|^{2} )\nonumber \\
                                                                               &+ \frac{1}{2}(|\bra{\psi}L_{\mu}\ket{\psi}|^{2} - \bra{\psi}L_{\mu}^{*}L_{\mu}\ket{\psi}) + \frac{1}{2}(|\bra{\bar{\phi}}L_{\mu}\ket{\bar{\phi}}|^{2}) - \bra{\bar{\phi}}L_{\mu}^{*}L_{\mu}\ket{\bar{\phi}})\Big) \nonumber \\
                                                                               &= \frac{1}{2}\sum_{\mu=1}^{m}\Big(\braket{\ket{\psi}\bra{\psi}, \re\D_{\mu}(\ket{\psi}\bra{\psi})} + \braket{\ket{\bar{\phi}}\bra{\bar{\phi}}, \re\D_{\mu}(\ket{\bar{\phi}}\bra{\bar{\phi}})} -|\bra{\psi}L_{\mu}\ket{\psi} - \bra{\bar{\phi}}L_{\mu}\ket{\bar{\phi}}|^{2}\Big) \nonumber \\
                                                                               &= \frac{1}{2}\braket{\ket{\psi}\bra{\psi}, \re\D(\ket{\psi}\bra{\psi})} + \frac{1}{2}\braket{\ket{\bar{\phi}}\bra{\bar{\phi}}, \re\D(\ket{\bar{\phi}}\bra{\bar{\phi}})} - \frac{1}{2}\sum_{\mu=1}^{m}|\bra{\psi}L_{\mu}\ket{\psi} - \bra{\bar{\phi}}L_{\mu}\ket{\bar{\phi}}|^{2} \nonumber \\
                                                                               &\leq \frac{1}{2}\braket{\ket{\psi}\bra{\psi}, \re\D(\ket{\psi}\bra{\psi})} + \frac{1}{2}\braket{\ket{\bar{\phi}}\bra{\bar{\phi}}, \re\D(\ket{\bar{\phi}}\bra{\bar{\phi}})},
    \end{align}
    which implies
    \begin{equation}
        \label{eq:negation}
        -\braket{\ket{\psi}\bra{\bar{\phi}}, \re\D(\ket{\psi}\bra{\bar{\phi}})} \geq -\frac{1}{2}\braket{\ket{\psi}\bra{\psi}, \re\D(\ket{\psi}\bra{\psi})} - \frac{1}{2}\braket{\ket{\bar{\phi}}\bra{\bar{\phi}}, \re\D(\ket{\bar{\phi}}\bra{\bar{\phi}})}.
    \end{equation}
    Since $(J, a, b)$ is a \emph{(reject-case)} instance of \LL{$k$} we have for all $\ket{\chi}$,
    \begin{equation}
        |\braket{\ket{\chi}\bra{\chi}, \re\D(\ket{\chi}\bra{\chi})}| \geq b
    \end{equation}
    which implies 
    \begin{equation}
        \label{eq:nonabs}
        -\braket{\ket{\chi}\bra{\chi}, \re\D(\ket{\chi}\bra{\chi})} \geq b.
    \end{equation}
    because of Lemma~\ref{lm:repartL}. Thus we find using equation~\eqref{eq:negation}
    \begin{equation}
        \label{eq:bbb}
        -\braket{\ket{\psi}\bra{\bar{\phi}}, \re\D(\ket{\psi}\bra{\bar{\phi}})} \geq b
    \end{equation}
    Therefore, combining equation~\eqref{eq:bbb} with equation~\eqref{eq:hamtodecay} gives us
    \begin{equation}
        (\bra{\psi}_{A} \otimes \bra{\phi}_{B}) H(\ket{\psi}_{A} \otimes \ket{\phi}_{B}) \geq \left(b - \sum_{\mu=1}^{m}\lambda_{0}(\tilde{H}_{\mu})\right) / \sup_{\mu\in[m]}\left(\lVert \tilde{H}_{\mu}\rVert + |\lambda_{0}(\tilde{H}_{\mu})|\right) = b' 
    \end{equation}
    as required.
\end{proof}

Proposition~\ref{prop:SEPHinQMA} implies there exists a $\mathrm{QMA}$ verification protocol for \SEPH{$k$}, the details of which can be found in Algorithm~\ref{alg:QMAver} of Appendix~\ref{app:H}. Consequently, Proposition~\ref{prop:LLtoSEPH} implies we can construct a simple two-step $\mathrm{QMA}$ verification protocol for \LL{$k$}. Succinctly described, we simply run the reduction described in the proof of Proposition~\ref{prop:LLtoSEPH} to convert our instance of \LL{$k$} into an instance of \SEPH{$k$}, then run the $\mathrm{QMA}$ verification protocol for \SEPH{$k$}. Thus we have the following corollary.
\begin{corollary}
    \label{cor:LLinQMA}
    For $k \leq \log(n)$, $\LL{$k$} \in \mathrm{QMA}$
\end{corollary}

\section{Hardness of the Local Lindbladian Problem}
\label{sec:hardness}

In order to show that the \LL{$k$} problem is $\mathrm{QMA}$-Hard we must be able to reduce a problem that is already $\mathrm{QMA}$-Hard problem to \LL{$k$}. Given the structure of the problem it may seem natural to choose a $\mathrm{QMA}$-Hard local Hamiltonian problem, such as \textsc{Separable $k$-Local Hamiltonian} \cite{chailloux2012complexity}, to transform into our Lindbladian problem. Taking such an approach, we can map $\re\mathcal{L}$ into a Hamiltonian acting on a doubled Hilbert space,
 but the symmetry that Lindbladians require in-addition to the pure state constraint on the problem make the reduction not clear. Instead we choose to generalize the approach taken by Kitaev in \cite{kitaev1999quantum, kitaev2002classical} to show $\mathrm{QMA}$-Hardness for the local Hamiltonian problem through reducing from \textsc{Quantum Circuit-SAT} defined below.
\begin{definition}[\textsc{Quantum Circuit-SAT} Problem]
    \label{def:QCSAT}
    Let $V=U_{T-1}U_{T-2}\cdots U_{0}$ be a quantum circuit acting on a proof register $M$ of $n$ qubits and an ancillary register $A$ of $k=\poly(n)$ qubits, decide
    \begin{itemize}
        \item \emph{(accept-case)} there exists an $n$ qubit state $\ket{\psi}$ such that $\Tr\left(\ket{1}\bra{1}_{M_{0}}V\ket{\psi}\bra{\psi}_{M} \otimes \ket{0}\bra{0}_{A}V^{*}\right) \geq b$
        \item \emph{(reject-case)} for all $n$ qubit states $\ket{\psi}$, $\Tr\left(\ket{1}\bra{1}_{M_{0}}V\ket{\psi}\bra{\psi}_{M} \otimes \ket{0}\bra{0}_{A}V^{*}\right) \leq a$
    \end{itemize}
    promised one of these to be the case, where $b-a \geq 1 / \poly(n)$.
\end{definition}
We define an \emph{instance of} \textsc{Quantum Circuit-SAT} as a tuple $(V, b, a)$ where $V$ is a quantum circuit and $b, a \in [0, 1]$ with $V$, $b$ and $a$ satisfying their respective conditions in Definition~\ref{def:QCSAT}. 

By Defintion~\ref{def:QMA} \textsc{Quantum Circuit-SAT} is $\mathrm{QMA}$-Complete. Thus, if we can efficiently transform an instance of \textsc{Quantum Circuit-SAT} into an instance of \LL{$k$} that maps the \emph{(accept-case)} and \emph{(reject-case)} of \textsc{Quantum Circuit-SAT} to their respective cases for \LL{$k$} then we have shown \LL{$k$} to be $\mathrm{QMA}$-Hard. Similar to Kitaev's Hamiltonian construction \cite{kitaev1999quantum, kitaev2002classical}, from the quantum circuit $V=U_{T-1}U_{T-2}\cdots U_{0}$, we encode the correct execution of the circuit into history states, quantum states on a larger Hilbert space encoding the full computation an input state undergoes under the quantum circuit $V=U_{T-1}U_{T-2}\cdots U_{0}$. We then engineer a Lindbladian dissipator that ``checks" the correct execution of the circuit by driving the system into a steady-state space consisting of only the history states associated to the quantum circuit. We then build a Lindbladian dissipator that preserves the purity of this steady-state space if the output of the circuit is that of an \emph{(accept-case)} of \textsc{Quantum Circuit-SAT}, mixing the steady-state space of the system if it is not.
\begin{figure}[ht!]
    \label{fig:hardnessconstruction}
    \includesvg[width=0.85\linewidth]{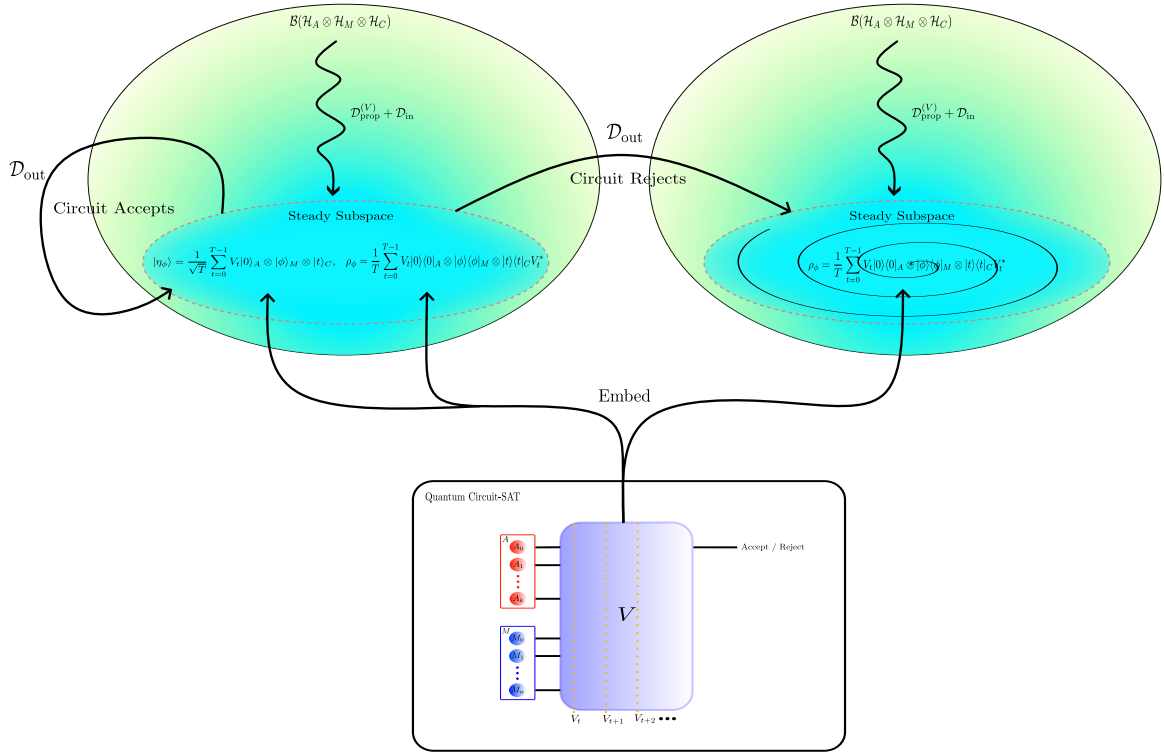}
    \caption{A visualization of our \textsc{Quantum Circuit-SAT} disispator construction.}
\end{figure}

We start by designing the dissipator which checks the correct propagation of the quantum circuit $V=U_{T-1}U_{T-2}\cdots U_{0}$, by driving the system into a steady-state subspace of history states. It is important to note that at this stage we are not sure whether the quantum circuit we are given is an \emph{(accept-case)} instance of \textsc{Quantum Circuit-SAT} or a \emph{(reject-case)} instance. Thus the steady-state subspace must contain both history states that are pure and ones that are mixed. Luckily we have notions of a history state that fit both of these constraints. We define a \emph{pure history state} associated to quantum circuit $V=U_{T-1}U_{T-2} \cdots U_{0}$ acting on Hilbert space $\mathcal{H}$ with input state $\ket{\phi} \in \mathcal{H}$ as the quantum state $\ket{\eta_{\phi}} \in \mathcal{H} \otimes \mathcal{H}_{C}$, where
\begin{equation}
    \label{eq:purehistorystate}
    \ket{\eta_{\phi}} = \frac{1}{\sqrt{T}}\sum_{t=0}^{T-1}V_{t}\ket{\phi}\otimes \ket{t}_{C} \text{ and } V_{t} = U_{t}U_{t-1}\cdots U_{0}.
\end{equation}
Similarly, we define a \emph{mixed history state} associated to quantum circuit $V=U_{T-1}U_{T-2}\cdots U_{0}$ acting on Hilbert space $\mathcal{H}$ with input state $\ket{\phi} \in \mathcal{H}$ as the quantum state $\rho_{\phi} \in \mathcal{B}(\mathcal{H}\otimes\mathcal{H}_{C})$ where 
\begin{equation}
    \label{eq:mixedhistorystate}
    \rho_{\phi} = \frac{1}{T}\sum_{t=0}^{T-1}V_{t}\ket{\phi}\bra{\phi}V_{t}^{*}\otimes \ket{t}\bra{t}_{C} \text{ and } V_{t} = U_{t}U_{t-1}\cdots U_{0}.
\end{equation}
With our history state definitions explicated, we can proceed towards designing the dissipator that ``checks" the correct propagation of a quantum circuit. More specifically, given a quantum circuit $V=U_{T-1}U_{T-2}\cdots U_{0}$, the corresponding dissipator $\D_{\prop}^{(V)}$ should drive the system into a steady-state subspace containing only the pure and mixed history states associated to $V$. To do this, we break our dissipator into terms $\D_{t}$ where each term $\D_{t}$ ``checks" the correct propagation of the gate $U_{t}$ in our circuit. Interestingly, to construct each term $\D_{t}$, we simply take the Hamiltonian term $H_{t}$ used in Kitaev's Hamiltonian construction \cite{kitaev2002classical}, and re-appropriate it as jump operator generating a Lindbladian dissipator. Formally, we define our propagation dissipator below.
\begin{definition}[Propagation Dissipator]
    \label{def:propdiss}
    Let $V = U_{T-1}U_{T-2}\cdots U_{0}$ be a quantum circuit acting on a register $S$. Define a $\lceil \log(T)\rceil$ qubit clock register $C$, the propagation dissipator associated to $V$ is defined as a Lindbladian dissipator $\D_{\prop}^{(V)}$ acting on $\mathcal{B}(\mathcal{H}_{S}\otimes \mathcal{H}_{C})$ where
    \begin{equation}
        \label{eq:DpropD}
        \mathcal{D}_{\mathrm{prop}}^{(V)}(\bullet) = \sum_{t=0}^{T-2}\D_{t}(\bullet) = \sum_{t=0}^{T-2}\left(H_{t}\bullet H_{t}^{*} - \frac{1}{2}\left\{H_{t}^{*}H_{t}, \bullet\right\}\right),
    \end{equation}
    where $H_{t} = \left([U_{t+1}]_{S}\otimes \ket{t+1}_{C} - \mathbb{I}_{S} \otimes \ket{t}_{C}\right)\left([U_{t+1}^{*}]_{S}\otimes \bra{t+1}_{C} - \mathbb{I}_{S} \otimes \bra{t}_{C}\right)$.
\end{definition}
To give credence to our discussion above, we find the following Propositions regarding our propagation dissipator.
\begin{proposition}[Kernel of Propagation Dissipator]
    \label{prop:propdisker}
    Let $V=U_{T-1}U_{T-2}\cdots U_{0}$ be a quantum circuit acting on register $S$ and $\D_{\prop}^{(V)}$ be its propagation dissipator defined according to equation~\eqref{eq:DpropD} of Definition~\ref{def:propdiss}, then
    \begin{equation}
        \ker{\D_{\prop}^{(V)}} = \mathcal{U}^{(V)}\left[\mathcal{B}(\mathcal{H}_{S}) \otimes \Span{\left\{\mathbb{I}_{C}, \ket{p}\bra{p}_{C}\right\}}\right]
    \end{equation}
    where $\ket{p} = \frac{1}{\sqrt{T}}\sum_{t=0}^{T-1}\ket{t}$ and $\U^{(V)}$ is the circuit rotation associated $V$ defined according to equation~\eqref{eq:circrot} of Appendix~\ref{app:A}
\end{proposition}
\begin{proof}
    See Appendix~\ref{app:A}
\end{proof}
\begin{proposition}[Gap of Propagation Dissipator]
    \label{prop:propdissgap}
    Let $V=U_{T-1}U_{T-2}\cdots U_{0}$ be a quantum circuit acting on register $S$ and $\D_{\prop}^{(V)}$ be its propagation dissipator defined according to equation~\eqref{eq:DpropD} of Definition~\ref{def:propdiss}, then
    \begin{equation}
        \gap{\D_{\prop}^{(V)}} \geq \frac{1}{8T^{4}}.
    \end{equation}
\end{proposition}
\begin{proof}
    See Appendix~\ref{app:A}
\end{proof}

Proposition~\ref{prop:propdisker} confirms our intuition that the steady-states of the propagation dissipator for a quantum circuit $V$ contain both the mixed and pure history states of the circuit. This follows from the fact that for pure history state $\ket{\eta_{\phi}}$ and mixed history state $\rho_{\phi}$, 
\begin{equation}
    \begin{cases}
        \U^{(V)}(\ket{\phi}\bra{\phi} \otimes \mathbb{I}_{C}) = \rho_{\phi}, \\
        \U^{(V)}(\ket{\phi}\bra{\phi} \otimes \ket{p}\bra{p}_{C}) = \ket{\eta_{\phi}}\bra{\eta_{\phi}},
    \end{cases}
\end{equation}
where $\mathcal{U}^{(V)}$ is defined according to equation~\eqref{eq:circrot} of Appendix~\ref{app:A}

The quantum circuits $V$ associated to instances of \textsc{Quantum Circuit-SAT} act on both a proof register $M$ of $n$ qubits and ancillary register $A$ of $k$ qubits, with the ancillary register initialized in the state $\ket{0}^{\otimes k}$. Thus, in-addition to the propagation dissipator associated to the circuit $V$, we also require dissipators which ``check" the initialization of the ancillary register as well as the output of quantum circuit. To ``check" the output of the circuit we require a dissipator that mixes the steady-state subspace of the propagation dissipator when the quantum circuit is associated to a \emph{(reject-case)} instance of \textsc{Quantum Circuit-SAT} and conversely, preserves the steady-state subspace of the propagation dissipator when circuit is associated to an \emph{(accept-case)} instance of the problem. If the 0\textsuperscript{th} qubit of the proof register $M_{0}$ is defined as the output qubit of \textsc{Quantum Circuit-SAT}, then the Lindbladian dissipator generated by the jump operator $\ket{0}\bra{0}_{M_{0}}\otimes \ket{T-1}\bra{T-1}_{C}$ realizes this described behavior. Similarly, we can ``check" the initialization of the $i$-th ancillary qubit $A_{i}$, using the Lindbladian dissipator generated by the jump operator $\ket{1}\bra{1}_{A_{i}} \otimes \ket{0}\bra{0}_{C}$. Thus we can define the full Lindbladian dissipator which embeds the \textsc{Quantum Circuit-SAT} problem below.
\begin{definition}[\textsc{Quantum Circuit-SAT} Dissipator]
    \label{def:QCSATdis}
    Let $(V, b,a)$ be an instance of \textsc{Quantum Circuit-SAT}, define
    \begin{equation}
        \label{eq:DQCSAT}
        \D^{(V, b, a)} = \Gamma_{\inp}\D_{\inp} + \Gamma_{\prop}\D_{\prop}^{(V)} + \Gamma_{\out}\D_{\out}
    \end{equation}
    where $\D_{\prop}^{(V)}$ is defined according to equation~\eqref{eq:DpropD} of Definition~\ref{def:propdiss},
    \begin{equation}
        \label{eq:Dinp}
        \D_{\inp}(\bullet) = \sum_{i=0}^{k-1}\left((\ket{1}\bra{1}_{A_{i}}\otimes\ket{0}\bra{0}_{C})\bullet (\ket{1}\bra{1}_{A_{i}} \otimes \ket{0}\bra{0}_{C}) - \frac{1}{2} \left\{\ket{1}\bra{1}_{A_{i}} \otimes \ket{0}\bra{0}_{C}, \bullet \right\}\right),
    \end{equation}
    and
    \begin{equation}
        \label{eq:Dout}
        \D_{\out}(\bullet) = (\ket{0}\bra{0}_{M_{0}} \otimes \ket{T-1}\bra{T-1}_{C}) \bullet (\ket{0}\bra{0}_{M_{0}} \otimes \ket{T-1}\bra{T-1}_{C}) - \frac{1}{2}\left\{\ket{0}\bra{0}_{M_{0}}\otimes \ket{T-1}\bra{T-1}_{C}, \bullet\right\}, 
    \end{equation}
    with dissipation rates
    \begin{equation}
        \label{eq:dissrates}
        \begin{cases}
            \Gamma_{\out} = T \\
            \Gamma_{\inp} = 32T^{4} + 8T^{2} \\
            \Gamma_{\prop} = \frac{1}{T-2}\left((32kT^{7} + 8kT^{5} + T^{4})^{2} + 1024kT^{9} + 256 k T^{7} + 32 T^{6}\right)
        \end{cases}.
    \end{equation}
\end{definition}

The main tool we leverage in our prove of hardness for the \LL{$k$} problem is a generalization of the projection lemma of Kemp et al. \cite{kempe2006thecomplexity} which we use when lower-bounding the decay a pure state undergoes when evolving under our \textsc{Quantum Circuit-SAT} dissipator. Similar to the original projection lemma, our pure state decay projection lemma permits us to cut out parts of $\mathcal{B}(\mathcal{H})$ that decay quickly under our Lindbladian dissipator. This allows us the luxury of focusing our analysis on a smaller subset of $\mathcal{B}(\mathcal{H})$, with the important addition that this smaller subset contains only pure quantum states.
\begin{lemma}[Pure State Decay Projection Lemma]
    \label{lm:psdecayproj}
    Let $\mathcal{L}_{1}$ and $\mathcal{L}_{2}$ be Hermitian Lindbladians acting on $\mathcal{B}(\mathcal{H}) = \ker{\mathcal{L}_{1}} \oplus (\ker{\mathcal{L}_{1}})^{\perp}$, define subspace $\mathcal{H}_{1} = \mathrm{span}(\{\ket{\psi} \in \mathcal{H} \; | \; \ket{\psi}\bra{\psi} \in \ker{\mathcal{L}_{1}} \})$ which denotes the span of all pure steady states of $\mathcal{L}_{1}$. If $\mathcal{H}_{1} \neq \emptyset$ and $(1 - \cos^{2}\theta)\gap{\mathcal{L}_{1}} \geq 2\norm{\mathcal{L}_{2}}_{\infty}$, then for all $\ket{\psi} \in \mathcal{H}$,
    \begin{equation}
        |\braket{\ket{\psi}\bra{\psi}, \mathrm{Re}(\mathcal{L}_{1} + \mathcal{L}_{2})(\ket{\psi}\bra{\psi})}| \geq \inf_{\ket{h} \in \mathcal{H}_{1}}|\braket{\ket{h}\bra{h}, \mathrm{Re}\mathcal{L}_{2}|_{\mathcal{B}(\mathcal{H}_{1})}(\ket{h}\bra{h})}| - \frac{2\norm{\mathcal{L}_{2}}_{\infty}^{2}}{(1 - \cos^{2}\theta)\gap{\mathcal{L}_{1}} - 2 \norm{\mathcal{L}_{2}}_{\infty}},
    \end{equation}
    where 
    \begin{equation}
        \label{eq:overlap}
        \cos^{2}\theta = \sup_{\substack{A_{0} \in \ker{\mathcal{L}_{1}}, A \in \mathsf{M}\\ \braket{A_{0},A_{0}} = \braket{A , A } = 1}} |\braket{A_{0},A}|^{2} 
    \end{equation}
    denotes the overlap between $\ker{\mathcal{L}_{1}}$ and
    \begin{equation}
        \mathsf{M} = \left\{ v_{1}\ket{h}\bra{\perp} + v_{2}\ket{\perp}\bra{h} + v_{3}\ket{\perp}\bra{\perp} \in \mathcal{B}(\mathcal{H}_{1})^{\perp} \; | \; (\ket{h}, \ket{\perp}, v) \in \mathcal{H}_{1} \times \mathcal{H}_{1}^{\perp} \times \mathbb{C}^{3} \right\}.
    \end{equation}
\end{lemma}
\begin{proof}
    See Appendix~\ref{app:E}
\end{proof}

The argument that the \LL{$k$} problem is $\mathrm{QMA}$-Hard roughly breaks up into two parts. In the first part, we show that our dissipator of Definition~\ref{def:QCSATdis} yields a reduction from \textsc{Quantum Circuit-SAT} to \LL{$\log(n)$}. In the second half, we add additional qubits into our clock register and encode the circuit's clock into unary. This allows us to replace the quasi-local jump operator terms with $k$-local ones. This modification of the jump operators forces the steady-subspace of the dissipator to allow for illegal clock register states. To circumvent this we add an additional dissipator which forces these operators to decay rapidly, letting us use our pure state decay projection lemma to cut-off this portion of $\mathcal{B}(\mathcal{H})$. Inside this subspace the $k$-local Lindbladian dissipator can be mapped to our quasi-local $\mathrm{QMA}$-Hard dissipator, allowing us to re-apply the analysis from the first part.

To show that the dissipator of Definition~\ref{def:QCSATdis} yields an appropriate reduction from \textsc{Quantum Circuit-SAT} to \LL{$\log(n)$} we must first show that the jump operators which generate the dissipator of Definition~\ref{def:QCSATdis} can be efficiently construct from an instance of \textsc{Quantum Circuit-SAT}, in-addition to forming an appropriate instance of \LL{$\log(n)$}. We then require showing the completeness and soundness of the reduction. In essence, given an \emph{(accept-case)} or \emph{(reject-case)} of \textsc{Quantum Circuit-SAT} the corresponding instance of \LL{$\log(n)$} constructed is the appropriate \emph{(accept-case)} or \emph{(reject-case)}. For completeness, we simply show that the pure history state associated to a high accept probability input state of \textsc{Quantum Circuit-SAT} has the appropriate decay under the dissipator associated to an \emph{(accept-case)} of \textsc{Quantum Circuit-SAT}. To show this we leverage the following lemma, which relates the decay of a pure history state with input state $\ket{\phi}$, to the probability that the \textsc{Quantum Circuit-SAT} circuit $V$ rejects this input state.
\begin{lemma}
    \label{lm:properdecay}
    Let $(V, b,a)$ be an instance of \textsc{Quantum Circuit-SAT} and $\D^{(V, b, a)}$ defined according to equation~\eqref{eq:DQCSAT} of Definition~\ref{def:QCSATdis}. Let $\ket{\eta_{\phi}}$ be a pure history state associated to $V$ with input state $\ket{\phi} \otimes \ket{0}_{A} \in \mathcal{H}_{M} \otimes \mathcal{H}_{A}$, then the decay of $\ket{\eta_{\phi}}\bra{\eta_{\phi}}$ under $\D^{(V, a, b)}$  is given by  
    \begin{align}
        |\braket{\ket{\eta_{\phi}}\bra{\eta_{\phi}}, \re\D^{(V, a, b)}(\ket{\eta_{\phi}}\bra{\eta_{\phi}})}| = \Big(1 - \frac{1}{T} &\Tr\left(\ket{0}\bra{0}_{M_{0}}V\ket{\phi}\bra{\phi}V^{*}\right)\Big) \Tr\left(\ket{0}\bra{0}_{M_{0}}V\ket{\phi}\bra{\phi}V^{*}\right).
    \end{align}
\end{lemma}
\begin{proof}
    See Appendix~\ref{app:F}
\end{proof}

For soundness we must argue when given a \emph{(reject-case)} instance of \textsc{Quantum Circuit-SAT}, that all pure states have a decay under the associated \textsc{Quantum Circuit-SAT} dissipator that is lower-bounded by some non-zero value, ie there does not exist pure steady-states. To argue this, let $\D^{(V, b, a)}$ be a \textsc{Quantum Circuit-SAT} dissipator acting on $\mathcal{B}(\mathcal{H})$ and consider the, increasing in size, family of subspaces $\mathcal{H}_{\qcirc} \subseteq \mathcal{H}_{\prop} \subseteq \mathcal{H}$. Roughly speaking, $\mathcal{H}_{\prop}$ is the subspace of $\mathcal{H}$ consisting of pure history states associated to the quantum circuit $V$. Similarly, $\mathcal{H}_{\qcirc}$ corresponds to the subspace of $\mathcal{H}_{\prop}$ where the ancillary register $A$ of the quantum circuit $V$ is initialized in $\ket{0}_{A}$. In essence, the pure history states in $\mathcal{H}_{\qcirc}$ have their input states properly initialized. Our argument proceeds by using the pure state decay projection lemma to ``project" us through this family of pure state subspaces until we reach $\mathcal{H}_{\mathrm{circ}}$. Inside $\mathcal{H}_{\mathrm{circ}}$, we can calculate the decay by leveraging the fact that $\D^{(V, b, a)}$ is associated to a \emph{(reject-case)} instance of \textsc{Quantum Circuit-SAT}. 

Before proceeding it will be useful to provide formal definitions of our pure state subspaces $\mathcal{H}_{\prop}$ and $\mathcal{H}_{\qcirc}$. Furthermore, it will be necessary to prove that our definitions do indeed realize the intuitive descriptions of $\mathcal{H}_{\prop}$ and $\mathcal{H}_{\qcirc}$ provided above. Below we provide our formal definition of the history state subspace $\mathcal{H}_{\prop}$.
\begin{definition}[History State Subspace]
    \label{def:histsubspace}
    Let $(V, b, a)$ be an instance of \textsc{Quantum Circuit-SAT} and $\D_{\prop}^{(V)}$ acting on $\mathcal{B}(\mathcal{H})$ be defined according to equation~\eqref{eq:DpropD} of Definition~\ref{def:propdiss}, then we define
    \begin{equation}
        \label{eq:Hprop}
        \mathcal{H}_{\prop} = \Span{\{\ket{\psi} \in \mathcal{H} \; | \; \ket{\psi}\bra{\psi} \in \ker{\mathcal{D}^{(V)}_{\prop}}\}}.
    \end{equation}
\end{definition}
In addition, our formal definition of the history state subspace where the ancillary register of the input state is properly initialized, $\mathcal{H}_{\qcirc}$ is given below.
\begin{definition}[Circuit History State Subspace]
    \label{def:circsubspace}
    Let $\D_{\inp}$ be defined according to equation~\eqref{eq:Dinp} of Definition~\ref{def:QCSATdis} and $\mathcal{H}_{\prop}$ be defined according to equation~\eqref{eq:Hprop} of Definition~\ref{def:histsubspace}, then we define
    \begin{equation}
        \label{eq:Hqcirc}
        \mathcal{H}_{\qcirc} = \Span{\{\ket{\psi} \in \mathcal{H}_{\prop} \; | \; \ket{\psi}\bra{\psi} \in \ker{\D_{\inp}}\}}.
    \end{equation}
\end{definition}
In Appendix~\ref{app:G} we show following useful statements regarding our history state and circuit history state subspaces.
\begin{lemma}
    \label{lm:onlyhist}
    Let $(V, b, a)$ be an instance of \textsc{Quantum Circuit-SAT} where $V=U_{T-1}U_{T-2}\cdots U_{0}$ with $T > 2$. Define subspace $\mathcal{H}_{\prop}$ according to equation~\eqref{eq:Hprop} of Definition~\ref{def:histsubspace} using $V$, then
    \begin{equation}
        \mathcal{H}_{\prop} = U^{(V)}[\mathcal{H}_{MA} \otimes \ket{p}_{C}]
    \end{equation}
    where $\ket{p} = \frac{1}{\sqrt{T}}\sum_{t=0}^{T-1}\ket{t}$ and $U^{(V)}$ is given in equation~\eqref{eq:circrot} of Appendix~\ref{app:A}.
\end{lemma}
\begin{proof}
    See Appendix~\ref{app:G}
\end{proof}
\begin{lemma}
    \label{lm:onlycirc}
    Let $(V, b, a)$ be an instance of \textsc{Quantum Circuit-SAT} where $V=U_{T-1}U_{T-2}\cdots U_{0}$ with $T > 2$. Define $\mathcal{H}_{\qcirc}$ according to equation~\eqref{eq:Hqcirc} of Definition~\ref{def:circsubspace} using $V$, then
    \begin{equation}
        \mathcal{H}_{\qcirc} = U^{(V)}[\mathcal{H}_{M} \otimes \ket{0}_{A} \otimes \ket{p}_{C}]
    \end{equation}
    where $\ket{p} = \frac{1}{\sqrt{T}}\sum_{t=0}^{T-1}\ket{t}$ and $U^{(V)}$ is given in equation~\eqref{eq:circrot} of Appendix~\ref{app:A}.
\end{lemma}
\begin{proof}
    See Appendix~\ref{app:G}
\end{proof}
Lemma~\ref{lm:onlyhist} directly confirms our intuitive description of $\mathcal{H}_{\prop}$. In essence $\mathcal{H}_{\prop}$ consists of only pure history states associated to the quantum circuit $V$ which comes from our instance of \textsc{Quantum Circuit-SAT}. Essentially, all pure states that can be generated through combining both mixed and pure history states can only be pure history states. In addition, Lemma~\ref{lm:onlycirc} shows that $\mathcal{H}_{\qcirc}$ consists of only pure history states that have a properly initialized ancillary register $A$.  

In order to use our pure state decay projection lemma, we require lower-bounds on the gaps of the terms that sum to form our \textsc{Quantum Circuit-SAT} dissipator as well as a characterization of their kernels. Since we have already characterized our propagation dissipator, through Propositions~\ref{prop:propdisker} and~\ref{prop:propdissgap} we only need to characterize $\D_{\inp}$ projected into the appropriate subspace which we provide in the statements below.
\begin{lemma}
    \label{lm:DinpKer}
    Let $(V, b, a)$ be an instance of \textsc{Quantum Circuit-SAT} where $V=U_{T-1}U_{T-2}\cdots U_{0}$ with $T > 2$ and define $\mathcal{H}_{\prop}$ according to equation~\eqref{eq:Hprop} of Definition~\ref{def:histsubspace} using $V$ as well as $\D_{\inp}$ according to equation~\ref{eq:Dinp} of Definition~\ref{def:QCSATdis}, then
    \begin{equation}
        \ker{\D_{\inp}|_{\mathcal{B}(\mathcal{H}_{\prop})}} = \U^{(V)}[\mathcal{B}(\mathcal{H}_{M}) \otimes \ket{0}\bra{0}_{A} \otimes \ket{p}\bra{p}_{C})].
    \end{equation}
    where $\ket{p} = \frac{1}{\sqrt{T}}\sum_{t=0}^{T-1}\ket{t}$ and $\U^{(V)}$ the circuit rotation associated to $V$, is defined according to equation~\eqref{eq:circrot} of Appendix~\ref{app:A}.
\end{lemma}
\begin{proof}
    See Appendix~\ref{app:C}
\end{proof}
\begin{lemma}
    \label{lm:DinpGap}
    Let $(V, b, a)$ be an instance of \textsc{Quantum Circuit-SAT} where $V=U_{T-1}U_{T-2}\cdots U_{0}$ with $T > 2$ and define $\mathcal{H}_{\prop}$ according to equation~\eqref{eq:Hprop} of Definition~\ref{def:histsubspace} using $V$ as well as $\D_{\inp}$ according to equation~\ref{eq:Dinp} of Definition~\ref{def:QCSATdis}, then
    \begin{equation}
        \gap{\D_{\inp}|_{\mathcal{B}(\mathcal{H}_{\prop})}} \geq \frac{1}{2T}.
    \end{equation}
\end{lemma}
\begin{proof}
    See Appendix~\ref{app:C}
\end{proof}
Additionally, in order to use our pure state decay projection lemma, we also need upper-bounds on the overlap between the kernels of our dissipators which sum to the \textsc{Quantum Circuit-SAT} dissipator and set $\mathsf{M}$ which captures the directions in a pure state orthogonal to $\mathcal{B}(\mathcal{H}_{1})$. Since $\mathcal{H}_{1}$ functions like a decoherence free subspace of the dissipator, roughly speaking, what we are trying to control, is the overlap between the orthogonal complement of the dissipators decohernece free subspace and its kernel. This leakage out of $\mathsf{M}$ into the kernel inhibits the strength of the dissipator's gap. This could potentially reduce the strength of the dissipator's gap to be arbitrarily small, thus preventing our ability to control the perturbation to the pure state decay. In the following statements we bound how bad this leakage can be for each of the dissipators involved in our \textsc{Quantum Circuit-SAT} dissipator.
\begin{lemma}
    \label{lm:DpropLeakage}
    Let $V=U_{T-1}U_{T-2}\cdots U_{0}$ be a quantum circuit acting on register $S$ with $T > 2$ and $\mathcal{D}^{(V)}_{\prop}$ be defined according to equation~\eqref{eq:DpropD} of Definition~\ref{def:propdiss}. Define $\mathcal{H}_{\prop}$ according to equation~\eqref{eq:Hprop} of Definition~\ref{def:histsubspace} and 
    \begin{equation}
        \mathsf{M} = \left\{v_{1}\ket{h}\bra{\perp} + v_{2}\ket{\perp}\bra{h} + v_{3}\ket{h}\bra{\perp} \; | \; (\ket{h}, \ket{\perp}, v) \in \mathcal{H}_{\prop} \times \mathcal{H}_{\prop}^{\perp} \times \mathbb{C}^{3} \right\},
    \end{equation}
    then the overlap between $\ker{\D_{\prop}^{(V)}}$ and $\mathsf{M}$ is given by
    \begin{equation}
        \cos^{2}\theta = \sup_{\substack{A_{0} \in \ker{\D_{\prop}^{(V)}}, A \in \mathsf{M}\\ \braket{A_{0},A_{0}} = \braket{A , A } = 1}} |\braket{A_{0},A}|^{2} \leq \frac{1}{T-1}.
    \end{equation}
\end{lemma}
\begin{proof}
    See Appendix~\ref{app:A}
\end{proof}
\begin{lemma}
    \label{lm:DinpLeakage}
    Let $(V, b, a)$ be an instance of \textsc{Quantum Circuit-SAT} where $V=U_{T-1}U_{T-2}\cdots U_{0}$ with $T > 2$ and define $\mathcal{H}_{\prop}$ according to equation~\eqref{eq:Hprop} of Definition~\ref{def:histsubspace} as well as $\D_{\inp}$ according to equation~\ref{eq:Dinp} of Definition~\ref{def:QCSATdis} and $\mathcal{H}_{\qcirc}$ according to equation~\eqref{eq:Hqcirc} of Definition~\ref{def:circsubspace}. In addition define
    \begin{equation}
        \mathsf{M} = \left\{v_{1}\ket{h}\bra{\perp} + v_{2}\ket{\perp}\bra{h} + v_{3}\ket{\perp}\bra{\perp} \; | \; (\ket{h}, \ket{\perp}, v) \in \mathcal{H}_{\qcirc} \times \mathcal{H}_{\qcirc}^{\perp} \times \mathbb{C}^{3} \right\},
    \end{equation}
    then the overlap between $\ker{\D_{\inp}|_{\mathcal{B}(\mathcal{H}_{\prop})}}$ and $\mathsf{M}$ is given by 
    \begin{equation}
        \cos^{2}\theta = \sup_{\substack{A_{0} \in \ker{\D_{\inp}|_{\mathcal{B}(\mathcal{H}_{\prop})}}, A \in \mathsf{M}\\ \braket{A_{0},A_{0}} = \braket{A , A } = 1}} |\braket{A_{0},A}|^{2} =0 
    \end{equation}
\end{lemma}
\begin{proof}
    See Appendix~\ref{app:C}
\end{proof}

With the rest of our dissipators involved in the \textsc{Quantum Circuit-SAT} dissipator characterized we proceed to proving our first hardness result.
\begin{theorem}
    \label{thm:quasihardness}
    \LL{$\log(n)$} is $\mathrm{QMA}$-$\mathrm{Hard}$.
\end{theorem}
\begin{proof}
    To show \LL{$\log(n)$} is $\mathrm{QMA}$-$\mathrm{Hard}$ we show there exists a reduction from \textsc{Quantum Circuit-SAT} to \LL{$\log(n)$}.


    Let $(Q, b, a)$ be an instance of \textsc{Quantum Circuit-SAT}. Our reduction is two steps. From $(Q, b, a)$ we construct an equivalent \textsc{Quantum Circuit-SAT} instance $(V,b,a)$, then we construct our \LL{$\log(n)$} instance $(J, a', b')$. To do this note, we can always transform $Q=Q_{T'-1}Q_{T-2}\cdots Q_{0}$ into an equivalent circuit $V = U_{T-1}U_{T-2}\cdots U_{0}$ that outputs the same quantum state at the end of the computation by simply adding identity gates into the circuit $Q$. Thus we transform $(Q, b, a)$ into an equivalent instance $(V, b, a)$ by adding identity gates into the quantum circuit $Q = Q_{T'-1}Q_{T'-2}'\cdots Q_{0}'$ to produce quantum circuit $V=U_{T-1}U_{T-2}\cdots U_{0}$ where
    \begin{equation}
        \label{eq:T}
        T = T' + \frac{2-a}{b-a} + 2
    \end{equation}
    From $(V, b, a)$, let $\D^{(V, b,a)}$ be defined according to equation~\eqref{eq:DQCSAT} of Definition~\ref{def:QCSATdis}. We construct our instance of \LL{$\log(n)$} by letting $J$ be a set of jump operators generating $\D^{(V, b, a)}$ and setting 
    \begin{equation}
        \label{eq:newab}
        \begin{cases}
            a' = 1 - b,\\
            b' = (1 - \frac{1}{T})(1-a) - \frac{1}{T}.
        \end{cases}
    \end{equation}
    Note, we require $b'-a' \geq 1/\poly(n)$ which we confirm below,
    \begin{equation}
        b' - a' = \frac{(b - a)T  - (2 - a)}{T} = \frac{(b-a)(T'+2)}{T' + \frac{2-a}{b-a} + 2} \geq \frac{1}{\poly(n)},
    \end{equation}
    where the inequality follows from $T' = \poly(n)$ and $0 \leq a \leq 1$. Therefore we only require the set of jump operators $J$ to satisfy the conditions of the problem in-order for $(J, a' , b')$ to be a valid instance of \LL{$\log(n)$}.

    According to equation~\eqref{eq:DQCSAT} of Definition~\ref{def:QCSATdis}, $\D^{(V, b, a)} = \Gamma_{\inp} \D_{\inp} + \Gamma_{\prop} \D_{\prop}^{(V)} + \Gamma_{\out} \D_{\out}$ where $\D_{\inp}$ and $\D_{\out}$ are generated by jump operators 
    \begin{equation}
        \label{eq:JinpJout}
        \tilde{J}_{\inp} = \left\{\ket{1}\bra{1}_{A_{i}} \otimes \ket{0}\bra{0}_{C}\right\}_{i=1}^{k}, \text{ and } \tilde{J}_{\out} = \left\{\ket{0}\bra{0}_{M_{0}} \otimes \ket{T-1}\bra{T-1}_{C}\right\}
    \end{equation}
    respectively. According to Definition~\ref{def:propdiss} the propagation dissipator $\D_{\prop}^{(V)}$ is generated by jump operators 
    \begin{equation}
        \label{eq:Jprop}
        \tilde{J}_{\prop} = \left\{\frac{1}{\sqrt{2}}\left([U_{t+1}]_{MA}\otimes \ket{t+1}_{C} - \mathbb{I}_{MA} \otimes \ket{t}_{C}\right)\frac{1}{\sqrt{2}}\left([U_{t+1}^{*}]_{MA}\otimes \bra{t+1}_{C} - \mathbb{I}_{MA} \otimes \bra{t}_{C}\right)\right\}_{t=0}^{T-2}.
    \end{equation}
    Thus, re-scaling each of jump operators in the sets $J_{\inp}$, $J_{\out}$, and $J_{\prop}$ by their respective dissipation rates we define,
    \begin{equation}
        \label{eq:rescaledJs}
        \begin{cases}
            J_{\inp} = \{\sqrt{\Gamma_{\inp}}L \; | \; L \in \tilde{J}_{\inp}\}, \\
            J_{\out} = \{\sqrt{\Gamma_{\out}}L \; | \; L \in \tilde{J}_{\out}\, \\
            J_{\prop} = \{\sqrt{\Gamma_{\prop}}L \; | \; L \in \tilde{J}_{\prop}\}
        \end{cases}
    \end{equation}
    which generate $\Gamma_{\inp} \D_{\inp}$, $\Gamma_{\out} \D_{\out}$, and $\D^{(V)}_{\prop}$ respectively. If we set $J$ equal to the set of all jump operators in $J_{\inp}$, $J_{\out}$, and $J_{\prop}$, then $J$ generates $\D^{(V, b, a)}$. 

    Let $L \in J$, we have $\norm{L} \leq \poly(T, k)$, which follows from each of the jump operators in $\tilde{J}_{\inp}$, $\tilde{J}_{\out}$, and $\tilde{J}_{\prop}$ having an operator norm upper-bounded by a constant, and equation~\eqref{eq:dissrates} of Definition~\ref{def:QCSATdis}. Since $T' = \poly(n)$, equation~\eqref{eq:T} implies $T=\poly(n)$ too. In-addition $k=\poly(n)$, thus $\norm{L} \leq \poly(n)$ as required. Similarly, each $L$ is acts non-trivially on at most $\log(T)$ qubits given that each operator $\ket{t'}\bra{t}_{C}$ acts on $\log(T)$ many qubits. Since we have $T=\poly(n)$, we have that each $L$ acts non-trivially on at most $\log(n)$ qubits implying $J$ generates a $\log(n)$-local Lindbladian dissipator as required. Finally, since both $T=\poly(n)$ and $k=\poly(n)$, there are $\poly(n)$ many jump operators in the set $J$, thus $(J, a', b')$ forms a valid instance of \LL{$\log(n)$}.

       To complete the description of our reduction we require the algorithm that constructs $(J, a' ,b')$ from our instance of \textsc{Quantum Circuit-SAT} $(Q, b, a)$ to be efficient. Equations~\eqref{eq:T} and~\eqref{eq:newab} imply we can compute $a'$ and $b'$ in $\poly(n)$ time. In-addition, the transformation from $Q$ to $V$ requires only adding $(2 -a) / (b-a) = \poly(n)$ many identity gates into the description of the circuit $Q$ which can also be done in $\poly(n)$ time. Finally, note each of the jump operators in $J_{\inp}$ and $J_{\out}$ are matrices of dimension $\poly(n) \times \poly(n)$ and thus can be computed in $\poly(n)$ time from the dimension $\poly(n) \times \poly(n)$ matrix representations of $\ket{1}\bra{1}_{A_{i}}$, $\ket{0}\bra{0}_{M_{0}}$, and $\ket{t'}\bra{t}_{C}$ using standard matrix manipulation algorithms. A similar argument holds for each of the jump operators in $J_{\prop}$ since each of the unitary gates $U_{t}$ in the quantum circuit $V=U_{T-1}U_{T-2}\cdots U_{0}$ act on at most a constant number of qubits. Since there are $\poly(n)$ many jump operators in $J$, and each $L \in J$ can be computed in $\poly(n)$ time, the total runtime of the algorithm computing the set of jump operators $J$ from $(V, b, a)$ is $\poly(n)$ which is efficient. Thus there exists an efficient algorithm transforming our instance of \textsc{Quantum Circuit-SAT} to our instance of \LL{$\log(n)$}. \\

    \ul{\textbf{Completeness:}} \\

    Let $(Q, b, a)$ be an \emph{(accept-case)} instance of \textsc{Quantum Circuit-SAT}, and $(J, a', b')$ be the corresponding instance of \LL{$\log(n)$} generated according to the description above. The jump operators $J$ generate $\D^{(V, b, a)}$ defined in equation~\eqref{eq:DQCSAT} of Definition~\ref{def:QCSATdis} where $(V, b, a)$ is also an \emph{(accept-case)} instance of \textsc{Quantum Circuit-SAT} since both circuits $Q$ and $V$ output the same quantum state. For completeness, we require $(J, a', b')$ to be an \emph{(accept-case)} instance of \LL{$\log(n)$}, which requires there to exist a pure quantum state $\ket{\psi}$ such that 
    \begin{equation}
        |\braket{\ket{\psi}\bra{\psi}, \re\D^{(V, b, a)}(\ket{\psi}\bra{\psi})}| \leq a'.
    \end{equation}

    To show this, consider a pure history state $\ket{\eta_{\phi}}$ associated to the quantum circuit $V$. Choosing $\ket{\phi} = \ket{\psi}_{M} \otimes \ket{0}_{A}$ and applying Lemma~\ref{lm:properdecay} we find 
    \begin{align}
        |\braket{\ket{\eta_{\phi}}\bra{\eta_{\phi}}, \re\D^{(V, b, a)}(\ket{\eta_{\phi}}\bra{\eta_{\phi}})}| = \Big(1 - \frac{1}{T} &\Tr\left(\ket{0}\bra{0}_{M_{0}}V(\ket{\psi}\bra{\psi}_{M} \otimes \ket{0}\bra{0}_{A})V^{*}\right)\Big) \nonumber \\
                                                                                                                                               &\times\Tr\left(\ket{0}\bra{0}_{M_{0}}V(\ket{\psi}\bra{\psi}_{M} \otimes \ket{0}\bra{0}_{A}) V^{*}\right).
    \end{align}
    Since our $(V, b, a)$ forms an \emph{(accept-case)} instance of \textsc{Quantum Circuit-SAT} we can bound 
    \begin{equation}
        0 \leq \Tr\left(\ket{0}\bra{0}_{M_{0}}V(\ket{\psi}\bra{\psi}_{M} \otimes \ket{0}\bra{0}_{A})V^{*}\right) = 1 - \Tr\left(\ket{1}\bra{1}_{M_{0}}V(\ket{\psi}\bra{\psi}_{M} \otimes \ket{0}\bra{0}_{A})V^{*}\right) \leq 1 - b,
    \end{equation}
    implying
    \begin{equation}
        |\braket{\ket{\eta_{\phi}}\bra{\eta_{\phi}}, \re\D_{\QCSAT}^{(V, a, b)}(\ket{\eta_{\phi}}\bra{\eta_{\phi}})}| \leq 1 - b = a'
    \end{equation}
    as required. \\

    \ul{\textbf{Soundness:}} \\

    Let $(Q, b, a)$ be an \emph{(accept-case)} instance of \textsc{Quantum Circuit-SAT}, and $(J, a', b')$ be the corresponding instance of \LL{$\log(n)$} generated according to the description above. The jump operators $J$ generate $\D^{(V, b, a)}$ defined in equation~\eqref{eq:DQCSAT} of Definition~\ref{def:QCSATdis} where $(V, b, a)$ is also an \emph{(accept-case)} instance of \textsc{Quantum Circuit-SAT} since both circuits $Q$ and $V$ output the same quantum state. For soundness, we require $(J, a', b')$ to be an \emph{(reject-case)} instance of \LL{$\log(n)$} requiring for all pure quantum states $\ket{\psi}$, 
    \begin{equation}
        |\braket{\ket{\psi}\bra{\psi}, \re\D^{(V, b, a)}(\ket{\psi}\bra{\psi})}| \geq b'.
    \end{equation}

    To show this let $\D^{(V, b, a)} = \mathcal{L}_{1} + \mathcal{L}_{2}$, where $ \mathcal{L}_{2} = \Gamma_{\inp}\D_{\inp}  + \Gamma_{\out}\D_{\out}$ and $\mathcal{L}_{1} = \Gamma_{\prop}\D_{\prop}^{(V)}$. Since $\D^{(V, b, a)}$ is Hermitian along with $\mathcal{L}_{1}$ and $\mathcal{L}_{2}$, we can apply our pure state decay projection lemma, Lemma~\ref{lm:psdecayproj}, finding for all pure quantum states $\ket{\psi}$,
    \begin{equation}
        \label{eq:lwbd1}
        |\braket{\ket{\psi}\bra{\psi}, \re\D^{(V, b, a)}(\ket{\psi}\bra{\psi})}| \geq \inf_{\ket{h} \in \mathcal{H}_{1}}|\braket{\ket{h}\bra{h}, \mathrm{Re}\mathcal{L}_{2}|_{\mathcal{B}(\mathcal{H}_{1})}(\ket{h}\bra{h})}| - \frac{2\norm{\mathcal{L}_{2}}_{\infty}^{2}}{(1 - \cos^{2}\theta_{\prop})\gap{\mathcal{L}_{1}} - 2 \norm{\mathcal{L}_{2}}_{\infty}}.
    \end{equation}
    Since $\ker{\mathcal{L}_{1}} = \ker{\D_{\prop}^{(V)}}$, $\mathcal{H}_{1} = \mathcal{H}_{\prop}$ where $\mathcal{H}_{\prop}$ is defined according to equation~\eqref{eq:Hprop} of Definition~\ref{def:histsubspace}. Therefore equation~\eqref{eq:lwbd1} becomes for all pure quantum states $\ket{\psi}$,
    \begin{equation}
        \label{eq:lwbd2}
        |\braket{\ket{\psi}\bra{\psi}, \re\D^{(V, b, a)}(\ket{\psi}\bra{\psi})}| \geq \inf_{\ket{h} \in \mathcal{H}_{\prop}}|\braket{\ket{h}\bra{h}, \re(\Gamma_{\inp}\D_{\inp} + \Gamma_{\out}\D_{\out})|_{\mathcal{B}(\mathcal{H}_{\prop})}(\ket{h}\bra{h})}| - p_{1}
    \end{equation}
    where we have defined 
    \begin{equation}
        \label{eq:p1}
        p_{1} = \frac{2\norm{\Gamma_{\inp}\D_{\inp} + \Gamma_{\out}\D_{\out}}_{\infty}^{2}}{(1 - \cos^{2}\theta_{\prop})\Gamma_{\prop}\gap{\D_{\prop}^{(V)}} - 2 \norm{\Gamma_{\inp}\D_{\inp} + \Gamma_{\out}\D_{\out}}_{\infty}},
    \end{equation}
    Letting $\mathcal{L}_{1} = \Gamma_{\inp}\D_{\inp}|_{\mathcal{B}(\mathcal{H}_{\prop})}$ and $\mathcal{L}_{2} = \Gamma_{\out}\D_{\out}|_{\mathcal{B}(\mathcal{H}_{\prop})}$ and applying Lemma~\ref{lm:psdecayproj} again inside $\mathcal{B}(\mathcal{H}_{\prop})$ we find for all pure quantum states $\ket{\psi}$,
    \begin{equation}
        \label{eq:lwbd3}
        |\braket{\ket{\psi}\bra{\psi}, \re\D^{(V, b, a)}(\ket{\psi}\bra{\psi})}| \geq \inf_{\ket{h} \in \mathcal{H}_{1}}|\braket{\ket{h}\bra{h}, \mathrm{Re}\mathcal{L}_{2}|_{\mathcal{B}(\mathcal{H}_{1})}(\ket{h}\bra{h})}| -\frac{2\norm{\mathcal{L}_{2}}_{\infty}^{2}}{(1 - \cos^{2}\theta_{\inp})\gap{\mathcal{L}_{1}} - 2 \norm{\mathcal{L}_{2}}_{\infty}} - p_{1}.
    \end{equation}
    Note, since we have applied the statement inside $\mathcal{B}(\mathcal{H}_{\prop})$, $\mathcal{H}_{1} = \{\ket{\psi} \in \mathcal{H}_{\prop} \; | \; \ket{\psi}\bra{\psi} \in \ker{\D_{\inp}}\} = \mathcal{H}_{\qcirc}$ where $\mathcal{H}_{\qcirc}$ is defined according to equation~\eqref{eq:Hqcirc} of Definition~\ref{def:circsubspace}. Therefore our lower-bound for the decay of an arbitrary pure quantum state $\ket{\psi}$ becomes
    \begin{equation}
        \label{eq:lwbd4}
        |\braket{\ket{\psi}\bra{\psi}, \re\D^{(V, b, a)}(\ket{\psi}\bra{\psi})}| \geq \inf_{\ket{h} \in \mathcal{H}_{\qcirc}}|\braket{\ket{h}\bra{h}, \mathrm{Re}(\Gamma_{\out}\D_{\out})(\ket{h}\bra{h})}|  - p_{1} - p_{2}
    \end{equation}
    where we have defined
    \begin{equation}
        p_{2} =  \frac{2\norm{\Gamma_{\out}\D_{\out}|_{\mathcal{B}(\mathcal{H}_{\prop})}}_{\infty}^{2}}{(1 - \cos^{2}\theta_{\inp})\Gamma_{\inp}\gap{\D_{\inp}|_{\mathcal{B}(\mathcal{H}_{\prop})}} - 2 \norm{\Gamma_{\out}\D_{\out}|_{\mathcal{B}(\mathcal{H}_{\prop})}}_{\infty}}.
    \end{equation}
    Since for all $\ket{h} \in \mathcal{H}_{\qcirc}$, we have $\ket{h}\bra{h} \in \ker{\D_{\prop}^{(V)}}$ and $\ket{h}\bra{h} \in \ker{\D_{\inp}}$, we can re-write equation~\eqref{eq:lwbd4} giving us for all pure quantum states $\ket{\psi}$, 
    \begin{equation}
        \label{eq:lwbd5}
        |\braket{\ket{\psi}\bra{\psi}, \re\D^{(V, b, a)}(\ket{\psi}\bra{\psi})}| \geq \inf_{\ket{h} \in \mathcal{H}_{\qcirc}}|\braket{\ket{h}\bra{h}, \mathrm{Re}\D^{(V, b, a)}(\ket{h}\bra{h})}|  - p_{1} - p_{2}.
    \end{equation}

    Considering the decay of an arbitrary $\ket{h} \in \mathcal{H}_{\qcirc}$ under $\D^{(V, b, a)}$. Lemma~\ref{lm:onlycirc} implies
    we can write $\ket{h} = \ket{\eta_{\phi}}$ where $\ket{\eta_{\phi}}$ is a pure history state associated to $V$ with input state $\ket{\phi}=\ket{\psi}_{M} \otimes \ket{0}_{A}$ for some $\ket{\psi} \in \mathcal{H}_{M}$. Applying Lemma~\ref{lm:properdecay} we find for all $\ket{h} \in \mathcal{H}_{\qcirc}$ ,
    \begin{align}
        \label{eq:circdecay}
        |\braket{\ket{h}\bra{h}, \mathrm{Re}\D^{(V, b, a)}(\ket{h}\bra{h})}| = \Big(1 - \frac{1}{T} &\Tr\left(\ket{0}\bra{0}_{M_{0}}V(\ket{\psi}\bra{\psi}_{M}\otimes \ket{0}\bra{0}
        _{A})V^{*}\right)\Big) \nonumber \\
                                                                                                             &\times \Tr\left(\ket{0}\bra{0}_{M_{0}}V(\ket{\psi}\bra{\psi}_{M}\otimes \ket{0}\bra{0}_{A})V^{*}\right).
    \end{align}
    Since our $(V, b, a)$ forms a \emph{(reject-case)} instance of \textsc{Quantum Circuit-SAT} we can bound 
    \begin{equation}
        1 \geq \Tr\left(\ket{0}\bra{0}_{M_{0}}V(\ket{\psi}\bra{\psi}_{M} \otimes \ket{0}\bra{0}_{A})V^{*}\right) = 1 - \Tr\left(\ket{1}\bra{1}_{M_{0}}V(\ket{\psi}\bra{\psi}_{M} \otimes \ket{0}\bra{0}_{A})V^{*}\right) \geq 1 - a
    \end{equation}
    implying for all $\ket{h} \in \mathcal{H}_{\qcirc}$
    \begin{equation}
        \label{eq:circdecay1}
        |\braket{\ket{h}\bra{h}, \mathrm{Re}\D^{(V, b, a)}(\ket{h}\bra{h})}| \geq \left(1 - \frac{1}{T}\right)(1-a).
    \end{equation}
    Applying the above equation to our lower-bound for the decay of an arbitrary pure quantum state, equation~\eqref{eq:lwbd5}, we find for all pure states $\ket{\psi}$,
    \begin{equation}
        \label{eq:lwbd6}
        |\braket{\ket{\psi}\bra{\psi}, \re\D^{(V, b, a)}(\ket{\psi}\bra{\psi})}| \geq \left(1 - \frac{1}{T}\right)(1 - a)  - p_{1} - p_{2}.
    \end{equation}

    Moving to calculating our perturbations $p_{1}$ and $p_{2}$. From Lemma~\ref{lm:Lupperbound} and equations~\eqref{eq:Dinp} \eqref{eq:Dout} and~\eqref{eq:dissrates} of Definition~\ref{def:QCSATdis} we find 
    \begin{equation}
        \label{eq:norms}
        \begin{cases}
            \norm{\Gamma_{\out}\D_{\out}|_{\mathcal{B}(\mathcal{H}_{\prop})}}_{\infty} \leq \Gamma_{\out}\norm{\D_{\out}}_{\infty} \leq 2\Gamma_{\out} \\
            \norm{\Gamma_{\inp}\D_{\inp} + \Gamma_{\out}\D_{\out}}_{\infty} \leq \Gamma_{\inp}\norm{\D_{\inp}}_{\infty} + \Gamma_{\out}\norm{\D_{\out}}_{\infty}  \leq 2(k\Gamma_{\inp} + \Gamma_{\out})
        \end{cases} .
    \end{equation}
    In-addition, from Proposition~\ref{prop:propdissgap} and Lemma~\ref{lm:DpropLeakage} as well as Lemma~\ref{lm:DinpGap} and Lemma~\ref{lm:DinpLeakage} we find 
    \begin{equation}
        \label{eq:gaps+leakage}
        \begin{cases}
            (1 - \cos^{2}\theta_{\prop})\Gamma_{\prop}\gap{\D_{\prop}^{(V)}} \geq \left(1 - \frac{1}{T-1}\right)\frac{\Gamma_{\prop}}{8T^{4}} \geq \frac{(T-2)\Gamma_{\prop}}{8T^{5}} \\
            (1 - \cos^{2}\theta_{\inp})\Gamma_{\inp}\gap{\D_{\inp}|_{\mathcal{B}(\mathcal{H}_{\prop})}} \geq \frac{\Gamma_{\inp}}{2T}
        \end{cases} .
    \end{equation}
    Thus for $p_{1}$ we calculate
    \begin{align}
        p_{1} &\leq \frac{8(k\Gamma_{\inp} + \Gamma_{\out})^{2}}{\frac{(T-2)\Gamma_{\prop}}{8T^{5}}- 4(k\Gamma_{\inp} + \Gamma_{\out})} = \frac{1}{2T}
    \end{align}
    where the last equality follows from the definitions of $\Gamma_{\inp}$, $\Gamma_{\prop}$, and $\Gamma_{\out}$ provided in equation~\eqref{eq:dissrates} of Definition~\ref{def:QCSATdis}. For perturbation $p_{2}$ we calculate
    \begin{equation}
        p_{2} \leq  \frac{8\Gamma_{\out}^{2}}{\frac{\Gamma_{\inp}}{2T} - 4\Gamma_{\out}}
 = \frac{1}{2T}
    \end{equation}
    where the last equality follows from our definitions of $\Gamma_{\inp}$ and $\Gamma_{\out}$ provided in equation~\eqref{eq:dissrates} of Definition~\ref{def:QCSATdis}. Applying our upper-bounds on $p_{1}$ and $p_{2}$ to equation~\eqref{eq:lwbd6} we find 
    \begin{equation}
        \label{eq:lwbd7}
        |\braket{\ket{h}\bra{h}, \mathrm{Re}\D^{(V, b, a)}(\ket{h}\bra{h})}| \geq   \left(1 - \frac{1}{T}\right)(1-a) - \frac{1}{T} = b'.
    \end{equation}
    as required.
\end{proof}

Thus we find from the above theorem combined with Corollary~\ref{cor:LLinQMA} the following statement.
\begin{theorem}
    \LL{$\log(n)$} is $\mathrm{QMA}$-$\mathrm{Complete}$
\end{theorem}

In-addition, the proof of Theorem~\ref{thm:quasihardness} also allows us to show a hardness result for \LLPS{$\log(n)$}.
\begin{theorem}
    \LLPS{$\log(n)$} is $\mathrm{QMA}_{1}$-$\mathrm{Hard}$
\end{theorem}
\begin{proof}
    Consider the proof of Theorem~\ref{thm:quasihardness} and the set of jump operators $J$ used to form the $\mathrm{QMA}$-Hard instance of \LL{$\log(n)$}. For all $L \in J$, it follows from equations~\eqref{eq:JinpJout} \eqref{eq:Jprop} and \eqref{eq:rescaledJs} that $L = L^{*}$. Thus the $\mathrm{QMA}$-Hard instance of \LL{$\log(n)$} is also a valid instance of \RLL{$\log(n)$}. Therefore \RLL{$k$} is $\mathrm{QMA}$-Hard using the proof of Theorem~\ref{thm:quasihardness}.

    In our instance of \textsc{Quantum Circuit-SAT} $(Q, b, a)$ in the proof of Theorem~\ref{thm:quasihardness} set $b=1$. Thus \RLL{$\log(n)$} with $a=0$ is $\mathrm{QMA}_{1}$-Hard. From Proposition~\ref{prop:rlltollps} of Section~\ref{sec:dfs} we find \LLPS{$\log(n)$} is $\mathrm{QMA}_{1}$-Hard.
\end{proof}

The quasi-localness of the dissipator used to show the hardness result of Theorem~\ref{thm:quasihardness} comes from the portions of the jump operators involving the clock register $C$.  To extend the result of Theorem~\ref{thm:quasihardness} to apply to the \LL{$k$} problem with constant $k$ we must modify our clock register in such a way as to allow the quasi-local terms in the Lindbladian to become local. To do this we take the standard approach of encoding the quantum circuit's clock in unary as opposed to binary. In essence, we map
\begin{equation}
    \label{eq:unary}
    \ket{t} \mapsto \ket{0}^{\otimes T - t} \otimes \ket{1}^{\otimes t},
\end{equation}
implying our clock register $C$ goes from containing $\log(T)$ qubits to containing $T$ qubits. Since our clock register contains $T$ qubits, there exists $T$-bit strings that no longer encode a time $t$. Therefore we define
\begin{equation}
    \label{eq:globalillandlegstrings}
    \mathrm{S}_{\leg}^{(T)} = \left\{0^{T-t}1^{t}\right\}_{t=0}^{T} \text{ and } \mathrm{S}_{\ill}^{(T)} = \left\{0, 1\right\}^{T} \setminus S_{\leg}^{(T)}
\end{equation}
where $\mathrm{S}_{\leg}^{(T)}$ denotes the $T$-bit strings that encode a valid time $t$ in unary. Our clock register's Hilbert space also decomposes as $\mathcal{H}_{C} = \mathcal{H}^{(T)}_{\leg} \oplus \mathcal{H}^{(T)}_{\ill}$ where we have defined 
\begin{equation}
    \label{eq:globalillandlegsubspaces}
    \mathcal{H}^{(T)}_{\leg} = \Span{\left\{\ket{x}\right\}_{x \in S_{\leg}^{(T)}}} \text{ and } \mathcal{H}^{(T)}_{\ill} = \Span{\left\{\ket{x}\right\}_{x \in S_{\ill}^{(T)}}}.
\end{equation}
Applying the mapping described in equation~\eqref{eq:unary} to equations~\eqref{eq:purehistorystate} and~\eqref{eq:mixedhistorystate} give us our new definitions of a \emph{pure history state} and a \emph{mixed history state} respectively, but applying this mapping to Definition~\ref{def:QCSATdis} does not gives us a Lindbladian dissipator with jump operators that act locally. To fix this, if we assume our clock register $C$ is constrained to exist only in valid unary clock states then we can map
\begin{equation}
    \label{eq:local}
    \begin{cases}
        \ket{t}\bra{t} \mapsto \ket{01}\bra{01}_{t+1,t}, & \ket{t+1}\bra{t} \mapsto \ket{011}\bra{001}_{t+2, t+1, t}, \\
        \ket{T-1}\bra{T-1} \mapsto \ket{1}\bra{1}_{T-1}, & \ket{T-1}\bra{T-2} \mapsto \ket{11}\bra{01}_{T-1, T-2}, \\
        \ket{0}\bra{0} \mapsto \ket{0}\bra{0}_{1}, & \ket{1}\bra{0} \mapsto \ket{01}\bra{00}_{2,1},
    \end{cases}
\end{equation}
inside Definition~\ref{def:QCSATdis}, resulting in a Lindbladian  dissipator with jump operators that do act locally. Performing this mapping creates a new problem that arises from the fact that the mapping described in equation~\eqref{eq:local} is only valid within $\mathcal{B}(\mathcal{H}_{\leg}^{(T)})$. To deal with this new problem we add an additional  Lindbladian dissipator $\Gamma_{\unary}\D_{\unary}$ that drives the system into this subspace.

Interestingly, in-order to engineer $\Gamma_{\unary}\D_{\unary}$ we must diverge from the approach previously taken. In our previous approach we would have taken the Hamiltonian $H_{\mathrm{stab}} = \sum_{t=0}^{T-2}\ket{10}\bra{10}_{t+1,t}$, used by Kiteav et al. in \cite{kitaev2002classical} to stabilize the legal unary clock subspace, and turned the individual Hamiltonian constraints $\ket{10}\bra{10}_{t+1,t}$ into jump operators generating $\D_{\unary}$. The problem with this approach is that the kernel of this $\D_{\unary}$ contains pure illegal clock register states. To see this note that the pure illegal clock register states can still be eigenvectors of all the jump operators $\ket{10}\bra{10}_{t+1,t}$. Applying Lemma~\ref{lm:restrictdfscond} implies those states also lie in the kernel of the dissipator. 

To construct $\Gamma_{\unary}\D_{\unary}$ consider the subspace of the clock register $C$ given by qubits $C_{t+2,t+1,t} = \{t+2, t+1,t\}$. Defining our local illegal and legal bit strings as 
\begin{equation}
    \label{eq:localillandlegstrings}
    \mathrm{S}_{\ill}^{\loc} = \{000, 001, 011, 111\} \text{ and } \mathrm{S}_{\leg}^{\loc} = \{100, 010, 101, 110\}, 
\end{equation}
allows us to partition $\mathcal{H}_{C_{t+2,t+1,t}} = \mathcal{H}_{\leg}^{\loc} \oplus \mathcal{H}_{\ill}^{\loc}$ where we have defined,
\begin{equation}
    \label{eq:localillandlegsubspaces}
    \mathcal{H}_{\leg}^{\loc} = \Span{\left\{\ket{x} \right\}_{ x \in S_{\leg}^{\loc}}} \text{ and } \mathcal{H}_{\ill}^{\loc} = \Span{\left\{\ket{x} \right\}_{ x \in S_{\ill}^{\loc}}}.
\end{equation}
We require a set $3$-local jump operators that collectively have each $\mathcal{H}_{\leg}^{\loc}$ as an eigenspace but not $\mathcal{H}_{\ill}^{\loc}$. As a side note, for a single Hermitian jump operator this is not possible since part of the spectrum would have to be assigned to vectors in $\mathcal{H}_{\ill}^{\loc}$. In addition, if we had considered to a $2$-qubit subspace instead of $3$-qubit, then the local illegal subspace would be $1$-dimensional. This forces every operator which has the local legal subspace as an eigenspace, too also have the local illegal subspace as an eigenspace. Thus we are forced to consider the $3$-local jump operators as a starting point with more than one jump operator in the set.

Define our set of $3$-local jump operators acting on register $C_{t+2,t+1,t}$ as $J^{(t)}$. Take for all $L \in J^{(t)}$, $L \in \mathcal{B}(\mathcal{H}^{\loc}_{\ill})$ as our ansatz. Clearly, $\mathcal{H}_{\leg}^{\loc}$ is a collective eigenspace with eigenvalue $0$ for all jump operators in $J^{(t)}$. Thus, we only require that none of the vectors in $\mathcal{H}_{\ill}^{\loc}$ are eigenvectors of all the jump operators in $J^{(t)}$ simultaneously. In order to satisfy this condition, we can simply take $J^{(t)}$ to be the set of jump operators generating global depolarization on the subspace $\mathcal{B}(\mathcal{H}_{\ill}^{\loc})$. 

The full $\D_{\unary}$ becomes the dissipator generated by the jump operators in each $J^{(t)}$ for $t \in \{0, 1, \cdots, T-3\}$. Since each 3-qubit window of a valid unary clock state lies in $\mathcal{B}(\mathcal{H}_{\leg}^{\loc})$ we immediately have that $\mathcal{B}(\mathcal{H}_{\leg}^{(T)})$ is contained in the kernel of $\D_{\unary}$. In addition , since the local depolarization dissipators overlapped, outside of $\mathcal{B}(\mathcal{H}_{\mathrm{legal}}^{(T)})$ the system is driven into the maximally mixed state supported on $\mathcal{H}_{\mathrm{illegal}}^{(T)}$. Thus the only pure steady-states lie in $\mathcal{B}(\mathcal{H}_{\leg})$ with the addition that  leakage outside of $\mathcal{B}(\mathcal{H}_{\leg}^{(T)})$ is small since the only steady-state supported on both portions of the kernel is maximally mixed state. Therefore, we can define our complete local dissipator which embeds \textsc{Quantum Circuit-SAT}, given below. 
\begin{definition}[Local \textsc{Quantum Circuit-SAT} Dissipator]
    \label{def:localDQCSAT}
    Let $(V, b, a)$ be an instance of \textsc{Quantum Circuit-SAT} and $\D^{(V, b, a)}$ be defined according to equation~\eqref{eq:DQCSAT} of Definition~\ref{def:QCSATdis}. Let $\tilde{\D}^{(V, b, a)}$ be the modification of $\D^{(V, b, a)}$ under the mapping defined in equation~\eqref{eq:local}, then the local \textsc{Quantum Circuit-SAT} dissipator is defined as
    \begin{equation}
        \label{eq:Dloc}
        \D_{\loc}^{(V, b, a)} = \Gamma_{\unary}\D_{\unary} + \tilde{\D}^{(V, b, a)},
    \end{equation}
    where 
    \begin{equation}
        \label{eq:Dunary}
        \D_{\unary}(\bullet) = \sum_{t=0}^{T-3}\D_{\loc}^{(t+2,t+1,t)}(\bullet) = \sum_{t=0}^{T-3}\sum_{x, y \in \mathrm{S}_{\ill}^{\loc}}\left(\frac{1}{4}(\ket{x}\bra{y}_{t+2,t+1,t}) \bullet (\ket{y}\bra{x}_{t+2,t+1,t}) - \frac{1}{8}\left\{\ket{y}\bra{y}_{t+2,t+1,t}, \bullet \right\}\right),
    \end{equation}
    with $\mathrm{S}_{\ill}^{\loc}$ defined according to equation~\eqref{eq:localillandlegstrings} and dissipation rate
    \begin{equation}
        \label{eq:Gunary}
        \Gamma_{\unary} = \frac{100T^{4}}{(T-1)}\left[4(k\Gamma_{\inp} + T \Gamma_{\prop} + \Gamma_{\out}) + 8T(k\Gamma_{\inp} + T \Gamma_{\prop} + \Gamma_{\out})^{2}\right]
    \end{equation}
    where $\Gamma_{\inp}$, $\Gamma_{\prop}$, and $\Gamma_{\out}$ are defined according to equation~\eqref{eq:dissrates} of Definition~\ref{def:QCSATdis}.
\end{definition}
Similar to our hardness argument in the quasi-local case, to prove completeness, we simply show that a pure history state associated to $V$ with an input state of high accept probability has the appropriate decay under our local \textsc{Quantum Circuit-SAT} dissipator. To do so we need a modified version of Lemma~\ref{lm:properdecay}, which we state below.
\begin{lemma}
    \label{lm:properdecayloc}
    Let $(V, b,a)$ be an instance of \textsc{Quantum Circuit-SAT} and $\D_{\loc}^{(V, a, b)}$ defined according to equation~\eqref{eq:Dloc} of Definition~\ref{def:localDQCSAT}.  Let $\ket{\eta_{\phi}}$ be a pure history state associated to $V$ with input state $\ket{\phi} = \ket{\psi}_{M} \otimes \ket{0}_{A} \in \mathcal{H}_{M}\otimes \mathcal{H}_{A}$, then the decay of $\ket{\eta_{\phi}}\bra{\eta_{\phi}}$ under $\D_{\loc}^{(V, a, b)}$  is given by  
    \begin{align}
        |\braket{\ket{\eta_{\phi}}\bra{\eta_{\phi}}, \re\D_{\loc}^{(V, a, b)}(\ket{\eta_{\phi}}\bra{\eta_{\phi}})}| = \Big(1 - \frac{1}{T} &\Tr\left(\ket{0}\bra{0}_{M_{0}}V\ket{\phi}\bra{\phi}V^{*}\right)\Big) \Tr\left(\ket{0}\bra{0}_{M_{0}}V\ket{\phi}\bra{\phi}V^{*}\right).
    \end{align}
\end{lemma}
\begin{proof}
    See Appendix~\ref{app:F}
\end{proof}
Our strategy to prove soundness again involves using our pure state decay projection lemma to project our analysis into the subspace of the properly initialized unary clock register. In this subspace, the local jump operators behave like the original quasi-local ones. From this point, we can simply reverse the mapping of equation~\eqref{eq:local} and apply our original analysis to lower-bound the decay in the quasi-local case. To use our pure state decay projection lemma, we need to characterize the kernel of $\D_{\unary}$. We also need control over $\gap{\D_{\unary}}$ as well as the leakage outside of the span of the pure steady-states. In the statements below we provide the following information on $\D_{\unary}$ needed in order to apply our pure state decay projection lemma.
\begin{proposition}[Kernel of Unary Clock Dissipator]
    \label{prop:DuKer}
    Let $\D_{\unary}$ be the Lindbladian dissipator defined according to equation~\eqref{eq:Dunary} of Definition~\ref{def:localDQCSAT} and $P_{\ill}^{(T)}$ be the projection onto $\mathcal{H}_{\ill}^{(T)}$ with $\mathcal{H}_{\ill}^{(T)}$ defined according to equation~\eqref{eq:globalillandlegsubspaces} , then 
    \begin{equation}
        \ker{\D_{\unary}} = \mathcal{B}(\mathcal{H}_{MA}) \otimes \left(\mathcal{B}(\mathcal{H}_{\leg}^{(T)}) \oplus P_{\ill}^{(T)}\right)
    \end{equation}
\end{proposition}
\begin{proof}
    See Appendix~\ref{app:D}
\end{proof}
\begin{proposition}[Gap of Unary Clock Dissipator]
    \label{prop:DuGap}
    Let $\D_{\unary}$ be the Lindbladian dissipator defined according to equation~\eqref{eq:Dunary} of Definition~\ref{def:localDQCSAT}, then 
    \begin{equation}
        \gap{\D_{\unary}} = 1 / 96T^{3}.
    \end{equation}
\end{proposition}
\begin{proof}
    See Appendix~\ref{app:D}
\end{proof}
\begin{lemma}
    \label{lm:DuLeakage}
    Let $\D_{\unary}$ be the Lindbladian dissipator defined according to equation~\eqref{eq:Dunary} of Definition~\ref{def:localDQCSAT} and define $\mathcal{H}_{\leg}^{(T)}$ according to equation~\eqref{eq:globalillandlegsubspaces} and
    \begin{equation}
        \mathsf{M} = \left\{v_{1}\ket{h}\bra{\perp} + v_{2}\ket{\perp}\bra{h} + v_{3}\ket{\perp}\bra{\perp} \; \big| \; (\ket{h}, \ket{\perp}, v) \in \mathcal{H}_{MA} \otimes \mathcal{H}_{\leg}^{(T)} \times [\mathcal{H}_{MA} \otimes \mathcal{H}_{\leg}^{(T)}]^{\perp} \times \mathbb{C}^{3} \right\},
    \end{equation}
    then the overlap between $\ker{\D_{\unary}}$ and $\mathsf{M}$ is given by 
    \begin{equation}
        \cos^{2}\theta = \sup_{\substack{A_{0} \in \ker{\D_{\unary}}, A \in \mathsf{M} \\ \braket{A_{0}, A_{0}} = \braket{A, A} = 1}} |\braket{A_{0}, A}|^{2} \leq \frac{1}{2^{T} - T}
    \end{equation}
\end{lemma}
\begin{proof}
    See Appendix~\ref{app:D}
\end{proof}
Thus we can proceed towards proving that \LL{$k$} is $\mathrm{QMA}$-Hard.
\begin{theorem}
    \label{thm:localhardness}
    For $k \geq 5$, \LL{$k$} is $\mathrm{QMA}$-$\mathrm{Hard}$.
\end{theorem}
\begin{proof}
    To show \LL{$k$} is $\mathrm{QMA}$-$\mathrm{Hard}$ we show there exists a reduction from \textsc{Quantum Circuit-SAT} to \LL{$k$} for $k \geq 5$.\\


    Let $(Q, b, a)$ be an instance of \textsc{Quantum Circuit-SAT}. Similar to our reduction in Theorem~\ref{thm:quasihardness}, from $(Q, b, a)$ we construct an instance of \LL{$k$} by transforming the \textsc{Quantum Circuit-SAT} instance $(Q,b,a)$ into an equivalent instance $(V,b,a)$ then into the \LL{$k$} instance $(J, a', b')$ for $k \geq 5$. Similarly, to transform $(Q, b, a)$ into an equivalent instance $(V, b, a)$ we add identity gates into the quantum circuit $Q = Q_{T'-1}Q_{T'-2}\cdots Q_{0}$ to produce quantum circuit $V=U_{T-1}U_{T-2}\cdots U_{0}$ where
    \begin{equation}
        \label{eq:Tlocal}
        T = T' + \frac{3-a}{b-a} + 2.
    \end{equation}
    From \textsc{Quantum Circuit-SAT} instance $(V, b, a)$, let $\D_{\loc}^{(V, b,a)}$ be the local \textsc{Quantum Circuit-SAT} dissipator defined in equation~\eqref{eq:Dloc} of Definition~\ref{def:localDQCSAT}. We construct our instance of \LL{$k$} by letting $J$ be a set of jump operators which generate $\D_{\loc}^{(V, b, a)}$ and setting 
    \begin{equation}
        \label{eq:abloc}
        \begin{cases}
            a' = 1 - b,\\
            b' = (1 - \frac{1}{T})(1-a) - \frac{2}{T}.
        \end{cases}
    \end{equation}
    Similar to our proof of Theorem~\ref{thm:quasihardness} we confirm,
    \begin{equation}
        b' - a' = \frac{(b - a)T  - (3 - a)}{T} = \frac{(b-a)(T' + 2)}{T' + \frac{3-a}{b-a} + 2} \geq \frac{1}{\poly(n)}. 
    \end{equation}
    Thus we only require the set of jump operators $J$ to satisfy the conditions of the problem in-order for $(J, a' , b')$ to be a valid instance of \LL{$k$}.

    According to equation~\eqref{eq:Dloc} of Definition~\ref{def:localDQCSAT}, $\D_{\loc}^{(V, b, a)} = \Gamma_{\unary}\D_{\unary} + \tilde{\D}^{(V, b, a)}$ where $\D_{\unary}$ is generated by jump operators 
    \begin{equation}
        \label{eq:Junary}
        \tilde{J}_{\unary} = \left\{\frac{1}{2\sqrt{2}}(\ket{x}\bra{y}_{t+2,t+1,t} +  \ket{y}\bra{x}_{t+2,t+1,t}), \frac{i}{2\sqrt{2}}(\ket{x}\bra{y}_{t+2,t+1,t} -  \ket{y}\bra{x}_{t+2,t+1,t})\; \Big| \; x, y \in \mathrm{S}_{\ill}^{\loc}\right\}_{t=0}^{T-3}
    \end{equation}
    and $\tilde{\D}^{(V, b, a)}$ is equal to $\D^{(V, b, a)}$ after having applied the mapping of equation~\eqref{eq:local}. $\D^{(V,b,a)}$ is defined according to equation~\eqref{eq:DQCSAT} of Definition~\ref{def:QCSATdis}. Equation~\eqref{eq:DQCSAT} of Definition~\ref{def:QCSATdis} implies we can write $\tilde{\D}^{(V, b, a)} = \Gamma_{\out}\tilde{\D}_{\out} + \Gamma_{\prop}\tilde{\D}_{\prop}^{(V)} + \Gamma_{\inp}\tilde{\D}_{\inp}$ where $\tilde{\D}_{\inp}$, $\tilde{\D}^{(V)}_{\prop}$, and $\tilde{\D}_{\out}$  are equal to $\D_{\inp}$, $\D_{\prop}^{(V)}$, and $\D_{\out}$ of equation~\eqref{eq:DQCSAT} after having applied the mapping of equation~\eqref{eq:local}. Thus we find that $\tilde{\D}_{\inp}$ and $\tilde{\D}_{\out}$ can be generated by jump operators
    \begin{equation}
        \label{eq:JinpandJoutlocal}
        \tilde{J}_{\inp} = \left\{\ket{1}\bra{1}_{A_{i}} \otimes \ket{0}\bra{0}_{1}\right\}_{i=1}^{k}, \text{ and } \tilde{J}_{\out} = \left\{\ket{0}\bra{0}_{M_{0}} \otimes \ket{1}\bra{1}_{T-1}\right\}
    \end{equation}
    respectively, along with $\tilde{\D}_{\prop}^{(V)}$ being generated by jump operators
    \begin{align}
        \label{eq:Jproplocal}
        \tilde{J}_{\prop} = \Big\{ \frac{1}{2}\big(\ket{01}\bra{01}_{t+2, t+1} -& [U_{t+1}]_{MA} \otimes \ket{011}\bra{001}_{t+2,t+1,t} \nonumber \\
        &- [U^{*}_{t+1}]_{MA} \otimes \ket{001}\bra{011}_{t+2,t+1,t} + \ket{01}\bra{01}_{t+1,t}\big)\Big\}_{t=0}^{T-2}.
    \end{align}
    Re-scaling each set of jump operators according to their respective dissipation rates we can define jump operator sets 
    \begin{equation}
        \label{eq:jumpslocal}
        \begin{cases}
            J_{\unary} = \{\sqrt{\Gamma_{\unary}}L \; | \; L \in \tilde{J}_{\unary}\}, \\
            J_{\inp} = \{\sqrt{\Gamma_{\inp}}L \; | \; L \in \tilde{J}_{\inp}\}, \\
            J_{\out} = \{\sqrt{\Gamma_{\out}}L \; | \; L \in \tilde{J}_{\out}\}, \\
            J_{\prop} = \{\sqrt{\Gamma_{\prop}}L \; | \; L \in \tilde{J}_{\prop}\} 
        \end{cases}
    \end{equation}
    which generate $\Gamma_{\unary}\D_{\unary}$, $\Gamma_{\inp}\tilde{\D}_{\inp}$, $\Gamma_{\out}\tilde{\D}_{\out}$, and $\Gamma_{\prop}\tilde{\D}^{(V)}_{\prop}$ respectively. Thus $\D_{\loc}^{(V, b, a)}$ can be generated by the set of jump operators $J$, equal to the set of all jump operators contained in the sets $J_{\unary}$, $J_{\inp}$, $J_{\out}$, and $J_{\prop}$.

    Note, each of the jump operators in $\tilde{J}_{\unary}$, $\tilde{J}_{\inp}$, $\tilde{J}_{\out}$ and $\tilde{J}_{\prop}$ have an operator norm upper-bounded by a constant. In-addition, all of the dissipation rates are $\poly(T, k)$, thus for all $L \in J$ we have $\norm{L} \leq \poly(n)$. We also find, since each of the jump operators in $\tilde{J}_{\unary}$, $\tilde{J}_{\inp}$, and $\tilde{J}_{\out}$ act on $k\leq 3$ qubits, each of the jump operators in $J_{\unary}$, $J_{\inp}$, and $J_{\out}$ also act on $k\leq 3$ qubits. Since each gate $U_{t}$ in our quantum circuit $V = U_{T-1}U_{T-2}\cdots U_{0}$ acts on $k \geq 2$, we find that each jump operator in $J_{\prop}$ acts on $k \geq 5$ qubits. Thus for all $L \in J$, the jump operator $L$ acts non-trivially on at most $k$ qubits for $k \geq 5$ as required. In-addition, since $T=\poly(n)$ and $k=\poly(n)$, we have $\poly(n)$ many jump operators in $J$. Thus$(J, a', b')$ forms a valid instance of \LL{$k$} for $k \geq 5$.

    To complete our description of the reduction we require the existence of an efficient algorithm that can construct $(J, a' ,b')$ from our instance of \textsc{Quantum Circuit-SAT} $(Q, b, a)$. Equations~\eqref{eq:Tlocal} and~\eqref{eq:abloc} imply we can compute $a'$ and $b'$ in $\poly(n)$ time. In-addition, the transformation from $Q$ to $V$ only requires adding $[(3 - a) / (b-a)] + 2= \poly(n)$ many identity gates into the description of the circuit $Q$ which can also be done in $\poly(n)$ time. Since each of the jump operators in $J$ are $k$-local for a constant $k\geq 5$, they can be represented as matrices of constant dimension and can be computed in constant time from the matrix representations of their various components. Since there are $\poly(n)$ many jump operators in $J$, the total runtime of the algorithm computing the set of jump operators $J$ from $(V, b, a)$ is $\poly(n)$ which is efficient. Thus there exists an efficient algorithm transforming an instance of \textsc{Quantum Circuit-SAT} to an instance of \LL{$k$} for $k\geq 5$. \\

    \ul{\textbf{Completeness:}} \\

    Let $(Q, b, a)$ be an \emph{(accept-case)} instance of \textsc{Quantum Circuit-SAT}, and $(J, a', b')$ be the corresponding instance of \LL{$k$} generated according to the description above. The jump operators $J$ generate $\D_{\loc}^{(V, b, a)}$ defined in equation~\eqref{eq:Dloc} of Definition~\ref{def:localDQCSAT} where $(V, b, a)$ is also an \emph{(accept-case)} instance \textsc{Qunatum Circuit-SAT}. For completeness, we require $(J, a', b')$ to be an \emph{(accept-case)} instance of \LL{$k$} implying there exists a pure quantum state $\ket{\psi}$ such that 
    \begin{equation}
        |\braket{\ket{\psi}\bra{\psi}, \re\D_{\loc}^{(V, b, a)}(\ket{\psi}\bra{\psi})}| \leq a'.
    \end{equation}

    To show this, consider a pure history state $\ket{\eta_{\phi}}$ associated to $V$. Choosing $\ket{\phi}_{MA} = \ket{\psi}_{M} \otimes \ket{0}_{A}$ and applying Lemma~\ref{lm:properdecayloc} we find 
    \begin{align}
        |\braket{\ket{\eta_{\phi}}\bra{\eta_{\phi}}, \re\D_{\loc}^{(V, b, a)}(\ket{\eta_{\phi}}\bra{\eta_{\phi}})}| = \Big(1 - \frac{1}{T} &\Tr\left(\ket{0}\bra{0}_{M_{0}}V(\ket{\psi}\bra{\psi}_{M} \otimes \ket{0}\bra{0}_{A})V^{*}\right)\Big) \nonumber \\
                                                                                                                                               &\times \Tr\left(\ket{0}\bra{0}_{M_{0}}V(\ket{\psi}\bra{\psi}_{M} \otimes \ket{0}\bra{0}_{A}) V^{*}\right).
    \end{align}
    Since our $(V, b, a)$ forms an \emph{(accept-case)} instance of \textsc{Quantum Circuit-SAT} we can bound 
    \begin{equation}
        0 \leq \Tr\left(\ket{0}\bra{0}_{M_{0}}V(\ket{\psi}\bra{\psi}_{M} \otimes \ket{0}\bra{0}_{A})V^{*}\right) = 1 - \Tr\left(\ket{1}\bra{1}_{M_{0}}V(\ket{\psi}\bra{\psi}_{M} \otimes \ket{0}\bra{0}_{A})V^{*}\right) \leq 1 - b,
    \end{equation}
    implying
    \begin{equation}
        |\braket{\ket{\eta_{\phi}}\bra{\eta_{\phi}}, \re\D_{\loc}^{(V, a, b)}(\ket{\eta_{\phi}}\bra{\eta_{\phi}})}| \leq  1 - b = a'
    \end{equation}
    as required. \\

    \ul{\textbf{Soundness:}} \\

    Let $(Q, b, a)$ be a \emph{(reject-case)} instance of \textsc{Quantum Circuit-SAT}, and $(J, a', b')$ be the corresponding instance of \LL{$k$} generated according to the description above. The jump operators $J$ generate $\D_{\loc}^{(V, b, a)}$ defined in equation~\eqref{eq:Dloc} of Definition~\ref{def:localDQCSAT} with $(V, b, a)$ being also a \emph{(reject-case)} instance of \textsc{Quantum Circuit-SAT}. For soundness, we require $(J, a', b')$ to be a \emph{(reject-case)} instance of \LL{$k$} implying for all pure quantum states $\ket{\psi}$, 
    \begin{equation}
        |\braket{\ket{\psi}\bra{\psi}, \re\D_{\loc}^{(V, b, a)}(\ket{\psi}\bra{\psi})}| \geq b'.
    \end{equation}

    To show this let $\D_{\loc}^{(V, b, a)} = \mathcal{L}_{1} + \mathcal{L}_{2}$, where $ \mathcal{L}_{2} = \tilde{\D}^{(V, b, a)}$ and $\mathcal{L}_{1} = \Gamma_{\unary}\D_{\unary}$. Applying Lemma~\ref{lm:psdecayproj} we find for all pure quantum states $\ket{\psi}$,
    \begin{align}
        \label{eq:prjtounary}
        |\braket{\ket{\psi}\bra{\psi}, \re\D_{\loc}^{(V, b, a)}(\ket{\psi}\bra{\psi})}| \geq \inf_{\ket{h} \in \mathcal{H}_{1}}&|\braket{\ket{h}\bra{h}, \mathrm{Re}\mathcal{L}_{2}|_{\mathcal{B}(\mathcal{H}_{1})}(\ket{h}\bra{h})}| - \frac{2\norm{\mathcal{L}_{2}}_{\infty}^{2}}{(1 - \cos^{2}\theta)\gap{\mathcal{L}_{1}} - 2 \norm{\mathcal{L}_{2}}_{\infty}} 
    \end{align}
    where $\mathcal{H}_{1} = \Span{\{\ket{\psi}\in\mathcal{H} \; | \; \ket{\psi}\bra{\psi} \in \ker{\mathcal{L}_{1}}\}}$. Since $\ker{\mathcal{L}_{1}} = \ker{\D_{\unary}}$, we can apply Proposition~\ref{prop:DuKer} implying $\mathcal{H}_{1} = \mathcal{H}_{MA} \otimes \mathcal{H}_{\leg}^{(T)}$ where $\mathcal{H}_{\leg}^{(T)}$ is defined according to equation~\eqref{eq:globalillandlegsubspaces}. Therefore equation~\eqref{eq:prjtounary} becomes for all pure quantum states $\ket{\psi}$,
    \begin{equation}
        \label{eq:prjtounary1}
        |\braket{\ket{\psi}\bra{\psi}, \re\D_{\loc}^{(V, b, a)}(\ket{\psi}\bra{\psi})}| \geq \inf_{\ket{h} \in \mathcal{H}_{MA}\otimes \mathcal{H}_{\leg}^{(T)}}|\braket{\ket{h}\bra{h}, \re\tilde{\D}^{(V, b, a)}(\ket{h}\bra{h})}| - p_{u}
    \end{equation}
    where we have defined 
    \begin{equation}
        \label{eq:pu}
        p_{u} = \frac{2\norm{ \tilde{\D}^{(V, b, a)} }_{\infty}^{2}}{(1 - \cos^{2}\theta)\Gamma_{\unary}\gap{\D_{\unary}} - 2\norm{ \tilde{\D}^{(V, b, a)} }_{\infty}}
    \end{equation}

    Proceeding to evaluate $p_{u}$, from Lemma~\ref{lm:Lupperbound} and Definitions~\ref{def:localDQCSAT} and~\ref{def:QCSATdis} we find
    \begin{align}
        \norm{\tilde{\D}^{(V,b,a)}}_{\infty} &\leq \Gamma_{\inp}\norm{\tilde{\D}_{\inp}}_{\infty} + \Gamma_{\prop}\norm{\tilde{\D}_{\prop}^{(V)}}_{\infty} + \Gamma_{\out} \norm{\tilde{\D}_{\out}}_{\infty} \nonumber \\
                                                &\leq 2(k\Gamma_{\inp}  + T\Gamma_{\prop} + \Gamma_{\out}).
    \end{align}
    From Proposition~\ref{prop:DuGap} and Lemma~\ref{lm:DuLeakage} we find
    \begin{equation}
        (1 - \cos^{2}\theta)\Gamma_{\unary}\gap{\D_{\unary}} \geq  \left(1 - \frac{1}{T}\right)\frac{\Gamma_{\unary}}{100T^{3}} = \frac{(T-1)\Gamma_{\unary}}{100T^{4}}
    \end{equation}
    Thus upper-bounding $p_{u}$ using the above equation we find
    \begin{equation}
        \label{eq:pu1}
        p_{u} \leq \frac{8(k\Gamma_{\inp} + T\Gamma_{\prop} + \Gamma_{\out})^{2}}{\frac{(T-1)}{100T^{4}}\Gamma_{\unary} - 4(k\Gamma_{\inp} + T\Gamma_{\prop} + \Gamma_{\out})} \leq  \frac{1}{T}
    \end{equation}
    where the last inequality follows from applying our definition of $\Gamma_{\unary}$ found in equation~\eqref{eq:Gunary} of Definition~\ref{def:localDQCSAT}. Thus from equation~\eqref{eq:prjtounary1} we find for all pure states $\ket{\psi}$, 
    \begin{equation}
        \label{eq:prjtounary2}
        |\braket{\ket{\psi}\bra{\psi}, \re\D_{\loc}^{(V, b, a)}(\ket{\psi}\bra{\psi})}| \geq \inf_{\ket{h} \in \mathcal{H}_{MA}\otimes \mathcal{H}_{\leg}^{(T)}}|\braket{\ket{h}\bra{h}, \re\tilde{\D}^{(V, b, a)}(\ket{h}\bra{h})}| - \frac{1}{T}.
    \end{equation}
    Note, inside $\mathcal{H}_{MA}\otimes \mathcal{H}_{\leg}^{(T)}$, we can reverse the mapping described in equation~\eqref{eq:local} given that both sides of the mapping act on the basis of $\mathcal{H}_{MA}\otimes \mathcal{H}_{\leg}^{(T)}$ in the same way. Reversing the mapping described in equation~\eqref{eq:Dunary}, allows us to map $\mathcal{H}_{MA}\otimes \mathcal{H}_{\leg}^{(T)}$ back to $\mathcal{H}_{\mathrm{old}}$ which is defined as the Hilbert space our original quasi-local dissipator of Definition~\ref{def:QCSATdis} acts on. Therefore the first term in the lower-bound of equation~\eqref{eq:prjtounary2} can be bounded as
    \begin{equation}
        \label{eq:oldanalysis}
        \inf_{\ket{h} \in \mathcal{H}_{\unary}}|\braket{\ket{h}\bra{h}, \re\tilde{\D}^{(V, b, a)}(\ket{h}\bra{h})}| \geq \inf_{\ket{h} \in \mathcal{H}_{\mathrm{old}}}|\braket{\ket{h}\bra{h}, \re\D^{(V, b, a)}(\ket{h}\bra{h})}|
    \end{equation}
    where $\D^{(V, b, a)}$ is defined according to equation~\eqref{eq:DQCSAT} of Definition~\ref{def:QCSATdis}. Applying equation~\eqref{eq:lwbd7} found in our proof of Theorem~\ref{thm:quasihardness} we find, for all pure states $\ket{\psi}$,
    \begin{equation}
        \label{eq:prjtounary3}
        |\braket{\ket{\psi}\bra{\psi}, \re\D_{\loc}^{(V, b, a)}(\ket{\psi}\bra{\psi})}| \geq \left(1 - \frac{1}{T}\right)(1-a) - \frac{2}{T} = b' 
    \end{equation}
    as required.
\end{proof}

Similar to the results that follow from our quasi-local hardness result of Theorem~\ref{thm:quasihardness},  we find from Corollary~\ref{cor:LLinQMA} and Theorem~\ref{thm:localhardness} the following statement.
\begin{theorem}
    For $k \geq 5$, \LL{$k$} is $\mathrm{QMA}$-$\mathrm{Complete}$
\end{theorem}
Similarly we also find, a hardness result for the local version of our pure steady-state problem.
\begin{theorem}
    For $k \geq 5$, \LLPS{$k$} is $\mathrm{QMA}_{1}$-$\mathrm{Hard}$
\end{theorem}
\begin{proof}
    Consider the proof of Theorem~\ref{thm:localhardness} and the set of jump operators $J$ used to form the $\mathrm{QMA}$-Hard instance of \LL{$k$}. For all $L \in J$, it follows from equations~\eqref{eq:Junary}, \eqref{eq:JinpandJoutlocal}, \eqref{eq:Jproplocal}, and \eqref{eq:jumpslocal} that $L = L^{*}$. Thus the $\mathrm{QMA}$-Hard instance of \LL{$k$} is also a valid instance of \RLL{$k$}. Therefore for $k \geq 5$, \RLL{$k$} is $\mathrm{QMA}$-Hard using the proof of Theorem~\ref{thm:localhardness}.

    In our instance of \textsc{Quantum Circuit-SAT} $(Q, b, a)$ in the proof of Theorem~\ref{thm:localhardness} set $b=1$. Thus for $k \geq 5$, \RLL{$k$} with $a=0$ is $\mathrm{QMA}_{1}$-Hard. From Proposition~\ref{prop:rlltollps} of Section~\ref{sec:dfs} we find for $k \geq 5$, \LLPS{$k$} is $\mathrm{QMA}_{1}$-Hard.
\end{proof}

\section{Discussion}
\label{sec:disc}

In this work we have studied the computational complexity of deciding whether a given open quantum system, modeled by a Lindbladian, has a decoherence free subspace. To make headway towards this goal, we introduced and analyzed three problems: the \LLPS{$k$} problem, the \LL{$k$} problem, and the \SEPH{$k$} problem, along with their interrelations depicted in FIG~\ref{fig:mainresults}. Our main results establish that for $k \geq 5$, the \LLPS{$k$} problem is $\mathrm{QMA}_{1}$-Hard with the addition that the \LL{$k$} problem is $\mathrm{QMA}$-Complete. Since the \LLPS{$k$} problem is concerned with deciding whether a given Lindbladian admits a 1-dimensional decoherence free subspace our results suggest that finding a decoherence free subspaces is intractable for quantum computation. In addition, even the relaxed form of the problem, regarding finding slowly decaying pure states, where the decay rate is computed only up to first-order, is still intractable for quantum computation under the assumption $\mathrm{BQP} \neq \mathrm{QMA}$.

There are a number of natural directions for future work. From a complexity-theoretic perspective, it would be interesting to determine whether the \LLPS{$k$} problem is $\mathrm{QMA}_{1}$-Complete, which would require showing inclusion in $\mathrm{QMA}_{1}$. Moreover, developing a perturbation gadget approach, similar to what was developed in \cite{kempe2006thecomplexity}, could prove useful in establishing further intractability results for Lindbladians which model more relevant physical systems. Another interesting open question is with regards to locality. Is the situation similar to the local Hamiltonian problem where the \textsc{$2$-Local Hamiltonian} problem is still $\mathrm{QMA}$-Complete? Or is the situation something different, where the minimum locality to guarantee $\mathrm{QMA}$-Hardness is greater than $2$? From a more physical perspective, the intractability results established here motivate the need for finding more efficiently identifiable conditions that guarantee the presence or absence of a decoherence free subspace. Finally, as briefly mentioned in the introduction and expanded upon in \cite{albert2014symmetries}, the decoherence free subspace is only one type of steady-state structure that time-independent Lindbladian dynamics admits. Given the $\mathrm{QMA}_{1}$-Hardness of the decoherence free subspace problem, it could be interesting to explore the complexity theoretic consequences that the other types of Lindbladian steady-state structures imply.

\begin{acknowledgments}
I would like to thank my advisor Milad Marvian for raising the question posed in this paper, as well as Cole Maurer and Chaithanya Rayudu for the helpful conversations we had regarding this project. This work is supported by Sandia National Laboratories’ Laboratory Directed Research and Development program (Contract \#2534192). Additional support from DOE Express Award No. DE-SC0024685 (2023) and DOE Early Career Research Award No. DE-SC0026373 is gratefully acknowledged. No AI/LLMs were used in the creation of this idea, the proofs or drafting this paper.
\end{acknowledgments}

\bibliography{refs}

\clearpage
\setcounter{table}{0}
\renewcommand{\thetable}{\arabic{table}}%
\setcounter{figure}{0}
\renewcommand{\thefigure}{\arabic{figure}}%
\setcounter{section}{0}
\setcounter{equation}{0}
\renewcommand{\theequation}{\arabic{equation}}%

\onecolumngrid

\appendix

\section{Analysis of Propagation Dissipator}
\label{app:A}

We present an analysis of the propagation dissipator associated to a quantum circuit $V$ as defined in Definition~\ref{def:propdiss}. More specifically, our goal is to provide an analysis of the dissipator's kernel, its spectral gap as well as its leakage outside of the span of its own pure steady-states. Similar to what is done in the Feynman-Kitaev Hamiltonian construction, our first step is to ``rotate away" the specifics of the quantum circuit our propagation dissipator is checking. To do so we define the \emph{circuit rotation} for the quantum circuit $V = U_{T-1}U_{T-2}\cdots U_{0}$ acting on register $S$ as the unitary
    \begin{equation}
        \label{eq:circrot}
        \U^{(V)}(\bullet) = U^{(V)}\bullet U^{(V)\dagger} = \left(\sum_{t=0}^{T-1} [V_{t}]_{S} \otimes \ket{t}\bra{t}_{C}\right)\bullet\left(\sum_{t=0}^{T-1} [V_{t}]_{S} \otimes \ket{t}\bra{t}_{C}\right)^{*} \text{ where } V_{t} = U_{t}U_{t-1}\cdots U_{0}.
    \end{equation}
After rotating away the specifics of the quantum circuit, we analyze the kernel and bound the gap of the resulting dissipator which call the rotated propagation dissipator. We can then rotate the kernel back and leverage the unitary invariance of the Hilbert-Schmidt inner product in order to answer the same questions about the original propagation dissipator. To analyze the leakage, we utilize rotated propagation dissipator in a similar fashion. 

In the remainder of this appendix we define our rotated propagation dissipator and connect it to a Laplacian of a graph. In the two sections following, we provide proofs of Propositions~\ref{prop:propdisker} and~\ref{prop:propdissgap} giving us an analysis of the propagation dissipator's kernel and its spectral gap. In the third and final section we prove Lemma~\ref{lm:DpropLeakage} giving us a statement regarding the propagation dissipator's leakage.

Applying our circuit rotation for quantum circuit $V$ to its associated propagation dissipator we prove the following statement.
\begin{lemma}
    \label{lm:rotpropdef}
    Let $V = U_{T-1}U_{T-2}\cdots U_{0}$ be a quantum circuit acting on register $S$ and $\D_{\prop}^{(V)}$ be its propagation dissipator defined according to equation~\eqref{eq:DpropD} of Definition~\ref{def:propdiss}. Let $\U^{(V)}$ be the  circuit rotation for $V$, then
    \begin{equation}
        \label{eq:lmrot}
        \U^{(V)\dagger} \circ \D_{\prop}^{(V)} \circ \U^{(V)}(\bullet) = \mathcal{I}_{S} \otimes \left[\sum_{t=0}^{T-2}\left(\ket{-_{t}}\bra{-_{t}}_{C} \bullet \ket{-_{t}}\bra{-_{t}}_{C} - \frac{1}{2}\left\{\ket{-_{t}}\bra{-_{t}}_{C}, \bullet \right\}\right)\right],
    \end{equation}
    where $\ket{-_{t}} = \frac{1}{\sqrt{2}}(\ket{t} - \ket{t+1})$.
\end{lemma}
\begin{proof}
    Proceeding by direct calculation we find
    \begin{equation}
        \U^{(V)\dagger} \circ \D_{\prop}^{(V)} \circ \U^{(V)}(\bullet) = \sum_{t=0}^{T-2}\left(\tilde{H}_{t}\bullet \tilde{H}_{t}^{*} - \frac{1}{2}\left\{\tilde{H}_{t}^{*}\tilde{H}_{t}, \bullet\right\}\right)
    \end{equation}
    where
    \begin{align}
        \tilde{H}_{t} &= \U^{(V)\dagger}(H_{t}) = \left(\sum_{t'=0}^{T-1}[V_{t'}]_{S}\otimes \ket{t'}\bra{t'}_{C}\right)^{*}H_{t} \left(\sum_{t=0}^{T-1}[V_{t''}]_{S}\otimes \ket{t''}\bra{t''}_{C}\right) \\
                      &= \frac{1}{2}([V^{*}_{t+1}U_{t+1}]_{S} \otimes \ket{t+1}_{C} - [V_{t}]_{S}^{*}\otimes\ket{t}_{C})([V^{*}_{t+1}U_{t+1}]_{S}^{*} \otimes \bra{t+1}_{C} - [V_{t}]_{S}\otimes\bra{t}_{C}) \\
                      &= \frac{1}{2}\left([V_{t}]_{S}^{*}\otimes (\ket{t+1}_{C} - \ket{t}_{C})\right)\left([V_{t}]_{S}\otimes(\bra{t+1}_{C} - \bra{t}_{C})\right) =  \mathbb{I}_{S} \otimes \ket{-_{t}}\bra{-_{t}}.
    \end{align}
    The result then follows.
\end{proof}

The above statement confirms our intuition about $\U^{(V)}$ ``rotating away" the specifics of quantum circuit $V$ given that the right-hand side of equation~\eqref{eq:lmrot} does not depend on $V$. Defining the \emph{rotated propagation dissipator} as 
\begin{equation}
    \tilde{\D}_{\prop}(\bullet) = \sum_{t=0}^{T-2}\left(\ket{-_{t}}\bra{-_{t}}_{C} \bullet \ket{-_{t}}\bra{-_{t}}_{C} - \frac{1}{2}\left\{\ket{-_{t}}\bra{-_{t}}_{C}, \bullet \right\}\right)
\end{equation}
allows us to reduce the analysis of our propagation dissipator associated with an arbitrary quantum circuits $V$ to that of the rotated propagation dissipator which is independent of the specific quantum circuit. We further transform our problem into a graph theoretic one through the following statement.
\begin{lemma}
    \label{lm:rotadistograph}
    Let $\tilde{\D}_{\prop}$ be the rotated propagation dissipator and define a basis for $\mathcal{B}(\mathcal{H}_{C})$ given by
    \begin{equation}
        \{P_{(t_{A}, t_{B})} \; | \; (t_{A}, t_{B}) \in \{0, 1, \cdots, T-1\}\} \text{ where } P_{(t_{A}, t_{B})} = \ket{t_{A}}\bra{t_{B}}.
    \end{equation}
    Let $L$ be the matrix defined by elements
    \begin{equation}
        L_{v_{1},v_{2}} = -4\braket{P_{v_{1}}, \tilde{\D}_{\prop}(P_{v_{2}})} \text{ for } v_{1}, v_{2} \in \{0, 1, \cdots, T-1\}^{\times 2},
    \end{equation}
    then $L$ is a Laplacian for a graph $G$ with vertices $V = \{0, 1,\cdots, T-1\}^{\times 2}$.
\end{lemma}
\begin{proof}
    Let $(t_{A}, t_{B})', (t_{A}, t_{B}) \in V$, calculating the matrix element of $L$ we find
    \begin{align}
        \label{eq:part1}
        L_{(t_{A}, t_{B})', (t_{A}, t_{B})} &= -4\braket{P_{(t_{A}, t_{B})'}, \tilde{\D}_{\prop}(P_{(t_{A}, t_{B})})} = -4\braket{\ket{t_{A}'}\bra{t_{B}'}, \tilde{\D}_{\prop}(\ket{t_{A}}\bra{t_{B}})} \nonumber \\
                                            &= -4\sum_{t=0}^{T-2}\left(\braket{t_{A}'|-_{t}}\braket{-_{t}|t_{A}}\braket{t_{B}|-_{t}}\braket{-_{t}|t_{B}'} - \frac{1}{2}\braket{t_{A}'|-_{t}}\braket{-_{t}| t_{A}}\delta_{t_{B},t_{B}'} - \frac{1}{2}\braket{t'_{B}|-_{t}}\braket{-_{t}|t_{B}}\delta_{t_{A}, t_{A}'}\right).
    \end{align}
    Given that $\ket{-_{t}} = \frac{1}{\sqrt{2}}(\ket{t} - \ket{t-1})$, we find $L_{(t_{A}, t_{B})', (t_{A}, t_{B})} \in \mathbb{R}$. Taking the transpose of $L$ we find,
    \begin{align}
        \label{eq:intercal}
        [L^{\intercal}]_{(t_{A}, t_{B})', (t_{A}, t_{B})} &= L_{(t_{A}, t_{B}), (t_{A}, t_{B})'} = -4\braket{P_{(t_{A}, t_{B})}, \tilde{\D}_{\prop}(P_{(t_{A}, t_{B})'})} = -4 \overline{\braket{P_{(t_{A},t_{B})'}, \tilde{\D}_{\prop}^{\dagger}(P_{(t_{A}, t_{B})})}} \nonumber \\
                                                          &= -4 \overline{\braket{P_{(t_{A}, t_{B})'}, \tilde{\D}_{\prop}(P_{(t_{A}, t_{B})})}} = \overline{L}_{(t_{A}, t_{B})', (t_{A}, t_{B})} = L_{(t_{A}, t_{B})', (t_{A}, t_{B})}, 
    \end{align}
    where the second line follows from $\tilde{\D}_{\prop}^{\dagger} = \tilde{\D}_{\prop}$. Equation~\eqref{eq:intercal} implies,
    \begin{equation}
        \label{eq:prop1}
        L^{\intercal} = L.
    \end{equation}

    Let $(t_{A}, t_{B})' \neq (t_{A}, t_{B})$, which implies three cases
    \begin{enumerate}
        \item $(t_{A} \neq t_{A}')$ and $(t_{B} = t_{B}')$
        \item $(t_{A} = t_{A}')$ and $(t_{B} \neq t_{B}')$
        \item $(t_{A} \neq t_{A}')$ and $(t_{B} \neq t_{B}')$.
    \end{enumerate}
    Consider the first of the three cases given above, applying equation~\eqref{eq:part1} we find, 
    \begin{align}
        L_{(t_{A}, t_{B})', (t_{A}, t_{B})} &= 2 \sum_{t=0}^{T-2}\Big(\braket{t_{A}' | -_{t}}\braket{-_{t}| t_{A}}(1 - 2|\braket{t_{B}| -_{t}}|^{2})\Big) = 2 \sum_{t=0}^{T-2}\Big(\braket{t_{A}' | -_{t}}\braket{-_{t}| t_{A}}(1 - |\delta_{t_{B},t} - \delta_{t_{B},t+1}|^{2})\Big) \nonumber \\
                                            &= \sum_{t=0}^{T-2}\Big((\delta_{t_{A}', t} - \delta_{t_{A}', t+1})(\delta_{t_{A},t} - \delta_{t_{A},t+1})(1 - \delta_{t_{B},t} - \delta_{t_{B},t+1})\Big) \nonumber \\
                                            &= - \delta_{t_{A},t_{A}'+1}(1 - \delta_{t_{B},t_{A}'} - \delta_{t_{B},t_{A}'+1}) - \delta_{t_{A}',t_{A}+1}(1 - \delta_{t_{B},t_{A}} - \delta_{t_{B},t_{A}+1})
    \end{align}
    Note, for all $(t_{A}, t_{B})$ and $(t_{A}, t_{B})'$ both $(1 - \delta_{t_{B},t_{A}'} - \delta_{t_{B},t_{A}'+1})$ and $(1 - \delta_{t_{B},t_{A}} - \delta_{t_{B},t_{A}+1})$ take values only in the set $\{0, 1\}$ implying $L_{(t_{A}, t_{B})', (t_{A}, t_{B})} \in \{0, -1\}$.  
    Consider the second of the three cases, similarly, we find
    \begin{align}
        L_{(t_{A}, t_{B})', (t_{A}, t_{B})} &= 2 \sum_{t=0}^{T-2}\Big(\braket{t_{B}' | -_{t}}\braket{-_{t}| t_{B}}(1 - 2|\braket{t_{A}| -_{t}}|^{2})\Big) = 2 \sum_{t=0}^{T-2}\Big(\braket{t_{B}' | -_{t}}\braket{-_{t}| t_{B}}(1 - |\delta_{t_{A},t} - \delta_{t_{A},t+1}|^{2})\Big) \nonumber \\
                                            &= \sum_{t=0}^{T-2}\Big((\delta_{t_{B}', t} - \delta_{t_{B}', t+1})(\delta_{t_{B},t} - \delta_{t_{B},t+1})(1 - \delta_{t_{A},t} - \delta_{t_{A},t+1})\Big) \nonumber \\
                                            &= - \delta_{t_{B},t_{B}'+1}(1 - \delta_{t_{A},t_{B}'} - \delta_{t_{A},t_{B}'+1}) - \delta_{t_{B}',t_{B}+1}(1 - \delta_{t_{A},t_{B}} - \delta_{t_{A},t_{B}+1}) \in \{0, -1\}.
    \end{align}
    Finally, consider the third of the three cases discussed above, we find 
    \begin{align}
        L_{(t_{A},t_{B})', (t_{A},t_{B})} &= -4\sum_{t=0}^{T-1}\braket{t_{A}'|-_{t}}\braket{-_{t}| t_{A}}\braket{t_{B}| -_{t}}\braket{-_{t}|t_{B}'} \nonumber \\
                                          &= -\sum_{t=0}^{T-1}\left((\delta_{t_{A}',t} - \delta_{t_{A}', t+1})(\delta_{t_{A},t} - \delta_{t_{A},t+1})(\delta_{t_{B},t} - \delta_{t_{B},t+1})(\delta_{t_{B}',t} - \delta_{t_{B}',t+1})\right) \nonumber \\
                                          &= \delta_{t_{A},t_{A}'+1}(\delta_{t_{B},t_{A}'} - \delta_{t_{B},t_{A}'+1})(\delta_{t_{B}',t_{A}'} - \delta_{t_{B}',t_{A}'+1}) + \delta_{t_{A},t_{A}'-1}(\delta_{t_{B},t_{A}'-1} - \delta_{t_{B},t_{A}'})(\delta_{t_{B}',t_{A}'-1} - \delta_{t_{B}',t_{A}'}) \nonumber \\
                                          &= -\delta_{t_{A},t_{A}'+1}(\delta_{t_{B},t_{A}'}\delta_{t_{B}',t_{A}'+1} + \delta_{t_{B},t_{A}'+1}\delta_{t_{B}',t_{A}'}) - \delta_{t_{A},t_{A}'-1}(\delta_{t_{B},t_{A}'-1}\delta_{t_{B}',t_{A}'} + \delta_{t_{B},t_{A}'}\delta_{t_{B}',t_{A}'-1}).
    \end{align}
    The equation above implies 
    \begin{equation}
        L_{(t_{A},t_{B})', (t_{A},t_{B})} = \begin{cases}-(\delta_{t_{B},t_{A}'}\delta_{t_{B}',t_{A}'+1} + \delta_{t_{B},t_{A}'+1}\delta_{t_{B}',t_{A}'}) \text{ for } t_{A} = t_{A}'+1  \\ - (\delta_{t_{B},t_{A}'-1}\delta_{t_{B}',t_{A}'} + \delta_{t_{B},t_{A}'}\delta_{t_{B}',t_{A}'-1}) \text{ for } t_{A} = t_{A}'-1 \\ 0 \text{ otherwise }.
        \end{cases}
    \end{equation}
    It is clear that $(\delta_{t_{B},t_{A}'}\delta_{t_{B}',t_{A}'+1} + \delta_{t_{B},t_{A}'+1}\delta_{t_{B}',t_{A}'}) \in \{0, 1\}$; similarly, $(\delta_{t_{B},t_{A}'-1}\delta_{t_{B}',t_{A}'} + \delta_{t_{B},t_{A}'}\delta_{t_{B}',t_{A}'-1}) \in \{0, 1\}$. Therefore we have $L_{(t_{A},t_{B})', (t_{A},t_{B})} \in \{0, -1\}$.  Thus we find for all $v_{1}, v_{2} \in V$ and $v_{1} \neq v_{2}$,
    \begin{equation}
        \label{eq:prop2}
        L_{v_{1}, v_{2}} \in \{0, -1\}
    \end{equation}

    Consider the sum $\sum_{(t_{A}, t_{B}) \in V} L_{(t_{A}, t_{B})', (t_{A}, t_{B})}$, applying equation~\eqref{eq:part1}, we find,
    \begin{align}
        \sum_{(t_{A}, t_{B}) \in V} L_{(t_{A}, t_{B})', (t_{A}, t_{B})} &= 2\sum_{t=0}^{T-2}\Big(\sum_{(t_{A}, t_{B}) \in V}\big(\braket{t_{A}'|-_{t}}\braket{-_{t}| t_{A}}(\delta_{t_{B},t_{B}'} - \braket{t_{B}| -_{t}}\braket{-_{t}|t_{B}'})\big) \nonumber \\
                                                                        &\quad + \sum_{(t_{A},t_{B})\in V}\big(\braket{t_{B}|-_{t}}\braket{-_{t}|t_{B}'}(\delta_{t_{A},t_{A}'} - \braket{t_{A}'|-_{t}}\braket{-_{t}| t_{A}})\big)\Big) \nonumber \\
                                                                        &= \frac{2}{\sqrt{2}}\sum_{t=0}^{T-2}\Big(\sum_{(t_{A}, t_{B})\in V}\big(\braket{t_{A}'|-_{t}}(\delta_{t,t_{A}} - \delta_{t+1,t_{A}})(\delta_{t_{B},t_{B}'} - \braket{t_{B}| -_{t}}\braket{-_{t}|t_{B}'})\big) \nonumber \\
                                                                        &\quad + \sum_{(t_{A},t_{B})\in V}\big(\braket{-_{t}|t_{B}'}(\delta_{t,t_{B}} - \delta_{t+1,t_{B}})(\delta_{t_{A},t_{A}'} - \braket{t_{A}'|-_{t}}\braket{-_{t}| t_{A}} )\big)\Big) .
    \end{align} 
    Note, $\sum_{t_{A}=0}^{T-1}(\delta_{t,t_{A}} - \delta_{t+1,t_{A}}) = \sum_{t_{B}=0}^{T-1}(\delta_{t,t_{B}} - \delta_{t+1,t_{B}}) =0$ for all $0 \leq t < T-1$, therefore
    \begin{equation}
        \label{eq:prop3}
        \sum_{(t_{A},t_{B}) \in V} L_{(t_{A}, t_{B})', (t_{A},t_{B})} = 0.
    \end{equation}

    Given that matrix $L$ satisfies equations~\eqref{eq:prop1},~\eqref{eq:prop2}, and~\eqref{eq:prop3} we can apply Lemma~\ref{lm:mattograph} of Appendix~\ref{app:B} implying $L$ is a Laplacian for a graph $G$ over vertices $V$. The statement of Lemma~\ref{lm:rotadistograph} follows.
\end{proof}

Lemma~\ref{lm:rotadistograph} allows us to connect various properties of the rotated propagation dissipator to those of its associated graph. Some examples of the graph $G$ defined in Lemma~\ref{lm:rotadistograph} are depicted in FIG~\ref{fig:G}.
\begin{figure}[ht!]
    \includesvg[width=0.7\linewidth]{figures/rotdisgraph.svg}
    \caption{Graphs associated with the rotated propagation dissipator for $T=3$ and $T=4$}
    \label{fig:G}
\end{figure}

\subsection{Characterization of Kernel}

Our strategy to characterize the kernel of the propagation dissipator involves characterizing the kernel of the rotated propagation disipator, then rotating the subspace back under the circuit rotation rotation unitary for the specific quantum circuit. To characterize the kernel of the rotated propagation dissipator our strategy is two fold. We first calculate the nullity of the rotated propagation dissipator. We subsequently construct a number of linearly independent steady-states equal to the nullity, therefore proving the kernel equal to the span of these steady-states.

The graphs depicted in FIG~\ref{fig:G} suggest that $G$ has two connected components. It is known that the nullity for a graph's Laplacian is equal to the number of connected components of its associated graph. Therefore FIG~\ref{fig:G} suggest that $\dim\ker{\tilde{\D}_{\prop}} = 2$ implying the nullity of the rotated propagation dissipator is also $2$. To construct the linearly-independent steady-states of the rotated propagation dissipator, we take both the pure and mixed history states of the quantum circuit we want to embed and rotate them under the associated circuit rotation. We then take as an ansatz the portion of these operators supported on only the clock register.

To proceed with our analysis we require the following statements regarding the graph associated to our rotated propagation dissipator.
\begin{lemma}
    \label{lm:mixcomp}
    Let $G$ be the graph defined in Lemma~\ref{lm:rotadistograph} and $G_{m}$ be the subgraph associated to vertices $V_{m} = \{(t,t)\}_{t=0}^{T-1}$, $G_{m}$ is a connected component of $G$.
\end{lemma}
\begin{proof}
    See Appendix~\ref{app:B}
\end{proof}

\begin{lemma}
    \label{lm:purecomp}
    Let $G$ be the graph defined in Lemma~\ref{lm:rotadistograph} and $G_{p}$ be the subgraph associated to vertices $V_{p} = V \setminus \{(t, t)\}_{t=0}^{T-1}$, $G_{p}$ is a connected component of $G$.
\end{lemma}
\begin{proof}
    See Appendix~\ref{app:B}
\end{proof}

The above two Lemmas imply the following corollary regarding graph $G$.

\begin{corollary}
    \label{cor:comps}
    Let $G$ be the graph defined in Lemma~\ref{lm:rotadistograph}, then $G$ has only two connected components.
\end{corollary}

Using our connection between the rotated propagation dissipator $\tilde{\D}_{\prop}$ and the graph $G$ associated to it, we can leverage the above statements regarding $G$ in order to prove results about the rotated propagation dissipator. In the statement which follows we give a full characterization of the rotated propagation dissipator's kernel and in the second we extend it to the propagation dissipator, proving Proposition~\ref{prop:propdisker}.

\begin{lemma}
    \label{lm:kerrotprop}
    Let $\tilde{\D}_{\prop}$ be the rotated propagation dissipator, then
    \begin{equation}
        \ker{\tilde{\D}_{\prop}} = \Span{\{\mathbb{I}_{C}, \ket{p}\bra{p}_{C}\}}
    \end{equation}
    where $\ket{p} = \frac{1}{\sqrt{T}}\sum_{t=0}^{T-1}\ket{t}_{C}$.
\end{lemma}
\begin{proof}
    Consider the action 
    \begin{equation}
        \tilde{\D}_{\prop}(\ket{p}\bra{p}_{C}) =  \sum_{t=0}^{T-2}\left(\ket{-_{t}}_{C}\braket{-_{t}|p}\braket{p|-_{t}}\bra{-_{t}}_{C} - \frac{1}{2}\left\{\ket{-_{t}}\bra{-_{t}}_{C}, \ket{p}\bra{p}_{C} \right\}\right),
    \end{equation}
    calculating $\braket{p|-_{t}}$ we find
    \begin{equation}
        \braket{p|-_{t}}  = \frac{1}{\sqrt{2}}\sum_{t'=0}^{T-1}(\delta_{t',t} - \delta_{t',t+1}) = 0 \text{ for all } t\in \{0, 1, \cdots, T-2\}, 
    \end{equation}
    thus $\tilde{\D}_{\prop}(\ket{p}\bra{p}_{C}) = 0$  implying $\ket{p}\bra{p}_{C} \in \ker{\tilde{\D}_{\prop}}$. Similarly,
    \begin{equation}
        \tilde{\D}_{\prop}(\mathbb{I}_{C}) =  \sum_{t=0}^{T-2}\left(\ket{-_{t}}_{C}\braket{-_{t}|-_{t}}\bra{-_{t}}_{C} - \frac{1}{2}\left\{\ket{-_{t}}\bra{-_{t}}_{C}, \mathbb{I}_{C} \right\}\right) = 0,
    \end{equation}
    therefore $\mathbb{I}_{C} \in \ker{\tilde{\D}_{\prop}}$. The set $\{\mathbb{I}_{C}, \ket{p}\bra{p}\}$ is linearly independent, thus $\dim\Span{\{\mathbb{I}_{C}, \ket{p}\bra{p}_{C}\}} = 2$, in addition $\Span{\{\mathbb{I}_{C}, \ket{p}\bra{p}_{C}\}} \subseteq \ker{\tilde{\D}_{\prop}}$. 

    Lemma~\ref{lm:rotadistograph} implies $\dim\ker{\tilde{\D}_{\prop}} = \dim\ker{L_{G}}$ where $L_{G}$ is the Laplacian associated to graph $G$ defined in Lemma~\ref{lm:rotadistograph}. The algebraic multiplicity of the zero eigenvalue of $L_{G}$ is equal to the $\dim\ker{L_{G}}$, thus applying Theorem 3.10 of \cite{marsden2013eigenvalues} implies $\dim\ker{L_{G}}$ is equal to the number of connected components of graph $G$. Thus we find 
    \begin{equation}
        \dim\ker{\tilde{\D}_{\prop}} = \dim\ker{L_{G}} = 2,
    \end{equation}
    following from Corollary~\ref{cor:comps}. Assume $\Span{\{\mathbb{I}_{C}, \ket{p}\bra{p}_{C}\}}$ is strictly a subspace of $\ker{\tilde{\D}_{\prop}}$. We require
    \begin{equation}
        \dim\Span{\{\mathbb{I}_{C}, \ket{p}\bra{p}_{C}\}} < \dim\ker{\tilde{\D}_{\prop}}
    \end{equation}
    but clearly $\dim\Span{\{\mathbb{I}_{C}, \ket{p}\bra{p}_{C}\}} = 2= \dim\ker{\tilde{\D}_{\prop}}$ which is a contradiction. Thus,
    \begin{equation}
        \Span{\{\mathbb{I}_{C}, \ket{p}\bra{p}_{C}\}} = \ker{\tilde{\D}_{\prop}}
    \end{equation}
    must hold.
\end{proof}
\begin{proposition*}[Kernel of Propagation Dissipator, Proposition~\ref{prop:propdisker}]
    Let $V = U_{T-1}U_{T-2}\cdots U_{0}$ be a quantum circuit acting on register $S$ and $\D_{\prop}^{(V)}$ be its propagation dissipator defined according to equation~\eqref{eq:DpropD} of Definition~\ref{def:propdiss}, then
    \begin{equation}
        \ker{\D_{\prop}^{(V)}} = \mathcal{U}^{(V)}\left[\mathcal{B}(\mathcal{H}_{S}) \otimes \Span{\left\{\mathbb{I}_{C}, \ket{p}\bra{p}_{C}\right\}}\right]
    \end{equation}
    where $\ket{p} = \frac{1}{\sqrt{T}}\sum_{t=0}^{T-1}\ket{t}$ and $\U^{(V)}$ is the circuit rotation associated to $V$.
\end{proposition*}
\begin{proof}
    Let $\U^{(V)}$ be the circuit rotation for quantum circuit $V$ acting on $\mathcal{H}_{MA}$ and let $\tilde{\D}_{\prop}$ be the rotated propagation dissipator. Lemma~\ref{lm:rotpropdef} implies
    \begin{equation}
        \U^{(V)\dagger} \circ \D^{(V)}_{\prop} \circ \U^{(V)} = \mathcal{I}_{S} \otimes \tilde{\D}_{\prop},
    \end{equation}
    which further implies 
    \begin{equation}
        \label{eq:rotatedker}
        \ker{\U^{(V)\dagger} \circ \D^{(V)}_{\prop} \circ \U^{(V)}} = \mathcal{B}(\mathcal{H}_{S}) \otimes \ker{\tilde{\D}_{\prop}} = \mathcal{B}(\mathcal{H}_{S}) \otimes \Span{\{\mathbb{I}_{C}, \ket{p}\bra{p}_{C}\}}
    \end{equation}
    where the last equality follows from Lemma~\ref{lm:kerrotprop}. The statement of Proposition~\ref{prop:propdisker} follows.
\end{proof}

\subsection{Analysis of Spectral Gap}

Our argument for lower-bounding the gap of our propagation dissipator starts with lower-bounding the gap of the rotated propagation dissipator, then extending that argument to the original propagation disispator. Similar to the previous section, our analysis lower-bounding the gap of the rotated propagation dissipator relies on our ability to map it to a graph Laplacian and having characterized the graphs connected components. In the statement which follows we provide a lower-bound on the gap of the rotated propagation dissipator.

\begin{lemma}
    \label{lm:gaprotprop}
    Let $\tilde{\D}_{\prop}$ be the rotated propagation dissipator, then
    \begin{equation}
        \gap{\tilde{\D}_{\prop}} \geq \frac{1}{2}\left(1 - \cos\left(\frac{\pi}{T(T-1)}\right)\right).
    \end{equation}
\end{lemma}
\begin{proof}
    From Lemma~\ref{lm:rotadistograph} we directly have
    \begin{equation}
        \gap{\tilde{\D}_{\prop}} = \frac{1}{4}\gap{L_{G}}
    \end{equation}
    where $L_{G}$ is the Laplacian of the graph $G$ defined in Lemma~\ref{lm:rotadistograph}. Applying part (vi) of Lemma 1.7 from \cite{chung1997spectral} as well as Corollary~\ref{cor:comps} implies $\Sp{L_{G}} = \Sp{L_{G_{m}}} \cup \Sp{L_{G_{p}}}$. Thus we can write,
    \begin{equation}
        \label{eq:graphcomps}
        \gap{\tilde{\D}_{\prop}} = \frac{1}{4}\inf\left\{\gap{L_{G_{m}}} , \gap{L_{G_{p}}} \right\}
    \end{equation}

    Let $\lambda_{i}(A)$ denote the $i$-th eigenvalue of Hermitian operator $A$, ordered according to $\lambda_{0}(A) \leq \lambda_{1}(A) \leq \cdots \lambda_{i}(A) \leq \cdots$. Lemmas~\ref{lm:mixcomp} and~\ref{lm:purecomp} imply subgraphs $G_{m}$ and $G_{p}$ are connected components. Applying part (iv) of Lemma 1.7 from \cite{chung1997spectral} allows us to simplify equation~\eqref{eq:graphcomps} as 
    \begin{equation}
        \gap{\tilde{\D}_{\prop}} = \frac{1}{4}\inf\left\{\lambda_{1}(L_{G_{m}}), \lambda_{1}(L_{G_{p}})\right\}.
    \end{equation}
    Applying Theorem 6.1 of \cite{zhang2011laplacian} allows us to lower-bound $\lambda_{1}(L_{G_{m}})$ and $\lambda_{1}(L_{G_{p}})$ since $G_{p}$ and $G_{m}$ are simple connected graphs. We find 
    \begin{equation}
        \gap{\tilde{\D}_{\prop}} \geq \frac{1}{2}\min\left\{\kappa'(G_{m})\left(1 - \cos\left(\frac{\pi}{|V_{m}|}\right)\right), \kappa'(G_{p})\left(1 - \cos\left(\frac{\pi}{|V_{p}|}\right)\right) \right\},
    \end{equation}
    where $V_{m}$ and $V_{p}$ denote the sets of vertices associated to graphs $G_{m}$ and $G_{p}$ respectively, and $\kappa'(G)$ denotes the edge-connectivity of graph $G$. From the definitions of the graphs $G_{m}$ and $G_{p}$ found in Lemmas~\ref{lm:mixcomp} and~\ref{lm:purecomp} we find 
    \begin{equation}
        \begin{cases}
            |V_{m}| = |\{(t,t)\}_{t=0}^{T-1}| = T   \\
            |V_{p}| = |\{(t_{A}, t_{B})\}_{t_{A}, t_{B}=0}^{T-1}| - |V_{m}| = T(T-1),
        \end{cases}
    \end{equation}
    Since both graphs $G_{m}$ and $G_{p}$ are connected according to Lemmas~\ref{lm:mixcomp} and~\ref{lm:purecomp} respectively we have
    \begin{equation}
        1 \leq \kappa'(G_{m}), \kappa'(G_{p}).
    \end{equation}
    Therefore we find
    \begin{equation}
        \gap{\tilde{\D}_{\prop}} \geq \frac{1}{2}\min\left\{\left(1 - \cos\left(\frac{\pi}{T}\right)\right), \left(1 - \cos\left(\frac{\pi}{T(T-1)}\right)\right) \right\}.
    \end{equation}
    For $T > 1$, we find 
    \begin{equation}
        \gap{\tilde{\D}_{\prop}} \geq \frac{1}{2}\left(1 - \cos\left(\frac{\pi}{T(T-1)}\right)\right)
    \end{equation}
    and the statement of Lemma~\ref{lm:gaprotprop} follows.
\end{proof}

With the our lower-bound on the rotated propagation dissipator proven we continue on to proving our lower-bound on the propagation dissipator given in the statement below.
\begin{proposition*}[Gap of Propagation Dissipator, Proposition~\ref{prop:propdissgap}]
    Let $V = U_{T-1}U_{T-2}\cdots U_{0}$ be a quantum circuit acting on register $S$ and $\D_{\prop}^{(V)}$ be its propagation dissipator defined according to equation~\eqref{eq:DpropD} of Definition~\ref{def:propdiss}, then
    \begin{equation}
        \gap{\D_{\prop}^{(V)}} \geq \frac{1}{8T^{4}}
    \end{equation}
\end{proposition*}
\begin{proof}
    Let $\U^{(V)}$ be the circuit rotation for $V$ and let $\tilde{\D}_{\prop}$ be the rotated propagation dissipator. Lemma~\ref{lm:rotpropdef} allows us to write
    \begin{align}
        \gap{\D_{\prop}^{(V)}} &= \gap{\U^{(V)} \circ \mathcal{I}_{S} \otimes \tilde{\D}_{\prop} \circ \U^{(V)\dagger}} = \gap{\tilde{\D}_{\prop}} \geq \frac{1}{2}\left(1 - \cos\left(\frac{\pi}{T(T-1)}\right)\right)
    \end{align}
    where the second equality follows from $\U^{(V)} \circ \bullet \circ \U^{(V)\dagger}$ being a similarity transformation and thus spectrum preserving and the final inequality from Lemma~\ref{lm:gaprotprop}. For $T > 1$ we can write
    \begin{align}
        \gap{\D_{\prop}^{(V)}}  &\geq \frac{1}{2}\left(1 - \cos\left(\frac{\pi}{T(T-1)}\right)\right) \geq \frac{1}{8T^{4}}
    \end{align}
    implying the statement.
\end{proof}

\subsection{Leakage}

With both the kernel and gap of our propagation dissipator analyzed we conclude this appendix by analyzing the leakage of our pure steady-state subspace out into the rest of the propagation dissipator's kernel. More specifically, in the statement which follows, we upper-bound the overlap between the subspace associated to the pure steady-states of the propagation dissipator and the rest of its kernel. This allows a way to quantify how much the span of our pure steady-states leaks out into the rest of the propagation dissipator's kernel.

\begin{lemma*}[Lemma~\ref{lm:DpropLeakage}]
    Let $V=U_{T-1}U_{T-2}\cdots U_{0}$ be a quantum circuit acting on register $S$ with $T > 2$ and $\mathcal{D}^{(V)}_{\prop}$ be its propagation dissipator defined according to equation~\eqref{eq:DpropD} of Definition~\ref{def:propdiss}. Define $\mathcal{H}_{\prop}$ according to equation~\eqref{eq:Hprop} of Definition~\ref{def:histsubspace} and 
    \begin{equation}
        \mathsf{M} = \left\{v_{1}\ket{h}\bra{\perp} + v_{2}\ket{\perp}\bra{h} + v_{3}\ket{h}\bra{\perp} \; | \; (\ket{h}, \ket{\perp}, v) \in \mathcal{H}_{\prop} \times \mathcal{H}_{\prop}^{\perp} \times \mathbb{C}^{3} \right\},
    \end{equation}
    then the overlap between $\ker{\D_{\prop}^{(V)}}$ and $\mathsf{M}$ is given by
    \begin{equation}
        \cos^{2}\theta = \sup_{\substack{A_{0} \in \ker{\D_{\prop}^{(V)}}, A \in \mathsf{M}\\ \braket{A_{0},A_{0}} = \braket{A , A } = 1}} |\braket{A_{0},A}|^{2} \leq \frac{1}{T-1}.
    \end{equation}
\end{lemma*}
\begin{proof}
    Let $\U^{(V)}$ be the circuit rotation associated to $V$. Calculating the overlap we find 
    \begin{align}
        \label{eq:cos}
        \cos^{2}\theta &= \sup_{\substack{A_{0} \in \ker{\D_{\prop}^{(V)}}, A \in \mathsf{M}\\ \braket{A_{0},A_{0}} = \braket{A , A } = 1}} |\braket{A_{0},A}|^{2} = \sup_{\substack{A_{0} \in \ker{\D_{\prop}^{(V)}}, A \in \mathsf{M}\\ \braket{A_{0},A_{0}} = \braket{A , A } = 1}} |\braket{\U^{(V)\dagger}(A_{0}),\U^{(V)\dagger}(A)}|^{2} \nonumber \\ \\
                       &= \sup_{\substack{A_{0} \in \U^{(V)\dagger}[\ker{\D_{\prop}^{(V)}}], A \in \U^{(V)\dagger}[\mathsf{M}]\\ \braket{A_{0},A_{0}} = \braket{A , A } = 1}}|\braket{A_{0},A}|^{2}.
    \end{align}
    From Proposition~\ref{prop:propdisker} we have
    \begin{equation}
        \U^{(V)\dagger}[\ker{\D_{\prop}^{(V)}}] = \mathcal{B}(\mathcal{H}_{S})\otimes \Span{\{\mathbb{I}_{C}, \ket{p}\bra{p}_{C}\}}
    \end{equation}
    where $\ket{p} = \frac{1}{\sqrt{T}}\sum_{t=0}^{T-1}\ket{t}$. Therefore for all $A_{0} \in \U^{(V)\dagger}[\ker{\D_{\prop}^{(V)}}]$ we have 
    \begin{equation}
        \label{eq:finalA0}
        A_{0} = B_{S} \otimes (\mathbb{I}_{C} - \ket{p}\bra{p}_{C}) + C_{S} \otimes \ket{p}\bra{p}_{C} = B_{S} \otimes P^{(\perp)}_{C} + C_{S} \otimes \ket{p}\bra{p}_{C},
    \end{equation}
    where we have defined the projection onto the orthogonal complement to the subspace $\ket{p}_{C}$ as $P^{(\perp)}$. For $A_{0} \in \U^{(V)\dagger}$ and $\braket{A_{0}, A_{0}} =1$ we find
    \begin{equation}
        \braket{A_{0}, A_{0}} = \braket{B, B}\braket{P^{(\perp)}, P^{(\perp)}} + \braket{C,C}\braket{\ket{p}\bra{p},\ket{p}\bra{p}} = \braket{B, B}(T - 1) + \braket{C, C}  = 1.
    \end{equation}
    which implies 
    \begin{equation}
        \label{eq:upperonB}
        \begin{cases}
            \braket{B, B} \leq 1 / (T - 1) \\
            \braket{C, C} \leq 1 
        \end{cases}.
    \end{equation}
    
    Similarly we have for all $A \in \U^{(V)\dagger}[\mathsf{M}]$, 
    \begin{align}
        A &= v_{1}\U^{(V)\dagger}(\ket{h}\bra{\perp}) + v_{2} \U^{(V)\dagger}(\ket{\perp}\bra{h}) + v_{3}\U^{(V)\dagger}(\ket{\perp}\bra{\perp}) \\
          &= v_{1}U^{(V)*}\ket{h}\bra{\perp}U^{(V)} + v_{2} U^{(V)*} \ket{\perp}\bra{h}U^{(V)} + v_{3}U^{(V)*}\ket{\perp}\bra{\perp}U^{(V)}
    \end{align}
    where $\ket{h} \in \mathcal{H}_{\prop}$ and $\ket{\perp} \in \mathcal{H}_{\prop}^{\perp}$. Since Lemma~\ref{lm:onlyhist} implies $\mathcal{H}_{\prop} = U[\mathcal{H}_{S} \otimes \ket{p}]$ we have $U^{(V)*}\ket{h} \in \mathcal{H}_{S} \otimes \ket{p}_{C}$. Thus we can write $A \in \U^{(V)\dagger}[\mathsf{M}]$ as 
    \begin{align}
        A = v_{1}(\ket{\phi}_{S}\otimes\ket{p}_{C})\bra{\perp}U^{(V)} + v_{2} U^{(V)*} \ket{\perp}(\bra{\phi}_{S}\otimes \bra{p}_{C}) + v_{3}U^{(V)*}\ket{\perp}\bra{\perp}U^{(V)},
    \end{align}
    for some $\ket{\phi} \in \mathcal{H}_{S}$. To begin analyzing $U^{(V)*}\ket{\perp} \in U^{(V)*}[\mathcal{H}_{\prop}^{\perp}] $ note that since $U^{(V)*}$ is a unitary and $\mathcal{H}_{\prop}^{\perp}$ is the orthogonal complement of $\mathcal{H}_{\prop}$ we have 
    \begin{equation}
        U^{(V)*}[\mathcal{H}_{\prop}^{\perp}] = U^{(V)*}[\mathcal{H}_{\prop}]^{\perp} = (\mathcal{H}_{S} \otimes \ket{p}_{C})^{\perp} = \mathcal{H}_{S} \otimes \ket{p}_{C}^{\perp},
    \end{equation}
    where $\ket{p}^{\perp}_{C}$ denotes the subspace of $\mathcal{H}_{C}$ which is the orthogonal complement of the subspace $\ket{p}_{C}$. Therefore we can write $A \in \U^{(V)\dagger}[\mathsf{M}]$ as 
    \begin{equation}
        \label{eq:finalA}
        A = v_{1}(\ket{\phi}_{S}\otimes\ket{p}_{C})\bra{\perp'} + v_{2} \ket{\perp'}(\bra{\phi}_{S}\otimes \bra{p}_{C}) + v_{3}\ket{\perp'}\bra{\perp'},
    \end{equation}
    where $\ket{\perp'} = U^{(V)*}\ket{\perp} \in \mathcal{H}_{S} \otimes \ket{p}_{C}^{\perp}$.

    Moving back to calculating $\cos^{2}\theta$, Applying equations~\eqref{eq:finalA0} and~\eqref{eq:finalA}, for all $A_{0} \in \U^{(V)\dagger}[\ker{\D_{\prop}^{(V)}}]$ and $A \in \U^{(V)\dagger}[\mathsf{M}]$ such that $\braket{A_{0}, A_{0}} = \braket{A, A} = 1$ we have 
    \begin{equation}
        |\braket{A_{0},A}|^{2} =  |\braket{B_{S}\otimes P_{C}^{(\perp)} + C_{S}\otimes\ket{p}\bra{p}_{C}, v_{1}(\ket{\phi}_{S}\otimes\ket{p}_{C})\bra{\perp'} + v_{2} \ket{\perp'}(\bra{\phi}_{S}\otimes \bra{p}_{C}) + v_{3}\ket{\perp'}\bra{\perp'}}|^{2}
    \end{equation}
    which gives us six terms inside the absolute value. Calculating four of these six terms gives us
    \begin{equation}
        \begin{cases}
            \braket{B_{S}\otimes P_{C}^{(\perp)}, (\ket{\phi}_{S}\otimes\ket{p}_{C})\bra{\perp'}} = \Tr((B_{S}^{*}\ket{\phi}_{S} \otimes P_{C}^{(\perp)}\ket{p}) \bra{\perp'}) = 0 \\
            \braket{B_{S}\otimes P_{C}^{(\perp)}, \ket{\perp'}(\bra{\phi}_{S} \otimes \bra{p}_{C})} =  \Tr((B_{S}^{*} \otimes P_{C}^{(\perp)})\ket{\perp'}(\bra{\phi}_{S} \otimes \bra{p}_{C})) = \Tr(\ket{\perp'}(\bra{\phi}_{S}B_{S}^{*} \otimes \bra{p}_{C}P_{C}^{(\perp)})) = 0 \\
            \braket{C_{S} \otimes \ket{p}\bra{p}_{C}, (\ket{\phi}_{S} \otimes \ket{p}_{C})\bra{\perp'}} = \bra{\perp'}((C^{*}\ket{\phi})_{S}\otimes \ket{p}_{C}) = \bra{\perp}U^{(V)}(\ket{C^{*}\phi}_{S} \otimes \ket{p}_{C}) = \braket{\perp|\eta_{C^{*}\phi}} = 0  \\
        \braket{C_{S} \otimes \ket{p}\bra{p}, \ket{\perp'}(\bra{\phi}_{S}\otimes\bra{p}_{C})} = ((\bra{\phi}C^{*})_{S} \otimes \bra{p}_{C}) \ket{\perp'} = (\bra{C\phi}_{S} \otimes \bra{p}_{C})U^{(V)*}\ket{\perp} = \braket{\eta_{C\phi}|\perp} = 0.
        \end{cases}
    \end{equation}
    Therefore we can write 
    \begin{equation}
        |\braket{A_{0},A}|^{2} =  |v_{3}|^{2}|\braket{B_{S}\otimes P_{C}^{(\perp)} + C_{S}\otimes\ket{p}\bra{p}_{C}, \ket{\perp'}\bra{\perp'}}|^{2}.
    \end{equation}
    Calculating $\braket{C_{S}\otimes\ket{p}\bra{p}_{C}, \ket{\perp'}\bra{\perp'}}$ we find
    \begin{equation}
        \braket{C_{S}\otimes\ket{p}\bra{p}_{C}, \ket{\perp'}\bra{\perp'}} = \bra{\perp'}(C_{S} \otimes \ket{p}\bra{p}_{C})\ket{\perp'} = \sum_{ij}c_{ij}\bra{\perp'}(C\ket{i})_{S} \otimes (\ket{p}\braket{p|j})_{C} = 0,
    \end{equation}
    with the last equality following from $\ket{\perp'} \in \mathcal{H}_{S} \otimes \ket{p}_{C}^{\perp}$ implying $\ket{j} \in \ket{p}_{C}^{\perp}$. Thus we find,
    \begin{align}
        |\braket{A_{0},A}|^{2} &=  |v_{3}|^{2}|\braket{B_{S}\otimes P_{C}^{(\perp)}, \ket{\perp'}\bra{\perp'}}|^{2} \leq |\Tr(B^{*}_{S} \otimes P_{C}^{(\perp)} \ket{\perp'}\bra{\perp'})|^{2} \nonumber \\
                               &\leq |\sum_{ij}c_{ij}\Tr((B^{*}\ket{i})_{S} \otimes (P^{(\perp)}\ket{j})_{C}\bra{\perp'})|^{2} =  |\sum_{ij}c_{ij}\Tr((B^{*}\ket{i})_{S} \otimes \ket{j}_{C}\bra{\perp'})|^{2} = |\Tr(B^{*}_{S} \otimes \mathbb{I}_{C} \ket{\perp'}\bra{\perp'})|^{2} \nonumber \\
                               & \leq |\Tr\left(B^{*} \Tr_{C}(\ket{\perp'}\bra{\perp'})\right)|^{2} = |\braket{B, \Tr_{C}(\ket{\perp'}\bra{\perp'})}|^{2}.
    \end{align}
    Since $\ket{\perp'}\bra{\perp'}$ is a quantum state, $\Tr_{C}(\ket{\perp'}\bra{\perp'}) = \sigma$ is also a quantum state, thus
    \begin{equation}
        |\braket{A_{0},A}|^{2} \leq \braket{B, B} \braket{\sigma, \sigma} \leq \braket{B, B} \leq \frac{1}{T-1}.
    \end{equation}
    where the final inequality follows from equation~\ref{eq:upperonB}. Combining the above equation with equation~\eqref{eq:cos} implies the statement.
\end{proof}

\section{Analysis of the Graph associated with the Rotated Propagation Dissipator}
\label{app:B}

Define the set of vertices $V = \{0, 1, \cdots, T-1\}^{\times 2}$ and map $P:V\mapsto \mathcal{B}(\mathcal{H}_{C})$ as $P_{(t_{A}, t_{B})} = \ket{t_{A}}\bra{t_{B}}$. As described in Lemma~\ref{lm:rotadistograph} of Appendix~\ref{app:A}, the rotated propagation dissipator expanded in the basis $\{P_{v}\}_{v \in V}$ is a Laplacian for a graph $G$ over vertices $V$. Throughout Appendix~\ref{app:A} we leveraged various properties about this graph $G$ in-order to prove statements about the rotated propagation dissipator. In this Appendix we provide an analysis of the graph $G$ defined in Lemma~\ref{lm:rotadistograph}, providing proof of the various statements used throughout Appendix~\ref{app:A}.

Before beginning our analysis it is useful to revisit some common definitions related to graph theory, for reference see \cite{chung1997spectral}. A \emph{graph} $G$ can be defined as a tuple $(V, E)$ where $V$ contains the graph's vertices and $E \subseteq \binom{V}{2}$ its edges. Set $\{v_{1}, v_{2}\}$ is in $E$ if and only if there exists an edge connecting vertices $v_{1}$ and $v_{2}$. The degree of a vertex $v \in V$, denoted $\deg(v)$ is equal to the number of edges connected to $v$.

The \emph{Laplacian} for a graph $G$ is the matrix associated to the elements
\begin{equation}
    [L_{G}]_{v_{1}, v_{2}} =  \begin{cases}
        \deg(v_{1}) \text{ if } v_{1} = v_{2} \\
        -1 \text{ if } v_{1} \neq v_{2} \text{ and } \{v_{1}, v_{2}\} \in E \\
        0 \text{ otherwise }
    \end{cases}
\end{equation}

From the above definition of a Laplacian we can prove the following Lemma regarding when a matrix happens to be the Laplacian associated to a graph.
\begin{lemma}
    \label{lm:mattograph}
    Let $A$ be an $n \times n$ matrix If
    \begin{enumerate}
        \item $A = A^{\intercal}$
        \item for all $i \neq j$, matrix elements $[A]_{ij} \in \{-1, 0\}$
        \item $\sum_{j=1}^{n}[A]_{ij} = 0$ 
    \end{enumerate}
    then there exists a 
    graph $G$ over $n$ vertices such that the Laplacian for graph $G$ is matrix $A$.
\end{lemma}
\begin{proof}
    We first build the graph $G$ associated to matrix $A$, then show the Laplacian for graph $G$ is equal to $A$. 

    Let $V = [n]$ be a set of vertices and define the set of edges $E = \{\{i, j\} \in \binom{V}{2} \; | \; A_{ij} = -1 \}$. Consider the Laplacian associated to graph $G = (V, E)$ denoted $L_{G}$. Since $|V| = n$, $L_{G}$ is an $n \times n$ matrix. Let $i, j \in V$ with $i \neq j$. Since $L_{G}$ is a Laplacian, $[L_{G}]_{ij} \in \{-1, 0\}$. If $[L_{G}]_{ij} = -1$ then $\{i, j\} \in E$ implying $A_{ij} = -1$ therefore $[L_{G}]_{ij} = A_{ij}$. In the other case $[L_{G}]_{ij} = 0$ implying $\{i, j\} \notin E$ implying  $A_{ij} \neq -1$. Since $i \neq j$, we must have $A_{ij}=0$ implying $[L_{G}]_{ij} = A_{ij}$ as well. Thus for $i \neq j$,
    \begin{equation}
        \label{eq:offdiag}
        [L_{G}]_{ij} = A_{ij}.
    \end{equation}
    Let $i \in V$. Since $L_{G}$ is a Laplacian $\sum_{j=1}^{n}[L_{G}]_{ij} = 0$ thus we have 
    \begin{equation}
        \sum_{j=1}^{n}[L_{G}]_{ij} = \sum_{j=1}^{n}A_{ij} 
    \end{equation}
    which implies 
    \begin{equation}
        [L_{G}]_{ii} + \sum_{j\neq i}^{n} [L_{G}]_{ij} = A_{ii} + \sum_{j\neq i}^{n} A_{ij}.
    \end{equation}
    Using equation~\eqref{eq:offdiag}, we find 
    \begin{equation}
        [L_{G}]_{ii} + \sum_{j\neq i}^{n} A_{ij} = A_{ii} + \sum_{j\neq i}^{n} A_{ij}.
    \end{equation}
    which implies $[L_{G}]_{ii} = A_{ii}$. Thus for all $i, j \in V$
    \begin{equation}
        [L_{G}]_{ij} = A_{ij}
    \end{equation}
    implying the Laplacian for graph $G$ is the $n \times n$ matrix $A$.
\end{proof}

Some other definitions that will be used include that of a subgraph, connectedness, and a component. Formally, let $G$ be a graph with vertices $V$ and edges $E$. If $V_{s} \subseteq V$, then the \emph{subgraph} associated to $V_{s}$ is the graph given by vertices $V_{s}$ and edges $\binom{V}{2} \cap E$. A \emph{walk} in $G$ connecting vertices $v_{s}$ and $v_{f}$ is a finite sequence of edges $\{e_{i}\}_{i=1}^{N} \subseteq E$ such that $e_{1}$ contains $v_{s}$, $e_{N}$ contains $v_{f}$, and $e_{i}$ and $e_{i+1}$ share a vertex. Similarly, a \emph{path} is a walk that never revists an edge or vertex. Graph $G$ is called \emph{connected} if and only if for any pair of vertices $v_{s}, v_{f}\in V$ there exists a walk in $G$ connecting $v_{s}$ to $v_{f}$. Conversely, let $G_{s}$ be the subgraph associated to vertices $V_{s} \subseteq V$, $G_{s}$ is called a \emph{component} of $G$ if and only if for all vertices $v_{s} \in V_{s}$ and $v_{o} \in V \setminus V_{S}$  there does not exists a walk in $G$ connecting vertices $v_{s}$ and $v_{o}$, ie $G_{s}$ is disconnected from the rest of the graph. A subgraph is called a \emph{connected component} if it is both a component of $G$ and a connected subgraph of $G$. In-addition, if $G = (V, E)$ is a connected graph, the \emph{edge-connectivity} of $G$ is the size of the smallest set $E_{c} \subseteq E$ such that the graph $(V, E \setminus E_{c})$ is not connected.

With these common definitions reviewed, we proceed to proving the following Lemmas characterizing the graph $G$ associated with the rotated propagation dissipator. Examples of $G$ for various values of $T$ can be found in FIG~\ref{fig:G} of Appendix~\ref{app:A}.

We begin by showing the ``mixed component" of $G$ depicted below, is indeed a connected component. We call this component ``mixed" because the portion of the rotated propagation dissipator's kernel generated by $\mathbb{I}_{C}$ is associated to the kernel of this components Laplacian. Since $\mathbb{I}_{C}$ generates the family of mixed history states when rotated back under our circuit rotation unitary discussed in Appendix~\ref{app:A}, this component of the $G$ gives us the mixed history states associated to our propagation dissipator. 
\begin{figure}[ht!]
    \label{fig:mixedcomp}
    \includesvg[width=0.6\linewidth]{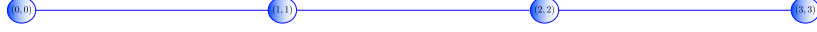}
    \caption{Mixed component of the graph associated to the rotated propagation dissipator for $T=4$}
\end{figure}

\begin{lemma*}[Lemma~\ref{lm:mixcomp} of Appendix~\ref{app:A}]
    Let $G$ be the graph defined in Lemma~\ref{lm:rotadistograph} and $G_{m}$ be the subgraph associated to $V_{m} = \{(t,t)\}_{t=0}^{T-1}$, $G_{m}$ is a connected component of $G$.
\end{lemma*}
\begin{proof}
    To show that $G_{m}$ is a component of $G$ it suffices to show that for all $v_{m} \in V_{m}$ and $v_{p} \in V \setminus V_{m}$ there does not exist an edge $\{v_{m}, v_{p}\} \in E$ where $E$ is the set of edges of $G$. If we let $L_{G}$ denote the Laplacian asssociated to $G$ then the previous statement is true if and only if $L_{G}$ block diagonalizes into $L_{G_{m}}$ and $L_{G_{p}}$ where $G_{m}$ and $G_{p}$ denote the subgraphs over vertices $V_{m}$ and $V \setminus V_{m}$ respectively. 

    Lemma~\ref{lm:rotadistograph} defines $L_{G}$ as $-4\tilde{\D}_{\prop}$ expanded in the bases $\{P_{v}\}_{v \in V}$ where $P_{v} = \ket{t_{A}}\bra{t_{B}}$ for $v = (t_{A}, t_{B})$. Thus $L_{G_{m}}$ and $L_{G_{p}}$ can be generated by expanding $-4\tilde{\D}_{\prop}$ in the subsets $\{P_{v}\}_{v \in V_{m}}$ and $\{P_{v}\}_{v\in V \setminus V_{m}}$. Defining $\Pi_{m}(\bullet) = \sum_{t=0}^{T-1}\ket{t}\bra{t} \bullet \ket{t}\bra{t}$ as the projection map into the subspace $\Span{\{P_{v}\}_{v \in V_{m}}}$ we can define $L_{G_{m}}$ as $-4\mathcal{P}_{m}\circ\tilde{\D}_{\prop}\circ\mathcal{P}_{m}$ expanded in the basis $\{P_{v}\}_{v \in V_{m}}$ and similarly $L_{G_{p}}$ as $-4(\mathcal{I} - \mathcal{P}_{m})\circ\tilde{\D}_{\prop}\circ(\mathcal{I} - \mathcal{P}_{m})$ expanded in the basis $\{P_{v}\}_{v \in V \setminus V_{m}}$. The condition that $L_{G}$ block diagonalizes into $L_{G_{m}}$ and $L_{G_{p}}$ becomes equivalent to
    \begin{equation}
        \label{eq:condition}
        \Pi_{m} \circ (-4\tilde{\D}_{\prop}) \circ (\mathcal{I} - \Pi_{m}) = 0,
    \end{equation}
    implying the off-diagonal blocks of $-4\tilde{\D}_{\prop}$ expanded in the basis $\{P_{v}\}_{v\in V}$ vanish.

    Calculating $\Pi_{m} \circ \tilde{\D}_{\prop}$ we find,  
    \begin{align} 
        \Pi_{m} \circ \tilde{\D}_{\prop}(\bullet) &= \sum_{t=0}^{T-2}\left(\Pi_{m}(\ket{-_{t}}\bra{-_{t}}\bullet\ket{-_{t}}\bra{-_{t}}) - \frac{1}{2}\Pi_{m}\left(\{\ket{-_{t}}\bra{-_{t}}, \bullet\}\right)\right) \nonumber \\
                                                  &= \frac{1}{2}\sum_{t=0}^{T-2}\sum_{t'=0}^{T-1}\Big((\delta_{t',t} - \delta_{t',t+1})^{2}\ket{t'}\bra{-_{t}}\bullet\ket{-_{t}}\bra{t'} - \frac{1}{\sqrt{2}}(\delta_{t',t} - \delta_{t',t+1})\ket{t'}\bra{-_{t}} \bullet \ket{t'}\bra{t'} \nonumber \\
                                                  & \quad - \frac{1}{\sqrt{2}}(\delta_{t',t} - \delta_{t',t+1})\ket{t'}\bra{t'}\bullet\ket{-_{t}}\bra{t'}\Big) \nonumber \\ 
                                                  &= \frac{1}{2}\sum_{t=0}^{T-2}\Big(\ket{t}\bra{-_{t}}\bullet\ket{-_{t}}\bra{t} + \ket{t+1}\bra{-_{t}}\bullet\ket{-_{t}}\bra{t+1} - \frac{1}{\sqrt{2}}\ket{t}\bra{-_{t}} \bullet \ket{t}\bra{t} + \frac{1}{\sqrt{2}} \ket{t+1}\bra{-_{t}} \bullet \ket{t+1}\bra{t+1}  \nonumber\\
                                                  & \quad - \frac{1}{\sqrt{2}}\ket{t}\bra{t}\bullet\ket{-_{t}}\bra{t} + \frac{1}{\sqrt{2}} \ket{t+1}\bra{t+1}\bullet\ket{-_{t}}\bra{t+1}\Big) \nonumber \\
                                                  &= \frac{1}{4}\sum_{t=0}^{T-2}\Big(-\ket{t}\bra{t}\bullet\ket{t+1}\bra{t} + \ket{t}\bra{t+1}\bullet\ket{t+1}\bra{t} + \ket{t+1}\bra{t}\bullet\ket{t}\bra{t+1} - \ket{t+1}\bra{t+1}\bullet\ket{t}\bra{t+1} \nonumber \\ 
                                                  & \quad - \ket{t}\bra{t}\bullet\ket{t}\bra{t} + \ket{t}\bra{t}\bullet\ket{t+1}\bra{t} + \ket{t+1}\bra{t+1}\bullet\ket{t}\bra{t+1} - \ket{t+1}\bra{t+1}\bullet\ket{t+1}\bra{t+1}\Big) \nonumber \\
                                                  &= \frac{1}{4}\sum_{t=0}^{T-2}\Big(\ket{t}\bra{t+1}\bullet\ket{t+1}\bra{t} + \ket{t+1}\bra{t}\bullet\ket{t}\bra{t+1} \nonumber \\
                                                  &\quad - \ket{t}\bra{t}\bullet\ket{t}\bra{t} - \ket{t+1}\bra{t+1}\bullet\ket{t+1}\bra{t+1}\Big).
\end{align}
    Now considering $\Pi_{m} \circ \tilde{\D}_{\prop} \circ \Pi_{m}$ we find, 
    \begin{align}
        \Pi_{m} \circ \tilde{\D}_{\prop}\circ\Pi_{m}(\bullet) &= \frac{1}{4}\sum_{t=0}^{T-1}\Big(\ket{t}\bra{t+1}\Pi_{m}(\bullet)\ket{t+1}\bra{t} + \ket{t+1}\bra{t}\Pi_{m}(\bullet)\ket{t}\bra{t+1} \nonumber \\
                                                              & \quad - \ket{t}\bra{t}\Pi_{m}(\bullet)\ket{t}\bra{t} - \ket{t+1}\bra{t+1}\Pi_{m}(\bullet)\ket{t+1}\bra{t+1}\Big) \nonumber \\
    &= \frac{1}{4}\sum_{t=0}^{T-1}\sum_{t'=0}^{T}\Big(\ket{t}\braket{t+1|t'}\bra{t'}\bullet \ket{t'}\braket{t'|t+1}\bra{t} + \ket{t+1}\braket{t|t'}\bra{t'}\bullet\ket{t'}\braket{t'|t}\bra{t+1} \nonumber \\
    &- \ket{t}\braket{t|t'}\bra{t'}\bullet\ket{t'}\braket{t'|t}\bra{t} - \ket{t+1}\braket{t+1|t'}\bra{t'}\bullet\ket{t'}\braket{t'|t+1}\bra{t+1}\Big) \nonumber \\
    &= \frac{1}{4}\sum_{t=0}^{T-1}\Big(\ket{t}\bra{t+1}\bullet \ket{t+1}\bra{t} + \ket{t+1}\bra{t}\bullet\ket{t}\bra{t+1} \nonumber \\
    & \quad - \ket{t}\bra{t}\bullet\ket{t}\bra{t} - \ket{t+1}\bra{t+1}\bullet\ket{t+1}\bra{t+1}\Big) = \Pi \circ \tilde{\mathcal{D}}_{\prop}(\bullet).
    \end{align}
    Therefore $\Pi_{m} \circ \tilde{\D}_{\prop} \circ (\mathcal{I} - \Pi_{m}) = 0$ implying equation~\eqref{eq:condition}, thus $L_{G}$ block diagonalizes into $L_{G_{m}}$ and $L_{G_{p}}$ implying both $G_{m}$ and $G_{p}$ are components of graph $G$.

    To show $G_{m}$ is connected, we consider a family of subgraphs $\{G_{m}^{(t)}\}_{t=1}^{T-1}$ where subgraph $G_{m}^{(t)}$ contains the set of vertices $\{(t',t')\}_{t'=0}^{t} \subseteq V_{m}$. Note, by definition $G_{m}^{T-1} = G_{m}$. To show $G_{m}$ is connected we proceed by induction. Consider the base case of $G_{m}^{(1)}$ containing vertices $\{(0,0), (1,1)\}$. $G_{m}^{(1)}$ is connected if there exists an edge connecting $(0,0)$ to $(1,1)$. For vertices $(0, 0)$ and $(1, 1)$ we find,
    \begin{equation}
        [L_{G}]_{(0,0), (1,1)} = -4\braket{P_{(0,0)}, \tilde{\D}_{\prop}(P_{(1,1)})}_{HS} = \braket{\ket{0}\bra{0}, - \ket{0}\bra{0} -\ket{2}\bra{2} + 2 \ket{1}\bra{1} }_{HS} = -1.
    \end{equation} 
    Thus there exists an edge connecting vertices $(0,0)$ to $(1,1)$ implying $G_{m}^{(1)}$ is connected. Assume induction hypothesis: $G_{m}^{(t)}$ is connected for all $t \in \{0, 1, \cdots, T-1\}$, and consider subgraph $G_{m}^{(t+1)}$. By definition $G_{m}^{(t)}$ is a subgraph of $G_{m}^{(t+1)}$ with $G_{m}^{(t+1)}$ containing all the vertices of $G_{m}^{(t)}$ along with vertex $(t+1,t+1)$. Thus assuming our induction hypothesis, $G_{m}^{(t+1)}$ is connected if there exists a vertex in $G_{m}^{(t)}$ with an edge in $G_{m}^{(t+1)}$ connecting it to vertex $(t+1,t+1)$.  To show this consider vertices $(t, t)$, associated to subgraph $G_{m}^{(t)}$, and $(t+1, t+1)$, associated to $G_{m}^{(t+1)}$, we find 
    \begin{equation}
        [L_{G}]_{(t,t), (t+1,t+1)} = -4\braket{P_{(t,t)}, \circ \tilde{\D}_{\prop}(P_{(t+1,t+1)})} = \braket{\ket{t}\bra{t}, -\ket{t}\bra{t} - \ket{t+2}\bra{t+2} + 2\ket{t+1}\bra{t+1}} = -1.
    \end{equation}
    Thus there exists an edge connecting vertices $(t, t)$ to $(t+1, t+1)$ implying $G_{m}^{(t+1)}$ is connected. Therefore for all $t \in \{0, 1, \cdots, T-1\}$ the subgraphs $G_{m}^{(t)}$ are connected implying $G_{m}^{(T-1)} = G_{m}$ is connected as well.
\end{proof}

 Before proceeding in our analysis of graph $G$, it is useful to introduce an ordering on the vertices of graph. The vertices of graph $G$ are given by tuples $(t_{A}, t_{B}) \in V$ for $t_{A}, t_{B} \in \{0, 1, \cdots, T-1\}$, we can order the tuples as such,
\begin{equation}
    (i+1, i) \rightarrow (i+1, i) \rightarrow \cdots \rightarrow (T-1, i) \rightarrow (0, i+1) \rightarrow (1, i+1) \rightarrow \cdots .
\end{equation}
We can then define intervals over $V$. For $v_{1}, v_{2} \in V$, the set $(v_{1}, v_{2})$ contains all vertices in-between $v_{1}$ and $v_{2}$. Similarly, the set $[v_{1}, v_{2}]$ contains all vertices in-between $v_{1}$ and $v_{2}$ including $v_{1}$ and $v_{2}$.

Similar to the ``mixed component" of $G$ described above, the ``pure component" of $G$, defined as the subgraph $G_{p}$ associated to the vertices $V_{p} = V \setminus V_{m}$, gives rise to the pure history states that lie in the kernel of our propagation dissipator. More specifically, the portion of the rotated propagation dissipator's kernel generated by $\ket{p}_{C} = \frac{1}{\sqrt{T}}\sum_{t=0}^{T-1}\ket{t}_{C}$ is associated to the kernel of the Laplacian associated to $G_{p}$. When the  states with a clock register initialized as $\ket{p}_{C}$ are rotated back using our circuit rotation unitary defined in Appendix~\ref{app:A}, they become pure history states.

The proof for the connectedness of $G_{p}$ is more involved than that for $G_{m}$. Thus it is best to sketch out the argument with the help of the figure provide below.
\begin{figure}[ht!]
    \includesvg[width=0.6\linewidth]{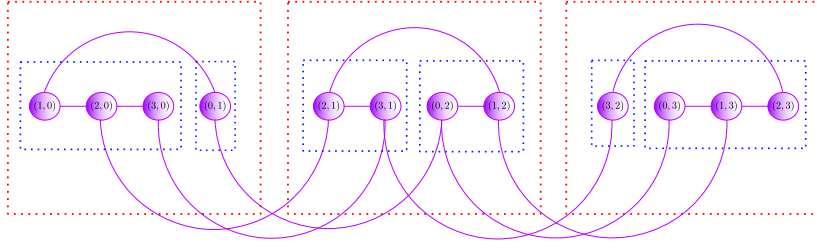}
    \caption{Pure component of the graph associated to the rotated propagation dissipator for $T=4$}
    \label{fig:purecomp}
\end{figure}

The graph $G_{p}$ depicted above can be broken-up into subgraphs $G_{p}^{(1)}$, $G_{p}^{(2)}$ and $G_{p}^{(3)}$ boxed in red in FIG~\ref{fig:purecomp}. For arbitrary $T$, the $t$-th subgraph $G_{p}^{(t)}$ is defined as the subgraph associated to the vertices $\left((t-1, t-1), (t, t)\right) \subseteq V_{p}$ where $t \in \{1, 2, \cdots, T-1\}$. As depicted in FIG~\ref{fig:purecomp}, these subgraphs are connected between vertices $(0, t)$, associated to $G_{p}^{(t)}$, and $(0, t+1)$, associated to $G_{p}^{(t+1)}$. Thus, the problem reduces to showing the graphs $G_{p}^{(t)}$, boxed in red in FIG~\ref{fig:purecomp}, are connected.

To show each $G_{p}^{(t)}$ is connected we can further split each $G_{p}^{(t)}$ into subgraphs $G_{A}^{(t)}$ and $G_{B}^{(t)}$, boxed in blue in FIG~\ref{fig:purecomp}. For arbitrary $T$, we define $G_{A}^{(t)}$ as the subgraph associated to vertices $((t-1, t-1), (T-1, t-1)]$ and $G_{B}^{(t)}$ as the subgraph associated to vertices $[(0, t), (t, t))$. As suggested in FIG~\ref{fig:purecomp} these subgraphs are connected through vertices $(t-1, t)$, associated to $G_{A}^{(t)}$, and $(t, t-1)$, associated to $G_{B}^{(t)}$, with the edges colored in yellow. This allows us to reduce the problem to showing that each $G_{p}^{A}$ and $G_{p}^{B}$ is connected which becomes more tractable.

Starting with subgraphs $G_{A}^{(t)}$ and $G_{B}^{(t)}$ defined above, we show the following Lemmas.
\begin{lemma}
    \label{lm:GA}
    For all $t \in \{1, 2, \cdots, T-1\}$, the subgraph $G_{A}^{(t)}$ is connected.
\end{lemma}
\begin{proof}
    For all $t \in \{1, 2, \cdots, T-1\}$, let $G_{A}^{(t)}$ denote the subgraph associated to the set of vertices $((t-1, t-1), (T-1, t-1)]$. Consider the family of subgraphs $\{[G_{A}^{(t)}]_{i}\}_{i=t}^{T-1}$ where $[G_{A}^{(t)}]_{i}$ denotes the subgraph of $G_{A}^{(t)}$ associated to vertices $((t-1, t-1), (i, t-1)]$. Note, by definition $G_{A}^{(t)} = [G_{A}^{(t)}]_{T-1}$. To prove $G_{A}^{(t)}$ is connected we proceed by induction. Consider the base case of $[G_{A}^{(t)}]_{t}$ associated to vertices $((t-1, t-1), (t, t-1)] = \{(t, t-1)\}$. By definition of a singleton graph $[G_{A}^{(t)}]_{t}$ is a connected. Assume induction hypothesis: $[G_{A}^{(t)}]_{i}$ is connected for all $i \in \{t, t+1,\cdots, T-1\}$, and consider subgraph $[G_{A}^{(t)}]_{i+1}$. By definition $[G_{A}^{(t)}]_{i}$ is a subgraph of $[G_{A}^{(t)}]_{i+1}$ containing all the vertices associated to $[G_{A}^{(t)}]_{i}$ in addition to $(i+1, t-1)$. Assuming our induction hypothesis, $[G_{A}^{(t)}]_{i+1}$ is connected if there exists a vertex associated to $[G_{A}^{(t)}]_{i}$ with an edge associated to $[G_{A}^{(t)}]_{i+1}$ connecting it to vertex $(i+1, t-1)$. To show this consider vertices $(i, t-1)$, associated to $[G_{A}^{(t)}]_{i}$, and $(i+1, t-1)$, associated to $[G_{A}^{(t)}]_{i}$, we find 
    \begin{align}
        \label{eq:GA}
        \left[L_{[G_{A}^{(t)}]_{i+1}}\right]_{(i, t-1), (i+1, t-1)} &= -4\braket{\ket{i}\bra{t-1}, \tilde{\D}_{\prop}(\ket{i+1}\bra{t-1})} = \sum_{t'=0}^{T-2}\left(\braket{i|-_{t'}}\braket{-_{t'}|i+1}\left(2 - 4|\braket{t-1|-_{t'}}|^{2}\right)\right) \nonumber \\
                                                                  &= \sum_{t'=0}^{T-2}\left((\delta_{i,t'} - \delta_{i, t'+1})(\delta_{t', i+1} - \delta_{t'+1,i+1} )\left(1- |\delta_{t-1,t'} - \delta_{t-1,t'+1}|^{2}\right)\right) \nonumber \\
                                                                  &= \sum_{t'=0}^{T-2}\left((\delta_{i,t'}\delta_{t',i+1} - \delta_{i,t'}\delta_{t'+1,i+1} - \delta_{i, t'+1}\delta_{t',i+1} + \delta_{i, t'+1}\delta_{t'+1,i+1}) \left(1- \delta_{t-1,t'} - \delta_{t-1,t'+1}\right)\right) \nonumber  \\
                                                                  &= -1 + \delta_{t-1,i} + \delta_{t-1, i+1} = -1,
    \end{align}
    where the last equality follows from $i \geq t$. From the definition of the Laplacian for $[G_{A}^{(t)}]_{i+1}$, equation~\eqref{eq:GA} implies there exists an edge associated to $[G_{A}^{(t)}]_{i+1}$ connecting vertex $(i, t-1)$ to $(i+1,t-1)$ implying $[G_{A}^{(t)}]_{i+1}$ is connected. Therefore for all $i \in \{t, t+1, \cdots, T-1\}$ the subgraph $[G_{A}^{(i)}]_{i}$ is connected. Therefore for all $t \in \{1, 2, \cdots, T-1\}$, the subgraph $[G_{A}^{(t)}]_{T-1} = G_{A}^{(t)}$ is connected, and the statement of Lemma~\ref{lm:GA} holds.
\end{proof}

\begin{lemma}
    \label{lm:GB}
    For all $t \in \{1, 2, \cdots, T-1\}$, the subgraph $G_{B}^{(t)}$ is connected.
\end{lemma}
\begin{proof}
    For all $t \in \{1, 2, \cdots, T-1\}$, let $G_{B}^{(t)}$ denote the subgraph associated to the set of vertices $[(0, t), (t,t))$. Consider the family of subgraphs $\{[G_{B}^{(t)}]_{i}\}_{i=1}^{t}$ where $[G_{B}^{(t)}]_{i}$ denotes the subgraph of $G_{B}^{(t)}$ associated to vertices $[(0, t), (i,t))$. Note, by definition $G_{B}^{(t)} = [G_{B}^{(t)}]_{t}$. To prove $G_{B}^{(t)}$ is connected we proceed by induction. Consider the base case of $[G_{B}^{(t)}]_{1}$ associated to vertices $[(0, t), (1, t)) = \{(0, t)\}$. By definition of a singleton graph $[G_{B}^{(t)}]_{0}$ is connected. Assume induction hypothesis: $[G_{B}^{(t)}]_{i}$ is connected for all $i \in \{1,\cdots, t-1\}$, and consider subgraph $[G_{B}^{(t)}]_{i+1}$. By definition $[G_{B}^{(t)}]_{i}$ is a subgraph of $[G_{B}^{(t)}]_{i+1}$ containing all the vertices associated to $[G_{B}^{(t)}]_{i}$ in addition to $(i, t)$. Assuming our induction hypothesis, $[G_{B}^{(t)}]_{i+1}$ is connected if there exists a vertex associated to $[G_{B}^{(t)}]_{i}$ with an edge associated to $[G_{B}^{(t)}]_{i+1}$ connecting it to vertex $(i, t)$. To show this consider vertices $(i-1, t)$, associated to $[G_{B}^{(t)}]_{i}$, and $(i, t)$, associated to $[G_{B}^{(t)}]_{i+1}$, we find 
    \begin{align}
        \label{eq:GB}
        \left[L_{[G_{B}^{(t)}]_{i+1}}\right]_{(i-1, t), (i, t)} &= -4\braket{\ket{i-1}\bra{t}, \tilde{\D}_{\prop}(\ket{i}\bra{t})} = \sum_{t'=0}^{T-2}\left(\braket{i-1|-_{t'}}\braket{-_{t'}|i}\left(2 - 4|\braket{t|-_{t'}}|^{2}\right)\right) \nonumber \\
                                                                  &= \sum_{t'=0}^{T-2}\left((\delta_{i-1,t'} - \delta_{i-1, t'+1})(\delta_{t', i} - \delta_{t'+1,i} )\left(1- |\delta_{t,t'} - \delta_{t,t'+1}|^{2}\right)\right) \nonumber \\
                                                                  &= \sum_{t'=0}^{T-2}\left((\delta_{i-1,t'}\delta_{t',i} - \delta_{i-1,t'}\delta_{t'+1,i} - \delta_{i-1, t'+1}\delta_{t',i} + \delta_{i-1, t'+1}\delta_{t'+1,i}) \left(1- \delta_{t,t'} - \delta_{t,t'+1}\right)\right) \nonumber  \\
                                                                  &= -1 + \delta_{t,i-1} + \delta_{t, i} = -1,
    \end{align}
    where the last equality follows from $i \leq t-1$. From the definition of the Laplacian for $[G_{B}^{(t)}]_{i+1}$, equation~\eqref{eq:GB} implies there exists an edge associated to $[G_{B}^{(t)}]_{i+1}$ connecting vertex $(i, t)$ to $(i+1,t)$ implying $[G_{A}^{(t)}]_{i+1}$ is connected. Therefore for all $i \in \{1, 2, \cdots, t\}$ the subgraph $[G_{B}^{(i)}]_{i}$ is connected. Therefore for all $t \in \{1, 2, \cdots, T-1\}$, the subgraph $[G_{B}^{(t)}]_{T-1} = G_{A}^{(t)}$ is connected, and the statement of Lemma~\ref{lm:GB} holds.
\end{proof}

Given that our base subgraphs $G_{A}^{(t)}$, and $G_{B}^{(t)}$ have been shown to be internally connected, we can glue these subgraphs together in-order to show that each of the subgraphs $G_{p}^{(t)}$ are connected as well. We summarize this result in the following Lemma proven below.
\begin{lemma}
    \label{lm:subgraph}
    For all $t \in \{1, 2, \cdots, T-1\}$, the subgraph $G_{p}^{(t)}$ is connected.
\end{lemma}
\begin{proof}
    For all $t \in \{1, 2, \cdots, T-1\}$, let $G_{p}^{(t)}$ be the subgraph of $G_{p}$ associated to the vertices $\left((t-1,t-1), (t, t)\right)$. We can partition the vertices into subsets $((t-1, t-1), (T-1, t-1)]$ and $[(T-1, t-1), (t-1, t-1))$ and define subgraphs $G_{A}^{(t)}$ and $G_{B}^{(t)}$ respectively. Let $v_{1}$ and $v_{2}$ be vertices associated to graph $G_{p}^{(t)}$. By definition of $G_{A}^{(t)}$ and $G_{B}^{(t)}$, vertices $v_{1}$ and $v_{2}$ must be associated to either $G_{A}^{(t)}$ or $G_{B}^{(t)}$. 

    Assume the case when $v_{1}$ and $v_{2}$ are associated to the same graph. Consider the sub-case when this graph is $G_{A}^{(t)}$. Since Lemma~\ref{lm:GA} implies $G_{A}^{(t)}$ is connected, there exists a walk in $G_{A}$ connecting $v_{1}$ to $v_{2}$. Similarly consider the sub-case when this graph is $G_{B}^{(t)}$. Lemma~\ref{lm:GB} implies $G_{B}^{(t)}$ is connected, thus there exists a walk in $G_{B}^{(t)}$ connecting $v_{1}$ to $v_{2}$. Therefore in both sub-cases there exists a walk in $G_{p}^{(t)}$ connecting $v_{1}$ to $v_{2}$.

    Now assume the case when $v_{1}$ and $v_{2}$ are associated to different graphs. Without loss of generality take $v_{1}$ as associated to $G_{A}^{(t)}$ and $v_{2}$ as associated to $G_{B}^{(t)}$. Assume for all $t \in \{1, 2, \cdots, T-1\}$ there exists an edge in $G_{p}^{(t)}$ connecting a vertex $v_{A}$ associated to subgraph $G_{A}$ to a vertex $v_{B}$ associated to subgraph $G_{B}^{(t)}$.  Since Lemma~\ref{lm:GA} implies $G_{A}^{(t)}$ is connected, there exists a walk in $G_{A}^{(t)}$ connecting $v_{1}$ to $v_{A}$. Using our assumption, there exists a walk in $G_{p}^{(t)}$ connecting $v_{1}$ to $v_{B}$. Since Lemma~\ref{lm:GB} implies $G_{B}^{(t)}$ is connected, there exists a walk in $G_{B}^{(t)}$ connecting $v_{B}$ to $v_{2}$. Thus there exists a walk in $G_{p}^{(t)}$ connecting $v_{1}$ to $v_{2}$. Therefore it suffices to show our assumption in-order to prove there exists a walk in $G_{p}^{(t)}$ connecting $v_{1}$ to $v_{2}$.

    To prove our assumption consider vertices $v_{A} = (t-1, t)$ and $v_{B} = (t, t-1)$ associated to $G_{A}^{(t)}$ and $G_{B}^{(t)}$ respectively. From Lemma~\ref{lm:rotadistograph}, the Laplacian for subgraph $G_{p}^{(t)}$, denoted $L_{G_{p}^{(t)}}$, is $-4\tilde{\D}_{\prop}$ expanded in the basis $\{P_{v}\}_{v \in ((t-1, t-1), (t,t))}$ for $P:V\mapsto \mathcal{B}(\mathcal{H})$ defined by $(t_{A}, t_{B}) \in V \mapsto P_{(t_{A}, t_{B})} = \ket{t_{A}}\bra{t_{B}}$. Thus we find,
    \begin{align}
        \label{eq:GPT}
        \left[L_{G_{p}^{(t)}}\right]_{v_{A}, v_{B}} &=  -4 \braket{\ket{t-1}\bra{t}, \tilde{\D}_{\prop}(\ket{t}\bra{t-1})} = -4\sum_{t'=0}^{T-2}\braket{t-1|-_{t'}}\braket{-_{t'}|t}\braket{t-1|-_{t'}}\braket{-_{t'}|t} \\
&= -\sum_{t'=0}^{T-2}(\delta_{t',t-1} - \delta_{t'+1,t-1})^{2}(\delta_{t',t} - \delta_{t'+1,t})^{2} = -\sum_{t'=0}^{T-2}(\delta_{t',t-1} + \delta_{t'+1,t-1})(\delta_{t',t} + \delta_{t'+1,t}) \\
&= -\sum_{t'=0}^{T-2}(\delta_{t', t-1}\delta_{t',t} + \delta_{t',t-1}\delta_{t'+1,t} + \delta_{t'+1, t-1}\delta_{t',t} + \delta_{t'+1,t-1} \delta_{t'+1, t}) = -1.
    \end{align}
    From the definition of the Laplacian for $G_{p}^{(t)}$, equation~\eqref{eq:GPT} implies there exists an edge associated to $G_{p}^{(t)}$ connecting $v_{A}$ to $v_{B}$. Thus for all $t \in \{1, 2, \cdots, T-1\}$, there exists an edge associated to $G_{p}^{(t)}$ connecting a vertex $v_{A}$, associated to $G_{A}^{(t)}$, to a vertex $v_{B}$, associated to $G_{B}^{(t)}$.

    Since in both cases there exists a walk in $G_{p}^{(t)}$ that connects vertex $v_{1}$ to $v_{2}$, $G_{p}^{(t)}$ is connected by definition, and the statement of Lemma~\ref{lm:subgraph} follows.
\end{proof} 

Armed with the above Lemma, we can proceed towards proving the last result regarding the connectedness of the  ``pure component" of graph associated to the rotated propagation dissipator $G$.
\begin{lemma*}[Lemma~\ref{lm:purecomp} of Appendix~\ref{app:A}]
    Let $G$ be the graph defined in Lemma~\ref{lm:rotadistograph} and $G_{p}$ be the subgraph associated to $V_{p} = V \setminus \{(t, t)\}_{t=0}^{T-1}$, $G_{p}$ is a connected component of $G$.
\end{lemma*}
\begin{proof}
    Since Lemma~\ref{lm:mixcomp} states that the subgraph $G_{m}$ associated to the set vertices $V_{m} = \{(t, t)\}_{t=0}^{T-1}$ is a component of $G$, then the subgraph associated to the set of vertices of $V \setminus V_{m}$ must also be a component of $G$. Therefore $G_{p}$ is a component of $G$.

    To show $G_{p}$ is connected consider the family of subgraphs of $G_{p}$, $\{G_{p}^{(t)}\}_{t=1}^{T-1}$ where each subgraph $G_{p}^{(t)}$ is associated to vertices $((t-1,t-1), (t, t))$. Let $v_{1}$ and $v_{2}$ be two vertices associated to graph $G_{p}$. By the definition of subgraphs $\{G_{p}^{(t)}\}_{t=1}^{T-1}$ there exists $t_{1}, t_{2} \in \{1, 2, \cdots T-1\}$ such that $v_{1}$ is associated to subgraph $G_{p}^{(t_{1})}$ and $v_{2}$ is associated to subgraph $G_{p}^{(t_{2})}$. Assume for all $t \in \{1, 2, \cdots, T-1\}$ there exists an edge associated to graph $G_{p}$ connecting a vertex $v_{ex}^{(t)}$ associated to $G_{p}^{(t)}$ to a vertex $v_{en}^{(t+1)}$ associated to $G_{p}^{(t+1)}$. If $t_{1} = t_{2}$, vertices $v_{1}$ and $v_{2}$ are associated to graph $G_{p}^{(t_{1})}$. Since Lemma~\ref{lm:subgraph} implies $G_{p}^{(t_{1})}$ is connected, there exists a walk in $G_{p}^{(t_{1})}$ connecting $v_{1}$ to $v_{2}$. If $t_{1} \neq t_{2}$, without loss of generality we can assume $t_{2} > t_{1}$, then vertex $v_{1}$ is associated to $G_{p}^{(t_{1})}$ and vertex $v_{2}$ is associated to $G_{p}^{(t_{2})}$ with $G_{p}^{(t_{1})} \neq G_{p}^{(t_{2})}$. By our assumption, there exists an edge connecting $v_{ex}^{(t_{1})}$ to $v_{en}^{(t_{1}+1)}$, and by using the connectedness of $G_{p}^{(t_{1}+1)}$ implied by Lemma~\ref{lm:subgraph} there exists walk in $G_{p}^{(t_{1}+1)}$ connecting $v_{en}^{(t_{1}+1)}$ to $v_{ex}^{(t_{1}+1)}$, therefore there exists a walk in $G$ connecting $v_{ex}^{(t_{1})}$ to $v_{ex}^{(t_{1}+1)}$. Applying this result $t_{2} - t_{1}$ times gives us a walk $G$ connecting $v_{ex}^{(t_{1})}$ to $v_{ex}^{(t_{2})}$. Since Lemma~\ref{lm:subgraph} implies $G_{p}^{(t_{1})}$ is connected, there exists a walk in $G_{p}^{(t_{1})}$ connecting $v_{1}$ to $v_{ex}^{(t_{1})}$, therefore there exists a walk in $G_{p}$ connecting $v_{1}$ to $v_{ex}^{(t_{2})}$. Similarly, since Lemma~\ref{lm:subgraph} implies $G_{p}^{(t_{2})}$ is connected, there exists a walk in $G_{p}^{(t_{2})}$ connecting $v_{ex}^{(t_{2})}$ to $v_{2}$, therefore there exists a walk in $G_{p}$ connecting $v_{1}$ to $v_{2}$. Since in all cases, for $v_{1}, v_{2} \in V$ there exists a walk in $G_{p}$ connecting $v_{1}$ to $v_{2}$, $G_{p}$ is connected. Thus it suffices to show our assumption; for all $t \in \{1, 2, \cdots, T-1\}$  there exists an edge associated to graph $G_{p}$ connecting a vertex $v_{ex}^{(t)}$ associated to $G_{p}^{(t)}$ to a vertex $v_{en}^{(t+1)}$ associated $G_{p}^{(t+1)}$, in-order to show graph $G_{p}$ is connected.

    To prove our assumption, let $t  \in \{1, 2, \cdots, T-1\}$ and consider vertices $v_{ex}^{(t)} = (0, t)$ and $v_{en}^{(t+1)} = (0,t+1)$ associated to graph $G_{p}$. From the definitions of $G_{p}^{(t)}$ and $G_{p}^{(t+1)}$, vertex $v_{ex}^{(t)}$ is associated to graph $G_{p}^{(t)}$ and vertex $v_{en}^{(t+1)}$ is associated to $G_{p}^{(t+1)}$.  From Lemma~\ref{lm:rotadistograph}, the Laplacian for subgraph $G_{p}$, denoted $L_{G_{p}}$, is $-4\tilde{\D}_{\prop}$ expanded in the basis $\{P_{v}\}_{v \in V_{p}}$  for $P: V \mapsto \mathcal{B}(\mathcal{H}_{C})$ defined by $(t_{A}, t_{B}) \in V \mapsto P_{(t_{A}, t_{B})} = \ket{t_{A}}\bra{t_{B}}$. Thus we find, 
    \begin{align}
        \label{eq:edget}
        \left[L_{G_{p}}\right]_{v_{ex}^{(t)}, v_{en}^{(t+1)}} &= -4\braket{\ket{0}\bra{t}, \tilde{\mathcal{D}}_{prop}(\ket{0}\bra{t+1})} = -4\sum_{t'=0}^{T-2}\left(|\braket{-_{t'}|0}|^{2}\braket{t+1|-_{t'}}\braket{-_{t'}|t} - \frac{1}{2}\braket{t+1|-_{t'}}\braket{-_{t'}|t}\right) \nonumber \\
                                                              &= \sum_{t'=0}^{T-2}\left(\left(2 - 4|\braket{-_{t'}|0}|^{2}\right)\braket{t+1|-_{t'}}\braket{-_{t'}|t} \right) = 2\sum_{t'=0}^{T-2}\left(\left(1 - |\delta_{t',0} - \delta_{t'+1,0}|^{2}\right)\braket{t+1|-_{t'}}\braket{-_{t'}|t}\right) \nonumber \\
                                                              &= \sum_{t'=0}^{T-2}\left(\left(1 - \delta_{t',0} + \delta_{t'+1,0}\right)(\delta_{t',t+1} - \delta_{t'+1,t+1})(\delta_{t',t} - \delta_{t'+1,t})\right) \nonumber \\
                                                              &= \sum_{t'=0}^{T-2}\left(\left(1 - \delta_{t',0} + \delta_{t'+1,0}\right)(\delta_{t',t+1}\delta_{t',t} - \delta_{t',t+1}\delta_{t'+1,t} - \delta_{t'+1,t+1}\delta_{t',t} + \delta_{t'+1,t+1}\delta_{t'+1,t})\right) \nonumber \\
                                                              &= \sum_{t'=0}^{T-2}(\delta_{t',t+1}\delta_{t',t} - \delta_{t',t+1}\delta_{t'+1,t} - \delta_{t'+1,t+1}\delta_{t',t} + \delta_{t'+1,t+1}\delta_{t'+1,t}) = -1
\end{align}
where the 2\textsuperscript{nd} to last equality follows from $t \in \{1, 2, \cdots, T-1\}$. From the definition of Laplacian $L_{G_{p}}$, equation~\eqref{eq:edget} implies there exists an edge associated to $G_{p}$ connecting $v_{ex}^{(t)}$ to $v_{en}^{(t+1)}$. Thus for all $t \in \{1, 2, \cdots, T-1\}$, there exists an edge associated to graph $G_{p}$ connecting a vertex $v_{p}^{(t)}$ associated to $G_{p}^{(t)}$ to a vertex $v_{en}^{(t+1)}$ associated to $G_{p}^{(t+1)}$, implying $G_{p}$ is connected. The statement of Lemma~\ref{lm:purecomp} thus follows.
\end{proof}

\section{Analysis of Input Dissipator}
\label{app:C}

The input dissipator $\D_{\inp}$ defined in equation~\eqref{eq:Dinp} of Definition~\ref{def:QCSATdis} is used in our hardness construction in order to drive the system into a steady-state subspace with a properly initialized ancillary register. In this appendix we analyze $\D_{\inp}$, charcterizing its kernel inside the history state subspace, lower-bounding its spectral gap, as well as analyzing its leakage outside of it is pure steady-state subspace. These statements appear as Lemmas~\ref{lm:DinpKer},~\ref{lm:DinpGap}, and~\ref{lm:DinpLeakage} in the main text. In the first of the statements provided below, we characterizes the kernel of $\D_{\inp}$ after it has been restricted to the subspace of pure history states.

\begin{lemma*}[Lemma~\ref{lm:DinpKer}]
    Let $(V, b, a)$ be an instance of \textsc{Quantum Circuit-SAT} where $V=U_{T-1}U_{T-2}\cdots U_{0}$ with $T > 2$ and define $\mathcal{H}_{\prop}$ according to equation~\eqref{eq:Hprop} of Definition~\ref{def:histsubspace} using $V$ as well as $\D_{\inp}$ according to equation~\ref{eq:Dinp} of Definition~\ref{def:QCSATdis}, then
    \begin{equation}
        \ker{\D_{\inp}|_{\mathcal{B}(\mathcal{H}_{\prop})}} = \U^{(V)}[\mathcal{B}(\mathcal{H}_{M}) \otimes \ket{0}\bra{0}_{A} \otimes \ket{p}\bra{p}_{C})].
    \end{equation}
    where $\ket{p} = \frac{1}{\sqrt{T}}\sum_{t=0}^{T-1}\ket{t}$ and $\U^{(V)}$ the circuit rotation associated to $V$, is defined according to equation~\eqref{eq:circrot} of Appendix~\ref{app:A}
\end{lemma*}
\begin{proof}
    Defining the rotated version of $\D_{\inp}$ as $\D_{\inp}^{(V)} = \U^{(V)\dagger} \circ \D_{\inp} \circ \U^{(V)}$ we find     
    \begin{equation}
        \ker{\D_{\inp}^{(V)}|_{\mathcal{U}^{(V)\dagger}[\mathcal{B}(\mathcal{H}_{\prop})]} } = \U^{(V)\dagger}[\ker{\D_{\inp}|_{\mathcal{B}(\mathcal{H}_{\prop})}}].
    \end{equation}
    From Lemma~\ref{lm:onlyhist} we can simplify the above equation as
    \begin{equation}
        \label{eq:Dinpkern}
        \ker{\D_{\inp}^{(V)}|_{\mathcal{B}(\mathcal{H}_{MA}) \otimes \ket{p}\bra{p}_{C}} } = \U^{(V)\dagger}[\ker{\D_{\inp}|_{\mathcal{B}(\mathcal{H}_{\prop})}}].
    \end{equation}

    Consider $\D_{\inp}^{(V)}$, we find
    \begin{align}
        \label{eq:roteDinp}
        \D_{\inp}^{(V)}(\bullet) &=  \U^{(V)\dagger} \circ \D_{\inp} \circ \U^{(V)}(\bullet) \nonumber \\
                                 &= \sum_{i=0}^{k-1}\left(\U^{(V)\dagger}(\ket{1}\bra{1}_{A_{i}} \otimes \ket{0}\bra{0}_{C}) \bullet \U^{(V)\dagger}(\ket{1}\bra{1}_{A_{i}} \otimes \ket{0}\bra{0}_{C})- \frac{1}{2}\{\U^{(V)\dagger}(\ket{1}\bra{1}_{A_{i}} \otimes \ket{0}\bra{0}_{C}), \bullet\} \right).
    \end{align}
    Note,
    \begin{equation}
        \U^{(V)\dagger}(\ket{1}\bra{1}_{A_{i}} \otimes \ket{0}\bra{0}_{C}) = V_{0}^{*}\ket{1}\bra{1}_{A_{i}} V_{0} \otimes \ket{0}\bra{0}_{C} = \ket{1}\bra{1}_{A_{i}} \otimes \ket{0}\bra{0}_{C}
    \end{equation}
    since $V_{0}$, the unitary acting on the circuit before any gates have been applied, is always $\mathbb{I}$. Thus we find from equation~\eqref{eq:roteDinp} that
    \begin{equation}
        \D_{\inp}^{(V)} = \D_{\inp} 
    \end{equation}
    implying that 
    \begin{equation}
        \D_{\inp}^{(V)}|_{\mathcal{B}(\mathcal{H}_{MA}) \otimes \ket{p}\bra{p}_{C}}(\bullet) = \D_{\inp}|_{\mathcal{B}(\mathcal{H}_{MA}) \otimes \ket{p}\bra{p}_{C}}(\bullet).
    \end{equation}
    Defining the projection onto $\mathcal{B}(\mathcal{H}_{MA}) \otimes \ket{p}\bra{p}_{C}$ as $\mathcal{P}(\bullet) = (\mathbb{I}_{MA} \otimes \ket{p}\bra{p}_{C}) \bullet (\mathbb{I}_{MA} \otimes \ket{p}\bra{p}_{C})$ we can calculate
    \begin{align}
        \mathcal{P} \circ \D_{\inp} \mathcal{P}(\bullet) =  \sum_{i=0}^{k-1}\Big(\frac{1}{T^{2}}&(\ket{1}\bra{1}_{A_{i}} \otimes \ket{p}\bra{p}_{C})\bullet (\ket{1}\bra{1}_{A_{i}} \otimes \ket{p}\bra{p}_{C}) - \frac{1}{2T}(\ket{1}\bra{1}_{A_{i}} \otimes \ket{p}\bra{p}_{C})\bullet (\mathbb{I}_{MA} \otimes \ket{p}\bra{p}_{C})  \nonumber \\
                                                                                                &- \frac{1}{2T}(\mathbb{I}_{MA} \otimes \ket{p}\bra{p}_{C}) \bullet (\ket{1}\bra{1}_{A_{i}} \otimes \ket{p}\bra{p}_{C}) \Big).
    \end{align}
    Since we are restricting $\mathcal{P} \circ \D_{\inp} \circ \mathcal{P}$ to the subspace $\mathcal{B}(\mathcal{H}_{MA}) \otimes \ket{p}\bra{p}_{C}$ we find 
    \begin{equation}
        \D_{\inp}^{(V)}|_{\mathcal{B}(\mathcal{H}_{MA}) \otimes \ket{p}\bra{p}_{C}}(\bullet) =  \D_{\inp}|_{\mathcal{B}(\mathcal{H}_{MA}) \otimes \ket{p}\bra{p}_{C}}(\bullet) = \sum_{i=0}^{k-1}\left(\frac{1}{T^{2}}\ket{1}\bra{1}_{A_{i}} \bullet \ket{1}\bra{1}_{A_{i}} - \frac{1}{2T}\{\ket{1}\bra{1}_{A_{i}}, \bullet\} \right).
    \end{equation}
    Defining $\mathcal{M}_{i}(\bullet) = \frac{1}{T^{2}}\ket{1}\bra{1}_{A_{i}} \bullet \ket{1}\bra{1}_{A_{i}} - \frac{1}{2T}\{\ket{1}\bra{1}_{A_{i}}, \bullet\} $ allows us to write,
    \begin{equation}
        \D_{\inp}^{(V)}|_{\mathcal{B}(\mathcal{H}_{MA}) \otimes \ket{p}\bra{p}_{C}} = \sum_{i=0}^{k-1}\mathcal{M}_{i}.
    \end{equation}

    To analyze the kernel of $\D_{\inp}^{(V)}|_{\mathcal{B}(\mathcal{H}_{MA}) \otimes \ket{p}\bra{p}_{C}}$ note 
    \begin{equation}
        [\mathcal{M}_{i}, \mathcal{M}_{j}] = 0.
    \end{equation}
    In addition, $\mathcal{M}_{i} \preceq 0$ since for all $B \in \mathcal{B}(\mathcal{H}_{MA}) \otimes \ket{p}\bra{p}_{C}$,
    \begin{align}
        \braket{B,\mathcal{M}_{i}(B)} &= \frac{1}{T^{2}}\left(\braket{\ket{1}\bra{1}_{A_{i}} B, B \ket{1}\bra{1}_{A_{i}}} - \frac{T}{2}\braket{B, \ket{1}\bra{1}_{A_{i}}B} - \frac{T}{2}\braket{B, B\ket{1}\bra{1}_{A_{i}}} \right) \nonumber \\
                                      &= \frac{1}{T^{2}}\left(\braket{\ket{1}\bra{1}_{A_{i}} B, B \ket{1}\bra{1}_{A_{i}}} - \frac{T}{2}\braket{\ket{1}\bra{1}_{A_{i}} B, \ket{1}\bra{1}_{A_{i}} B} - \frac{T}{2}\braket{B \ket{1}\bra{1}_{A_{i}}, B\ket{1}\bra{1}_{A_{i}}} \right) \nonumber \\
                                      &\leq  - \frac{1}{2T}\left(\sqrt{\braket{\ket{1}\bra{1}_{A_{i}} B, \ket{1}\bra{1}_{A_{i}} B}} - \sqrt{\braket{B \ket{1}\bra{1}_{A_{i}}, B\ket{1}\bra{1}_{A_{i}}}} \right)^{2} \leq 0.
    \end{align}
    Therefore $\ker{\D_{\inp}^{(V)}|_{\mathcal{B}(\mathcal{H}_{MA}) \otimes \ket{p}\bra{p}_{C}}}$ is given by the intersection of all $\ker{\mathcal{M}_{i}}$. Each $\mathcal{M}_{i}$ acts on a single qubit $A_{i}$ and is diagonalized in the basis $\{\ket{0}\bra{0}_{A_{i}}, \ket{0}\bra{1}_{A_{i}}, \ket{1}\bra{0}_{A_{i}}, \ket{1}\bra{1}_{A_{i}}\}$ with a unique zero eigenvalue assigned to $\ket{0}\bra{0}_{A_{i}}$ when $T > 2$. Thus 
    \begin{equation}
        \ker{\mathcal{M}_{i}} = \mathcal{B}(\mathcal{H}_{MA \setminus A_{i}}) \otimes \ket{0}\bra{0}_{A_{i}} \otimes \ket{p}\bra{p}_{C}
    \end{equation}
    implying 
    \begin{equation}
        \ker{\D_{\inp}^{(V)}|_{\mathcal{B}(\mathcal{H}_{MA}) \otimes \ket{p}\bra{p}_{C}}} = \mathcal{B}(\mathcal{H}_{M}) \otimes \ket{0}\bra{0}_{A} \otimes \ket{p}\bra{p}_{C}.
    \end{equation}
    From equation~\ref{eq:Dinpkern} we find 
    \begin{equation}
        \ker{\D_{\inp}|_{\mathcal{B}(\mathcal{H}_{\prop})}} = \U^{(V)}[\mathcal{B}(\mathcal{H}_{M}) \otimes \ket{0}\bra{0}_{A} \otimes \ket{p}\bra{p}_{C} ]
    \end{equation}
    as stated.
\end{proof}

The next statement we show is a lower-bound on the gap of our input dissipator when restricted to the subspace of pure history states.

\begin{lemma*}[Lemma~\ref{lm:DinpGap}]
    Let $(V, b, a)$ be an instance of \textsc{Quantum Circuit-SAT} where $V=U_{T-1}U_{T-2}\cdots U_{0}$ with $T > 2$ and define $\mathcal{H}_{\prop}$ according to equation~\eqref{eq:Hprop} of Definition~\ref{def:histsubspace} using $V$ as well as $\D_{\inp}$ according to equation~\ref{eq:Dinp} of Definition~\ref{def:QCSATdis}, then
    \begin{equation}
        \gap{\D_{\inp}|_{\mathcal{B}(\mathcal{H}_{\prop})}} \geq \frac{1}{2T}.
    \end{equation}
\end{lemma*}
\begin{proof}
    Defining the rotated version of $\D_{\inp}$ as $\D_{\inp}^{(V)} = \U^{(V)\dagger} \circ \D_{\inp} \circ \U^{(V)}$ we find 
    \begin{equation}
        \gap{\D_{\inp}|_{\mathcal{B}(\mathcal{H}_{\prop})}} = \gap{\D_{\inp}^{(V)}|_{\U^{(V)\dagger}[\mathcal{B}(\mathcal{H}_{\prop})]}} = \gap{\D_{\inp}^{(V)}|_{\mathcal{B}(\mathcal{H}_{MA})\otimes \ket{p}\bra{p}_{C}}}
    \end{equation}
    where the last equality follows from Lemma~\ref{lm:onlyhist}. As found in our proof of Lemma~\ref{lm:DinpKer}, calculating $\D_{\inp}^{(V)}|_{\mathcal{B}(\mathcal{H}_{MA})\otimes \ket{p}\bra{p}_{C}}$ we find
    \begin{equation}
        \label{eq:restrictedDinp}
        \D_{\inp}^{(V)}|_{\mathcal{B}(\mathcal{H}_{MA}) \otimes \ket{p}\bra{p}_{C}}(\bullet) = \sum_{i=0}^{k-1}\left(\frac{1}{T^{2}}\ket{1}\bra{1}_{A_{i}} \bullet \ket{1}\bra{1}_{A_{i}} - \frac{1}{2T}\{\ket{1}\bra{1}_{A_{i}}, \bullet\} \right).
    \end{equation}
    Defining $\mathcal{M}_{i}(\bullet) = \frac{1}{T^{2}}\ket{1}\bra{1}_{A_{i}} \bullet \ket{1}\bra{1}_{A_{i}} - \frac{1}{2T}\{\ket{1}\bra{1}_{A_{i}}, \bullet\} $ allows us to write,
    \begin{equation}
        \D_{\inp}^{(V)}|_{\mathcal{B}(\mathcal{H}_{MA}) \otimes \ket{p}\bra{p}_{C}} = \sum_{i=0}^{k-1}\mathcal{M}_{i}.
    \end{equation}
    As shown in our proof of Lemma~\ref{lm:DinpKer}, each $\mathcal{M}_{i} \preceq 0$, in addition $[\mathcal{M}_{i}, \mathcal{M}_{j}] = 0$, therefore 
    \begin{equation}
        \gap{\D_{\inp}|_{\mathcal{B}(\mathcal{H}_{MA} \otimes \ket{p}\bra{p}}} \geq \min_{i \in \{0, 1, \cdots, k-1\}}\gap{\mathcal{M}_{i}}.
    \end{equation}
    Since each $\mathcal{M}_{i}$ acts on a single qubit $A_{i}$ and is diagonalized in the basis $\{\ket{0}\bra{0}_{A_{i}}, \ket{0}\bra{1}_{A_{i}}, \ket{1}\bra{0}_{A_{i}}, \ket{1}\bra{1}_{A_{i}}\}$, we can read off $\gap{\mathcal{M}_{i}} = 1 / 2T$ implying
    \begin{equation}
        \gap{\D_{\inp}|_{\mathcal{B}(\mathcal{H}_{\prop})}} \geq \min_{i \in \{0, 1, \cdots, k-1\}}\gap{\mathcal{M}_{i}} = \frac{1}{2T}.
    \end{equation}
\end{proof}

We conclude this appendix proving a statement bounding how much the span of our pure steady-states leak out into the rest of the kernel of our restricted input dissipator.

\begin{lemma*}[Lemma~\ref{lm:DinpLeakage}]
    Let $(V, b, a)$ be an instance of \textsc{Quantum Circuit-SAT} where $V=U_{T-1}U_{T-2}\cdots U_{0}$ with $T > 2$ and define $\mathcal{H}_{\prop}$ according to equation~\eqref{eq:Hprop} of Definition~\ref{def:histsubspace} as well as $\D_{\inp}$ according to equation~\ref{eq:Dinp} of Definition~\ref{def:QCSATdis} and $\mathcal{H}_{\qcirc}$ according to equation~\eqref{eq:Hqcirc} of Definition~\ref{def:circsubspace}. In addition define
    \begin{equation}
        \mathsf{M} = \left\{v_{1}\ket{h}\bra{\perp} + v_{2}\ket{\perp}\bra{h} + v_{3}\ket{\perp}\bra{\perp} \; | \; (\ket{h}, \ket{\perp}, v) \in \mathcal{H}_{\qcirc} \times \mathcal{H}_{\qcirc}^{\perp} \times \mathbb{C}^{3} \right\},
    \end{equation}
    then the overlap between $\ker{\D_{\inp}|_{\mathcal{B}(\mathcal{H}_{\prop})}}$ and $\mathsf{M}$ is given by 
    \begin{equation}
        \cos^{2}\theta = \sup_{\substack{A_{0} \in \ker{\D_{\inp}|_{\mathcal{B}(\mathcal{H}_{\prop})}}, A \in \mathsf{M}\\ \braket{A_{0},A_{0}} = \braket{A , A } = 1}} |\braket{A_{0},A}|^{2} =0 .
    \end{equation}
\end{lemma*}
\begin{proof}
    From Lemma~\ref{lm:onlycirc} we find 
    \begin{equation}
        \mathcal{B}(\mathcal{H}_{\qcirc}) = \mathcal{B}(U^{(V)}[\mathcal{H}_{M} \otimes \ket{0}_{A} \otimes \ket{p}_{C}]) = \U^{(V)}[\mathcal{B}(\mathcal{H}_{M}) \otimes \ket{0}\bra{0}_{A} \otimes \ket{p}\bra{p}_{C}] = \ker{\D_{\inp}|_{\mathcal{B}(\mathcal{H}_{\prop})}}
    \end{equation}
    where the last equality follows from Lemma~\ref{lm:DinpKer}. Thus for all $A_{0} \in \ker{\D_{\inp}|_{\mathcal{B}(\mathcal{H}_{\prop})}} = \mathcal{B}(\mathcal{H}_{\qcirc})$ and for all $\ket{h} \in \mathcal{H}_{\qcirc}$ we have $A_{0}\ket{h} \in \mathcal{H}_{\qcirc}$ implying for $A_{0} \in \ker{\D_{\inp}|_{\mathcal{B}(\mathcal{H}_{\prop})}}$ and $A \in \mathsf{M}$ we have 
    \begin{equation}
        \braket{A_{0}, A} = \braket{A_{0} ,v_{1}\ket{h}\bra{\perp} + v_{2}\ket{\perp}\bra{h} + v_{3}\ket{\perp}\bra{\perp}} =  v_{1}\braket{\perp|h'} + v_{2}\braket{h''|\perp} + v_{3}\bra{\perp}A_{0}^{*}\ket{\perp}
    \end{equation}
    where $\ket{h'} = A_{0}^{*}\ket{h} \in \mathcal{H}_{\qcirc}$ and $\ket{h''} = A_{0}\ket{h} \in \mathcal{H}_{\qcirc}$. Since $\ket{\perp} \in \mathcal{H}_{\qcirc}^{\perp}$ we have 
    \begin{equation}
        \braket{A_{0},A} = v_{3}\bra{\perp}A_{0}^{*}\ket{\perp} = v_{3}\braket{A_{0}, \ket{\perp}\bra{\perp}}.
    \end{equation}
    By definition $\ket{\perp}\bra{\perp} \in \mathcal{B}(\mathcal{H}_{\qcirc}^{\perp}) \subseteq \mathcal{B}(\mathcal{H}_{\qcirc})^{\perp}$, thus $\braket{A_{0}, \ket{\perp}\bra{\perp}} = 0$ implying the statement of Lemma~\ref{lm:DinpLeakage}.
\end{proof}

\section{Analysis of Unary Clock Dissipator}
\label{app:D}

The unary clock dissipator $\D_{\unary}$, defined in equation~\eqref{eq:Dunary} of Definition~\ref{def:localDQCSAT} is used to drive our system into a steady-state subspace where the clock register exists in valid unary encoded clock states. This construction is useful when we extend our quasi-local hardness construction to become $k$-local in Section~\ref{sec:hardness}. In this appendix we provide the proofs of Propositions~\ref{prop:DuKer},~\ref{prop:DuGap}, and Lemma~\ref{lm:DuLeakage} which give us the statements characterizing the kernel of $\D_{\unary}$, the spectral gap of $\D_{\unary}$, and the leakage of the pure steady-state subspace of $\D_{\unary}$ out into the rest of its kernel respectively. The appendix is divided as follows. In the first section we prove Proposition~\ref{prop:DuKer} giving us a characterization of the unary clock dissipator's kernel. In the second section we provide proof of Proposition~\ref{prop:DuGap} giving us a lower-bound on the gap of our unary clock dissipator. In the final section we prove Lemma~\ref{lm:DuLeakage} giving us an analysis of the leakage outside of our unary clock dissipators' pure stead-state subspace.

\subsection{Characterization of Kernel}

Our general strategy to characterize the kernel of our unary clock dissipator follows an inductive argument over the number of qubits in our clock register $C$. We first prove the base case of our argument given in the following lemma.
\begin{lemma}
    \label{lm:DuKerbasecase}
    Let $\mathcal{H} = \mathcal{H}_{\leg}^{\loc} \oplus \mathcal{H}_{\ill}^{\loc}$ with $\mathcal{H}_{\leg}^{\loc}$ and $\mathcal{H}_{\ill}^{\loc}$ defined according to equation~\eqref{eq:localillandlegsubspaces}, and define $P_{\ill}^{\loc}$ to be the projector onto $\mathcal{H}_{\ill}^{\loc}$. Let $\D_{\loc}$ be a Lindbladian dissipator acting on $\mathcal{B}(\mathcal{H})$ defined as
    \begin{equation}
        \D_{\loc}(\bullet) = \sum_{(x, y) \in \mathrm{S}_{\ill}^{\loc \times 2}} \left(\frac{1}{4}\ket{x}\bra{y}\bullet\ket{y}\bra{x} - \frac{1}{8}\left\{\ket{y}\bra{y}, \bullet\right\}\right),
    \end{equation}
    where $S_{\ill}^{\loc}$ is defined according to equation~\eqref{eq:localillandlegstrings} then 
    \begin{equation}
        \ker{\D_{\loc}} = \mathcal{B}(\mathcal{H}_{\leg}^{\loc}) \oplus P_{\ill}^{\loc}.
    \end{equation}
\end{lemma}
\begin{proof} 
    Let $P_{\leg}^{\loc} = \mathbb{I} - P_{\ill}^{\loc}$ be the projector onto $\mathcal{H}_{\leg}^{\loc}$. We can thus decompose any $A \in \mathcal{B}(\mathcal{H})$ into
    \begin{equation}
        A = A_{\leg} + A_{\leg, \ill} + A_{\ill, \leg} + A_{\ill}, 
    \end{equation}
    where we have defined
    \begin{equation}
        \label{eq:localconstraints}
        \begin{cases} 
            A_{\leg} = P_{\leg}^{\loc}AP_{\leg}^{\loc} \\
            A_{\ill} = P_{\ill}^{\loc} A P_{\ill}^{\loc} \\
            A_{\leg, \ill} = P_{\leg}^{\loc} A P_{\ill}^{\loc} \\
            A_{\leg, \ill} = P_{\leg}^{\loc} A P_{\ill}^{\loc}
        \end{cases}.
    \end{equation}
    Calculating the action of $\D_{\loc}$ on $A_{\leg}$ we find
    \begin{equation}
        \label{eq:legblock}
        \D_{\unary}^{\loc}(A_{\leg}) = \sum_{(x, y) \in \mathrm{S}_{\ill}^{\loc \times 2}} \left(\frac{1}{4}\ket{x}\bra{y}P_{\leg}^{\loc}AP_{\leg}^{\loc}\ket{y}\bra{x} - \frac{1}{8}\left\{\ket{y}\bra{y}, P_{\leg}^{\loc}AP_{\leg}^{\loc}\right\}\right) = 0.
    \end{equation}
    We find the action of $\D_{\loc}$ on $A_{\leg, \ill}$ as 
    \begin{align}
        \label{eq:legillblock}
        \D_{\unary}^{\loc}(A_{\leg, \ill}) &= \sum_{(x, y) \in \mathrm{S}_{\ill}^{\loc \times 2}} \left(\frac{1}{4}\ket{x}\bra{y}P_{\leg}^{\loc}AP_{\ill}^{\loc}\ket{y}\bra{x} - \frac{1}{8}\left\{\ket{y}\bra{y}, P_{\leg}^{\loc}AP_{\ill}^{\loc}\right\}\right) \nonumber \\
                                           &= -\frac{1}{2}\left\{P_{\ill}^{\loc}, P_{\leg}^{\loc}AP_{\ill}^{\loc}\right\} = -\frac{1}{2}A_{\leg, \ill}.
    \end{align}
    Similarly, for $A_{\ill,\leg}$ we find 
    \begin{equation}
        \label{eq:illlegblock}
        \D_{\unary}^{\loc}(A_{\ill, \leg}) = -\frac{1}{2}\left\{P_{\ill}^{\loc}, P_{\ill}^{\loc}AP_{\leg}^{\loc}\right\} = -\frac{1}{2}A_{\ill, \leg}.
    \end{equation}
    Finally, calculating the action of $\D_{\loc}$ on $A_{\ill}$ we find,
    \begin{align}
        \label{eq:illblock}
        \D_{\unary}^{\loc}(A_{\ill}) &= \sum_{(x, y) \in \mathrm{S}_{\ill}^{\loc \times 2}} \left(\frac{1}{4}\ket{x}\bra{y}P_{\ill}^{\loc}AP_{\ill}^{\loc}\ket{y}\bra{x} - \frac{1}{8}\left\{\ket{y}\bra{y}, P_{\ill}^{\loc}AP_{\ill}^{\loc}\right\}\right) \nonumber \\
                                     &= \sum_{y\in \mathrm{S}_{\ill}^{\loc}}\bra{y}A_{\ill}\ket{y}\frac{P_{\ill}^{\loc}}{4} - A_{\ill}.
    \end{align}
    Combining equations~\eqref{eq:legblock},~\eqref{eq:legillblock},~\eqref{eq:illlegblock}, and~\eqref{eq:illblock} we find for all $A \in \mathcal{B}(\mathcal{H})$,
    \begin{equation}
        \D_{\unary}^{\loc}(A) =  -\frac{1}{2}A_{\leg,\ill} -\frac{1}{2}A_{\ill,\leg} + \left(\sum_{y\in \mathrm{S}_{\ill}^{\loc}}\bra{y}A_{\ill}\ket{y}\frac{P_{\ill}^{\loc}}{4} - A_{\ill}\right).
    \end{equation}

    Let $A \in \ker{\D_{\loc}}$, thus we have $\D_{\loc}(A) = 0$  implying
    \begin{equation}
        0 =  -\frac{1}{2}A_{\leg,\ill} -\frac{1}{2}A_{\ill,\leg} + \left(\sum_{y\in \mathrm{S}_{\ill}^{\loc}}\bra{y}A_{\ill}\ket{y}\frac{P_{\ill}^{\loc}}{4} - A_{\ill}\right).
    \end{equation}
    Since each of the three terms in the above equation lie in orthogonal subspaces of $\mathcal{B}(\mathcal{H})$, the above equation implies the following constraints,
    \begin{equation}
        \begin{cases}
            A_{\leg,\ill} = 0 \\
            A_{\ill, \leg} = 0 \\
            \sum_{y\in \mathrm{S}_{\ill}^{\loc}}\bra{y}A_{\ill}\ket{y}\frac{P_{\ill}^{\loc}}{4} - A_{\ill} = 0 
        \end{cases}.
    \end{equation}
    Consider the last of these three constraints, and assume $A_{\ill} \neq 0$. We find 
    \begin{equation}
        A_{\ill} = \sum_{y\in \mathrm{S}_{\ill}^{\loc}}\bra{y}A_{\ill}\ket{y}\frac{P_{\ill}^{\loc}}{4}
    \end{equation}
    which implies for all $A \in \ker{\D_{\unary}^{\loc}}$, there exists an $\alpha \in \mathbb{C}$ such that, $A = A_{\leg} + \alpha P_{\ill}^{\loc}$. Thus we have for all $A \in \ker{\D_{\unary}^{\loc}}$, we have $A \in \mathcal{B}(\mathcal{H}_{\leg}^{\loc}) \oplus P_{\ill}^{\loc}$. Since for the reverse case we have for all $B \in \mathcal{B}(\mathcal{H}_{\leg}^{\loc}) \oplus P_{\ill}^{\loc}$ that the conditions of equation~\eqref{eq:localconstraints} are satisfied, we have $\D_{\unary}^{\loc}(B) = 0$ implying $B \in \ker{\D_{\unary}^{\loc}}$. Therefore we conclude 
    \begin{equation}
        \ker{\D_{\unary}^{\loc}} = \mathcal{B}(\mathcal{H}_{\leg}^{\loc}) \oplus P_{\ill}^{\loc}.
    \end{equation}
\end{proof}

Before proveing the main statement characterizing the kernel of $\D_{\unary}$ we take a detour to prove Lemmas~\ref{lm:helper1} and~\ref{lm:helper2} which will be of use during our proof of the main statement.
\begin{lemma}
    \label{lm:helper1}
    Let $\D_{\loc}^{(T-1,T-2,T-3)}$ be the $T-3$-term in the definition of $\D_{\unary}$ given in equation~\eqref{eq:Dunary} of Definition~\ref{def:localDQCSAT} and define $\mathcal{H}_{\leg}^{(T)}$ and $\mathcal{H}_{\ill}^{(T)}$  according to equation~\eqref{eq:globalillandlegsubspaces}. Let $P_{\ill}^{\loc}$ be the projector onto $\mathcal{H}_{\ill}^{\loc}$ with $\mathcal{H}_{\ill}^{\loc}$ defined according to equation~\eqref{eq:localillandlegsubspaces}, then for all $T \geq 3$,
    \begin{align}
        \label{eq:blockchar}
        \D_{\loc}^{(T-1,T-2,T-3)} = &\D_{\loc}^{(T-1,T-2,T-3)}|_{\mathcal{B}(\mathcal{H}_{\leg}^{(T)})} \oplus \D_{\loc}^{(T-1,T-2,T-3)}|_{\mathcal{B}(\mathcal{H}_{\ill}^{(T)}, \mathcal{H}_{\leg}^{(T)})} \nonumber \\
                                  &\oplus \D_{\loc}^{(T-1,T-2,T-3)}|_{\mathcal{B}(\mathcal{H}_{\leg}^{(T)}, \mathcal{H}_{\ill}^{(T)})} \oplus \D_{\loc}^{(T-1,T-2,T-3)}|_{\mathcal{B}(\mathcal{H}_{\ill}^{(T)})}
    \end{align}
    with
    \begin{equation}
        \begin{cases}
            \D_{\loc}^{(T-1,T-2,T-3)}|_{\mathcal{B}(\mathcal{H}_{\leg}^{(T)})}(\bullet) = 0(\bullet) \\
            \D_{\loc}^{(T-1,T-2,T-3)}|_{\mathcal{B}(\mathcal{H}_{\ill}^{(T)}, \mathcal{H}_{\leg}^{(T)})}(\bullet) = - \frac{1}{2}[P_{\ill}^{\loc}]_{T-1,T-2,T-3} \bullet \\
            \D_{\loc}^{(T-1,T-2,T-3)}|_{\mathcal{B}(\mathcal{H}_{\leg}^{(T)}, \mathcal{H}_{\ill}^{(T)})}(\bullet) = - \frac{1}{2} \bullet [P_{\ill}^{\loc}]_{T-1,T-2,T-3}
        \end{cases}
    \end{equation}
    and
    \begin{equation}
        \label{eq:dlocill}
        \D_{\loc}^{(T-1,T-2,T-3)}|_{\mathcal{B}(\mathcal{H}_{\ill}^{(T)})}(\bullet) = \frac{1}{4}\sum_{(x, y) \in \mathrm{S}_{\ill}^{\loc \times 2}}(\ket{x}\bra{y}_{T-1, T-2, T-3}) \bullet (\ket{y}\bra{x}_{T-1, T-2, T-3}) - \frac{1}{2}\left\{[P_{\ill}^{\loc}]_{T-1,T-2, T-3},  \bullet \right\}.
    \end{equation}
\end{lemma}
\begin{proof}
    To analyze $\ker{\D_{\loc}^{(T-1, T-2, T-3)}}$ we first block diagonalize it with respect to the decomposition $\mathcal{H}_{C} = \mathcal{H}_{\leg}^{(T)} \oplus \mathcal{H}_{\ill}^{(T)}$. Let $P_{\leg}^{(T)} = \mathbb{I}_{C} - P_{\ill}^{(T+1)}$ be the projection onto $\mathcal{H}_{\leg}^{(T)}$ and define $\mathcal{P}_{\mathcal{B}(\mathcal{H}_{\leg}^{(T)})}$, $\mathcal{P}_{\mathcal{B}(\mathcal{H}_{\ill}^{(T)}, \mathcal{H}_{\leg}^{(T)})}$, $\mathcal{P}_{\mathcal{B}(\mathcal{H}_{\leg}^{(T)}, \mathcal{H}_{\ill}^{(T)})}$, and $\mathcal{P}_{\mathcal{B}(\mathcal{H}_{\ill}^{(T)})}$ as the projections onto $\mathcal{B}(\mathcal{H}^{(T)}_{\leg})$, $\mathcal{B}(\mathcal{H}^{(T)}_{\ill}, \mathcal{H}^{(T)}_{\leg})$, $\mathcal{B}(\mathcal{H}^{(T)}_{\leg}, \mathcal{H}^{(T)}_{\ill})$, and $\mathcal{B}(\mathcal{H}_{\ill}^{(T)})$ respectively. Calculating
    \begin{align}
        \D_{\loc}^{(T-1, T-2, T-3)} \circ \mathcal{P}_{\mathcal{B}(\mathcal{H}_{\leg}^{(T)})} = \frac{1}{4}&\sum_{(x, y) \in \mathrm{S}_{\ill}^{\loc \times 2}}(\ket{x}\bra{y}_{T-1, T-2, T-3})P_{\leg}^{(T)}\bullet P_{\leg}^{(T)}(\ket{y}\bra{x}_{T-1, T-2, T-3}) \nonumber \\
                                                                                          &- \frac{1}{2}\left\{[P_{\ill}^{\loc}]_{T-1,T-2,T-3}, P_{\leg}^{(T)} \bullet P_{\leg}^{(T)}\right\},
    \end{align}
    where we have defined $P_{\ill}^{\loc} = \sum_{s \in S_{\ill}^{\loc}} \ket{s}\bra{s}$. Note, let $s \in \mathrm{S}_{\ill}^{\loc}$ and $x \in \{0, 1\}^{T-2}$ we have $sx \in \mathrm{S}_{\ill}^{(T)}$ implying $\ket{sx} \in \mathcal{H}^{(T)}_{\ill}$ which further implies $[P_{\ill}^{\loc}]_{T-1,T-2,T-3}P_{\leg}^{(T)} = P_{\leg}^{(T)}[P_{\ill}^{\loc}]_{T-1,T-2,T-3} = 0$. Thus we find,
    \begin{equation}
        \label{eq:blah1}
        \D_{\loc}^{(T-1, T-2, T-3)} \circ \mathcal{P}_{\mathcal{B}(\mathcal{H}_{\leg}^{(T)})} = 0 \text{ and similarly } \mathcal{P}_{\mathcal{B}(\mathcal{H}_{\leg}^{(T)})} \circ \D_{\loc}^{(T-1, T-2 T-3)}  = 0.
    \end{equation}
    Calculating 
    \begin{align}
        \label{eq:blah2}
        \D_{\loc}^{(T-1, T-2, T-3)} \circ \mathcal{P}_{\mathcal{B}(\mathcal{H}_{\ill}^{(T)}, \mathcal{H}_{\leg}^{(T)})} = &\frac{1}{4}\sum_{(x, y) \in \mathrm{S}_{\ill}^{\loc \times 2}}(\ket{x}\bra{y}_{T-1, T-2, T-3})P_{\ill}^{(T)}\bullet P_{\leg}^{(T)}(\ket{y}\bra{x}_{T-1, T-2, T-3}) \nonumber \\
                                                                                          &- \frac{1}{2}\left\{[P_{\ill}^{\loc}]_{T-1,T-2,T-3}, P_{\ill}^{(T)} \bullet P_{\leg}^{(T)}\right\} = - \frac{1}{2}[P_{\ill}^{\loc}]_{T-1,T-2,T-3} \bullet P_{\leg}^{(T)}.
    \end{align}
    Through a similar calculation we find 
    \begin{align}
        \label{eq:blah3}
        \D_{\loc}^{(T-1, T-2, T-3)} \circ \mathcal{P}_{\mathcal{B}(\mathcal{H}_{\leg}^{(T)}, \mathcal{H}_{\ill}^{(T)})} = &\frac{1}{4}\sum_{(x, y) \in \mathrm{S}_{\ill}^{\loc \times 2}}(\ket{x}\bra{y}_{T-1, T-2, T-3})P_{\leg}^{(T)}\bullet P_{\ill}^{(T)}(\ket{y}\bra{x}_{T-1, T-2, T-3}) \nonumber \\
                                                                                          &- \frac{1}{2}\left\{[P_{\ill}^{\loc}]_{T-1,T-2,T-3}, P_{\leg}^{(T)} \bullet P_{\ill}^{(T)}\right\} = - \frac{1}{2}P_{\leg}^{(T)} \bullet [P_{\ill}^{\loc}]_{T-1,T-2,T-3}.
    \end{align}
    Calculating $\D_{\loc}^{(T-1, T-2, T-3)} \circ \mathcal{P}_{\mathcal{B}(\mathcal{H}^{(T)}_{\ill})}$ we find
    \begin{align}
        \label{eq:blah4}
        \D_{\loc}^{(T-1, T-2, T-3)} \circ \mathcal{P}_{\mathcal{B}(\mathcal{H}_{\ill}^{(T)})} = &\frac{1}{4}\sum_{(x, y) \in \mathrm{S}_{\ill}^{\loc \times 2}}(\ket{x}\bra{y}_{T-1, T-2, T-3}) \bullet (\ket{y}\bra{x}_{T-1, T-2, T-3}) \nonumber \\
                                                                                          &- \frac{1}{2}\left\{[P_{\ill}^{\loc}]_{T-1,T-2,T-3}, P_{\ill}^{(T)} \bullet P_{\ill}^{(T)}\right\} .
    \end{align}
    Combining equations~\eqref{eq:blah1},~\eqref{eq:blah2},~\eqref{eq:blah3}, and~\eqref{eq:blah4} we find that $\D_{\loc}^{(T-1,T-2,T-3)}$ block diagonalizes as 
    \begin{align}
        \label{eq:blockchar}
        \D_{\loc}^{(T-1,T-2,T-3)} = &\D_{\loc}^{(T-1,T-2,T-3)}|_{\mathcal{B}(\mathcal{H}_{\leg}^{(T)})} \oplus \D_{\loc}^{(T-1,T-2,T-3)}|_{\mathcal{B}(\mathcal{H}_{\ill}^{(T)}, \mathcal{H}_{\leg}^{(T)})} \nonumber \\
                                  &\oplus \D_{\loc}^{(T-1,T-2,T-3)}|_{\mathcal{B}(\mathcal{H}_{\leg}^{(T)}, \mathcal{H}_{\ill}^{(T)})} \oplus \D_{\loc}^{(T-1,T-2,T-3)}|_{\mathcal{B}(\mathcal{H}_{\ill}^{(T)})}
    \end{align}
    where we have 
    \begin{equation}
        \begin{cases}
            \D_{\loc}^{(T-1,T-2,T-3)}|_{\mathcal{B}(\mathcal{H}_{\leg}^{(T)})}(\bullet) = 0(\bullet) \\
            \D_{\loc}^{(T-1,T-2,T-3)}|_{\mathcal{B}(\mathcal{H}_{\ill}^{(T)}, \mathcal{H}_{\leg}^{(T)})}(\bullet) = - \frac{1}{2}[P_{\ill}^{\loc}]_{T-1,T-2,T-3} \bullet \\
            \D_{\loc}^{(T-1,T-2,T-3)}|_{\mathcal{B}(\mathcal{H}_{\leg}^{(T)}, \mathcal{H}_{\ill}^{(T)})}(\bullet) = - \frac{1}{2} \bullet [P_{\ill}^{\loc}]_{T-1,T-2,T-3}
        \end{cases}
    \end{equation}
    and 
    \begin{equation}
        \label{eq:dlocill}
        \D_{\loc}^{(T-1,T-2,T-3)}|_{\mathcal{B}(\mathcal{H}_{\ill}^{(T)})}(\bullet) = \frac{1}{4}\sum_{(x, y) \in \mathrm{S}_{\ill}^{\loc \times 2}}(\ket{x}\bra{y}_{T-1, T-2, T-3}) \bullet (\ket{y}\bra{x}_{T-1, T-2, T-3}) - \frac{1}{2}\left\{[P_{\ill}^{\loc}]_{T-1,T-2, T-3},  \bullet \right\}.
    \end{equation}
    as stated.
\end{proof}
\begin{lemma}
    \label{lm:helper2}
    Let $\D_{\loc}^{(T-1,T-2,T-3)}$ be the $T-3$-term in the definition of $\D_{\unary}$ given in equation~\eqref{eq:Dunary} of Definition~\ref{def:localDQCSAT} and define $\mathcal{H}_{\leg}^{(T)}$ and $\mathcal{H}_{\ill}^{(T)}$  according to equation~\eqref{eq:globalillandlegsubspaces}. Let $\mathcal{P}_{\mathcal{B}(\mathcal{H}_{\ill}^{(T)})}$ be the projector onto $\mathcal{B}(\mathcal{H}_{\ill}^{(T)})$ and $P_{\ill}^{(T)}$ be the projector onto $\mathcal{H}_{\ill}^{(T)}$, then for all $T \geq 3$,
    \begin{align}
        \ker{\D_{\loc}^{(T-1,T-2,T-3)}|_{\mathcal{B}(\mathcal{H}_{\ill}^{(T)})}} \cap \mathcal{P}_{\mathcal{B}(\mathcal{H}_{\ill}^{(T)})}\left[\mathcal{B}(\mathcal{H}_{T-1})\otimes\left(\mathcal{B}(\mathcal{H}_{\leg}^{(T-1)}) \oplus P_{\ill}^{(T-1)}\right)\right] = P_{\ill}^{(T)}.
    \end{align}
\end{lemma}
\begin{proof}
    Define $\mathbb{V} = \mathcal{B}(\mathcal{H}_{T-1})\otimes\left(\mathcal{B}(\mathcal{H}_{\leg}^{(T-1)}) \oplus P_{\ill}^{(T-1)}\right)$, from Lemma~\ref{lm:helper1} we have $\D_{\loc}^{(T-1,T-2,T-3)}|_{\mathcal{B}(\mathcal{H}_{\ill}^{(T)})}$ is a valid Lindbladian dissipator acting on $\mathcal{B}(\mathcal{H}_{\ill}^{(T)})$ generated by jump operators $\{\frac{1}{2}\ket{x}\bra{y}_{T-1,T-2,T-3}\}_{x,y \in \mathrm{S}_{\ill}^{\loc}}$ with $\mathrm{S}_{\ill}^{\loc}$ and $\mathrm{S}_{\ill}^{\loc}$ and $\mathrm{S}_{\leg}^{\loc}$ defined according to equation~\eqref{eq:localillandlegstrings}. In addition, it is also Hermitian with respect to the Hilbert-Schmidt inner product restricted to $\mathcal{B}(\mathcal{H}_{\ill}^{(T)})$. Thus we can apply Theorem 7.2 of \cite{wolf2012quantum} implying
    \begin{equation}
        \label{eq:wolfthmapplied}
        \ker{\D_{\loc}^{(T-1,T-2,T-3)}|_{\mathcal{B}(\mathcal{H}_{\ill}^{(T)})}} = \left\{\frac{1}{2}\ket{x}\bra{y}_{T-1,T-1,T-3}\right\}_{x,y \in \mathrm{S}_{\ill}^{\loc}}'
    \end{equation}
    where the commutant is taken with respect to $\mathcal{B}(\mathcal{H}_{\ill}^{(T)})$. Let
    \begin{equation}
        A \in \ker{\D_{\loc}^{(T-1,T-2,T-3)}|_{\mathcal{B}(\mathcal{H}_{\ill}^{(T)})}} \cap \mathcal{P}_{\mathcal{B}(\mathcal{H}_{\ill}^{(T)})}\left[\mathbb{V}\right].
    \end{equation}
    Calculating $\mathcal{P}_{\mathcal{B}(\mathcal{H}_{\ill}^{(T)})}\left[\mathbb{V}\right]$ we find,
    \begin{align}
        \label{eq:Villegalsubspace}
        \mathcal{P}_{\mathcal{B}(\mathcal{H}_{\ill}^{(T)})}\left[\mathbb{V}\right] &= \mathcal{P}_{\mathcal{B}(\mathcal{H}_{\ill}^{(T)})}\left[\ket{1}\bra{1}_{T-1} \otimes \mathcal{B}(\mathcal{H}_{\leg}^{(T-1)})\right] \oplus \mathcal{B}(\mathcal{H}_{T-1})\otimes P_{\ill}^{(T-1)} \nonumber \\
                                                                                     &= \mathcal{B}(\ket{10}_{T-1,T-2} \otimes \mathcal{H}_{\leg}^{(T-2)}) \oplus \mathcal{B}(\mathcal{H}_{T-1})\otimes P_{\ill}^{(T-1)}.
    \end{align}
    Equations~\eqref{eq:wolfthmapplied} and~\eqref{eq:Villegalsubspace} imply 
    \begin{equation}
        A \in \left\{\frac{1}{2}\ket{x}\bra{y}_{T-1,T-2,T-3}\right\}_{x,y \in \mathrm{S}_{\ill}^{\loc}}' \cap \left(\mathcal{B}(\ket{10}_{T-1,T-2} \otimes \mathcal{H}_{\leg}^{(T-2)}) \oplus \mathcal{B}(\mathcal{H}_{T-1})\otimes P_{\ill}^{(T-1)}\right),
    \end{equation}
    which further implies $A \in \mathcal{B}(\ket{10}_{T-1,T-2} \otimes \mathcal{H}_{\leg}^{(T-2)}) \oplus \mathcal{B}(\mathcal{H}_{T-1})\otimes P_{\ill}^{(T-1)}$, and 
    \begin{equation}
        \label{eq:requirement3}
        \left[\frac{1}{2}\ket{x}\bra{y}_{T-1,T-2,T-3}, A \right] = 0 
    \end{equation}
    for all $x, y \in \mathrm{S}_{\ill}^{\loc}$. 

    Consider the case of $x = y = 010$. We find
    \begin{align}
        \left[\frac{1}{2}\ket{010}\bra{010}_{T-1,T-2,T-3}, A\right] =  \frac{1}{2}&\left[\ket{010}\bra{010}_{T-1,T-2,T-3}, \ket{10}\bra{10}_{T-1,T-2} \otimes B\right] \nonumber \\
        &+ \frac{1}{2}\left[\ket{010}\bra{010}_{T-1,T-2,T-3}, C_{T+1} \otimes P_{\ill}^{(T-1)}\right],
    \end{align}
    where we have $B \in \mathcal{B}(\mathcal{H}_{\leg}^{(T-2)})$. The above equation simplifies to 
    \begin{align}
        \left[\frac{1}{2}\ket{010}\bra{010}_{T-1,T-2,T-3}, A\right] &= \frac{1}{2}\left[\ket{0}\bra{0}, C\right]_{T-1} \otimes \ket{10}\bra{10}_{T-2,T-3}.
    \end{align}
    Since we require equation~\eqref{eq:requirement3} we find 
    \begin{equation}
        [\ket{0}\bra{0}, C] = 0
    \end{equation}
    A similar calculation for the case $x = y = 110$ we find
    \begin{equation}
        [\ket{1}\bra{1}, C] = 0
    \end{equation}
    which implies 
    \begin{equation}
        C = c_{0}\ket{0}\bra{0} + c_{1}\ket{1}\bra{1}.
    \end{equation}

    Consider the case of $x=010$ and $y=110$, we find
    \begin{align}
        \left[\frac{1}{2}\ket{010}\bra{110}_{T-1,T-2,T-3}, A\right] =  \frac{1}{2}&\left[\ket{010}\bra{110}_{T-1,T-2,T-3}, \ket{10}\bra{10}_{T-1,T-2} \otimes B\right] \nonumber \\
        &+ \frac{1}{2}\left[\ket{010}\bra{110}_{T-1,T-2,T-3}, C_{T-1} \otimes P_{\ill}^{(T-1)}\right].
    \end{align}
    Simplifying the above equation gives us 
    \begin{align}
        \left[\frac{1}{2}\ket{010}\bra{110}_{T-1,T-2,T-3}, A\right] = \frac{1}{2}\left[\ket{0}\bra{1}, C\right]_{T-1} \otimes \ket{10}\bra{10}_{T-2,T-3} = \frac{1}{2}\left(c_{1}  - c_{0}\right) \ket{010}\bra{110}_{T+1, T,T-1}.
    \end{align}
    From our requirement of equation~\eqref{eq:requirement3} we have $c_{0} = c_{1} = c$ implying
    \begin{equation}
        C = c \mathbb{I}_{T-1}.
    \end{equation}

    Consider the case of $x = y = 101$, calculating the commutator we find,
    \begin{align}
        \left[\frac{1}{2}\ket{101}\bra{101}_{T-1,T-2,T-3}, A\right] =  \frac{1}{2}&\left[\ket{101}\bra{101}_{T-1,T-2,T-3}, \ket{10}\bra{10}_{T-1,T-2} \otimes B\right] \nonumber \\
        &+ \frac{1}{2}\left[\ket{101}\bra{101}_{T-1,T-2,T-3}, C_{T-1} \otimes P_{\ill}^{(T-1)}\right]
    \end{align}
    which becomes
    \begin{align}
        \left[\frac{1}{2}\ket{101}\bra{101}_{T-1,T-2,T-3}, A\right] &=  \frac{1}{2}\ket{10}\bra{10}_{T-1, T-2} \otimes \left[\ket{1}\bra{1}_{T-3}, B\right] + \frac{c}{2}\ket{1}\bra{1}_{T-1} \otimes \left[\ket{01}\bra{01}_{T-2,T-3}, P_{\ill}^{(T-1)}\right] \nonumber \\
                                                                  &= \frac{}{2}\ket{10}\bra{10}_{T-1, T-2} \otimes \left[\ket{1}\bra{1}_{T-3}, B\right].
    \end{align}
    Setting the above equation equal to zero implies 
    \begin{equation}
        \label{eq:constraintonB}
        \left[\ket{1}\bra{1}_{T-3}, B\right] = 0.
    \end{equation}
    Now consider the case of $x=y=100$, similarly we find 
    \begin{align}
        \left[\frac{1}{2}\ket{100}\bra{100}_{T-1,T-2,T-3}, A\right] =  \frac{1}{2}&\left[\ket{100}\bra{100}_{T-1,T-2,T-3}, \ket{10}\bra{10}_{T-1,T-2} \otimes B\right] \nonumber \\
        &+ \frac{1}{2}\left[\ket{100}\bra{100}_{T-1,T-2,T-3}, C_{T-1} \otimes P_{\ill}^{(T-1)}\right]
    \end{align}
    which simiplifies to become
    \begin{align}
        \left[\frac{1}{2}\ket{100}\bra{100}_{T-1,T-2,T-3}, A\right] &=  \frac{1}{2}\ket{10}\bra{10}_{T-1,T-2} \otimes \left[\ket{0}\bra{0}_{T-3}, B\right] + \frac{c}{2}\ket{1}\bra{1}_{T-1} \otimes \left[\ket{00}\bra{00}_{T-2,T-3}, P_{\ill}^{(T-1)}\right] \nonumber \\
                                                                  &= \frac{1}{2}\ket{10}\bra{10}_{T-1,T-2} \otimes \left[\ket{0}\bra{0}_{T-3}, B\right] 
    \end{align}
    Asserting our requirement of equation~\eqref{eq:requirement3} implies 
    \begin{equation}
        \label{eq:constraintonB1}
        \left[\ket{0}\bra{0}_{T-3}, B\right]=0.
    \end{equation}
    Since $B \in \mathcal{B}(\mathcal{H}_{\leg}^{(T-2)})$ we can write 
    \begin{align}
        B = \sum_{x,y \in \mathrm{S}_{\leg}^{(T-2)}}b_{x,y}\ket{x}\bra{y}.
    \end{align}
    The constraint given by equation~\eqref{eq:constraintonB} gives us 
    \begin{align}
        [\ket{1}\bra{1}_{T-3} , B ] = \sum_{y \in \mathrm{S}_{\leg}^{(T-2)}}b_{1^{T-2}, y}\ket{1^{T-2}}\bra{y} - \sum_{x \in \mathrm{S}_{\leg}^{(T-2)}}b_{x,1^{T-2}}\ket{x}\bra{1^{T-2}} = 0
    \end{align}
    which implies 
    \begin{align}
        \label{eq:constr1}
        b_{1^{T-2}, x} = b_{x, 1^{T-2}} = 0  \text{ for } x \in \mathrm{S}_{\leg}^{(T-2)} \text{ and } x \neq 1^{T-2}.
    \end{align}
    The constraint given by equation~\eqref{eq:constraintonB1} entails
    \begin{align}
        [\ket{0}\bra{0}_{T-1} , B ] &= \sum_{x_{<T-2} \in \mathrm{S}_{\leg}^{(T-3)}}\sum_{y \in \mathrm{S}_{\leg}^{(T-2)}}b_{0x_{< T-1}, y}\ket{0x_{<T-2}}\bra{y} - \sum_{y_{<T-2} \in \mathrm{S}_{\leg}^{(T-3)}}\sum_{x \in \mathrm{S}_{\leg}^{(T-2)}}b_{x,0y_{<T-2}}\ket{x}\bra{0y_{< T-2}} = 0,
    \end{align}
    where for the first equality we have used the fact that $x \in \mathrm{S}_{\leg}^{(t)}$ implies $x_{< t} \in \mathrm{S}_{\leg}^{(t-1)}$. The above equaiton implies 
    \begin{align}
        \label{eq:constr2}
        b_{0x,y} = b_{y, 0x} = 0 \text{ for } x \in \mathrm{S}_{\leg}^{(T-3)}, y \in \mathrm{S}_{\leg}^{(T-2)}, \text{ and } 0x \neq y.
    \end{align}
    Combining the constraints of equation~\eqref{eq:constr1} and~\eqref{eq:constr2} we find 
    \begin{align}
        b_{x,y} = 0 \text{ for } x,y \in \mathrm{S}_{\leg}^{(T-2)} \text{ and } x \neq y,
    \end{align}
    implying we can write $B$ as 
    \begin{align}
        B = \sum_{s \in \mathrm{S}_{\leg}^{(T-2)}} b_{s} \ket{s}\bra{s}.
    \end{align}

    The next case we consider fixes $x = 100$ and $y=110$. Calculating our commutator and applying or constraint of equation~\eqref{eq:requirement3} we find
    \begin{align}
        \left[\frac{1}{2}\ket{100}\bra{110}_{T-1,T-2,T-3}, A\right] &= \frac{1}{2}\left[\ket{100}\bra{110}_{T-1,T-2,T-3}, \ket{10}\bra{10}_{T-1,T-2}\otimes B\right] + \frac{c}{2}\ket{1}\bra{1}_{T-1} \otimes \left[\ket{00}\bra{10}_{T-2,T-3},  P_{\ill}^{(T-j)}\right] \nonumber \\
                                                                  &= -\frac{1}{2}\ket{10}\bra{11}_{T-1,T-2} \otimes (B\ket{0}\bra{0}_{T-3}) + \frac{c}{2}\ket{100}\bra{110}_{T-1,T-2,T-3} \nonumber \\
                                                                  & \quad - \frac{c}{2}\ket{1}\bra{1}_{T-1}\otimes (P_{\ill}^{(T-1)}\ket{00}\bra{10}_{T-2,T-3}) \nonumber \\
                                                                  &= -\frac{1}{2}\ket{10}\bra{11}_{T-1,T-2} \otimes (B\ket{0}\bra{0}_{T-3}) + \frac{c}{2}\ket{100}\bra{110}_{T-1,T-2,T-3} \nonumber \\
                                                                  & \quad - \frac{c}{2}\ket{1}\bra{1}_{T-1}\otimes \sum_{s\in\{0,1\}^{T-3}}\left(P_{\ill}^{(T-1)}\ket{00s}\bra{10s}\right) \nonumber \\
                                                                  &= -\frac{1}{2}\ket{10}\bra{11}_{T-1,T-2} \otimes (B\ket{0}\bra{0}_{T-3}) + \frac{c}{2}\ket{100}\bra{110}_{T-1,T-2,T-3}\otimes(\mathbb{I}^{(T-3)}  -  P_{\ill}^{(T-3)}) \nonumber \\
                                                                  &= -\frac{1}{2}\ket{10}\bra{11}_{T-1,T-2} \otimes (B\ket{0}\bra{0}_{T-3}) + \frac{c}{2}\ket{100}\bra{110}_{T-1,T-2,T-3}\otimes P_{\leg}^{(T-3)} \nonumber \\
                                                                  &= -\frac{1}{2}\ket{10}\bra{11}_{T-1,T-2} \otimes \sum_{ s \in \mathrm{S}_{\leg}^{(T-3)}} b_{0s}\ket{0s}\bra{0s} + \frac{c}{2}\ket{100}\bra{110}_{T-1,T-2,T-3}\otimes P_{\leg}^{(T-3)} \nonumber \\
                                                                  &= -\frac{1}{2}\ket{100}\bra{110}_{T-1,T-2, T-3} \otimes \sum_{s \in \mathrm{S}_{\leg}^{(T-3)}} b_{0s}\ket{s}\bra{s} + \frac{c}{2}\ket{100}\bra{110}_{T-1,T-2,T-3}\otimes P_{\leg}^{(T-3)} \nonumber \\
                                                                  &= 0.
    \end{align}
    The above equation implies
    \begin{equation}
        b_{0s} = c \text{ for } s \in \mathrm{S}_{\leg}^{(T-3)}
    \end{equation}
    Therefore $B$ simplifies to become 
    \begin{equation}
        B =  b_{1^{T-2}} \ket{1^{T-2}}\bra{1^{T-2}} + c \ket{0}\bra{0}_{T-3} \otimes P_{\leg}^{(T-3)}.
    \end{equation}

    Consider the final case of $x=110$ and $y=101$. We find 
    \begin{align}
        \left[\frac{1}{2}\ket{110}\bra{101}_{T-1,T-2,T-3}, A \right] &= \frac{1}{2}\left[\ket{110}\bra{101}_{T-1,T-2,T-3}, \ket{10}\bra{10}_{T-1,T-2} \otimes B\right] \nonumber \\
        & \quad + \frac{1}{2}\left[\ket{110}\bra{101}_{T-1,T-2,T-3}, C_{T-1}\otimes P_{\ill}^{(T-1)}\right] \nonumber \\
                                                                   &= \frac{1}{2}\ket{11}\bra{10}_{T-1,T-2} \otimes\left(\ket{0}\bra{1}_{T-3} B\right)  + \frac{c}{2}\ket{1}\bra{1}_{T-1} \otimes \left[\ket{10}\bra{01}_{T-2,T-3}, P_{\ill}^{(T-1)}\right] \nonumber \\
                                                                   &= \frac{b_{1^{T-2}}}{2} \otimes\ket{1101^{T-3}}\bra{101^{T-2}} \nonumber \\
                                                                   &\quad + \frac{c}{2}\ket{1}\bra{1}_{T-1} \otimes \left(\ket{10}\bra{01}_{T-2,T-3} P_{\ill}^{(T-1)} - P_{\ill}^{(T-1)}\ket{10}\bra{01}_{T-2,T-3} \right) \nonumber \\
                                                                   &= \frac{b_{1^{T-2}}}{2} \otimes\ket{1101^{T-3}}\bra{101^{T-2}} \nonumber \\
                                                                   & \quad + \frac{c}{2}\ket{1}\bra{1}_{T-1} \otimes \left(\sum_{s\in\{0,1\}^{T-3}}\ket{10s}\bra{01s} P_{\ill}^{(T-1)} - \ket{10}\bra{01}_{T-2,T-3} \right) \nonumber \\
                                                                   &= \frac{b_{1^{T-2}}}{2} \otimes\ket{1101^{T-3}}\bra{101^{T-2}} + \frac{c}{2}\ket{1}\bra{1}_{T-1} \otimes \left(\sum_{s\in\{0,1\}^{T-3} \setminus \{1^{T-3}\}}\ket{10s}\bra{01s}  - \ket{10}\bra{01}_{T-2,T-3} \right) \nonumber \\
                                                                   &= \frac{b_{1^{T-2}}}{2} \otimes\ket{1101^{T-3}}\bra{101^{T-2}} - \frac{c}{2} \ket{1101^{T-3}}\bra{101^{T-2}} = 0
    \end{align} 
    which implies $b_{1^{T-2}} = c$. Thus we find 
    \begin{equation}
        B =  c\left( \ket{1^{T-2}}\bra{1^{T-2}} +  \ket{0}\bra{0}_{T-3} \otimes P_{\leg}^{(T-3)}\right)
    \end{equation}
    implying 
    \begin{align}
        A = \ket{10}\bra{10}_{T-1,T-2} \otimes B + C \otimes P_{\ill}^{(T-1)} &= c \left(\ket{101^{T-2}}\bra{101^{T-2}} +  \ket{100}\bra{100}_{T-1,T-2,T-3} \otimes P_{\leg}^{(T-3)} + \mathbb{I}_{T-1} \otimes P_{\ill}^{(T-1)} \right) \nonumber \\
                                                                          &= c \left(\ket{10}\bra{10}_{T-1,T-2} \otimes P_{\leg}^{(T-2)} + \mathbb{I}_{T-1} \otimes P_{\ill}^{(T-1)} \right) = c P_{\ill}^{(T)}.
    \end{align}

    Therefore for all $A \in \ker{\D_{\loc}^{(T-1,T-2, T-3)}|_{\mathcal{B}(\mathcal{H}_{\ill}^{(T)})}} \cap \mathcal{P}_{\mathcal{B}(\mathcal{H}_{\ill}^{(T)})}\left[\mathbb{V}\right]$ , 
    \begin{equation}
        A = c P_{\ill}^{(T)}
    \end{equation}
    for some $c \in \mathbb{C}$ implying 
    \begin{equation}
        \ker{\D_{\loc}^{(T-1,T-2, T-3)}|_{\mathcal{B}(\mathcal{H}_{\ill}^{(T)})}} \cap \mathcal{P}_{\mathcal{B}(\mathcal{H}_{\ill}^{(T)})}\left[\mathbb{V}\right] = P_{\ill}^{(T)}
    \end{equation}
    as stated.
\end{proof}

We conclude this section with our proof of the statement fully characterizing the kernel of our unary clock dissipator.
\begin{proposition*}[Kernel of Unary Clock Dissipator, Proposition~\ref{prop:DuKer}]
    Let $\D_{\unary}$ be the Lindbladian dissipator defined according to equation~\eqref{eq:Dunary} of Definition~\ref{def:localDQCSAT} and $P_{\ill}^{(T)}$ be the projection onto $\mathcal{H}_{\ill}^{(T)}$ with $\mathcal{H}_{\leg}^{(T)}$ and $\mathcal{H}_{\ill}^{(T)}$ defined according to equation~\eqref{eq:globalillandlegsubspaces} , then 
    \begin{equation}
        \ker{\D_{\unary}} = \mathcal{B}(\mathcal{H}_{MA}) \otimes \left(\mathcal{B}(\mathcal{H}_{\leg}^{(T)}) \oplus P_{\ill}^{(T)}\right).
    \end{equation}
\end{proposition*}
\begin{proof}
    Since $\D_{\unary}$ only acts non-trivially on the clock register $C$ it suffices to restrict $\D_{\unary}$ to $\mathcal{B}(\mathcal{H}_{C})$ and show 
    \begin{equation}
        \label{eq:inducthyp}
        \ker{\D_{\unary}|_{\mathcal{B}(\mathcal{H}_{C})}} = \mathcal{B}(\mathcal{H}_{\leg}^{(T)}) \oplus P_{\ill}^{(T)}
    \end{equation}
    We proceed via induction. Consider the base case of $T=3$, from equation~\eqref{eq:Dunary} of Definition~\ref{def:localDQCSAT} we have 
    \begin{equation}
        \D_{\unary}|_{\mathcal{B}(\mathcal{H}_{C})}(\bullet) = \D_{\loc}^{(2,1,0)}(\bullet) =  \sum_{(x, y) \in \mathrm{S}_{\ill}^{\loc \times 2}} \left(\frac{1}{4}\ket{x}\bra{y}\bullet\ket{y}\bra{x} - \frac{1}{8}\left\{\ket{y}\bra{y}, \bullet\right\}\right)
    \end{equation}
    where $S_{\ill}^{\loc}$ is defined according to equation~\eqref{eq:localillandlegstrings}. Applying Lemma~\ref{lm:DuKerbasecase} we find our base case holds.

    We now proceed to proving the inductive step. Assume the statement of equation~\eqref{eq:inducthyp} holds for all $T \geq 3$. For a clock register $C$ of size $T+1$ we find from equation~\eqref{eq:Dunary} of Definition~\ref{def:localDQCSAT} 
    \begin{equation}
        \D_{\unary}|_{\mathcal{B}(\mathcal{H}_{C})} =  \D_{\loc}^{(T, T-1, T-2)} + \sum_{t=0}^{T-3}\D_{\loc}^{(t+2, t+1, t)} = \D_{\loc}^{(T, T-1, T-2)} + \D^{(T)},
    \end{equation}
    where we have defined $\D^{(T)} = \sum_{t=0}^{T-3}\D_{\loc}^{(t+2, t+1, t)}$. Since both $\D_{\loc}^{(T, T-1, T-2)}$ and $\D^{(T)}$ are valid Lindbladian dissipators acting on $\mathcal{B}(\mathcal{H}_{C})$ and are Hermitian with respect to the Hilbert-Schmidt inner product, we have $\D_{\loc}^{(T, T-1, T-2)}, \D^{(T)} \preceq 0$ implying 
    \begin{equation}
        \ker{\D_{\unary}|_{\mathcal{B}(\mathcal{H}_{C})}} =  \ker{\D_{\loc}^{(T, T-1, T-2)}} \cap \ker{\D^{(T)}}.
    \end{equation}
    Note, $\D^{(T)}$ is simply $\D_{\unary}|_{\mathcal{B}(\mathcal{H}_{C})}$ with a clock register of size $T$ acting on a register of size $T+1$. Therefore applying our induction hypothesis to the above equation we find
    \begin{equation}
        \ker{\D_{\unary}|_{\mathcal{B}(\mathcal{H}_{C})}} = \ker{\D_{\loc}^{(T, T-1, T-2)}} \cap \left[\mathcal{B}(\mathcal{H}_{T}) \otimes \left(\mathcal{B}(\mathcal{H}_{\leg}^{(T)}) \oplus P_{\ill}^{(T)}\right)\right].
    \end{equation}
    Lemma~\ref{lm:helper1} implies $\ker{\D_{\loc}^{(T,T-1,T-2)}} = \mathcal{B}(\mathcal{H}_{\leg}^{(T+1)}) \oplus \ker{\D_{\loc}^{(T,T-1,T-2)}|_{\mathcal{B}(\mathcal{H}_{\leg}^{(T+1)})^{\perp}}}$ allowing us to write
    \begin{equation}
        \label{eq:inductivestepapplied}
        \ker{\D_{\unary}|_{\mathcal{B}(\mathcal{H}_{C})}} = \left(\mathcal{B}(\mathcal{H}_{\leg}^{(T+1)}) \oplus \ker{\D_{\loc}^{(T,T-1,T-2)}|_{\mathcal{B}(\mathcal{H}_{\leg}^{(T+1)})^{\perp}}}\right)
         \cap \left[\mathcal{B}(\mathcal{H}_{T}) \otimes \left(\mathcal{B}(\mathcal{H}_{\leg}^{(T)}) \oplus P_{\ill}^{(T)}\right)\right].
    \end{equation}
    Since $\mathcal{B}(\mathcal{H}_{\leg}^{(T+1)}) \subseteq \mathcal{B}(\mathcal{H}_{T}) \otimes \mathcal{B}(\mathcal{H}_{\leg}^{(T)})$ we can apply the modular law to equation~\eqref{eq:inductivestepapplied} giving us
    \begin{align}
        \label{eq:inductionstep1}
        \ker{\D_{\unary}|_{\mathcal{B}(\mathcal{H}_{C})}} = \mathcal{B}(\mathcal{H}_{\leg}^{(T+1)}) \oplus \left(\ker{\D_{\loc}^{(T,T-1,T-2)}|_{\mathcal{B}(\mathcal{H}_{\leg}^{(T+1)})^{\perp}}} \cap \mathcal{B}(\mathcal{H}_{T}) \otimes \left(\mathcal{B}(\mathcal{H}^{(T)}_{\leg}) \oplus P_{\ill}^{(T)}\right)\right).
    \end{align}
    Defining $\mathbb{V} = \mathcal{B}(\mathcal{H}_{T}) \otimes \left(\mathcal{B}(\mathcal{H}^{(T)}_{\leg}) \oplus P_{\ill}^{(T)}\right)$ and applying our block characterization of $\D_{\loc}^{(T,T-1,T-2)}$ given in Lemma~\ref{lm:helper1} we find
    \begin{align}
        \label{eq:compofinduction}
        \ker{\D_{\loc}^{(T,T-1,T-2)}|_{\mathcal{B}(\mathcal{H}_{\leg}^{(T+1)})^{\perp}}} \cap \mathbb{V} = \Big(&\ker{\D_{\loc}^{(T,T-1,T-2)}|_{\mathcal{B}(\mathcal{H}_{\ill}^{(T+1)}, \mathcal{H}_{\leg}^{(T+1)})}} \oplus \ker{\D_{\loc}^{(T,T-1,T-2)}|_{\mathcal{B}(\mathcal{H}_{\leg}^{(T+1)}, \mathcal{H}_{\ill}^{(T+1)})}} \nonumber \\
                                                                                                                  &\oplus \ker{\D_{\loc}^{(T,T-1,T-2)}|_{\mathcal{B}(\mathcal{H}_{\ill}^{(T+1)})}}\Big) \cap \mathbb{V}
    \end{align}
    In order to distribute the intersection we compute $\mathcal{P}_{\mathcal{B}(\mathcal{H}_{\ill}^{(T+1)}, \mathcal{H}_{\leg}^{(T+1)})}\left[\mathbb{V}\right]$, $\mathcal{P}_{\mathcal{B}(\mathcal{H}_{\leg}^{(T+1)}, \mathcal{H}_{\ill}^{(T+1)})}\left[\mathbb{V}\right]$, and $\mathcal{P}_{\mathcal{B}(\mathcal{H}_{\ill}^{(T+1)})}\left[\mathbb{V}\right]$. We find
    \begin{align}
        \mathcal{P}_{\mathcal{B}(\mathcal{H}_{\ill}^{(T+1)}, \mathcal{H}_{\leg}^{(T+1)})}\left[\mathbb{V}\right] &= \mathcal{P}_{\mathcal{B}(\mathcal{H}_{\ill}^{(T+1)}, \mathcal{H}_{\leg}^{(T+1)})}\left[\mathcal{B}(\mathcal{H}_{T}) \otimes \mathcal{B}(\mathcal{H}_{\leg}^{(T)})\right] \\
                                                                                                                   &= \mathcal{P}_{\mathcal{B}(\mathcal{H}_{\ill}^{(T+1)}, \mathcal{H}_{\leg}^{(T+1)})}\left[\mathcal{B}(\mathcal{H}_{T}) \otimes \mathcal{B}(\ket{1^{T}} \oplus \ket{0}_{T-1} \otimes \mathcal{H}^{(T-1)}_{\leg})\right] \nonumber \\
                                                                                                                     &= \mathcal{P}_{\mathcal{B}(\mathcal{H}_{\ill}^{(T+1)}, \mathcal{H}_{\leg}^{(T+1)})}\left[\mathcal{B}(\mathcal{H}_{T}) \otimes \ket{1^{T}}\bra{1^{T}}\right] \oplus \mathcal{P}_{\mathcal{B}(\mathcal{H}_{\ill}^{(T+1)}, \mathcal{H}_{\leg}^{(T+1)})}\left[\mathcal{B}(\mathcal{H}_{T})\otimes \mathcal{B}(\ket{1^{T}}, \ket{0}_{T-1}\otimes\mathcal{H}_{\leg}^{(T-1)}) \right] \nonumber \\
                                                                                                                       &\quad \oplus \mathcal{P}_{\mathcal{B}(\mathcal{H}_{\ill}^{(T+1)}, \mathcal{H}_{\leg}^{(T+1)})}\left[\mathcal{B}(\mathcal{H}_{T})\otimes \mathcal{B}(\ket{0}_{T-1}\otimes\mathcal{H}_{\leg}^{(T-1)}, \ket{1^{T}} )\right] \nonumber \\
                                                                                                                         &\quad \oplus \mathcal{P}_{\mathcal{B}(\mathcal{H}_{\ill}^{(T+1)}, \mathcal{H}_{\leg}^{(T+1)})}\left[\mathcal{B}(\mathcal{H}_{T})\otimes \mathcal{B}(\ket{0}_{T-1}\otimes\mathcal{H}_{\leg}^{(T-1)})\right] \nonumber \\
                                                                                                                     &= \mathcal{P}_{\mathcal{B}(\mathcal{H}_{\ill}^{(T+1)}, \mathcal{H}_{\leg}^{(T+1)})}\left[\mathcal{B}(\mathcal{H}_{T})\otimes \mathcal{B}(\ket{0}_{T-1}\otimes\mathcal{H}_{\leg}^{(T-1)}, \ket{1^{T}} )\right] \nonumber \\
                                                                                                                         &\quad \oplus \mathcal{P}_{\mathcal{B}(\mathcal{H}_{\ill}^{(T+1)}, \mathcal{H}_{\leg}^{(T+1)})}\left[\mathcal{B}(\mathcal{H}_{T})\otimes \mathcal{B}(\ket{0}_{T-1}\otimes\mathcal{H}_{\leg}^{(T-1)})\right] \nonumber \\
                                                                                                                           &= \mathcal{B}(\ket{10}_{T,T-1}\otimes\mathcal{H}_{\leg}^{(T-1)}, \ket{01^{T}} )\oplus  \mathcal{B}(\ket{10}_{T,T-1}\otimes\mathcal{H}_{\leg}^{(T-1)}, \ket{1^{T+1}} )  \nonumber \\
                                                                                                                             &\quad \oplus \mathcal{B}(\ket{10}_{T,T-1}\otimes\mathcal{H}_{\leg}^{(T-1)}, \ket{00}_{T,T-1}\otimes\mathcal{H}_{\leg}^{(T-1)}) \nonumber \\
                                                                                                                               &= \mathcal{B}(\ket{10}_{T,T-1}\otimes\mathcal{H}_{\leg}^{(T-1)}, \mathcal{H}_{\leg}^{(T+1)} ) \subseteq \mathcal{B}(\mathcal{H}_{T})\otimes \mathcal{B}(\mathcal{H}_{\leg}^{(T)}) \subseteq \mathbb{V}
    \end{align}
    Through a similar calculation we find
    \begin{align}
        \mathcal{P}_{\mathcal{B}(\mathcal{H}_{\leg}^{(T+1)}, \mathcal{H}_{\ill}^{(T+1)})}\left[\mathbb{V}\right] &= \mathcal{B}(\mathcal{H}_{\leg}^{(T+1)}, \ket{10}_{T,T-1}\otimes\mathcal{H}_{\leg}^{(T-1)}) \subseteq \mathbb{V},
    \end{align}
    and
    \begin{align}
        \label{eq:billegalsubspace}
        \mathcal{P}_{\mathcal{B}(\mathcal{H}_{\ill}^{(T+1)})}\left[\mathbb{V}\right] &= \mathcal{P}_{\mathcal{B}(\mathcal{H}_{\ill}^{(T+1)})}\left[\ket{1}\bra{1}_{T} \otimes \mathcal{B}(\mathcal{H}_{\leg}^{(T)})\right] \oplus \mathcal{B}(\mathcal{H}_{T})\otimes P_{\ill}^{(T)} \nonumber \\
                                                                                     &= \mathcal{B}(\ket{10}_{T,T-1} \otimes \mathcal{H}_{\leg}^{(T-1)}) \oplus \mathcal{B}(\mathcal{H}_{T})\otimes P_{\ill}^{(T)} \subseteq \mathbb{V}.
    \end{align}
    Thus we can distribute the intersection in equation~\eqref{eq:compofinduction} giving us
    \begin{align}
        \label{eq:compofinduction1}
        \ker{\D_{\loc}^{(T,T-1,T-2)}|_{\mathcal{B}(\mathcal{H}_{\leg}^{(T+1)})^{\perp}}} \cap \mathbb{V} = &\left(\ker{\D_{\loc}^{(T,T-1,T-2)}|_{\mathcal{B}(\mathcal{H}_{\ill}^{(T+1)}, \mathcal{H}_{\leg}^{(T+1)})}} \cap \mathcal{P}_{\mathcal{B}(\mathcal{H}_{\ill}^{(T+1)}, \mathcal{H}_{\leg}^{(T+1)})}\left[\mathbb{V}\right] \right) \nonumber \\
                                                                                                             &\oplus \left(\ker{\D_{\loc}^{(T,T-1,T-2)}|_{\mathcal{B}(\mathcal{H}_{\leg}^{(T+1)}, \mathcal{H}_{\ill}^{(T+1)})}} \cap \mathcal{P}_{\mathcal{B}(\mathcal{H}_{\leg}^{(T+1)}, \mathcal{H}_{\ill}^{(T+1)})}\left[\mathbb{V}\right]\right) \nonumber \\
                                                                                                               &\oplus \left(\ker{\D_{\loc}^{(T,T-1,T-2)}|_{\mathcal{B}(\mathcal{H}_{\ill}^{(T+1)})}} \cap \mathcal{P}_{\mathcal{B}(\mathcal{H}_{\ill}^{(T+1)})}\left[\mathbb{V}\right]\right)
    \end{align}
    We proceed to break down each of the terms in the direct sum given above.

    Consider the first term in our direct sum of equation~\eqref{eq:compofinduction1}. Let
    \begin{equation}
        A \in \mathcal{P}_{\mathcal{B}(\mathcal{H}_{\ill}^{(T+1)}, \mathcal{H}_{\leg}^{(T+1)})}\left[\mathbb{V}\right] = \mathcal{B}(\ket{10}_{T,T-1}\otimes\mathcal{H}_{\leg}^{(T-1)}, \mathcal{H}_{\leg}^{(T+1)}).
    \end{equation}
    Thus we can write
    \begin{equation}
        A = \sum_{s \in \mathrm{S}_{\leg}^{(T-1)}}\sum_{x \in \mathrm{S}_{\leg}^{(T+1)}}a_{sx}\ket{10s}\bra{x}.
    \end{equation}
    For $A$ to be in $\ker{\D_{\loc}^{(T,T-1,T-2)}|_{\mathcal{B}(\mathcal{H}_{\ill}^{(T+1)}, \mathcal{H}_{\leg}^{(T+1)})}}  \cap \mathcal{P}_{\mathcal{B}(\mathcal{H}_{\ill}^{(T+1)}, \mathcal{H}_{\leg}^{(T+1)})}\left[\mathbb{V}\right]$ we require
    \begin{equation}
        \label{eq:requirement1}
        \D_{\loc}^{(T,T-1,T-2)}|_{\mathcal{B}(\mathcal{H}_{\ill}^{(T+1)}, \mathcal{H}_{\leg}^{(T+1)})}(A) = 0.
    \end{equation}
    We can calculate 
    \begin{align}
        \D_{\loc}^{(T,T-1,T-2)}|_{\mathcal{B}(\mathcal{H}_{\ill}^{(T+1)}, \mathcal{H}_{\leg}^{(T+1)})}(A) &= -\frac{1}{2}\left[P_{\ill}^{\loc}\right]_{T,T-1,T-2}A \nonumber \\
                                                                                                            &= -\frac{1}{2}\sum_{s \in \mathrm{S}_{\leg}^{(T-1)}}\sum_{x \in \mathrm{S}_{\leg}^{(T+1)}}a_{sx}\left[P_{\ill}^{\loc}\right]_{T,T-1,T-2}\ket{10s}\bra{x} = -\frac{1}{2}A ,
    \end{align}
    implying $A = 0$ if we require equation~\eqref{eq:requirement1}. Thus we find
    \begin{equation}
        \label{eq:zeropart}
        \ker{\D_{\loc}^{(T,T-1,T-2)}|_{\mathcal{B}(\mathcal{H}_{\ill}^{(T+1)}, \mathcal{H}_{\leg}^{(T+1)})}}  \cap \mathcal{P}_{\mathcal{B}(\mathcal{H}_{\ill}^{(T+1)}, \mathcal{H}_{\leg}^{(T+1)})}\left[\mathbb{V}\right] =  \{0\}.
    \end{equation}

    Similarly for the second term in our direct sum of equation~\eqref{eq:compofinduction1}, let
    \begin{equation}
        A \in \mathcal{P}_{\mathcal{B}(\mathcal{H}_{\leg}^{(T+1)}, \mathcal{H}_{\ill}^{(T+1)})}\left[\mathbb{V}\right] = \mathcal{B}(\mathcal{H}_{\leg}^{(T+1)}, \ket{10}_{T,T-1}\otimes\mathcal{H}_{\leg}^{(T-1)}).
    \end{equation}
    implying we can write
    \begin{equation}
        A = \sum_{s \in \mathrm{S}_{\leg}^{(T-1)}}\sum_{x \in \mathrm{S}_{\leg}^{(T+1)}}a_{xs}\ket{x}\bra{10s}.
    \end{equation}
    For $A$ to be in $\ker{\D_{\loc}^{(T,T-1,T-2)}|_{\mathcal{B}(\mathcal{H}_{\leg}^{(T+1)}, \mathcal{H}_{\ill}^{(T+1)})}}  \cap \mathcal{P}_{\mathcal{B}(\mathcal{H}_{\leg}^{(T+1)}, \mathcal{H}_{\ill}^{(T+1)})}\left[\mathbb{V}\right]$ we require
    \begin{equation}
        \label{eq:requirement2}
        \D_{\loc}^{(T,T-1,T-2)}|_{\mathcal{B}(\mathcal{H}_{\leg}^{(T+1)}, \mathcal{H}_{\ill}^{(T+1)})}(A) = 0.
    \end{equation}
    We can calculate 
    \begin{align}
        \D_{\loc}^{(T,T-1,T-2)}|_{\mathcal{B}(\mathcal{H}_{\leg}^{(T+1)}, \mathcal{H}_{\ill}^{(T+1)})}(A) &= -\frac{1}{2}A\left[P_{\ill}^{\loc}\right]_{T,T-1,T-2} \nonumber \\
                                                                                                            &= -\frac{1}{2}\sum_{s \in \mathrm{S}_{\leg}^{(T-1)}}\sum_{x \in \mathrm{S}_{\leg}^{(T+1)}}a_{xs}\ket{x}\bra{10s}\left[P_{\ill}^{\loc}\right]_{T,T-1,T-2} = -\frac{1}{2}A ,
    \end{align}
    implying $A = 0$ if we require equation~\eqref{eq:requirement2}. Thus we find
    \begin{equation}
        \label{eq:zeropart1}
        \ker{\D_{\loc}^{(T,T-1,T-2)}|_{\mathcal{B}(\mathcal{H}_{\leg}^{(T+1)}, \mathcal{H}_{\ill}^{(T+1)})}}  \cap \mathcal{P}_{\mathcal{B}(\mathcal{H}_{\leg}^{(T+1)}, \mathcal{H}_{\ill}^{(T+1)})}\left[\mathbb{V}\right] =  \{0\}.
    \end{equation}

    Applying equations~\eqref{eq:zeropart} and~\eqref{eq:zeropart1} to equation~\eqref{eq:compofinduction1} we find
    \begin{align}
        \ker{\D_{\loc}^{(T,T-1,T-2)}|_{\mathcal{B}(\mathcal{H}_{\leg}^{(T+1)})^{\perp}}} \cap \mathbb{V} = \ker{\D_{\loc}^{(T,T-1,T-2)}|_{\mathcal{B}(\mathcal{H}_{\ill}^{(T+1)})}} \cap \mathcal{P}_{\mathcal{B}(\mathcal{H}_{\ill}^{(T+1)})}\left[\mathbb{V}\right].
    \end{align}
    Thus our induction step becomes
    \begin{align}
        \label{eq:inductionstep2}
        \ker{\D_{\unary}|_{\mathcal{B}(\mathcal{H}_{C})}} = \mathcal{B}(\mathcal{H}_{\leg}^{(T+1)}) \oplus \left(\ker{\D_{\loc}^{(T,T-1,T-2)}|_{\mathcal{B}(\mathcal{H}_{\ill}^{(T+1)})}} \cap \mathcal{P}_{\mathcal{B}(\mathcal{H}_{\ill}^{(T+1)})}\left[\mathbb{V}\right]\right).
    \end{align}
    Applying Lemma~\ref{lm:helper2} we find 
    \begin{align}
        \ker{\D_{\unary}|_{\mathcal{B}(\mathcal{H}_{C})}} = \mathcal{B}(\mathcal{H}_{\leg}^{(T+1)}) \oplus P_{\ill}^{(T+1)}
    \end{align}
    implying our induction hypothesis holds for $T \geq 3$ which suffices to prove the statement.
\end{proof}

\subsection{Analysis of Spectral Gap}
Our argument for lower-bounding the spectral gap of our unary clock can be broken up into roughly two parts. The first part of our argument consists of block diagonalizing our  dissipator across the partition $\mathcal{B}(\mathcal{H}_{C}) = \mathbb{O} + \mathbb{D}_{\leg} +\mathbb{D}_{\ill}$ where we have defined 
\begin{equation}
    \label{eq:subspacedefs}
    \begin{cases}
        \mathbb{D}_{\leg} =\Span{\{\ket{x}\bra{x} \; | \; x \in \mathrm{S}_{\leg}^{(T)} \}} \\
        \mathbb{D}_{\ill} =\Span{\{\ket{x}\bra{x} \; | \; x \in \mathrm{S}_{\ill}^{(T)} \}} \\
        \mathbb{O} = \Span{\{\ket{x}\bra{y} \; | \; x,y \in \{0,1\}^{T} \text{ and } x \neq y\}}
    \end{cases},
\end{equation}
with $\mathrm{S}_{\leg}^{(T)}$ and $\mathrm{S}_{\ill}^{(T)}$ being defined according to equation~\eqref{eq:globalillandlegstrings} We subsequently proceed to lower-bound the spectral gap of for each of the blocks. In the statement below we prove the first part of our argument.
\begin{lemma}
    \label{lm:Dublockdiag}
    Let $\D_{\unary}$ be the Lindbladian dissipator defined according to equation~\eqref{eq:Dunary} of Definition~\ref{def:localDQCSAT} and subspaces $\mathbb{O}$, $\mathbb{D}_{\leg}$, and $\mathbb{D}_{\ill}$ be defined according to equation~\eqref{eq:subspacedefs} then
    \begin{equation}
        \D_{\unary}|_{\mathcal{B}(\mathcal{H}_{C})} = \D_{\unary}|_{\mathbb{D}_{\leg}} \oplus \D_{\unary}|_{\mathbb{D}_{\ill}}  \oplus \D_{\unary}|_{\mathbb{O}}.
    \end{equation}
\end{lemma}
\begin{proof}
    Define $\mathbb{D} = \mathbb{D}_{\leg} \oplus \mathbb{D}_{\ill}$. We first show that $\D_{\unary}|_{\mathcal{B}(\mathcal{H}_{C})}$ admits $\mathbb{D}$ and $\mathbb{O}$ as invariant subspaces implying $\D_{\unary}|_{\mathcal{B}(\mathcal{H}_{C})}$ block diagonalizes across the decomposition $\mathcal{B}(\mathcal{H}_{C}) = \mathbb{D} \oplus \mathbb{O}$. Subsequently we show that $\D_{\unary}|_{\mathbb{D}}$ admits $\mathbb{D}_{\ill}$ and $\mathbb{D}_{\leg}$ as invariant subspaces thus implying the statement.

    Let $A \in \mathbb{D}$, computing the action of $\D_{\unary}|_{\mathcal{B}(\mathcal{H}_{C})}$ on $A$ we find
    \begin{align}
        \D_{\unary}|_{\mathcal{B}(\mathcal{H}_{C})}(A) &= \sum_{x \in \{0,1\}^{T}}a_{x}\D_{\unary}|_{\mathcal{B}(\mathcal{H}_{C})}(\ket{x}\bra{x})  = \sum_{t=0}^{T-3}\sum_{x \in \{0,1\}^{T}}a_{x}\D_{\loc}^{(t+2,t+1,t)}(\ket{x}\bra{x}) \nonumber \\
        &= \sum_{t=0}^{T-3}\sum_{x \in \{0,1\}^{T}}a_{x}\ket{x_{>t+2}}\bra{x_{>t+2}}\otimes \D_{\loc}(\ket{x_{t+2}x_{t+1}x_{t}}\bra{x_{t+2}x_{t+1}x_{t}})\otimes \ket{x_{<t}}\bra{x_{<t}}.
    \end{align}
    Since 
    \begin{equation}
        \D_{\loc}(\ket{s}\bra{s}) = \begin{cases}
            \frac{P_{\ill}^{\loc}}{4} - \ket{s}\bra{s} \text{ for } s \in \mathrm{S}_{\ill}^{\loc} \\
            0 \text{ otherwise }
        \end{cases},
    \end{equation}
    if we define the set of indices $\tau(x) \subseteq \{0, 1, \cdots T-3\}$ such that $t \in \tau(x)$ implies $x_{t+2}x_{t+1}x_{t} \in \mathrm{S}_{\ill}^{\loc}$ then we write 
    \begin{align}
        \D_{\unary}|_{\mathcal{B}(\mathcal{H}_{C})}(A) &= \sum_{x \in \{0,1\}^{T}}a_{x}\sum_{t \in \tau(x)}\left(\frac{1}{4}\ket{x_{>t+2}}\bra{x_{>t+2}}\otimes [P_{\ill}^{\loc}]_{t+2,t+1,t} \otimes \ket{x_{<t}}\bra{x_{<t}} - \ket{x}\bra{x}\right).
    \end{align}
    Since for all $t \in \{0, 1, \cdots, T-3\}$,
    \begin{align}
        \ket{x_{>t+2}}\bra{x_{>t+2}}\otimes [P_{\ill}^{\loc}]_{t+2,t+1,t} \otimes \ket{x_{<t}}\bra{x_{<t}} = \sum_{s \in \mathrm{S}_{\ill}^{\loc}} \ket{x_{>t+2}sx_{<t}}\bra{x_{>t+2}sx_{<t}} \in \mathbb{D}
    \end{align}
    it follows that
    \begin{equation}
        \D_{\unary}|_{\mathcal{B}(\mathcal{H}_{C})}(A) \in \mathbb{D}
    \end{equation}
    implying $\mathbb{D}$ is an invariant subspace. Since for all $A \in \mathbb{D}$ and $B \in \mathbb{O}$ we have $\braket{A, B} = \sum_{x} \sum_{y\neq z}a_{x}b_{y,z}\braket{\ket{x}\bra{x}, \ket{y}\bra{z}} = 0$, we have $\mathbb{O} = \mathbb{D}^{\perp}$. Therefore we find $\mathbb{O}$ is also an invariant subspace of $\mathcal{D}_{\unary}|_{\mathcal{B}(\mathcal{H}_{C})}$, since $\mathcal{D}_{\unary}|_{\mathcal{B}(\mathcal{H}_{C})}$ is Hermitian with respect to the Hilbert-Schmidt inner product. Thus we have
    \begin{equation}
        \label{eq:partitionbreakup1}
        \D_{\unary}|_{\mathcal{B}(\mathcal{H}_{C})} = \D_{\unary}|_{\mathbb{D}} \oplus \D_{\unary}|_{\mathbb{O}}.
    \end{equation}

    Consider $\D_{\unary}$ restricted to $\mathbb{D} = \mathbb{D}_{\leg} \oplus \mathbb{D}_{\ill}$ Since for all $A \in \mathbb{D}_{\ill}$ and $B \in \mathbb{D}_{\leg}$ we find $\braket{A, B} = \sum_{x\in \mathrm{S}_{\ill}^{(T)}}\sum_{y \in \mathrm{S}_{\leg}^{(T)}}a_{x}b_{y}\braket{\ket{x}\bra{x}, \ket{y}\bra{y}} = 0$, we have that $\mathbb{D}_{\ill} = \mathbb{D}_{\leg}^{\perp_{\mathbb{D}}}$. Proposition~\ref{prop:DuKer} implies $\mathbb{D}_{\leg} \subseteq \ker{\D_{\unary}|_{\mathbb{D}}}$ therefore $\mathbb{D}_{\leg}$ is an invariant subspace of $\D_{\unary}|_{\mathbb{D}}$. Similarly, since $\D_{\unary}|_{\mathbb{D}}$ is Hermitian with respect to the Hibert-Schmidt inner product on $\mathbb{D}$, $\mathbb{D}_{\leg}$ is an invariant subspace of $\D_{\unary}|_{\mathbb{D}}$, and $\mathbb{D}_{\ill}$ is the orthogonal complement of $\mathbb{D}_{\leg}$ with respect to $\mathbb{D}$, we find that $\mathbb{D}_{\ill}$ must also be an invariant subspace of $\D_{\unary}|_{\mathbb{D}}$. Therefore 
    \begin{equation}
        \D_{\unary}|_{\mathbb{D}} = \D_{\unary}|_{\mathbb{D}_{\ill}} \oplus \D_{\unary}|_{\mathbb{D}_{\leg}}.
    \end{equation}
    The above equation combined with~\eqref{eq:partitionbreakup1} implies the statement of this lemma.
\end{proof}

Note, that $\mathbb{D}_{\leg}$  is contain within the kernel of $\D_{\unary}$ characterized in Proposition~\ref{prop:DuKer}. Therefore $\D_{\unary}|_{\mathbb{D}_{\leg}}$ is simply the zero operator. Thus we have that the only components that contribute to the gap of $\D_{\unary}$ are its restrictions to $\mathbb{D}_{\ill}$ and $\mathbb{O}$. 

We begin by lower-bounding the gap of $\D_{\unary}$ restricted to the subspace $\mathbb{O}$. Similar to our global strategy, our approach starts by showing that $\mathbb{O}$ breaks up into invariant subspaces then lower-bound that gap of each. Defining 
\begin{equation}
    \label{eq:diffO}
    \mathbb{O}_{d} = \Span{\{\ket{x}\bra{x \oplus d } \; | \; x \in \{0,1\}^{T} \}},
\end{equation}
allows us to write 
\begin{equation}
    \mathbb{O} = \bigoplus_{d \neq 0^{T}} \mathbb{O}_{d}.
\end{equation}
We show that each one of the subspace $\mathbb{O}_{d}$ is invariant under $\D_{\unary}|_{\mathbb{O}}$.
\begin{lemma}
    \label{lm:Oblocks}
    Let $\D_{\unary}$ be the Lindbladian dissipator defined according to equation~\eqref{eq:Dunary} of Definition~\ref{def:localDQCSAT} and define the subspaces $\mathbb{O}$ and $\mathbb{O}_{d}$ according to equations~\eqref{eq:subspacedefs} and~\eqref{eq:diffO} respectively, then
    \begin{equation}
        \D_{\unary}|_{\mathbb{O}} = \bigoplus_{d\neq 0^{T}} \D_{\unary}|_{\mathbb{O}_{d}}.
    \end{equation}
\end{lemma}
\begin{proof}
    Fix $t \in \{0, 1, \cdots T-3\}$ and $d \in \{0,1\}^{T}$ such that $d \neq 0^{T}$. Calculating,
    \begin{align}
        \label{eq:casesford}
        &\D_{\loc}^{(t+2,t+1,t)}(\ket{x}\bra{x\oplus d}) = \ket{x_{>t+2}}\bra{[x \oplus d]_{>t+2}} \otimes \D_{\loc}(\ket{x}\bra{x \oplus d}_{t+2,t+1,t}) \otimes \ket{x_{<t}}\bra{[x\oplus d]_{<t}} \nonumber \\
        &= \begin{cases}
            0 \text{ for } x_{t+2}x_{t+1}x_{t},\; x_{t+2}x_{t+1}x_{t} \oplus d_{t+2}d_{t+1}d_{t} \in \mathrm{S}_{\leg}^{\loc} \\
            -\frac{1}{2}\ket{x}\bra{x\oplus d} \text{ for } x_{t+2}x_{t+1}x_{t} \in \mathrm{S}_{\leg}^{\loc},\; x_{t+2}x_{t+1}x_{t} \oplus d_{t+2}d_{t+1}d_{t}  \in \mathrm{S}_{\ill}^{\loc} \\
            -\frac{1}{2}\ket{x}\bra{x \oplus d} \text{ for } x_{t+2}x_{t+1}x_{t} \in \mathrm{S}_{\ill}^{\loc},\; x_{t+2}x_{t+1}x_{t} \oplus d_{t+2}d_{t+1}d_{t}  \in \mathrm{S}_{\leg}^{\loc} \\
            - \ket{x}\bra{x\oplus d} \text{ for } x_{t+2}x_{t+1}x_{t}, \; x_{t+2}x_{t+1}x_{t} \oplus d_{t+2}d_{t+1}d_{t} \in \mathrm{S}_{\ill}^{\loc}, x_{t+2}x_{t+1}x_{t} \neq x_{t+2}x_{t+1}x_{t} \oplus d_{t+2}d_{t+1}d_{t} \\
            B - \ket{x}\bra{x \oplus d}\text{ for } x_{t+2}x_{t+1}x_{t},\; x_{t+2}x_{t+1}x_{t} \oplus d_{t+2}d_{t+1}d_{t} \in \mathrm{S}_{\ill}^{\loc}, x_{t+2}x_{t+1}x_{t} \neq x_{t+2}x_{t+1}x_{t} \oplus d_{t+2}d_{t+1}d_{t}
        \end{cases}
    \end{align}
    where we have defined 
    \begin{equation}
        B = \frac{1}{4}\ket{x_{>t+2}}\bra{[x \oplus d]_{>t+2}} \otimes [P_{\ill}^{\loc}]_{t+2,t+1,t} \otimes \ket{x_{<t}}\bra{[x\oplus d]_{<t}}.
    \end{equation}
    Note $B \in \mathbb{O}_{d}$, therefore equation~\eqref{eq:casesford} implies $\D_{\loc}^{(t+2,t+1,t)}(\ket{x}\bra{x\oplus d}) \in \mathbb{O}_{d}$ for all $t \in \{0, 1, \cdots, T-3\}$. Thus we find 
    \begin{align}
        \D_{\unary}|_{\mathbb{O}}(\ket{x}\bra{x\oplus d}) = \D_{\unary}(\ket{x}\bra{x\oplus d}) = \sum_{t=0}^{T-3}\D_{\loc}^{(t+2,t+1,t)}(\ket{x}\bra{x\oplus d}) \in \mathbb{O}_{d}
    \end{align}
    implying for all $A \in \mathbb{O}_{d}$ we have $\D_{\unary}|_{\mathbb{O}}(A) \in \mathbb{O}_{d}$. This statement holds for all $d \neq 0^{T}$ therefore the subspaces $\mathbb{O}_{d}$ are invariant under $\D_{\unary}|_{\mathbb{O}}$ implying the statement of the lemma.
\end{proof}
We next lower-bound the gap of $\D_{\unary}|_{\mathbb{O}}$ when restricted to $\mathbb{O}_{d}$. To do so we first show the following lemma.
\begin{lemma}
    \label{lm:helperlemma}
    Let $\D_{\unary}$ be the Lindbladian dissipator defined according to equation~\eqref{eq:Dunary} of Definition~\ref{def:localDQCSAT} and define the subspace $\mathbb{O}_{d}$ for $d\neq 0^{T}$ according to equation~\eqref{eq:diffO}. Let 
    \begin{equation}
        A = \sum_{x} a_{x} \ket{x}\bra{x \oplus d} \in \mathbb{O}_{d}
    \end{equation}
    then 
    \begin{align}
        \braket{A, \D_{\unary}|_{\mathbb{O}_{d}}(A)} &\leq -\frac{1}{4}\sum_{(x,y) \in \mathrm{S}_{E}}|a_{x} - a_{y}|^{2} -\frac{1}{2}\sum_{x \in \mathrm{S}_{B}}|a_{x}|^{2}. 
    \end{align}
    where
    \begin{equation}
        \label{eq:edges}
        \mathrm{S}_{E} = \left\{\{x,y\} \in \binom{\{0,1\}^{T}}{2} \; \Big| \; \exists t \in \{0, 1, \cdots, T-3\}, \substack{d_{t+2}d_{t+1}d_{t} = 000, \\ x_{>t+2} = y_{>t+2}, \\  x_{<t} = y_{<t}, \\ \{x_{t+2}x_{t+1}x_{t},y_{t+2}y_{t+1}y_{t}\} \in E}  \right\} 
    \end{equation}
    with $E = \left\{\{010, 100\}, \{110, 101\}\right\}$ and 
    \begin{equation}
        \label{eq:boundary}
        \mathrm{S}_{B} = \left\{x \in \{0, 1\}^{T} \; \Big| \; \exists t \in \{0, 1, \cdots, T-3\}, \substack{d_{t+2}d_{t+1}d_{t} \neq 000, \\ x_{t+2}x_{t+1}x_{t} \in \mathrm{S}_{\ill}^{\loc} \,\mathrm{ or }\\ x_{t+2}x_{t+1}x_{t} \oplus d_{t+2}d_{t+1}d_{t} \in \mathrm{S}_{\ill}^{\loc}}\right\}.
    \end{equation}
\end{lemma}
\begin{proof}
    Calculating $\braket{\ket{y}\bra{y\oplus d}, \D_{\loc}^{(t+2,t+1,t)}(\ket{x}\bra{x\oplus d})}$ gives us two cases, either $d_{t+2}d_{t+1}d_{t} = 0$ or $d_{t+2}d_{t+1}d_{t} \neq 0$.
    Consider the case of $d_{t+2}d_{t+1}d_{t} = 000$, for $x_{t+2}x_{t+1}x_{t}, y_{t+2}y_{t+1}y_{t} \in \mathrm{S}_{\ill}^{\loc}$ we find 
    \begin{align}
        \braket{\ket{y}\bra{y\oplus d}, \D_{\loc}^{(t+2,t+1,t)}(\ket{x}\bra{x\oplus d})} &= \delta_{y_{>t+2}, x_{>t+2}}\delta_{y_{<t}, x_{<t}}\left(\frac{1}{4}- \delta_{y_{t+2}y_{t+1}y_{t}, x_{t+2}x_{t+1}x_{t}}\right).
    \end{align}
    For $x_{t+2}x_{t+1}x_{t}, y_{t+2}y_{t+1}y_{t} \notin \mathrm{S}_{\ill}^{\loc}$ we have 
    \begin{equation}
        \braket{\ket{y}\bra{y\oplus d}, \D_{\loc}^{(t+2,t+1,t)}(\ket{x}\bra{x\oplus d})} = 0.
    \end{equation}
    Consider the case of $d_{t+2}d_{t+1}d_{t} \neq 000$   we find
    \begin{align}
        \braket{\ket{y}\bra{y\oplus d}, \D_{\loc}^{(t+2,t+1,t)}(\ket{x}\bra{x\oplus d})} &= -\frac{\delta_{x,y}}{2}\left(\mathbf{1}[x_{t+2}x_{t+1}x_{t} \in \mathrm{S}_{\ill}^{\loc}] + \mathbf{1}[x_{t+2}x_{t+1}x_{t} \oplus d_{t+2}d_{t+1}d_{t} \in \mathrm{S}_{\ill}^{\loc}]\right)
    \end{align}
    where $\mathbf{1}$ denotes the indicator function.

    Now consider calculating $\braket{A, \D_{\loc}^{(t+2,t+1,t)}(A)}$ for the case of $d_{t+2}d_{t+1}d_{t} = 000$, we find
    \begin{align}
        \braket{A, \D_{\loc}^{(t+2,t+1,t)}(A)} &= \sum_{x,y} a_{y}^{*}a_{x}\braket{\ket{y}\bra{y\oplus d}, \D_{\loc}^{(t+2,t+1,t)}(\ket{x}\bra{x\oplus d})} \nonumber \\
        &= \sum_{x_{>t+2}}\sum_{x_{<t}}\sum_{s_{y},s_{x} \in \{0,1\}^{3}}a_{x_{>t+2}s_{y}x_{<t}}^{*}a_{x_{>t+2}s_{x}x_{<t}}\braket{\ket{s_{y}}\bra{s_{y}}, \D_{\loc}(\ket{s_{x}}\bra{s_{x}})} \nonumber \\
        &= \sum_{x_{>t+2}}\sum_{x_{<t}}\sum_{s_{y},s_{x} \in \mathrm{S}_{\ill}^{\loc}}a_{x_{>t+2}s_{y}x_{<t}}^{*}a_{x_{>t+2}s_{x}x_{<t}}\left(\frac{1}{4} - \delta_{s_{x}, s_{y}}\right) \nonumber \\
        &= \frac{1}{4}\sum_{x_{>t+2}}\sum_{x_{<t}}\left(\sum_{s_{y},s_{x} \in \mathrm{S}_{\ill}^{\loc}}a_{x_{>t+2}s_{y}x_{<t}}^{*}a_{x_{>t+2}s_{x}x_{<t}} - 4\sum_{s_{x} \in \mathrm{S}_{\ill}^{\loc}}|a_{x_{>t+2}s_{x}x_{<t}}|^{2} \right) \nonumber \\
        &= -\frac{1}{4}\sum_{x_{>t+2}}\sum_{x_{<t}}\left(3\sum_{s_{x} \in \mathrm{S}_{\ill}^{\loc}}|a_{x_{>t+2}s_{x}x_{<t}}|^{2} - \sum_{s_{y},s_{x} \in \mathrm{S}_{\ill}^{\loc}; s_{x} \neq s_{y}}a_{x_{>t+2}s_{y}x_{<t}}^{*}a_{x_{>t+2}s_{x}x_{<t}} \right).
    \end{align}
    Order $\mathrm{S}_{\ill}^{\loc}$ according to $010 < 100 < 101 < 110$ allows us to simplify the above expression as
    \begin{align}
        \braket{A, \D_{\loc}^{(t+2,t+1,t)}(A)} &= -\frac{1}{4}\sum_{x_{>t+2}}\sum_{x_{<t}}\left(3\sum_{s_{x} \in \mathrm{S}_{\ill}^{\loc}}|a_{x_{>t+2}s_{x}x_{<t}}|^{2} - \sum_{s_{y},s_{x} \in \mathrm{S}_{\ill}^{\loc}; s_{x} < s_{y}}2 \re (a_{x_{>t+2}s_{y}x_{<t}}^{*}a_{x_{>t+2}s_{x}x_{<t}}) \right) \nonumber \\
        &= -\frac{1}{4}\sum_{x_{>t+2}}\sum_{x_{<t}}\sum_{s_{y}, s_{x} \in \mathrm{S}_{\ill}^{\loc}; s_{x} < s_{y}}\left(|a_{x_{>t+2}s_{x}x_{<t}}|^{2} + |a_{x_{>t+2}s_{y}x_{<t}}|^{2} - 2 \re (a_{x_{>t+2}s_{y}x_{<t}}^{*}a_{x_{>t+2}s_{x}x_{<t}}) \right) \nonumber \\
        &= -\frac{1}{4}\sum_{x_{>t+2}}\sum_{x_{<t}}\sum_{s_{y}, s_{x} \in \mathrm{S}_{\ill}^{\loc}; s_{x} < s_{y}}|a_{x_{>t+2}s_{x}x_{<t}} - a_{x_{>t+2}s_{y}x_{<t}}|^{2}.
    \end{align}
    Now consider calculating $\braket{A, \D_{\loc}^{(t+2,t+1,t)}(A)}$ for the case of $d_{t+2}d_{t+1}d_{t} \neq 000$, we find
    \begin{align}
        \braket{A, \D_{\loc}^{(t+2,t+1,t)}(A)} &= \sum_{x,y} a_{y}^{*}a_{x}\braket{\ket{y}\bra{y\oplus d}, \D_{\loc}^{(t+2,t+1,t)}(\ket{x}\bra{x\oplus d})} \nonumber \\
        &= -\frac{1}{2}\sum_{x}|a_{x}|^{2}\left(\mathbf{1}[x_{t+2}x_{t+1}x_{t} \in \mathrm{S}_{\ill}^{\loc}] + \mathbf{1}[x_{t+2}x_{t+1}x_{t} \oplus d_{t+2}d_{t+1}d_{t} \in \mathrm{S}_{\ill}^{\loc}]\right).
    \end{align}

    To calculate $\braket{A, \D_{\unary}|_{\mathbb{O}_{d}}(A)}$ we first define sets of indices
    \begin{equation}
        \begin{cases}
            \mathcal{T}^{0} = \{t \in \{0, 1, \cdots, T-3\} \; | \; d_{t+2}d_{t+1}d_{t} = 000\} \\
            \mathcal{T}^{\neq} = \{t \in \{0, 1, \cdots, T-3\} \; | \; d_{t+2}d_{t+1}d_{t} \neq 000\}
        \end{cases}.
    \end{equation}
    We find
    \begin{align}
        \braket{A, \D_{\unary}|_{\mathbb{O}_{d}}(A)} &= \sum_{t=0}^{T-3}\braket{A, \D_{\loc}^{(t+2,t+1,t)}(A)} = \sum_{t\in \mathcal{T}^{0}}\braket{A, \D_{\loc}^{(t+2,t+1,t)}(A)} + \sum_{t\in \mathcal{T}^{\neq}}\braket{A, \D_{\loc}^{(t+2,t+1,t)}(A)}
    \end{align}
    which becomes 
    \begin{align}
        \braket{A, \D_{\unary}|_{\mathbb{O}_{d}}(A)} = -&\frac{1}{4}\sum_{t \in \mathcal{T}^{0}}\sum_{x_{>t+2}}\sum_{x_{<t}}\sum_{s_{y}, s_{x} \in \mathrm{S}_{\ill}^{\loc}; s_{x} < s_{y}}|a_{x_{>t+2}s_{x}x_{<t}} - a_{x_{>t+2}s_{y}x_{<t}}|^{2} \nonumber \\
        &-\frac{1}{2}\sum_{t\in\mathcal{T}^{\neq}}\sum_{x}|a_{x}|^{2}\left(\mathbf{1}[x_{t+2}x_{t+1}x_{t} \in \mathrm{S}_{\ill}^{\loc}] + \mathbf{1}[x_{t+2}x_{t+1}x_{t} \oplus d_{t+2}d_{t+1}d_{t} \in \mathrm{S}_{\ill}^{\loc}]\right).
    \end{align}
    Using our definition of $\mathrm{S}_{B}$ and fixing $x \in \mathrm{S}_{B}$ the second term in the above equation contributes at most $- |a_{x}|^{2} / 2$. Therefore we can write
    \begin{align}
        \braket{A, \D_{\unary}|_{\mathbb{O}_{d}}(A)} &\leq -\frac{1}{4}\sum_{t \in \mathcal{T}^{0}}\sum_{x_{>t+2}}\sum_{x_{<t}}\sum_{s_{y}, s_{x} \in \mathrm{S}_{\ill}^{\loc}; s_{x} < s_{y}}|a_{x_{>t+2}s_{x}x_{<t}} - a_{x_{>t+2}s_{y}x_{<t}}|^{2} -\frac{1}{2}\sum_{x \in \mathrm{S}_{B}}|a_{x}|^{2} \nonumber \\
        &\leq -\frac{1}{4}\sum_{t \in \mathcal{T}^{0}}\sum_{x_{>t+2}}\sum_{x_{<t}}\sum_{\{s_{x},s_{y}\} \in E}|a_{x_{>t+2}s_{x}x_{<t}} - a_{x_{>t+2}s_{y}x_{<t}}|^{2} -\frac{1}{2}\sum_{x \in \mathrm{S}_{B}}|a_{x}|^{2}
    \end{align}
    where we have used our definition of $E$. Similarly, using our definition of $\mathrm{S}_{E}$ and fixing a pair $\{x, y\} \in \mathrm{S}_{E}$ the first term in the above equation contributes at most $- |a_{x} - a_{y}|^{2} / 4$. Thus we conclude,
    \begin{align}
        \braket{A, \D_{\unary}|_{\mathbb{O}_{d}}(A)} &\leq -\frac{1}{4}\sum_{\{x,y\} \in \mathrm{S}_{E}}|a_{x} - a_{y}|^{2} -\frac{1}{2}\sum_{x \in \mathrm{S}_{B}}|a_{x}|^{2}. 
    \end{align}
\end{proof}
We can now proceed to lower-bound the gap of $\D_{\unary}|_{\mathbb{O}}$.
\begin{lemma}
    \label{lm:diffOgap}
    Let $\D_{\unary}$ be the Lindbladian dissipator defined according to equation~\eqref{eq:Dunary} of Definition~\ref{def:localDQCSAT} and define the subspace $\mathbb{O}_{d}$ for $d \neq 0^{T}$ according to equation~\eqref{eq:diffO}, then 
    \begin{equation}
        \gap{\D_{\unary}|_{\mathbb{O}_{d}}} \geq 1 / 4T^{3}
    \end{equation}
\end{lemma}
\begin{proof}
    Our strategy is to bound the Rayleigh quotient. From Proposition~\ref{prop:DuKer} and our definition of $\mathbb{O}_{d}$ we find 
    \begin{align}
        \ker{\D_{\unary}|_{\mathbb{O}_{d}}} &=  \ker{\D_{\unary}|_{\mathcal{B}(\mathcal{H}_{C})}} \cap \mathbb{O}_{d} = \left(\mathcal{B}(\mathcal{H}_{\leg}^{(T)}) \oplus P_{\ill}^{(T)}\right) \cap \mathbb{O}_{d} = \mathcal{B}(\mathcal{H}_{\leg}^{(T)}) \cap \mathbb{O}_{d} \nonumber \\
        &= \Span{\{\ket{x}\bra{x \oplus d} \; | \;  x \in \mathrm{S}_{\leg}^{(T)} \text{ and }  x \oplus d \in \mathrm{S}_{\leg}^{(T)} \}}.
    \end{align}
    which implies
    \begin{align}
        \ker{\D_{\unary}|_{\mathbb{O}_{d}}}^{\perp_{\mathbb{O}_{d}}} &= \Span{\{\ket{x}\bra{x \oplus d} \; | \;  x \in \mathrm{S}_{\ill}^{(T)} \text{ or }  x \oplus d \in \mathrm{S}_{\ill}^{(T)} \}} = \Span{\{\ket{x}\bra{x \oplus d} \; | \; x \in \mathrm{S}_{O} \}}
    \end{align}
    where we have defined 
    \begin{equation}
        \mathrm{S}_{O} = \{x \in \{0,1\}^{T} \; | \; x \in \mathrm{S}_{\ill}^{(T)} \text{ or } x \oplus d \in \mathrm{S}_{\ill}^{(T)}\}
    \end{equation}
    Let $A \in\ker{\D_{\unary}|_{\mathbb{O}_{d}}}^{\perp_{\mathbb{O}_{d}}}$ expanding in the computational basis
    \begin{equation}
        A = \sum_{x \in \mathrm{S}_{O}} a_{x} \ket{x}\bra{x \oplus d}.
    \end{equation}
     Define $\mathrm{S}_{E}$ and $\mathrm{S}_{B}$ according to equations~\eqref{eq:edges} and~\eqref{eq:boundary} of Lemma~\ref{lm:helperlemma} respectively. Note, $\mathrm{S}_{B} \subseteq \mathrm{S}_{O}$ and $\mathrm{S}_{E} \subseteq \binom{\mathrm{S}_{O}}{2}$. Consider the graph $G = (\mathrm{S}_{O}, \mathrm{S}_{E})$. For each vertex $x \in \mathrm{S}_{O}$ we can construct a path from $x$ to a boundary vertex $b \in \mathrm{S}_{B}$ as follows. If $x \in \mathrm{S}_{O}$ and $x \in \mathrm{S}_{B}$ then our path is simply $\{\}$. Thus we take $x \in \mathrm{S}_{O}$ and $x \notin \mathrm{S}_{B}$ which implies there exists an index $s \in \{0, 1, \cdots, T-3\}$ such that $d_{s+2}d_{s+1}d_{s} = 000$ and $x_{t+2}x_{t+1}x_{t} \in \mathrm{S}_{\ill}^{\loc}$. Define the set $F(d) = \{t \in \{0, 1, \cdots, T-3\} \; | \; d_{t+2}d_{t+1}d_{t} \neq 000\}$. Since $d \neq 0^{T}$ we know that $F(d) \neq \emptyset$. Take the index $f = \arg\min_{f \in F(d)}|f - s|$, if $f - s > 0$ then we define our path according to the sequence of vertices
     \begin{equation}
         \label{eq:path1}
         x = x_{>s+2}x_{s+2}10x_{<s} \rightarrow x_{>s+2}10x_{s+2}x_{<s} \rightarrow \cdots  \rightarrow x_{> f}10x_{\geq s+2}^{\leq f}x_{<s} \rightarrow x_{> f+1}10x_{\geq s+2}^{\leq f+1}x_{<s} \in \mathrm{S}_{B}
     \end{equation}
     if $f - s < 0$ then we define our path according to the sequence of vertices 
     \begin{equation}
         \label{eq:path2}
         x = x_{>s+2}10x_{s}x_{<s} \rightarrow x_{>s+2}x_{s}10x_{<s} \rightarrow \cdots \rightarrow  x_{>s+2}x_{\geq s}^{\leq f - 1}10x_{<f+1} \rightarrow x_{>s+2}x_{\geq s}^{\leq f }10x_{<f} \in \mathrm{S}_{B},
     \end{equation}
     where we have removed cycles to ensure the resulting sequence of edges are valid paths. By construction each pair of vertices is an edge in $\mathrm{S}_{E}$ since $f$ is the index of the closest 3-bit window to $s$ that is not zero.  

     Let $\gamma_{x}$ be the path associated to $x \in \mathrm{S}_{O}$ defined according to equations~\eqref{eq:path1} and~\eqref{eq:path2} along with $b_{x} \in \mathrm{S}_{B}$ being the vertex in $\mathrm{S}_{B}$ that $\gamma_{x}$ connects $x \in \mathrm{S}_{O}$ to. Define $x_{i}$ to be the $i$-th vertex occurring in the path $\gamma_{x}$, implying $x_{|\gamma_{x}|} = b_{x}$ and $x_{0} = x$, for all $x \in \mathrm{S}_{O}$ we can write
     \begin{equation}
         a_{x} = a_{b_{x}} + \sum_{i=0}^{|\gamma_{x}|-1}(a_{x_{i}} - a_{x_{i+1}}).
     \end{equation}
     From Cauchy-Schwarz we find
     \begin{equation}
         |a_{x}|^{2} \leq (|\gamma_{x}| + 1)\left(|a_{b_{x}}|^{2} + \sum_{i=0}^{|\gamma_{x}|-1}|a_{x_{i}} - a_{x_{i+1}}|^{2}\right).
     \end{equation}
     From our definition of $\gamma_{x}$, the length of the path is upper-bounded by $|f - s|$. Since both $f, s \in \{0, 1, \cdots, T-3\}$ we have $|\gamma_{x}| \leq T- 3$ implying
     \begin{equation}
         |a_{x}|^{2} \leq T\left(|a_{b_{x}}|^{2} + \sum_{i=0}^{|\gamma_{x}|-1}|a_{x_{i}} - a_{x_{i+1}}|^{2}\right) = T \left(|a_{b_{x}}|^{2} + \sum_{\{x_{i}, x_{i+1}\} \in \gamma_{x}}|a_{x_{i}} - a_{x_{i+1}}|^{2}\right).
     \end{equation}
     Summing over $x \in \mathrm{S}_{O}$ we find
     \begin{equation}
         \sum_{x\in \mathrm{S}_{O}}|a_{x}|^{2} = \braket{A, A} \leq T \left(\sum_{x \in \mathrm{S}_{O}}|a_{b_{x}}|^{2} + \sum_{x \in \mathrm{S}_{O}}\sum_{\{x_{i}, x_{i+1}\} \in \gamma_{x}}|a_{x_{i}} - a_{x_{i+1}}|^{2}\right).
     \end{equation}
     Defining $N_{b} = \{x \in \mathrm{S}_{O} \; | \; b = b_{x}\}$ we can write
     \begin{equation}
         \braket{A, A} \leq T \left(\sum_{b \in \mathrm{S}_{B}}|N_{b}||a_{b}|^{2} + \sum_{x \in \mathrm{S}_{O}}\sum_{\{x_{i}, x_{i+1}\} \in \gamma_{x}}|a_{x_{i}} - a_{x_{i+1}}|^{2}\right).
     \end{equation}
     Similarly, defining $N_{e} = \{x \in \mathrm{S}_{O} \; | \; e \in \gamma_{x} \}$ we can write
     \begin{equation}
         \braket{A, A} \leq T \left(\sum_{b \in \mathrm{S}_{B}}|N_{b}||a_{b}|^{2} + \sum_{\{x, y\} \in \mathrm{S}_{E}}|N_{\{x,y\}}||a_{x} - a_{y}|^{2}\right).
     \end{equation}

     Consider $N_{b}$, fix indices $f, s \in \{0, 1, \cdots, T-3\}$ as used in the construction of our paths $\gamma_{x}$. For each pair $f$ and $s$ we can reverse the sequence of vertices given in equation~\eqref{eq:path1} or equation~\eqref{eq:path2} to arrive at a unique $x \in \mathrm{S}_{O}$. Since there are $T-2$ choices for both $f$ and $s$ we have $|N_{b}| \leq  T^{2}$ for all $b \in \mathrm{S}_{O}$ therefore we find 
     \begin{equation}
         \braket{A, A} \leq T \left(2T^{2}\sum_{b \in \mathrm{S}_{B}}|a_{b}|^{2} + \sum_{\{x, y\} \in \mathrm{S}_{E}}|N_{\{x,y\}}||a_{x} - a_{y}|^{2}\right).
     \end{equation}
    
    Consider $N_{e}$, for $e \notin \gamma_{x}$ for all $x \in \mathrm{S}_{O}$ we have $|N_{e}| = 0$. For the case of $e \in \gamma_{x}$ for some $x \in \mathrm{S}_{O}$, if we fix $f, s \in \{0, 1, \cdots, T-3\}$, then we can reverse the sequence of vertices in either equation~\eqref{eq:path1} or equation~\eqref{eq:path2} starting at $e$ until we reach index $s$ arriving at a unique $x \in \mathrm{S}_{O}$. Since there are $T-2$ choices for both $s$ and $f$ we find $|N_{e}|\leq T^{2}$ implying 
     \begin{equation}
         \label{eq:fin}
         \braket{A, A} \leq T \left(2T^{2}\sum_{b \in \mathrm{S}_{B}}|a_{b}|^{2} + T^{2}\sum_{\{x, y\} \in \mathrm{S}_{E}}|a_{x} - a_{y}|^{2}\right).
     \end{equation}

     From Lemma~\ref{lm:helperlemma} which holds for all $A \in \mathbb{O}_{d}$ we find 
     \begin{align}
         -\braket{A, \D_{\unary}|_{\mathbb{O}_{d}}(A)} \geq  \frac{1}{4}\left(2\sum_{b \in \mathrm{S}_{B}}|a_{b}|^{2} + \sum_{\{x, y\} \in \mathrm{S}_{E}}|a_{x} - a_{y}|^{2}\right) \geq \frac{1}{4T^{3}}\braket{A,A}.
     \end{align}
     Therefore we have 
     \begin{equation}
         \gap{\D_{\unary}|_{\mathbb{O}_{d}}} = \inf_{\substack{A \in \ker{\D_{\unary}|_{\mathbb{O}_{d}}}^{\perp_{\mathbb{O}_{d}}} \\ \braket{A, A} \neq 0}}\left(-\frac{\braket{A, \D_{\unary}|_{\mathbb{O}_{d}}(A)}}{\braket{A,A}}\right) \geq  \frac{1}{4T^{3}}
     \end{equation}
     as stated.
\end{proof}
We return to lower-bounding the gap $\D_{\unary}|_{\mathbb{O}}$ concluding with the following statement.
\begin{lemma}
    \label{lm:DunaryO}
    Let $\D_{\unary}$ be the Lindbladian dissipator defined according to equation~\eqref{eq:Dunary} of Definition~\ref{def:localDQCSAT} and define the subspace $\mathbb{O}$ according to equation~\eqref{eq:subspacedefs}, then
    \begin{equation}
        \gap{\D_{\unary}|_{\mathbb{O}}} \geq 1 / 4 T^{3}.
    \end{equation}
\end{lemma}
\begin{proof}
    Define $\mathbb{O}_{d}$ according to equation~\eqref{eq:diffO}. From Lemma~\ref{lm:Oblocks} we find
    \begin{equation}
        \gap{\D_{\unary}|_{\mathbb{O}}} \geq \inf_{d \neq 0^{T}}\gap{\D_{\unary}|_{\mathbb{O}_{d}}}.
    \end{equation}
    From Lemma~\ref{lm:diffOgap} we conclude
    \begin{equation}
        \gap{\D_{\unary}|_{\mathbb{O}}}\geq \frac{1}{4T^{3}}.
    \end{equation}
\end{proof}

In the next stage of our argument we focus on lower-bounding the gap of $\D_{\unary}$ when restricted to $\mathbb{D}_{\ill}$. In-order to do so, we first argue that $\D_{\unary}|_{\mathbb{D}_{\ill}}$ as matrix expanded in the computational basis can be interpreted as a transition-rate matrix for a continuous-time Markov chain. We then apply the standard canonical paths approach \cite{} to argue for a lower-bound on the gap of this transition-rate matrix.

A \emph{transition-rate matrix} for a finite continuous-time Markov chain, sometimes called a \emph{Q-matrix}, is defined as an $|S| \times |S|$ matrix $Q$, with matrix elements $Q_{x,y}$ satisfying
\begin{equation}
    \label{eq:QmatrixDef}
    \begin{cases}
        0 \leq -Q_{x, x} < \infty \\
        Q_{x,y} \geq 0 \text{ for } x \neq y \\
        \sum_{x \in S}Q_{y,x} = 0
    \end{cases}.
\end{equation}
A steady-state of the continuous-time chain is a probability distribution over $S$ which satisfies $Q \pi = 0$ For every transition-rate matrix, $Q$ we can define a $|S| \times |S|$ stochastic matrix generating a discrete-time Markov chain over the state space $S$  called a \emph{jump chain}. The specific jump chain we will consider is generated by the stochastic matrix
\begin{equation}
    \label{eq:jumpchain}
    P = \mathbb{I} + \frac{Q}{\norm{Q}}.
\end{equation}
A steady-state of the discrete-time chain satisfies $P\pi = \pi$. Note, by construction our $P$ defined above is positive semi-definite. A finite discrete-time Markov chain is called \emph{reversible} if for a steady-state $\pi$ we have $P_{x,y}\pi_{y} = P_{y,x}\pi_{x}$. In addition, it is sufficient to call the chain \emph{irreducible} if it has a unique steady-state $\pi$ that has full support over the state space $S$. For a good reference on both continuous-time and discrete-time Markov chains see \cite{}

In the statement to follow we prove that when $\D_{\unary}$ is restricted to $\mathbb{D}_{\ill}$ then it becomes a transition-rate matrix for a continuous-time Markov chain over the state space of bit strings contained in $\mathrm{S}_{\ill}^{(T)}$. This allows us to lower-bound the gap of $\D_{\unary}$ restricted to $\mathbb{D}$ by lower-bounding the gap of the transition-rate matrix of the associated continuous-time Markov chain.
\begin{lemma}
    \label{lm:blockismarkovchain}
    Let $\D_{\unary}$ be the Lindbladian dissipator defined according to equation~\eqref{eq:Dunary} of Definition~\ref{def:localDQCSAT} and define subspace  $\mathbb{D}_{\ill}$ according to equation~\eqref{eq:subspacedefs}. Let $Q$ be the matrix formed by elements
    \begin{equation}
        \label{eq:Q}
        Q_{x,y} = \braket{\ket{x}\bra{x}, \D_{\unary}|_{\mathbb{D}_{\ill}}(\ket{y}\bra{y})} \text{ for } x, y \in \mathrm{S}_{\ill}^{(T)},
    \end{equation}
    then $Q$ is a transition-rate matrix for a continuous-time Markov chain over the state space $\mathrm{S}_{\ill}^{(T)}$ with a reversible and irreducible jump chain given by equation~\eqref{eq:jumpchain} with steady-state
    \begin{equation}
        \label{eq:steadystate}
        \pi = \sum_{x \in \mathrm{S}_{\ill}^{(T)}} \frac{1}{|\mathrm{S}_{\ill}^{(T)}|} \hat{x}.
    \end{equation}
\end{lemma}
\begin{proof}
    By construction the matrix elements of $Q$ form an $|\mathrm{S}_{\ill}^{(T)}| \times |\mathrm{S}_{\ill}^{(T)}|$ matrix, therefore it suffices to show the conditions of equation~\eqref{eq:QmatrixDef} to prove $Q$ is a transition-rate matrix associated to a continuous-time Markov chain over the state space $\mathrm{S}_{\ill}^{(T)}$.
    
    Consider the first of the conditions given in equation~\eqref{eq:QmatrixDef}. Since $\D_{\unary}|_{\mathbb{D}_{\ill}}$ is Hermitian with respect to the Hilbert-Schmidt inner product on $\mathbb{D}_{\ill}$ and is a valid Lindbladian dissipator we have $\D_{\unary}|_{\mathbb{D}_{\ill}} \preceq 0$ implying $Q_{x,x} \leq 0$. Since $\D_{\unary}$ is a bounded linear operator the first condition of equation~\eqref{eq:QmatrixDef} follows. Consider the second condition of equation~\eqref{eq:QmatrixDef}. Calculating $\braket{\ket{x}\bra{x}, \D_{\loc}^{(t+2,t+1,t)}(\ket{y}\bra{y})}$ for $x, y \in \mathrm{S}_{\ill}^{(T)}$ and $t \in \{0, 1, \cdots, T-3\}$ we find
    \begin{align}
        \braket{\ket{x}\bra{x}, \D_{\loc}^{(t+2,t+1,t)}(\ket{y}\bra{y})} &= \braket{\ket{x}\bra{x}, \ket{y_{>t+2}}\bra{y_{>t+2}}\otimes \D_{\loc}(\ket{y_{t+2}y_{t+1}y_{t}}\bra{y_{t+2}y_{t+1}y_{t}}) \otimes \ket{y_{<t}}\bra{y_{<t}}} \nonumber \\
        &= \delta_{x_{>t+2},y_{>t+2}}\delta_{x_{<t},y_{<t}} \braket{\ket{x_{t+2}x_{t+1}x_{t}}\bra{x_{t+2}x_{t+1}x_{t}}, \D_{\loc}(\ket{y_{t+2}y_{t+1}y_{t}}\bra{y_{t+2}y_{t+1}y_{t}})} \nonumber \\
        &= \begin{cases}
            \delta_{x_{>t+2},y_{>t+2}}\left(\frac{1}{4} - \delta_{x_{t+2}x_{t+1}x_{t}, y_{t+2}y_{t+1}y_{t}}\right) \text{ for } x_{t+2}x_{t+1}x_{t}, y_{t+2}y_{t+1}y_{t} \in \mathrm{S}_{\ill}^{\loc} \\
            0 \text{ otherwise}
        \end{cases}
    \end{align}
    which implies for $x \neq y$ 
    \begin{equation}
        \braket{\ket{x}\bra{x}, \D_{\loc}^{(t+2,t+1,t)}(\ket{y}\bra{y})} \geq 0.
    \end{equation}
    Let $x, y \in \mathrm{S}_{\ill}^{(T)}$ with $x \neq y$ calculating $Q_{x,y}$ we find 
    \begin{align}
        Q_{x,y} = \braket{\ket{x}\bra{x}, \D_{\unary}|_{\mathbb{D}_{\ill}}(\ket{y}\bra{y})} =  \sum_{t=0}^{T-3}\braket{\ket{x}\bra{x}, \D_{\loc}^{(t+2,t+1,t)}(\ket{y}\bra{y})} \geq 0
    \end{align}
    which validates the second condition of equation~\eqref{eq:QmatrixDef}. Consider the third and final condition given in equation~\eqref{eq:QmatrixDef}. By direct calculation we validate
    \begin{align}
        \sum_{x \in \mathrm{S}_{\ill}^{(T)}} Q_{y,x} &= \sum_{x \in \mathrm{S}_{\ill}^{(T)}} \braket{\ket{y}\bra{y}, \D_{\unary}|_{\mathbb{D}_{\ill}}(\ket{x}\bra{x})} = \braket{\ket{y}\bra{y}, \D_{\unary}(P_{\ill}^{(T)})} = 0
    \end{align}
    where the last equality follows from Proposition~\ref{prop:DuKer}. Thus $Q$ is a valid transition-rate matrix for a continuous-time Markov chain over the state space $\mathrm{S}_{\ill}^{(T)}$.

    Consider the jump chain associated to $Q$ generated by stochastic matrix $P$ in equation~\eqref{eq:jumpchain}. For a finite discrete-time Markov chain it is sufficient to show there exists a unique steady-state with full support to prove that it is irreducible. Define vector
    \begin{equation}
        \pi = \sum_{x \in \mathrm{S}_{\ill}^{(T)}} \frac{1}{|\mathrm{S}_{\ill}^{(T)}|} \hat{x}.
    \end{equation}
    By construction $\pi$ is normalized since $\sum_{x}\pi_{x} = 1$ with $\pi_{x} \geq 0$ therefore $\pi$ is a valid state of the discrete-time Markov chain. By construction $\pi$ has full support over the state space $\mathrm{S}_{\ill}^{(T)}$ since for all $x \in \mathrm{S}_{\ill}^{(T)}$, $\pi_{x} > 0$. By direct calculation we find
    \begin{equation}
        P\pi = \sum_{x, y \in \mathrm{S}_{\ill}^{(T)}}P_{x, y}\pi_{y} \hat{x} = \sum_{x, y \in \mathrm{S}_{\ill}^{(T)}}\left(\delta_{x,y} + \frac{1}{\norm{Q}}Q_{x,y}\right)\pi_{y} \hat{x} = \pi + \frac{1}{\norm{Q}|\mathrm{S}_{\ill}^{(T)}|}\sum_{x, y \in \mathrm{S}_{\ill}^{(T)}}Q_{x,y} = \pi.
    \end{equation}
    Therefore $\pi$ is a steady-state for the jump chain with full support. To show uniqueness let $\pi'$ be an arbitrary steady-state of the jump chain. From equation~\eqref{eq:jumpchain} and the definition of steady-state we find
    \begin{equation}
        P\pi' = \pi' + \frac{1}{\norm{Q}}Q\pi' = \pi',
    \end{equation}
    implying
    \begin{equation}
        Q\pi' = 0.
    \end{equation}
    Therefore a steady-state $\pi'$ must lie in the kernel of $Q$. Since $Q$ is $\D_{\unary}|_{\mathbb{D}_{\ill}}$ expanded in the basis $\{\ket{x}\bra{x} \; x \in \mathrm{S}_{\ill}^{(T)}\}$ we have 
    \begin{equation}
        \dim\ker{Q} = \dim\ker{\D_{\unary}|_{\mathbb{D}_{\ill}}}.
    \end{equation}
    Proposition~\ref{prop:DuKer} implies $\ker{\D_{\unary}|_{\mathbb{D}_{\ill}}} = P_{\ill}^{(T)}$ implying 
    \begin{equation}
        \dim\ker{Q} = 1.
    \end{equation}
    Since the kernel of $Q$ is one dimensional it must contain only a single steady-state vector given by the $\pi$ we found. Therefore the jump chain associated to $Q$ is irreducible.

    To show that the jump chain is reversible we require $P_{x,y}\pi_{y} = P_{y,x}\pi_{x}$ for steady-state $\pi$. By direct calculation we validate
    \begin{align}
        P_{x,y}\pi_{y} &= \frac{1}{|\mathrm{S}_{\ill}^{(T)}|}P_{x,y} = \frac{1}{|\mathrm{S}_{\ill}^{(T)}|}\left(\delta_{x,y} + \frac{Q_{x,y}}{\norm{Q}}\right) = \frac{1}{|\mathrm{S}_{\ill}^{(T)}|}\left(\delta_{x,y} + \frac{1}{\norm{Q}}\braket{\ket{x}\bra{x}, \D_{\unary}|_{\mathbb{D}_{\ill}}(\ket{y}\bra{y})}\right) \nonumber \\
        &= \frac{1}{|\mathrm{S}_{\ill}^{(T)}|}\left(\delta_{x,y} + \frac{1}{\norm{Q}}\braket{\ket{y}\bra{y}, \D_{\unary}|_{\mathbb{D}_{\ill}}(\ket{x}\bra{x})}^{*}\right) = \frac{1}{|\mathrm{S}_{\ill}^{(T)}|}\left(\delta_{y,x} + \frac{Q_{y,x}^{*}}{\norm{Q}}\right) = \frac{1}{|\mathrm{S}_{\ill}^{(T)}|}P_{y,x} \nonumber \\
        &= P_{y,x}\pi_{x},
    \end{align}
    where the second line follows from $\D_{\unary}|_{\mathbb{D}_{\ill}}$ being Hermitian with respect to the Hilbert-Schmidt inner product. Therefore the jump chain is reversible implying the statement of Lemma.
\end{proof}
\begin{lemma}
    \label{lm:Qcalc}
    Let $Q$ be the transition-rate matrix defined according to equation~\eqref{eq:Q} of Lemma~\ref{lm:blockismarkovchain}. For all $x, y \in \mathrm{S}_{\ill}^{(T)}$ and $t \in \{0, 1, 2, \cdots, T-3 \}$, we have
    \begin{equation}
        Q_{x_{>t+2}ay_{<t}, x_{>t+2}by_{<t}} \geq 1 /4,
    \end{equation}
    for $a, b \in \mathrm{S}_{\ill}^{\loc}$ and $a \neq b$ where $\mathrm{S}_{\ill}^{\loc}$ is defined according to equation~\eqref{eq:localillandlegstrings}.
\end{lemma}
\begin{proof}
    We proceed by direct calculation,
    \begin{align}
        Q_{x_{>t+2}ay_{<t}, x_{>t+2}by_{<t}} &= \braket{\ket{x_{>t+2}ay_{<t}}\bra{x_{>t+2}ay_{<t}}, \D_{\unary}|_{\mathbb{D}}(\ket{x_{>t+2}by_{<t}}\bra{x_{>t+2}by_{<t}})} \nonumber \\
        &= \sum_{t'=0}^{T-3}\braket{\ket{x_{>t+2}ay_{<t}}\bra{x_{>t+2}ay_{<t}}, \D_{\loc}^{(t'+2,t'+1,t')}(\ket{x_{>t+2}by_{<t}}\bra{x_{>t+2}by_{<t}})}.
    \end{align}
    Since $a \neq b$ we find
    \begin{align}
        Q_{x_{>t+2}ay_{<t}, x_{>t+2}by_{<t}} = &\braket{\ket{a}\bra{a}, \D_{\loc}(\ket{b}\bra{b})} \nonumber \\
        &+ \braket{\ket{x_{t+3}a_{t+2}a_{t+1}}\bra{x_{t+3}a_{t+2}a_{t+1}}, \D_{\loc}(\ket{x_{t+3}b_{t+2}b_{t+1}}\bra{x_{t+3}b_{t+2}b_{t+1}})} \delta_{a_{t},b_{t}} \nonumber \\
        &+ \braket{\ket{x_{t+4}x_{t+3}a_{t+2}}\bra{x_{t+4}x_{t+3}a_{t+2}}, \D_{\loc}(\ket{x_{t+4}x_{t+3}b_{t+2}}\bra{x_{t+4}x_{t+3}b_{t+2}})} \delta_{a_{t},b_{t}} \delta_{a_{t+1},b_{t+1}} \nonumber \\
        &+ \braket{\ket{a_{t+1}a_{t}y_{t-1}}\bra{a_{t+1}a_{t}y_{t-1}}, \D_{\loc}(\ket{b_{t+1}b_{t}y_{t-1}}\bra{b_{t+1}b_{t}y_{t-1}})}\delta_{a_{t+2}, b_{t+2}} \nonumber \\
        &+ \braket{\ket{a_{t}y_{t-1}y_{t-2}}\bra{a_{t}y_{t-1}y_{t-2}}, \D_{\loc}(\ket{b_{t}y_{t-1}y_{t-2}}\bra{b_{t}y_{t-1}y_{t-2}})}\delta_{a_{t+2},b_{t+2}}\delta_{a_{t+1},b_{t+1}},
    \end{align}
    which simplifies to become
    \begin{align}
        \label{eq:fulllower}
        Q_{x_{>t+2}ay_{<t}, x_{>t+2}by_{<t}} = &\frac{1}{4} + \braket{\ket{x_{t+3}a_{t+2}a_{t+1}}\bra{x_{t+3}a_{t+2}a_{t+1}}, \D_{\loc}(\ket{x_{t+3}b_{t+2}b_{t+1}}\bra{x_{t+3}b_{t+2}b_{t+1}})} \delta_{a_{t},b_{t}} \nonumber \\
        &+ \braket{\ket{x_{t+4}x_{t+3}a_{t+2}}\bra{x_{t+4}x_{t+3}a_{t+2}}, \D_{\loc}(\ket{x_{t+4}x_{t+3}b_{t+2}}\bra{x_{t+4}x_{t+3}b_{t+2}})}  \delta_{a_{t+1}a_{t},b_{t+1}b_{t}} \nonumber \\
        &+ \braket{\ket{a_{t+1}a_{t}y_{t-1}}\bra{a_{t+1}a_{t}y_{t-1}}, \D_{\loc}(\ket{b_{t+1}b_{t}y_{t-1}}\bra{b_{t+1}b_{t}y_{t-1}})}\delta_{a_{t+2}, b_{t+2}} \nonumber \\
        &+ \braket{\ket{a_{t}y_{t-1}y_{t-2}}\bra{a_{t}y_{t-1}y_{t-2}}, \D_{\loc}(\ket{b_{t}y_{t-1}y_{t-2}}\bra{b_{t}y_{t-1}y_{t-2}})}\delta_{a_{t+2}a_{t+1},b_{t+2}b_{t+1}},
    \end{align}
    since $a, b \in \mathrm{S}_{\ill}^{\loc}$ and $a \neq b$. 

    Consider the term $\braket{\ket{x_{t+3}a_{t+2}a_{t+1}}\bra{x_{t+3}a_{t+2}a_{t+1}}, \D_{\loc}(\ket{x_{t+3}b_{t+2}b_{t+1}}\bra{x_{t+3}b_{t+2}b_{t+1}})} \delta_{a_{t},b_{t}} $. For $a_{t} \neq b_{t}$ this term is zero. For $a_{t} = b_{t}$ we have $a_{t+2}a_{t+1} \neq b_{t+2}b_{t+1}$ since $a \neq b$. For $x_{t+3}b_{t+2}b_{t+1} \in \mathrm{S}_{\ill}^{\loc}$ and $x_{t+3}a_{t+2}a_{t+1} \in \mathrm{S}_{\ill}^{\loc}$ we have 
    \begin{align}
        \braket{\ket{x_{t+3}a_{t+2}a_{t+1}}\bra{x_{t+3}a_{t+2}a_{t+1}}, \D_{\loc}(\ket{x_{t+3}b_{t+2}b_{t+1}}\bra{x_{t+3}b_{t+2}b_{t+1}})} &= \frac{1}{4}
    \end{align}
    otherwise 
    \begin{align}
        \braket{\ket{x_{t+3}a_{t+2}a_{t+1}}\bra{x_{t+3}a_{t+2}a_{t+1}}, \D_{\loc}(\ket{x_{t+3}b_{t+2}b_{t+1}}\bra{x_{t+3}b_{t+2}b_{t+1}})} &= 0
    \end{align}
    Thus we find
    \begin{align}
        \braket{\ket{x_{t+3}a_{t+2}a_{t+1}}\bra{x_{t+3}a_{t+2}a_{t+1}}, \D_{\loc}(\ket{x_{t+3}b_{t+2}b_{t+1}}\bra{x_{t+3}b_{t+2}b_{t+1}})}  \geq 0.
    \end{align}

    Consider the term $\braket{\ket{x_{t+4}x_{t+3}a_{t+2}}\bra{x_{t+4}x_{t+3}a_{t+2}}, \D_{\loc}(\ket{x_{t+4}x_{t+3}b_{t+2}}\bra{x_{t+4}x_{t+3}b_{t+2}})} \delta_{a_{t+1}a_{t},b_{t+1}b_{t}}$. For $a_{t+1}a_{t} \neq b_{t+1}b_{t}$ this term is zero. For $a_{t+1}a_{t} = b_{t+1}b_{t}$ we have $a_{t+2} \neq b_{t+2}$ since $a \neq b$. For $x_{t+4}x_{t+3}b_{t+2} \in \mathrm{S}_{\ill}^{\loc}$ and $x_{t+4}x_{t+3}a_{t+2} \in \mathrm{S}_{\ill}^{\loc}$ we have 
    \begin{align}
        \braket{\ket{x_{t+4}x_{t+3}a_{t+2}}\bra{x_{t+4}x_{t+3}a_{t+2}}, \D_{\loc}(\ket{x_{t+4}x_{t+3}b_{t+2}}\bra{x_{t+4}x_{t+3}b_{t+2}})} &= \frac{1}{4}
    \end{align}
    otherwise 
    \begin{align}
        \braket{\ket{x_{t+4}x_{t+3}a_{t+2}}\bra{x_{t+4}x_{t+3}a_{t+2}}, \D_{\loc}(\ket{x_{t+4}x_{t+3}b_{t+2}}\bra{x_{t+4}x_{t+3}b_{t+2}})} &= 0
    \end{align}
    Thus we find
    \begin{align}
        \braket{\ket{x_{t+4}x_{t+3}a_{t+2}}\bra{x_{t+4}x_{t+3}a_{t+2}}, \D_{\loc}(\ket{x_{t+4}x_{t+3}b_{t+2}}\bra{x_{t+4}x_{t+3}b_{t+2}})} \geq 0
    \end{align}

    Consider the term $\braket{\ket{a_{t+1}a_{t}y_{t-1}}\bra{a_{t+1}a_{t}y_{t-1}}, \D_{\loc}(\ket{b_{t+1}b_{t}y_{t-1}}\bra{b_{t+1}b_{t}y_{t-1}})} \delta_{a_{t+2},b_{t+2}}$. For $a_{t+2} \neq b_{t+2}$ this term is zero. For $a_{t+2} = b_{t+2}$ we have $a_{t+1}a_{t} \neq b_{t+1}b_{t}$ since $a \neq b$. For $b_{t+1}b_{t}y_{t-1} \in \mathrm{S}_{\ill}^{\loc}$ and $a_{t+1}a_{t}y_{t-1} \in \mathrm{S}_{\ill}^{\loc}$ we have 
    \begin{align}
        \braket{\ket{a_{t+1}a_{t}y_{t-1}}\bra{a_{t+1}a_{t}y_{t-1}}, \D_{\loc}(\ket{b_{t+1}b_{t}y_{t-1}}\bra{b_{t+1}b_{t}y_{t-1}})} &= \frac{1}{4}
    \end{align}
    otherwise 
    \begin{align}
        \braket{\ket{a_{t+1}a_{t}y_{t-1}}\bra{a_{t+1}a_{t}y_{t-1}}, \D_{\loc}(\ket{b_{t+1}b_{t}y_{t-1}}\bra{b_{t+1}b_{t}y_{t-1}})} &= 0
    \end{align}
    Thus we find
    \begin{align}
        \braket{\ket{a_{t+1}a_{t}y_{t-1}}\bra{a_{t+1}a_{t}y_{t-1}}, \D_{\loc}(\ket{b_{t+1}b_{t}y_{t-1}}\bra{b_{t+1}b_{t}y_{t-1}})} &\geq 0
    \end{align}

    Consider the final term $\braket{\ket{a_{t}y_{t-1}y_{t-2}}\bra{a_{t}y_{t-1}y_{t-2}}, \D_{\loc}(\ket{b_{t}y_{t-1}y_{t-2}}\bra{b_{t}y_{t-1}y_{t-2}})} \delta_{a_{t+2}a_{t+1},b_{t+2}b_{t+1}}$. For $a_{t+2}a_{t+1} \neq b_{t+2}b_{t+1}$ this term is zero. For $a_{t+2}a_{t+1} = b_{t+2}b_{t+1}$ have $a_{t} \neq b_{t}$ since $a \neq b$. For $b_{t}y_{t-1}y_{t-2} \in \mathrm{S}_{\ill}^{\loc}$ and $a_{t}y_{t-1}y_{t-2} \in \mathrm{S}_{\ill}^{\loc}$ we have 
    \begin{align}
        \braket{\ket{a_{t}y_{t-1}y_{t-2}}\bra{a_{t}y_{t-1}y_{t-2}}, \D_{\loc}(\ket{b_{t}y_{t-1}y_{t-2}}\bra{b_{t}y_{t-1}y_{t-2}})} &= \frac{1}{4}
    \end{align}
    otherwise 
    \begin{align}
        \braket{\ket{a_{t}y_{t-1}y_{t-2}}\bra{a_{t}y_{t-1}y_{t-2}}, \D_{\loc}(\ket{b_{t}y_{t-1}y_{t-2}}\bra{b_{t}y_{t-1}y_{t-2}})} &= 0
    \end{align}
    Thus we find
    \begin{align}
        \braket{\ket{a_{t}y_{t-1}y_{t-2}}\bra{a_{t}y_{t-1}y_{t-2}}, \D_{\loc}(\ket{b_{t}y_{t-1}y_{t-2}}\bra{b_{t}y_{t-1}y_{t-2}})} &\geq 0
    \end{align}

    Revisiting equation~\eqref{eq:fulllower} we conclude
    \begin{align}
        Q_{x_{>t+2}ay_{<t}, x_{>t+2}by_{<t}} \geq \frac{1}{4}.
    \end{align}
\end{proof}

To lower-bound the gap of the $Q$-matrix associated to $\D_{\unary}|_{\mathbb{D}_{\ill}}$ we upper-bound the second largest eigenvalue of the stochastic matrix generating its jump chain defined in equation~\eqref{eq:jumpchain}. To do so we use a canonical paths argument which requires analyzing the graph of allowable transitions our jump chain permits.

We define the \emph{transition graph} associated to a Markov chain generated by stochastic matrix $P$ as a directed graph consisting of the state space as vertices and edges between vertices $s_{1}$ and $s_{2}$ only if $P_{s_{1}, s_{2}} > 0$. We define a set of \emph{canonical paths} over a directed graph $G$ as a set of paths $\Gamma$ such that for all distinct vertices $s_{1}, s_{2}$ in the graph $G$, there exists a unique path $\gamma_{s_{1},s_{2}} \in \Gamma$ connecting $s_{1}$ to $s_{2}$. Using a set of canonical paths over a transition graph, we can directly upper-bound the second largest eigenvalue associated to the stochastic matrix generating the Markov chain, to how much the set of canonical paths is congested. 

For the specific jump chain associated to the transition-rate matrix appearing in equation~\eqref{eq:Q} of Lemma~\ref{lm:blockismarkovchain}, we define our set of canonical paths over its transition graph as follows. Let $x, y \in \mathrm{S}_{\ill}^{(T)}$ and define indices $s, f \in \{0, 1, \cdots, T-3\}$ such that $x_{s+2}x_{s+1}x_{s} \in \mathrm{S}_{\ill}^{\loc}$ and $y_{f+2}y_{f+1}y_{f} \in \mathrm{S}_{\ill}^{\loc}$ where $\mathrm{S}_{\ill}^{\loc}$ defined according to equation~\eqref{eq:localillandlegstrings}. If $f - s \geq 0$ we define the path $\gamma_{x,y}$ according to the sequence of vertices 
\begin{equation}
    \label{eq:seq1}
    \begin{cases}
        x \rightarrow x_{>s+1}10x_{< s} \rightarrow x_{>s}10x_{< s -1} \rightarrow \cdots \rightarrow x_{>1}10 \quad \text{Phase 1} \\
        x_{>1}10 \rightarrow x_{>2}10y_{0} \rightarrow x_{>3}10y_{1}y_{0} \rightarrow \cdots \rightarrow  10y_{<T-2} \quad \text{Phase 2} \\
        10y_{<T-2} \rightarrow y_{T-1}10y_{<T-3} \rightarrow y_{>T-3}10y_{<T-4} \rightarrow \cdots \rightarrow y_{>f+2}10y_{f}y_{<f} \rightarrow y \quad \text{Phase 3}
    \end{cases}, 
\end{equation}
and if $f - s < 0$ we define $\gamma_{x,y}$ according to the sequence of vertices 
\begin{equation}
    \label{eq:seq2}
    \begin{cases}
        x \rightarrow x_{>s+2}10x_{< s+1} \rightarrow x_{>s+3}10x_{< s +2} \rightarrow \cdots \rightarrow 10x_{< T-2} \quad \text{Phase 1} \\
        10x_{<T-2} \rightarrow y_{T-1}10x_{<T-3} \rightarrow y_{T-1}y_{T-2}10x_{<T-4} \rightarrow \cdots \rightarrow  y_{>1}10 \quad \text{Phase 2} \\
        y_{>1}10 \rightarrow y_{>2}10y_{0} \rightarrow y_{>3}10y_{<2}  \cdots \rightarrow y_{>f+2}y_{f+2}10y_{<f} \rightarrow y \quad \text{Phase 3}
    \end{cases}.
\end{equation}
In both cases we assume cycles have been removed and edges only connect distinct vertices so that $\gamma_{x,y}$ can form a valid path. In the statement which follows we prove $\gamma_{x,y}$ is indeed a path in our transition graph.
\begin{lemma}
    Let $G = (\mathrm{S}_{\ill}^{(T)}, E)$ be the transition graph associated to the jump chain of $Q$ 
    where $Q$ is defined according to equation~\eqref{eq:Q} of Lemma~\ref{lm:blockismarkovchain}. Let $x, y \in \mathrm{S}_{\ill}^{(T)}$ and $\gamma_{x,y}$ be the sequence of edges defined according to equations~\eqref{eq:seq1} and~\eqref{eq:seq2}, then $\gamma_{x,y}$ is a path from $x$ to $y$ in $G$.
\end{lemma}
\begin{proof}
    Assume every $e \in \gamma_{x,y}$ is contained in $E$. By construction $\gamma_{x,y}$ is walk from vertex $x$ to vertex $y$ with all cycles removed, therefore it is a path. Thus it suffices to show that for all $x,y \in \mathrm{S}_{\ill}^{(T)}$ if $e \in \gamma_{x,y}$ then $e \in E$. Since $x,y \in \mathrm{S}_{\ill}^{(T)}$ there exists at least one index $s \in \{0, 1, \cdots, T-3\}$ such that $x_{s+2}x_{s+1}x_{s} \in \mathrm{S}_{\ill}^{\loc}$, and similarly for $y$, there exists at least one index $f \in \{0, 1, \cdots, T-3\}$ such that $y_{f+2}y_{f+1}y_{f} \in \mathrm{S}_{\ill}^{\loc}$. Therefore, for $\gamma_{x,y}$ there exists two cases, either $f - s \geq 0$ or $f-s < 0$.

    Consider the case of $f-s \geq 0$. $\gamma_{x,y}$ is defined according to the sequence of vertices of equation~\eqref{eq:seq1}. Let $e \in \gamma_{x,y}$, we find three cases. Either $e$ is defined from the Phase 1 portion of the sequence, the Phase 2 portion of the sequence, or the Phase 3 portion of the sequence. Consider the first of these three cases. Either $e = (x, x_{>s+1}10x_{<s})$ or there exists some $t \in \{1, 2, \cdots, s\}$ such that $e = (x_{>t+1}10x_{<t}, x_{>t}10x_{<t-1})$. If $e = (x, x_{>s+1}10x_{<s})$, since $\gamma_{x,y}$ only contains edges between distinct vertices we have $x \neq x_{>s+1}10x_{<s}$ implying $x_{s+2}x_{s+1}x_{s} \neq x_{s+2}10$. Since $x_{s+2}x_{s+1}x_{s} \in \mathrm{S}_{\ill}^{\loc}$ and $x_{s+2}10 \in \mathrm{S}_{\ill}^{\loc}$ we can apply Lemma~\ref{lm:Qcalc} implying $Q_{e} \geq 1 /4$. Thus we have $P_{e} > 0$ from equation~\eqref{eq:jumpchain} implying $e \in E$. Similarly, if $e = (x_{>t+1}10x_{<t}, x_{>t}10x_{<t-1})$ for some $t \in \{1, 2, \cdots s \}$ we can apply Lemma~\ref{lm:Qcalc} as well, giving us $Q_{e} \geq 1 / 4$ and thus $P_{e} > 0$ implying $e \in E$. Consider the case of $e \in \gamma_{x,y}$ being defined through the Phase 2 sequence of vertices. Thus there exists some index $t \in \{0, 1, \cdots T-3\}$ such that $e = (x_{>t+1}10y_{<t}, x_{>t+2}10y_{<t+1})$. Similarly, we can apply Lemma~\ref{lm:Qcalc} giving us $Q_{e} \geq 1 / 4$ implying $P_{e} > 0$ and thus $e \in E$. Finally, consider the case of $e \in \gamma_{x,y}$ such that $e$ is defined through the Phase 3 sequence of vertices given in equation~\eqref{eq:seq1}. We have two cases, either there exists a $t \in \{f, f+1, \cdots, T-2\}$ such that $e=(y_{>t+1}10y_{<t}, y_{>t}10y_{<t-1})$ or $e = (y_{>f+1}10y_{<f}, y)$. Since the edges only exist in $\gamma_{x,y}$ if the vertices are distinct, and $y_{f+2}y_{f+1}y_{f}, y_{f+2}10 \in \mathrm{S}_{\ill}^{\loc}$ we can apply Lemma~\ref{lm:Qcalc} in both cases giving us $Q_{e} \geq 1 /4$ and thus $P_{e} > 0$ implying $e \in E$. Therefore in the case of $f- s \geq 0$, for all $e \in \gamma_{x,y}$, $e \in E$.

    Consider the case of $f-s < 0$. $\gamma_{x,y}$ is defined according to the sequence of vertices of equation~\eqref{eq:seq2}. Let $e \in \gamma_{x,y}$, we find three cases again. Either $e$ is defined from the Phase 1 portion of the sequence, the Phase 2 portion of the sequence, or the Phase 3 portion of the sequence. Consider the first of these three cases. Either $e = (x, x_{>s+2}10x_{<s+1})$ or there exists some $t \in \{s, s+1, \cdots, T-3\}$ such that $e = (x_{>t+1}10x_{<t}, x_{>t+2}10x_{<t+1})$. If $e = (x, x_{>s+2}10x_{<s+1})$, since $\gamma_{x,y}$ only contains edges between distinct vertices we have $x \neq x_{>s+2}10x_{<s+1}$ implying $x_{s+2}x_{s+1}x_{s} \neq 10x_{s}$. Since $x_{s+2}x_{s+1}x_{s} \in \mathrm{S}_{\ill}^{\loc}$ and $10x_{s} \in \mathrm{S}_{\ill}^{\loc}$ we can apply Lemma~\ref{lm:Qcalc} implying $Q_{e} \geq 1 /4$. Thus we have $P_{e} > 0$ from equation~\eqref{eq:jumpchain} implying $e \in E$. Similarly, if $e = (x_{>t+1}10x_{<t}, x_{>t+2}10x_{<t+1})$ for some $t \in \{s, s+1, \cdots, T-3 \}$ we can apply Lemma~\ref{lm:Qcalc} as well, giving us $Q_{e} \geq 1 / 4$ and thus $P_{e} > 0$ implying $e \in E$. Consider the case of $e \in \gamma_{x,y}$ being defined through the Phase 2 sequence of vertices. Thus there exists some index $t \in \{1, \cdots T-2\}$ such that $e = (y_{>t+1}10x_{<t}, y_{>t}10x_{<t-1})$. Similarly, we can apply Lemma~\ref{lm:Qcalc} giving us $Q_{e} \geq 1 / 4$ implying $P_{e} > 0$ and thus $e \in E$. Finally, consider the case of $e \in \gamma_{x,y}$ such that $e$ is defined through the Phase 3 sequence of vertices given in equation~\eqref{eq:seq1}. We have two cases, either there exists a $t \in \{0, 1, \cdots, f\}$ such that $e=(y_{>t+1}10y_{<t}, y_{>t+2}10y_{<t+1})$ or $e = (y_{>f+2}10y_{<f+1}, y)$. Since the edges only exist in $\gamma_{x,y}$ if the vertices are distinct, and $y_{f+2}y_{f+1}y_{f}, 10y_{f} \in \mathrm{S}_{\ill}^{\loc}$ we can apply Lemma~\ref{lm:Qcalc} in both cases giving us $Q_{e} \geq 1 /4$ and thus $P_{e} > 0$ implying $e \in E$. Therefore in the case of $f- s < 0$, for all $e \in \gamma_{x,y}$, $e \in E$.

    In both cases we find for all $e \in \gamma_{x,y}$, we have $e \in E$, we conclude that $\gamma_{x,y}$ is a path from $x$ to $y$ in $G$.
\end{proof}

In order to calculate the congestion associated to a set of canonical paths, we require both an upper-bound on the length of the longest path, as well as an upper-bound on the number of vertex pairs that use each edge contained in our set of paths. Intuitively, longer paths connecting vertices make it harder to move units of flow from one vertex to another using our paths. Similarly, the more a single edge is used by different vertex pairs, the more that edge can be overloaded by units of flow, thus congesting the network. In the first of the two statements provided below we upper-bound the path length while in the second we upper-bound the number of vertex pairs that use any single edge contained with in our paths.

\begin{lemma}
    \label{lm:pathlength}
    Let $G = (\mathrm{S}_{\ill}^{(T)}, E)$ be the transition graph associated to the jump chain of $Q$ 
    where $Q$ is defined according to equation~\eqref{eq:Q} of Lemma~\ref{lm:blockismarkovchain}. Let $x, y \in \mathrm{S}_{\ill}^{(T)}$ and $\gamma_{x,y}$ be the path defined according to equations~\eqref{eq:seq1} or~\eqref{eq:seq2}, then 
    \begin{equation}
        |\gamma_{x,y}| \leq 3T.
    \end{equation}
\end{lemma}
\begin{proof}
    Each $\gamma_{x,y}$ is defined via the associated walk given by the sequence of vertices in either equation~\eqref{eq:seq1} or equation~\eqref{eq:seq2}. In both cases since $s, f \in \{0, 1, \cdots, T-3\}$ there are at most $3(T-3)$ vertices in each sequence implying the walk defining $\gamma_{x,y}$ has at most $3(T-3) - 1$ edges. Since removing cycles and self-edges only reduces the number of edges in the resulting sequence we find
    \begin{equation}
        |\gamma_{x,y}| \leq 3T.
    \end{equation}
\end{proof}
\begin{lemma}
    \label{lm:edgeoverload}
    Let $G = (\mathrm{S}_{\ill}^{(T)}, E)$ be the transition graph associated to the jump chain of $Q$ 
    where $Q$ is defined according to equation~\eqref{eq:Q} of Lemma~\ref{lm:blockismarkovchain}. Let $e \in E$ and $\gamma_{x,y}$ be a path from $x$ to $y$ in $G$ defined according to equations~\eqref{eq:seq1} or~\eqref{eq:seq2}, then
    \begin{equation}
        |\{(x,y) \in \mathrm{S}_{\ill}^{(T)\times 2} \; | \; e \in \gamma_{x,y}\}| \leq 8|\mathrm{S}_{\ill}^{(T)}|
    \end{equation}
\end{lemma}
\begin{proof}
    Fix $e \in E$. Either there exists $(x,y) \in \mathrm{S}_{\ill}^{(T)\times 2}$ such that $e \in \gamma_{x,y}$ or for all $(x,y) \in \mathrm{S}_{\ill}^{(T)\times 2}, e \notin \gamma_{x,y}$. Consider the former case of $(x,y) \in \mathrm{S}_{\ill}^{(T)\times 2}$ such that $e \in \gamma_{x,y}$. We have three cases, either $e$ is defined according to the Phase 1 sequence of vertices, the Phase 2 sequence of vertices, or the Phase 3 sequence of vertices described in either equation~\eqref{eq:seq1} or equation~\eqref{eq:seq2}. Without loss of generality assume $\gamma_{x,y}$ is generated according to the sequence of vertices described in equation~\eqref{eq:seq1}.

    Consider the case of $e \in \gamma_{x,y}$ being defined according to the sequence given in Phase 1. Either $e = (x, x_{>s+1}10x_{<s})$ or there exists an index $t \in \{1, 2, \cdots, s\}$ such that $e = (x_{>t+1}10x_{<t}, x_{>t}10x_{<t-1})$. If $e = (x, x_{>s+1}10x_{<s})$ we find 
    \begin{equation}
        |\{(x,y) \in \mathrm{S}_{\ill}^{(T)\times 2} \; | \; e \in \gamma_{x,y}\}| \leq |\mathrm{S}_{\ill}^{(T)}|.
    \end{equation}
    If $e = (x_{>t+1}10x_{<t}, x_{>t}10x_{<t-1})$ for some $t \in \{1, 2, \cdots, s\}$ we find 
    \begin{equation}
        |\{(x,y) \in \mathrm{S}_{\ill}^{(T)\times 2} \; | \; e \in \gamma_{x,y}\}| \leq 2|\mathrm{S}_{\ill}^{(T)}|.
    \end{equation}
    Thus for all $e \in \gamma_{x,y}$ defined according to the sequence given in Phase 1 we have 
    \begin{equation}
        \label{eq:phase1}
        |\{(x,y) \in \mathrm{S}_{\ill}^{(T)\times 2} \; | \; e \in \gamma_{x,y}\}| \leq 2 |\mathrm{S}_{\ill}^{(T)}|.
    \end{equation}

    Consider the case of $e \in \gamma_{x,y}$ being defined according to the sequence given in Phase 2. There exists an index $t \in \{0, 1, \cdots, T-3\}$ such that $e = (x_{>t+1}10y_{<t}, x_{>t+2}10y_{<t+1})$. There exists three cases 
    \begin{enumerate}
        \item $s \leq t$ and $t \leq f$,
        \item $s \leq t$ and $f \leq t$,
        \item $t \leq s$ and $t \leq f$.
    \end{enumerate}

    Consider the first of these three cases. Since we are promised $x_{s+2}x_{s+1}x_{s} \in \mathrm{S}_{\ill}^{\loc}$ and $y_{f+2}y_{f+1}y_{f} \in \mathrm{S}_{\ill}^{\loc}$ we have $x_{\leq s+2} \in \mathrm{S}_{\ill}^{(s+3)}$ and $y_{\geq f} \in \mathrm{S}_{\ill}^{(T-f)}$ which implies 
    \begin{equation}
        \label{eq:case1constraints}
        x_{\leq t+2} \in \mathrm{S}_{\ill}^{(t+3)} \text{ and } y_{\geq t} \in \mathrm{S}_{\ill}^{(T-t)}.
    \end{equation}
    Since our edge $e = (x_{>t+1}10y_{<t}, x_{>t+2}10y_{<t+1})$ fixes $x_{>t+2}$ and $y_{<t}$ with only constraints on $x_{\leq t+2}$ and $y_{\geq t}$ being equation~\eqref{eq:case1constraints} we find
    \begin{align}
        \label{eq:case1e1}
        |\{(x,y) \in \mathrm{S}_{\ill}^{(T)\times 2} \; | \; e \in \gamma_{x,y}\}| &\leq |\mathrm{S}_{\ill}^{(t+3)}| |\mathrm{S}_{\ill}^{(T-t)}| = (2^{t+3} - (t+3))(2^{T-t} - (T-t)) \leq 2^{t+3}(2^{T-t} - (T-t)) \nonumber \\
        &\leq 2^{T + 3} - (T-t)2^{t+3} = 2^{T+3} - f(t).
    \end{align}
    Calculating $\Delta f$ we find
    \begin{align}
        \Delta f = f(t+1) - f(t) = (T-t-1)2^{t+4} - (T-t)2^{t+3} = [T - t - 2 ]2^{t+3} > 0 \text{ for all } t \in \{0, 1, \cdots, T-3\}.
    \end{align}
    Thus the minimum of $f(t)$ occurs at $t=0$ implying
    \begin{align}
        \label{eq:case1e2}
        |\{(x,y) \in \mathrm{S}_{\ill}^{(T)\times 2} \; | \; e \in \gamma_{x,y}\}| &\leq 8 (2^{T} - T) = 8 |\mathrm{S}_{\ill}^{(T)}|.
    \end{align}

    Consider the case of $s \leq t$ and $f \leq t$. Since we are promised $x_{s+2}x_{s+1}x_{s} \in \mathrm{S}_{\ill}^{\loc}$ and $y_{f+2}y_{f+1}y_{f} \in \mathrm{S}_{\ill}^{\loc}$  we have $x_{\leq s+2} \in \mathrm{S}_{\ill}^{(s+3)}$ and $y_{\leq f+2} \in \mathrm{S}_{\ill}^{(f+3)}$ which implies 
    \begin{equation}
        \label{eq:case2constraints}
        x_{\leq t+2} \in \mathrm{S}_{\ill}^{(t+3)} \text{ and } y_{\leq t+2} \in \mathrm{S}_{\ill}^{(t+3)}.
    \end{equation}
    Since our edge $e = (x_{>t+1}10y_{<t}, x_{>t+2}10y_{<t+1})$ fixes $x_{>t+2}$ and $y_{<t}$ with the only constraints on $x_{\leq t+2}$ and $y_{\geq t}$ being equation~\eqref{eq:case2constraints} and $y_{\geq t} \in \{0, 1\}^{T-t}$ we find 
    \begin{align}
        |\{(x,y) \in \mathrm{S}_{\ill}^{(T)\times 2} \; | \; e \in \gamma_{x,y}\}| &\leq |\mathrm{S}_{\ill}^{(t+3)}|2^{T-t} = (2^{t+3} - (t+3))2^{T-t} = 2^{T+3} - (t+3)2^{T-t} = 2^{T+3} - f(t).
    \end{align}
    Calculating $\Delta f$ we find
    \begin{align}
        \Delta f = f(t+1) - f(t) = (t+4)2^{T-t-1} - (t+3)2^{T-t} = \left(-\frac{t}{2} - 1\right)2^{T-t} < 0  \text{ for all } t \in \{0, 1, \cdots, T-3\}.
    \end{align}
    Thus the minimum of $f(t)$ occurs at $t=T-3$ implying
    \begin{align}
        |\{(x,y) \in \mathrm{S}_{\ill}^{(T)\times 2} \; | \; e \in \gamma_{x,y}\}| &\leq 8 (2^{T} - T) = 8 |\mathrm{S}_{\ill}^{(T)}|.
    \end{align}

    Consider the final case of $t \leq s$ and $t \leq f$. Since we are promised $x_{s+2}x_{s+1}x_{s} \in \mathrm{S}_{\ill}^{\loc}$ and $y_{f+2}y_{f+1}y_{f} \in \mathrm{S}_{\ill}^{\loc}$ we have $x_{\geq s} \in \mathrm{S}_{\ill}^{(T-s)}$ and $y_{\geq f} \in \mathrm{S}_{\ill}^{(T-f)}$ which implies 
    \begin{equation}
        \label{eq:case3constraints}
        x_{\geq t} \in \mathrm{S}_{\ill}^{(T-t)} \text{ and } y_{\geq t} \in \mathrm{S}_{\ill}^{(T-t)}.
    \end{equation}
    Since our edge $e = (x_{>t+1}10y_{<t}, x_{>t+2}10y_{<t+1})$ fixes $x_{>t+2}$ and $y_{<t}$ with the only constraints on $x_{\leq t+2}$ and $y_{\geq t}$ being equation~\eqref{eq:case3constraints} and $x_{\leq t+2} \in \{0, 1\}^{t+3}$, we find 
    \begin{align}
        |\{(x,y) \in \mathrm{S}_{\ill}^{(T)\times 2} \; | \; e \in \gamma_{x,y}\}| &\leq |\mathrm{S}_{\ill}^{(T-t)}|2^{t+3} = (2^{T-t} - (T-t))2^{t+3} \leq 8|\mathrm{S}_{\ill}^{(T)}|,
    \end{align}
    where the last inequality follows from applying equations~\eqref{eq:case1e1} and~\eqref{eq:case1e2}.

    Thus, for the case of $e \in \gamma_{x,y}$ being defined according to the sequence given in Phase 2, we conclude from all three cases that
    \begin{align}
        \label{eq:phase2}
        |\{(x,y) \in \mathrm{S}_{\ill}^{(T)\times 2} \; | \; e \in \gamma_{x,y}\}| &\leq  8|\mathrm{S}_{\ill}^{(T)}|.
    \end{align}

    Consider the final case of $e \in \gamma_{x,y}$ being defined according to the sequence of vertices of Phase 3. Either $e = (y_{>f+2}10y_{<f}, y)$ or there exists a $t \in \{f, f+1, \cdots T-3\}$ such that $e = (y_{>t+1}10y_{y<t}, y_{>t}10y_{<y-1})$. If $e = (y_{>f+2}10y_{<f}, y)$ we find
    \begin{align}
        |\{(x,y) \in \mathrm{S}_{\ill}^{(T)\times 2} \; | \; e \in \gamma_{x,y}\}| &\leq |\mathrm{S}_{\ill}^{(T)}|.
    \end{align}
    If $e = (y_{>t+1}10y_{y<t}, y_{>t}10y_{<y-1})$ for some $t \in \{f, f+1, \cdots, T-3\}$ we find 
    \begin{align}
        |\{(x,y) \in \mathrm{S}_{\ill}^{(T)\times 2} \; | \; e \in \gamma_{x,y}\}| &\leq 2|\mathrm{S}_{\ill}^{(T)}|.
    \end{align}
    Therefore, for all $e \in \gamma_{x,y}$ defined according to the sequence given in Phase 3 we have 
    \begin{align}
        \label{eq:phase3}
        |\{(x,y) \in \mathrm{S}_{\ill}^{(T)\times 2} \; | \; e \in \gamma_{x,y}\}| &\leq 2|\mathrm{S}_{\ill}^{(T)}|.
    \end{align}

    From equations~\eqref{eq:phase1}, \eqref{eq:phase2}, and~\eqref{eq:phase3} we conclude that for $e \in E$,
    \begin{align}
        \label{eq:phase3}
        |\{(x,y) \in \mathrm{S}_{\ill}^{(T)\times 2} \; | \; e \in \gamma_{x,y}\}| &\leq 8|\mathrm{S}_{\ill}^{(T)}|
    \end{align}
    as stated.
\end{proof}

With all of the pieces in place, we can proceed towards lower-bounding the gap of our unary clock dissipator when restricted to $\mathbb{D}_{\ill}$. 
\begin{lemma}
    \label{lm:DunaryD}
    Let $\D_{\unary}$ be the Lindbladian dissipator defined according to equation~\eqref{eq:Dunary} of Definition~\ref{def:localDQCSAT} and define subspace  $\mathbb{D}_{\ill}$ according to equation~\eqref{eq:subspacedefs}, then
    \begin{equation}
        \gap{\D_{\unary}|_{\mathbb{D}_{\ill}}} \geq 1 / 96T^{2}.
    \end{equation}
\end{lemma}
\begin{proof}
    Since $Q$ is $\D_{\unary}|_{\mathbb{D}}$ expanded in the basis $\{\ket{x}\bra{x} \; | \; x \in \mathrm{S}_{\ill}^{(T)} \}$ we immediately have
    \begin{equation}
        \label{eq:relationtoQ}
        \gap{\D_{\unary}|_{\mathbb{D}}} = \gap{Q}.
    \end{equation}
    From equation~\eqref{eq:jumpchain} we have $Q = \norm{Q}(P - \mathbb{I})$ where $P$ is the stochastic matrix generating the jump chain associated to $Q$. From Lemma~\ref{lm:blockismarkovchain} and $P \succeq 0$ we have,
    \begin{equation}
        \gap{Q} = \norm{Q}\left(1 - \lambda_{1}(P)\right),
    \end{equation}
    where $\lambda_{i}(P)$ denotes the $i$\textsuperscript{th} eigenvalue of $P$ ordered such that $\lambda_{i}(P) \geq \lambda_{i+1}(P)$ with $i \in \{0, 1, \cdots\}$.

    Let $G = (S_{\ill}^{(T)}, E)$ be the transition graph associated to the jump chain of $Q$. For all edges $e = (x,y) \in E$ define the weight of an edge $w(e) = w(x,y) = \pi_{x} P_{x,y}$ where $\pi$ is the steady-state of the jump chain found in equation~\eqref{eq:steadystate} of Lemma~\ref{lm:blockismarkovchain}. Let $\Gamma$ the set of canonical paths over $G$ where $\gamma_{x,y} \in \Gamma$ is the path from $x$ to $y$ defined according to the sequence of vertices given in equation~\eqref{eq:seq1} or equation~\eqref{eq:seq2}. Define the congestion associated to $\Gamma$ as 
    \begin{equation}
        \label{eq:congestion}
        \bar{\rho}(\Gamma) = \sup_{e \in E} \left(\frac{1}{w(e)} \sum_{\gamma_{x,y} \in \Gamma} \pi_{x}\pi_{y} |\gamma_{x,y}|\mathbf{1}[e \in \gamma_{x,y}]\right).
    \end{equation}
    where $\mathbf{1}$ denotes the indicator function. From Theorem 5 of \cite{sinclair1992Improved} we can upper-bound 
    \begin{equation}
        \lambda_{1}(P) \leq 1 - \frac{1}{\bar{\rho}(\Gamma)}
    \end{equation}
    implying 
    \begin{equation}
        \gap{Q} \geq \frac{\norm{Q}}{\bar{\rho}(\Gamma)}.
    \end{equation}

    We proceed to lower-bound $\gap{Q}$ by upper-bounding the congestion. From our definition of $w$ we find
    \begin{equation}
        \label{eq:wlowerbnd}
        w(e) = w(x,y) = \pi_{x}P_{x,y} = \frac{P_{x,y}}{|\mathrm{S}_{\ill}^{(T)}|} = \frac{Q_{x,y}}{|\mathrm{S}_{\ill}^{(T)}|\norm{Q}} \text{ for } x \neq y.
    \end{equation}

    Calculating the congestion defined in equation~\eqref{eq:congestion} we find
    \begin{equation}
        \label{eq:congestion1}
        \bar{\rho}(\Gamma) = \frac{1}{|\mathrm{S}_{\ill}^{(T)}|^{2}}\sup_{e \in E} \left(\frac{1}{w(e)} \sum_{\gamma_{x,y} \in \Gamma}  |\gamma_{x,y}|\mathbf{1}[e \in \gamma_{x,y}]\right).
    \end{equation} 
    Since the term inside the supremum is only non-zero for the edges contained in a path inside $\Gamma$, we have the additional constraints that $e = (x,y)$ where $x \neq y$ along with $e$ being defined according to the sequence of vertices given in equation~\eqref{eq:seq1} or equation~\eqref{eq:seq2}. Thus equation~\eqref{eq:congestion1} becomes
    \begin{equation}
        \bar{\rho}(\Gamma) \leq \frac{\norm{Q}}{|\mathrm{S}_{\ill}^{(T)}|}\sup_{e \in E} \left(\frac{1}{Q_{e}}\sum_{\gamma_{x,y} \in \Gamma}  |\gamma_{x,y}|\mathbf{1}[e \in \gamma_{x,y}]\right) \leq \frac{4\norm{Q}}{|\mathrm{S}_{\ill}^{(T)}|}\sup_{e \in E} \left(\sum_{\gamma_{x,y} \in \Gamma}  |\gamma_{x,y}|\mathbf{1}[e \in \gamma_{x,y}]\right) 
    \end{equation} 
    where the first inequality follows from equation~\eqref{eq:wlowerbnd} and the second follows from applying Lemma~\ref{lm:Qcalc}. Applying Lemma~\ref{lm:pathlength} we find 
    \begin{equation}
        \bar{\rho}(\Gamma) \leq \frac{12\norm{Q}T}{|\mathrm{S}_{\ill}^{(T)}|}\sup_{e \in E} \left( \sum_{\gamma_{x,y} \in \Gamma}  \mathbf{1}[e \in \gamma_{x,y}]\right) = \frac{12\norm{Q}T}{|\mathrm{S}_{\ill}^{(T)}|}\sup_{e \in E}|\{\gamma_{x,y} \in \Gamma \; | \; e \in \gamma_{x,y}\}|.
    \end{equation} 
    Since for all $x, y \in \mathrm{S}_{\ill}^{(T)}$ there exists a unique $\gamma_{x,y} \in \Gamma$ we can simplify the above equation as
    \begin{equation}
        \bar{\rho}(\Gamma) \leq  \frac{12\norm{Q}T}{|\mathrm{S}_{\ill}^{(T)}|}\sup_{e \in E}|\{(x,y) \in \mathrm{S}_{\ill}^{(T)\times2} \; | \; e \in \gamma_{x,y}\}|.
    \end{equation} 
    Applying Lemma~\ref{lm:edgeoverload} we find 
    \begin{equation}
        \bar{\rho}(\Gamma) \leq  96\norm{Q}T,
    \end{equation}
    which implies 
    \begin{equation}
        \label{eq:gapQ}
        \gap{Q} \geq \frac{1}{96\norm{Q}T}.
    \end{equation}

    Upper-bounding $\norm{Q}$ we find 
    \begin{align}
        \norm{Q} = \norm{\D_{\unary}|_{\mathbb{D}_{\ill}}}_{\infty} \leq \norm{\D_{\unary}}_{\infty} \leq  \sum_{t=0}^{T-3}\norm{\D_{\loc}^{(t+2,t+1,t)}}_{\infty} \leq T.
    \end{align}
    Combining the above equation along with equations~\eqref{eq:relationtoQ} and~\eqref{eq:gapQ} we conclude with
    \begin{equation}
        \gap{\D_{\unary}|_{\mathbb{D}_{\ill}}} \geq \frac{1}{96T^{2}}.
    \end{equation}
\end{proof}

Since we have characterized the spectral gaps of both $\D_{\unary}|_{\mathbb{O}}$ and $\D_{\unary}|_{\mathbb{D}_{\ill}}$ we can conclude our analysis on the gap of $\D_{\unary}$ in the following statement.
\begin{proposition*}[Gap of Unary Clock Dissipator, Proposition~\ref{prop:DuGap}]
    Let $\D_{\unary}$ be the Lindbladian dissipator defined according to equation~\eqref{eq:Dunary} of Definition~\ref{def:localDQCSAT}, then 
    \begin{equation}
        \gap{\D_{\unary}} = 1 / 96T^{3}.
    \end{equation}
\end{proposition*}
\begin{proof}
    Since $\D_{\unary}$ acts only on the clock register $C$  we have
    \begin{equation}
        \gap{\D_{\unary}} =  \gap{\D_{\unary}|_{\mathcal{B}(\mathcal{H}_{C})}}.
    \end{equation}
    From Lemma~\ref{lm:Dublockdiag} and Proposition~\ref{prop:DuKer} we find
    \begin{equation}
        \gap{\D_{\unary}} =  \inf\left\{\gap{\D_{\unary}|_{\mathbb{O}}}, \gap{\D_{\unary}|_{\mathbb{D}_{\ill}}}\right\}.
    \end{equation}
    Applying Lemmas~\ref{lm:DunaryO} and~\ref{lm:DunaryD} we conclude with
    \begin{equation}
        \gap{\D_{\unary}} \geq  \inf\left\{\frac{1}{4T^{3}}, \frac{1}{96T^{2}}\right\} = \frac{1}{96T^{3}}.
    \end{equation}
\end{proof}

\subsection{Leakage}

We conclude this appendix with our section analyzing how much the span of our span of the pure stead-states of our unary clock dissipator leaks out into the rest of the kernel. Our analysis on the leakage is provided in the following statement.
\begin{lemma*}[Lemma~\ref{lm:DuLeakage}]
    Let $\D_{\unary}$ be the Lindbladian dissipator defined according to equation~\eqref{eq:Dunary} of Definition~\ref{def:localDQCSAT} and define $\mathcal{H}_{\leg}^{(T)}$ according to equation~\eqref{eq:globalillandlegsubspaces} and
    \begin{equation}
        \mathsf{M} = \left\{v_{1}\ket{h}\bra{\perp} + v_{2}\ket{\perp}\bra{h} + v_{3}\ket{\perp}\bra{\perp} \; \big| \; (\ket{h}, \ket{\perp}, v) \in \mathcal{H}_{MA} \otimes \mathcal{H}_{\leg}^{(T)} \times [\mathcal{H}_{MA} \otimes \mathcal{H}_{\leg}^{(T)}]^{\perp} \times \mathbb{C}^{3} \right\},
    \end{equation}
    then the overlap between $\ker{\D_{\unary}}$ and $\mathsf{M}$ is given by 
    \begin{equation}
        \cos^{2}\theta = \sup_{\substack{A_{0} \in \ker{\D_{\unary}}, A \in \mathsf{M} \\ \braket{A_{0}, A_{0}} = \braket{A, A} = 1}} |\braket{A_{0}, A}|^{2} \leq \frac{1}{2^{T} -T}.
    \end{equation}
\end{lemma*}
\begin{proof}
    From Proposition~\ref{prop:DuKer} we have for all $A_{0} \in \ker{\D_{\unary}}$  
    \begin{equation}
        A_{0} = B + C_{MA}\otimes P_{\ill}^{(T)}
    \end{equation}
    for $B \in \mathcal{B}(\mathcal{H}_{MA}) \otimes \mathcal{B}(\mathcal{H}_{\leg}^{(T)})$ and $C = \mathcal{B}(\mathcal{H}_{MA})$. Let $\braket{A_{0}, A_{0}} = 1$, we find
    \begin{align}
        1 = \braket{B, B} + \braket{C, C}\left(2^{T} - T\right)
    \end{align}
    which implies
    \begin{equation}
        \label{eq:Cupperbound}
        \begin{cases}
            \braket{B, B} \leq 1 \\
            \braket{C, C} \leq 1 / (2^{T} - T)
        \end{cases}.
    \end{equation}
    Let $A \in \mathsf{M}$ and $\braket{A, A} = 1$, we have 
    \begin{equation}
        A = v_{1}\ket{h}\bra{\perp} + v_{2}\ket{\perp}\bra{h} + v_{3}\ket{\perp}\bra{\perp},
    \end{equation}
    and 
    \begin{equation}
        \label{eq:v3upper}
        1 = |v_{1}|^{2} + |v_{2}|^{2} + |v_{3}|^{2}.
    \end{equation}
    Directly calculating $\braket{A, A_{0}}$ we find 
    \begin{align}
        \braket{A, A_{0}} &=  v_{1}^{*}\bra{\perp}B\ket{h} + v_{1}^{*}\bra{\perp}C \otimes P_{\ill}^{(T)}\ket{h} + v_{2}^{*}\bra{h}B\ket{\perp} \nonumber \\
        &\quad + v_{2}^{*}\bra{h}C\otimes P_{\ill}^{(T)}\ket{\perp} + v_{3}^{*}\bra{\perp}B\ket{\perp} + v_{3}^{*}\bra{\perp}C \otimes P_{\ill}^{(T)}\ket{\perp} \nonumber \\
        &= v_{3}^{*}\bra{\perp}C \otimes P_{\ill}^{(T)}\ket{\perp} = \sum_{a}\sum_{x \in \mathrm{S}_{\ill}^{(T)}}v_{3}^{*}c_{a,x}\bra{\perp}C\otimes P_{\ill}^{(T)}\ket{a}_{MA}\otimes \ket{x}_{C} \nonumber \\
        &= v_{3}^{*}\Tr\left(C_{MA} \otimes \mathbb{I}_{C}\ket{\perp}\bra{\perp}\right) = v_{3}^{*}\Tr\left(C Tr_{C}\left(\ket{\perp}\bra{\perp}\right)\right) = v_{3}^{*}\Tr\left(C \rho\right)
    \end{align}
    where we have defined the quantum state $\rho = \Tr\left(\ket{\perp}\bra{\perp}\right)$.
    Thus we find 
    \begin{align}
        |\braket{A, A_{0}}|^{2} = |v_{3}|^{2} |\braket{C, \rho}|^{2} \leq |\braket{C, \rho}|^{2} \leq  \braket{C, C} \braket{\rho, \rho} \leq  \braket{C, C} \leq 1 / (2^{T} - T)
    \end{align}
    where the first inequality follows from equation~\eqref{eq:v3upper}, the second from Cauchy-Schwarz, the third from $\braket{\rho, \rho}$ being equal to the purity of the quantum state $\rho$ and thus upper-bounded by $1$, and the final from equation~\eqref{eq:Cupperbound}. Therefore we conclude,
    \begin{equation}
        \cos^{2}\theta = \sup_{\substack{A_{0} \in \ker{\D_{\unary}}, A \in \mathsf{M} \\ \braket{A_{0}, A_{0}} = \braket{A, A} = 1}} |\braket{A_{0}, A}|^{2} \leq \frac{1}{2^{T} - T}.
    \end{equation}
\end{proof}

\section{Pure State Decay Projection Lemma}
\label{app:E}

The pure state decay projection lemma is the main tool used throughout Section~\ref{sec:hardness} to prove hardness results for both the \LL{$\log(n)$} problem and the \LL{$k$} problem. In this appendix we provide a proof of the statement.

\begin{lemma*}[Pure State Decay Projection Lemma, Lemma~\ref{lm:psdecayproj}]
    Let $\mathcal{L}_{1}$ and $\mathcal{L}_{2}$ be Hermitian Lindbladians acting on $\mathcal{B}(\mathcal{H}) = \ker{\mathcal{L}_{1}} \oplus (\ker{\mathcal{L}_{1}})^{\perp}$, define subspace $\mathcal{H}_{1} = \mathrm{span}(\{\ket{\psi} \in \mathcal{H} \; | \; \ket{\psi}\bra{\psi} \in \ker{\mathcal{L}_{1}} \})$ which denotes the span of all pure steady states of $\mathcal{L}_{1}$. If $\mathcal{H}_{1} \neq \emptyset$ and $(1 - \cos^{2}\theta)\gap{\mathcal{L}_{1}} \geq 2\norm{\mathcal{L}_{2}}_{\infty}$, then for all $\ket{\psi} \in \mathcal{H}$,
    \begin{equation}
        |\braket{\ket{\psi}\bra{\psi}, \mathrm{Re}(\mathcal{L}_{1} + \mathcal{L}_{2})(\ket{\psi}\bra{\psi})}| \geq \inf_{\ket{h} \in \mathcal{H}_{1}}|\braket{\ket{h}\bra{h}, \mathrm{Re}\mathcal{L}_{2}|_{\mathcal{B}(\mathcal{H}_{1})}(\ket{h}\bra{h})}|- \frac{2\norm{\mathcal{L}_{2}}_{\infty}^{2}}{(1 - \cos^{2}\theta)\gap{\mathcal{L}_{1}} - 2 \norm{\mathcal{L}_{2}}_{\infty}},
    \end{equation}
    where 
    \begin{equation}
        \cos^{2}\theta = \sup_{\substack{A_{0} \in \ker{\mathcal{L}_{1}}, A \in \mathsf{M}\\ \braket{A_{0},A_{0}} = \braket{A , A } = 1}} |\braket{A_{0},A}|^{2} 
    \end{equation}
    denotes the overlap between $\ker{\mathcal{L}_{1}}$ and
    \begin{equation}
        \mathsf{M} = \left\{ v_{1}\ket{h}\bra{\perp} + v_{2}\ket{\perp}\bra{h} + v_{3}\ket{\perp}\bra{\perp} \in \mathcal{B}(\mathcal{H}_{1})^{\perp} \; | \; (\ket{h}, \ket{\perp}, v) \in \mathcal{H}_{1} \times \mathcal{H}_{1}^{\perp} \times \mathbb{C}^{3} \right\}.
    \end{equation}
\end{lemma*}
\begin{proof}
    Assume $\mathcal{H}_{1} \neq \emptyset$. We can partition $\mathcal{H} = \mathcal{H}_{1} \oplus \mathcal{H}_{1}^{\perp}$ implying for all $\ket{\psi} \in \mathcal{H}$ there exists a $\ket{h} \in \mathcal{H}_{1}$ and $\ket{\perp} \in \mathcal{H}_{1}^{\perp}$ such that $\ket{\psi} = \alpha_{1}\ket{h} + \alpha_{2}\ket{\perp}$ for $\alpha_{1}^{2} + \alpha_{2}^{2} = 1$. 
    Without loss of generality we take $\alpha_{1}, \alpha_{2} \geq 0$ implying $\alpha_{1}, \alpha_{2} \in [0,1]$. Calculating $\ket{\psi}\bra{\psi}$, we find 
    \begin{equation}
        \ket{\psi}\bra{\psi} = \alpha_{1}^{2}\ket{h}\bra{h} + \alpha_{1}\alpha_{2}\ket{h}\bra{\perp} + \alpha_{1}\alpha_{2}\ket{\perp}\bra{h} + \alpha_{2}^{2}\ket{\perp}\bra{\perp} = \alpha_{1}^{2} \ket{h}\bra{h} + \tilde{A}_{\perp},
    \end{equation}
    where we have defined $A_{\perp} = \alpha_{1}\alpha_{2}\ket{h}\bra{\perp} + \alpha_{1}\alpha_{2}\ket{\perp}\bra{h} + \alpha_{2}^{2}\ket{\perp}\bra{\perp}$. Similarly, we can calculate $\braket{\tilde{A}_{\perp},\tilde{A}_{\perp}}$ finding,
    \begin{equation}
        \label{eq:normofperp}
        \braket{\tilde{A}_{\perp},\tilde{A}_{\perp}} = 2\alpha_{1}^{2}\alpha_{2}^{2} + \alpha_{2}^{4} = (2 - \alpha_{2}^{2})\alpha_{2}^{2}.
    \end{equation}
    
    Assume $\braket{\tilde{A}_{\perp},\tilde{A}_{\perp}}=0$, we find from the above equation and $\alpha_{1},\alpha_{2} \in [0,1]$ that $\alpha_{2}=0$. Since $\ket{h} \in \mathcal{H}_{1}$ we have $\ket{h}\bra{h} \in \ker{\re\mathcal{L}_{1}}$ which implies the statement of the lemma. Thus we assume $\alpha_{2} \neq 0$ implying we can normalize $\tilde{A}_{\perp}$ allowing us to write
    \begin{equation}
        \ket{\psi}\bra{\psi} = \alpha_{1}^{2} \ket{h}\bra{h} + \sqrt{\braket{\tilde{A}_{\perp}, \tilde{A}_{\perp}}}A_{\perp} = \alpha_{1}^{2}\ket{h}\bra{h} + \alpha_{2}\sqrt{2 - \alpha_{2}^{2}}A_{\perp}
    \end{equation}
    where we have defined $A_{\perp} = \tilde{A}_{\perp} / \sqrt{\braket{\tilde{A}_{\perp}, \tilde{A}_{\perp}}}$.

    Calculating the decay of $\ket{\psi} \in \mathcal{H}$ under $\mathcal{L}_{1} + \mathcal{L}_{2}$, we find, 
    \begin{align}
        \label{eq:decay}
        |\braket{\ket{\psi}\bra{\psi}, \mathrm{Re}(\mathcal{L}_{1} + \mathcal{L}_{2})(\ket{\psi}\bra{\psi})}| &= - \braket{\ket{\psi}\bra{\psi}, \mathrm{Re}\mathcal{L}_{1}(\ket{\psi}\bra{\psi})} - \braket{\ket{\psi}\bra{\psi}, \mathrm{Re}\mathcal{L}_{2}(\ket{\psi}\bra{\psi})} \\
                                                                                                                &= - \alpha_{2}^{2}(2 - \alpha_{2}^{2})\braket{A_{\perp}, \mathrm{Re}\mathcal{L}_{1}(A_{\perp})} - \braket{\ket{\psi}\bra{\psi}, \mathrm{Re}\mathcal{L}_{2}(\ket{\psi}\bra{\psi})}
    \end{align}
    where the first equality follows from Lemma~\ref{lm:repartL}. Defining $g = -\braket{A_{\perp}, \mathrm{Re}\mathcal{L}_{1}(A_{\perp})}$ and $\mathrm{D} = - \braket{\ket{h}\bra{h}, \re\mathcal{L}_{2}(\ket{h}\bra{h})}$ we find,
    \begin{align}
        \label{eq:decay1}
        |\braket{\ket{\psi}\bra{\psi}, \mathrm{Re}(\mathcal{L}_{1} + \mathcal{L}_{2})(\ket{\psi}\bra{\psi})}| = &\alpha_{2}^{2}(2 - \alpha_{2}^{2})g + \mathrm{D}\alpha_{1}^{4} -2\re\braket{\ket{h}\bra{h}, \re\mathcal{L}_{2}(A_{\perp})}\alpha_{1}^{2}\alpha_{2}\sqrt{2-\alpha_{2}^{2}} \nonumber \\
                                                                                                                &- \braket{A_{\perp}, \re\mathcal{L}_{2}(A_{\perp})}\alpha_{2}^{2}(2-\alpha_{2}^{2}).
    \end{align}
    Since $\mathcal{L}_{2}$ is Hermitian, $\braket{A_{\perp}, \re\mathcal{L}_{2}(A_{\perp})} = \braket{A_{\perp}, \mathcal{L}_{2}(A_{\perp})} \leq 0$. Thus, we can drop the last term in equation~\eqref{eq:decay1} finding
    \begin{align}
        \label{eq:decay2}
        |\braket{\ket{\psi}\bra{\psi}, \mathrm{Re}(\mathcal{L}_{1} + \mathcal{L}_{2})(\ket{\psi}\bra{\psi})}| &\geq \alpha_{2}^{2}(2 - \alpha_{2}^{2})g + \mathrm{D}\alpha_{1}^{4} -2\re\braket{\ket{h}\bra{h}, \re\mathcal{L}_{2}(A_{\perp})}\alpha_{1}^{2}\alpha_{2}\sqrt{2-\alpha_{2}^{2}} \nonumber \\
                                                                                                                &\geq \alpha_{2}^{2}(2 - \alpha_{2}^{2})g + \mathrm{D}\alpha_{1}^{4} -2|\braket{\ket{h}\bra{h}, \mathcal{L}_{2}(A_{\perp})}|\alpha_{1}^{2}\alpha_{2}\sqrt{2-\alpha_{2}^{2}} \nonumber \\
                                                                                                                    &\geq \alpha_{2}^{2}(2 - \alpha_{2}^{2})g + \mathrm{D}\alpha_{1}^{4} - 2\norm{\mathcal{L}_{2}}_{\infty}\alpha_{1}^{2}\alpha_{2}\sqrt{2-\alpha_{2}^{2}},
    \end{align}
    where the last inequality follows from $\braket{A_{\perp}, A_{\perp}} = \braket{\ket{h}\bra{h}, \ket{h}\bra{h}} = 1$ and applying our definition of $\norm{\bullet}_{\infty}$. Rewriting the lower-bound with respect to $\alpha_{2}$ we find,
    \begin{align}
        \label{eq:decay3}
        |\braket{\ket{\psi}\bra{\psi}, \mathrm{Re}(\mathcal{L}_{1} + \mathcal{L}_{2})(\ket{\psi}\bra{\psi})}| &\geq \alpha_{2}^{2}(2 - \alpha_{2}^{2})g + \mathrm{D}(1-\alpha_{2}^{2})^{2} - 2\norm{\mathcal{L}_{2}}_{\infty}(1 - \alpha_{2}^{2})\alpha_{2}\sqrt{2-\alpha_{2}^{2}} \nonumber \\
                                                                                                                        &\geq \alpha_{2}^{2}(2 - \alpha_{2}^{2})g + \mathrm{D}(1-2\alpha_{2}^{2} + \alpha_{2}^{4}) - 2\sqrt{2}\norm{\mathcal{L}_{2}}_{\infty}(1 - \alpha_{2}^{2})\alpha_{2} \nonumber \\
                                                                                                                        &\geq \mathrm{D} + \alpha_{2}^{2}(2 - \alpha_{2}^{2})g - 2\mathrm{D}\alpha_{2}^{2} - 2\sqrt{2}\norm{\mathcal{L}_{2}}_{\infty}\alpha_{2} \nonumber \\
                                                                                                                          &\geq \mathrm{D} + \alpha_{2}^{2}(2 - \alpha_{2}^{2})g - 2\norm{\mathcal{L}_{2}}_{\infty}\alpha_{2}^{2} - 2\sqrt{2}\norm{\mathcal{L}_{2}}_{\infty}\alpha_{2}
    \end{align}
    where the third inequality follows from $\mathrm{D} = -\braket{\ket{h}\bra{h}, \re\mathcal{L}_{2}(\ket{h}\bra{h})} \geq 0$ which follows from Lemma~\ref{lm:repartL} and the fourth follows from $\mathrm{D} = -\braket{\ket{h}\bra{h}, \re\mathcal{L}_{2}(\ket{h}\bra{h})} = |\braket{\ket{h}\bra{h}, \mathcal{L}_{2}(\ket{h}\bra{h})}| \leq \norm{\mathcal{L}_{2}}_{\infty}$.

    To calculate $g$, note $\mathcal{B}(\mathcal{H}) = \ker{\mathcal{L}_{1}} \oplus \ker{\mathcal{L}_{1}}^{\perp}$, thus we can expand
    \begin{equation}
        A_{\perp} =  \beta_{1}A_{0} + \beta_{2}A_{g},
    \end{equation}
    where $A_{0} \in \ker{\mathcal{L}_{1}}$, $A_{g} \in \ker{\mathcal{L}_{1}}^{\perp}$, and $\braket{A_{0},A_{0}} = \braket{A_{g}, A_{g}} = 1$.
    Since $\braket{A_{\perp}, A_{\perp}}=1$ we find $|\beta_{1}|^{2} + |\beta_{2}|^{2} = 1$. Plugging this into our definition for $g$ we find,
    \begin{equation}
        g = - \braket{A_{\perp}, \re\mathcal{L}_{1}(A_{\perp})} = - |\beta_{2}|^{2}\braket{A_{g}, \mathcal{L}_{1}(A_{g})} \geq  (1 - |\beta_{1}|^{2})\gap{\mathcal{L}_{1}} \geq \left(1 - |\braket{A_{0}, A_{\perp}}|^{2}\right)\gap{\mathcal{L}_{1}}.
    \end{equation}
    Note, from definition we have $A_{\perp} \in \mathsf{M}$ thus we can further lower-bound $g$ as 
    \begin{equation}
        \label{eq:g}
        g \geq \left(1 - \sup_{\substack{A_{0} \in \ker{\mathcal{L}_{1}}, A_{\perp} \in \mathsf{M}\\ \braket{A_{0},A_{0}} = \braket{A_{\perp}, A_{\perp}} = 1}}|\braket{A_{0}, A_{\perp}}|^{2} \right)\gap{\mathcal{L}_{1}} = \left(1 - \cos^{2}\theta \right) \gap{\mathcal{L}_{1}} \geq 0.
    \end{equation}
    Therefore applying equation~\eqref{eq:g} to equation~\eqref{eq:decay3} we find 
    \begin{equation}
        |\braket{\ket{\psi}\bra{\psi}, \mathrm{Re}(\mathcal{L}_{1} + \mathcal{L}_{2})(\ket{\psi}\bra{\psi})}| \geq  \mathrm{D} - 2\sqrt{2}\norm{\mathcal{L}_{2}}_{\infty}\alpha_{2} + \left[(1 - \cos^{2}\theta) \gap{\mathcal{L}_{1}} - 2 \norm{\mathcal{L}_{2}}_{\infty} \right] \alpha_{2}^{2} = f(\alpha_{2}).
    \end{equation}

    Assume $(1 - \cos^{2}\theta)\gap{\mathcal{L}_{1}} > 2\norm{\mathcal{L}_{2}}$, calculating
    \begin{equation}
        0 = \frac{df}{d\alpha_{2}}\bigg|_{\alpha_{2} = \alpha} = -2\sqrt{2}\norm{\mathcal{L}_{2}}_{\infty} + 2\left[(1 - \cos^{2}\theta) \gap{\mathcal{L}_{1}} - 2\norm{\mathcal{L}_{2}}_{\infty}\right]\alpha,
    \end{equation}
    which implies $\alpha = \sqrt{2} \norm{\mathcal{L}_{2}} / \left[(1 - \cos^{2}\theta) \gap{\mathcal{L}_{1}} - 2\norm{\mathcal{L}_{2}}_{\infty}\right]$. Since 
    \begin{equation}
        \frac{d^{2}f}{d\alpha_{2}^{2}}\bigg|_{\alpha_{2} = \alpha} > 0,
    \end{equation}
    $f(\alpha)$ is a global minimum therefore we find 
    \begin{equation}
        |\braket{\ket{\psi}\bra{\psi}, \mathrm{Re}(\mathcal{L}_{1} + \mathcal{L}_{2})(\ket{\psi}\bra{\psi})}| \geq  f(\alpha) \geq \mathrm{D} - \frac{2\norm{\mathcal{L}_{2}}_{\infty}^{2}}{(1 - \cos^{2}\theta)\gap{\mathcal{L}_{1}} - 2 \norm{\mathcal{L}_{2}}_{\infty}},
    \end{equation}
    and the statement of Lemma~\ref{lm:psdecayproj} follows.
\end{proof}

\section{History State Decay under \textsc{Quantum Circuit-SAT} Dissipators}
\label{app:F}

A key feature in our proofs of the completeness and soundness of our reductions from \textsc{Quantum Circuit-SAT} to our various \textsc{Local Lindbladian} problems is our ability to relate the decay of a properly initialized history state, under the constructed Lindbladian, to the probability that the embedded quantum circuit accepts the input to the specific history state. Such a feature is captured in Lemmas~\ref{lm:properdecay} and~\ref{lm:properdecayloc} of Section~\ref{sec:hardness}. In this appendix we provide the proofs of these two statements. The first being the statement of this key feature for the quasi-local \textsc{Quantum Circuit-SAT} dissipator and the second for the local \textsc{Quantum Circuit-SAT} dissipator.

\begin{lemma*}[Lemma~\ref{lm:properdecay}]
    Let $(V, b,a)$ be an instance of \textsc{Quantum Circuit-SAT} and $\D^{(V, b, a)}$ defined according to equation~\eqref{eq:DQCSAT} of Definition~\ref{def:QCSATdis}. Let $\ket{\eta_{\phi}}$ be a pure history state associated to $V$ with input state $\ket{\phi} = \ket{\psi}_{M}\otimes\ket{0}_{A} \in \mathcal{H}_{M} \otimes \mathcal{H}_{A}$, then the decay of $\ket{\eta_{\phi}}\bra{\eta_{\phi}}$ under $\D^{(V, a, b)}$  is given by  
    \begin{align}
        |\braket{\ket{\eta_{\phi}}\bra{\eta_{\phi}}, \re\D^{(V, a, b)}(\ket{\eta_{\phi}}\bra{\eta_{\phi}})}| = \Big(1 - \frac{1}{T} &\Tr\left(\ket{0}\bra{0}_{M_{0}}V\ket{\phi}\bra{\phi}V^{*}\right)\Big) \Tr\left(\ket{0}\bra{0}_{M_{0}}V\ket{\phi}\bra{\phi}V^{*}\right).
    \end{align}
\end{lemma*}
\begin{proof}
    Consider a pure history state associated $V$ given by $\ket{\eta_{\phi}}$. From Definition~\ref{def:QCSATdis} we can calculate the action of $\D^{(V, b, a)}$ on $\ket{\eta_{\phi}}$ finding,
    \begin{align}
        \D^{(V, b, a)}(\ket{\eta_{\phi}}\bra{\eta_{\phi}}) &= \Gamma_{\inp}\D_{\inp}(\ket{\eta_{\phi}}\bra{\eta_{\phi}}) + \Gamma_{\prop}\D_{\prop}^{(V)}(\ket{\eta_{\phi}}\bra{\eta_{\phi}}) + \Gamma_{\out}\D_{\out}(\ket{\eta_{\phi}}\bra{\eta_{\phi}}) \nonumber \\
                                                                    &= \Gamma_{\inp}\D_{\inp}(\ket{\eta_{\phi}}\bra{\eta_{\phi}}) + \Gamma_{\out}\D_{\out}(\ket{\eta_{\phi}}\bra{\eta_{\phi}}),
    \end{align}
    where the last equality follows from applying Proposition~\ref{prop:propdisker}. Calculating $\D_{\inp}(\ket{\eta_{\phi}}\bra{\eta_{\phi}})$ we find,
    \begin{align}
        \D_{\inp}(\ket{\eta_{\phi}}\bra{\eta_{\phi}}) &= \sum_{i=1}^{k}\left((\ket{1}\bra{1}_{A_{i}}\otimes\ket{0}\bra{0}_{C})\ket{\eta_{\phi}}\bra{\eta_{\phi}} (\ket{1}\bra{1}_{A_{i}} \otimes \ket{0}\bra{0}_{C}) - \frac{1}{2} \left\{\ket{1}\bra{1}_{A_{i}} \otimes \ket{0}\bra{0}_{C}, \ket{\eta_{\phi}}\bra{\eta_{\phi}}\right\}\right) \nonumber \\
                                                      &= \sum_{i=1}^{k}\left(\frac{1}{T}\ket{\psi}\bra{\psi}_{M}\otimes(\ket{1}\bra{1}_{A_{i}})\ket{0}\bra{0}_{A}(\ket{1}\bra{1}_{A_{i}})\otimes\ket{0}_{C}\bra{0}_{C} - \frac{1}{2} \left\{\ket{1}\bra{1}_{A_{i}} \otimes \ket{0}\bra{0}_{C}, \ket{\eta_{\phi}}\bra{\eta_{\phi}}\right\}\right) \nonumber \\
                                                      &= -\frac{1}{2\sqrt{T}}\sum_{i=1}^{k}\left(\ket{\psi}_{M}\otimes(\ket{1}\bra{1}_{A_{i}})\ket{0}_{A} \otimes \ket{0}_{C}\bra{\eta_{\phi}} + \ket{\eta_{\phi}}\bra{\psi}_{M}\otimes\bra{0}_{A}(\ket{1}\bra{1}_{A_{i}}) \otimes \bra{0}_{C}\right) = 0.
    \end{align}
    Calculating $\D_{\out}(\ket{\eta_{\phi}}\bra{\eta_{\phi}})$ we find
    \begin{align}
        \D_{\out}(\ket{\eta_{\phi}}\bra{\eta_{\phi}}) = &(\ket{0}\bra{0}_{M_{0}}\otimes\ket{T-1}\bra{T-1}_{C})\ket{\eta_{\phi}}\bra{\eta_{\phi}} (\ket{0}\bra{0}_{M_{0}} \otimes \ket{T-1}\bra{T-1}_{C}) \nonumber \\
                                                        &- \frac{1}{2} \left\{\ket{0}\bra{0}_{M_{0}} \otimes \ket{T-1}\bra{T-1}_{C}, \ket{\eta_{\phi}}\bra{\eta_{\phi}}\right\},
    \end{align}
    which becomes,
    \begin{align}
        \label{eq:dout}
        \D_{\out}(\ket{\eta_{\phi}}\bra{\eta_{\phi}}) = \frac{1}{T}&(\ket{0}\bra{0}_{M_{0}})V(\ket{\phi}\bra{\phi}_{MA})V^{*}(\ket{0}\bra{0}_{M_{0}})\otimes\ket{T-1}_{C}\bra{T-1}_{C} \nonumber \\
                                                                   &- \frac{1}{2\sqrt{T}} (\ket{0}\bra{0}_{M_{0}})V\ket{\phi}_{MA} \otimes \ket{T-1}_{C}\bra{\eta_{\phi}} \nonumber \\
                                                                   &- \frac{1}{2\sqrt{T}}\ket{\eta_{\phi}}\bra{\phi}_{MA}V^{*}(\ket{0}\bra{0}_{M_{0}}) \otimes \bra{T-1}_{C}.
    \end{align}
    Thus we find $\D^{(V, b, a)}(\ket{\eta_{\phi}}\bra{\eta_{\phi}}) = \Gamma_{\out}\D_{\out}(\ket{\eta_{\phi}}\bra{\eta_{\phi}})$ where $\D_{\out}(\ket{\eta_{\phi}}\bra{\eta_{\phi}})$ is given by equation~\eqref{eq:dout}.

    Calculating the decay of $\ket{\eta_{\phi}}\bra{\eta_{\phi}}$ under $\D^{(V, b, a)}$ gives us,
    \begin{equation}
        |\braket{\ket{\eta_{\phi}}\bra{\eta_{\phi}}, \re\D^{(V, b, a)}(\ket{\eta_{\phi}}\bra{\eta_{\phi}})}| = |\braket{\ket{\eta_{\phi}}\bra{\eta_{\phi}}, \D^{(V, b, a)}(\ket{\eta_{\phi}}\bra{\eta_{\phi}})}| = \Gamma_{\out} |\braket{\ket{\eta_{\phi}}\bra{\eta_{\phi}}, \D_{\out}(\ket{\eta_{\phi}}\bra{\eta_{\phi}})}|.
    \end{equation}
    Applying equation~\eqref{eq:dout}, and our definition of $\Gamma_{\out}$ found in equation~\eqref{eq:dissrates} of Definition~\ref{def:QCSATdis}, we find
    \begin{equation}
        |\braket{\ket{\eta_{\phi}}\bra{\eta_{\phi}}, \re\D^{(V, b, a)}(\ket{\eta_{\phi}}\bra{\eta_{\phi}})}| =  \Big| \frac{1}{T}\bra{\phi}V^{*} (\ket{0}\bra{0}_{M_{0}})V\ket{\phi}\bra{\phi}V^{*}(\ket{0}\bra{0}_{M_{0}})V\ket{\phi} - \bra{\phi}V^{*}(\ket{0}\bra{0}_{M_{0}})V\ket{\phi} \Big|,
    \end{equation}
    which simplifies to 
    \begin{equation}
        |\braket{\ket{\eta_{\phi}}\bra{\eta_{\phi}}, \re\D^{(V, a, b)}(\ket{\eta_{\phi}}\bra{\eta_{\phi}})}| = \Big| \frac{1}{T} \Tr\left(\ket{0}\bra{0}_{M_{0}}V\ket{\phi}\bra{\phi}V^{*}\right) -  1\Big| \Tr\left(\ket{0}\bra{0}_{M_{0}}V\ket{\phi}\bra{\phi}V^{*}\right).
    \end{equation}
    Note, the term inside the absolute value is always less than zero since $0 \leq \Tr\left(\ket{0}\bra{0}_{M_{0}}V\ket{\phi}\bra{\phi}V^{*}\right) \leq 1$, thus
    \begin{equation}
        |\braket{\ket{\eta_{\phi}}\bra{\eta_{\phi}}, \re\D^{(V, a, b)}(\ket{\eta_{\phi}}\bra{\eta_{\phi}})}| = \Big(1 - \frac{1}{T} \Tr\left(\ket{0}\bra{0}_{M_{0}}V\ket{\phi}\bra{\phi}V^{*}\right)\Big)\Tr\left(\ket{0}\bra{0}_{M_{0}}V\ket{\phi}\bra{\phi}V^{*}\right),
    \end{equation}
    implying the statement of Lemma~\ref{lm:properdecay} holds.
\end{proof}

\begin{lemma*}[Lemma~\ref{lm:properdecayloc}]
    Let $(V, b,a)$ be an instance of \textsc{Quantum Circuit-SAT} and $\D_{\loc}^{(V, a, b)}$ defined according to equation~\eqref{eq:Dloc} of Definition~\ref{def:localDQCSAT}.  Let $\ket{\eta_{\phi}}$ be a pure history state associated to $V$ with input state $\ket{\phi} = \ket{\psi}_{M}\otimes\ket{0}_{A} \in \mathcal{H}_{M}\otimes\mathcal{H}_{A}$, then the decay of $\ket{\eta_{\phi}}\bra{\eta_{\phi}}$ under $\D_{\loc}^{(V, a, b)}$  is given by  
    \begin{align}
        |\braket{\ket{\eta_{\phi}}\bra{\eta_{\phi}}, \re\D_{\loc}^{(V, a, b)}(\ket{\eta_{\phi}}\bra{\eta_{\phi}})}| = \Big(1 - \frac{1}{T} &\Tr\left(\ket{0}\bra{0}_{M_{0}}V\ket{\phi}\bra{\phi}V^{*}\right)\Big) \Tr\left(\ket{0}\bra{0}_{M_{0}}V\ket{\phi}\bra{\phi}V^{*}\right).
    \end{align}
\end{lemma*}
\begin{proof}
    Consider a pure history state associated $V$ given by $\ket{\tilde{\eta}_{\phi}}$. From Definition~\ref{def:localDQCSAT} we can calculate the action of $\D_{\loc}^{(V, b, a)}$ on $\ket{\tilde{\eta}_{\phi}}$ finding,
    \begin{equation}
        \D_{\loc}^{(V, b, a)}(\ket{\tilde{\eta}_{\phi}}\bra{\tilde{\eta}_{\phi}}) = \Gamma_{\unary} \D_{\unary}(\ket{\tilde{\eta}_{\phi}}\bra{\tilde{\eta}_{\phi}}) + \Gamma_{\inp}\tilde{\D}_{\inp}(\ket{\tilde{\eta}_{\phi}}\bra{\tilde{\eta}_{\phi}}) + \Gamma_{\prop}\tilde{\D}_{\prop}^{(V)}(\ket{\tilde{\eta}_{\phi}}\bra{\tilde{\eta}_{\phi}}) + \Gamma_{\out}\tilde{\D}_{\out}(\ket{\tilde{\eta}_{\phi}}\bra{\tilde{\eta}_{\phi}}).
    \end{equation}
    From Proposition~\ref{prop:DuKer} we have $\ker{\D_{\unary}} = \mathcal{B}(\mathcal{H}_{MA}) \otimes \left(\mathcal{B}(\mathcal{H}_{\leg}^{(T)}) \oplus P_{\ill}^{(T)}\right)$. In-addition we have have $\ket{\tilde{\eta}_{\phi}} \in \mathcal{H}_{\qcirc} \subseteq \mathcal{B}(\mathcal{H}_{MA}\otimes \mathcal{H}_{\leg}^{(T)})$ where $\mathcal{H}_{\qcirc}$ is defined according to equation~\eqref{eq:Hqcirc} of Definition~\ref{def:circsubspace} after having applied the unary map of equation~\eqref{eq:unary} to our clock register $C$. Thus we have $\D_{\unary}(\ket{\tilde{\eta}_{\phi}}\bra{\tilde{\eta}_{\phi}}) = 0$. A similar argument holds for $\tilde{\D}_{\prop}^{(V)}(\ket{\tilde{\eta}_{\phi}}\bra{\tilde{\eta}_{\phi}})$ since the mappings given in equations~\eqref{eq:unary} and~\eqref{eq:local} can be reversed on $\mathcal{B}(\mathcal{H}_{MA} \otimes \mathcal{H}_{\leg}^{(T)})$ giving us $\tilde{\D}_{\prop}^{(V)}(\ket{\tilde{\eta}_{\phi}}\bra{\tilde{\eta}_{\phi}}) = \D_{\prop}^{(V)}(\ket{\eta_{\phi}}\bra{\eta_{\phi}}) = 0$ where $\ket{\eta_{\phi}}$ denotes the pure history state $\ket{\tilde{\eta}_{\phi}}$ with the clock register $C$ encoded in binary.

    Calculating $\tilde{\D}_{\inp}(\ket{\tilde{\eta}_{\phi}}\bra{\tilde{\eta}_{\phi}})$ we find,
    \begin{align}
        \tilde{\D}_{\inp}(\ket{\tilde{\eta}_{\phi}}\bra{\tilde{\eta}_{\phi}}) &= \sum_{i=0}^{k-1}\left((\ket{1}\bra{1}_{A_{i}}\otimes\ket{0}\bra{0}_{1})\ket{\tilde{\eta}_{\phi}}\bra{\tilde{\eta}_{\phi}} (\ket{1}\bra{1}_{A_{i}} \otimes \ket{0}\bra{0}_{1}) - \frac{1}{2} \left\{\ket{1}\bra{1}_{A_{i}} \otimes \ket{0}\bra{0}_{1}, \ket{\tilde{\eta}_{\phi}}\bra{\tilde{\eta}_{\phi}}\right\}\right) \nonumber \\
                                                      &= \sum_{i=0}^{k-1}\left(\frac{1}{T}\ket{\psi}\bra{\psi}_{M}\otimes(\ket{1}\bra{1}_{A_{i}})\ket{0}\bra{0}_{A}(\ket{1}\bra{1}_{A_{i}})\otimes\ket{0}\bra{0}^{\otimes T-1}\otimes\ket{1}\bra{1}_{0} - \frac{1}{2} \left\{\ket{1}\bra{1}_{A_{i}} \otimes \ket{0}\bra{0}_{C}, \ket{\tilde{\eta}_{\phi}}\bra{\tilde{\eta}_{\phi}}\right\}\right) \nonumber \\
                                                      &= -\frac{1}{2\sqrt{T}}\sum_{i=0}^{k-1}\left(\ket{\psi}_{M}\otimes(\ket{1}\bra{1}_{A_{i}})\ket{0}_{A} \otimes \ket{0}^{\otimes T-1}\otimes \ket{1}_{0} \bra{\tilde{\eta}_{\phi}} + \ket{\tilde{\eta}_{\phi}}\bra{\psi}_{M}\otimes\bra{0}_{A}(\ket{1}\bra{1}_{A_{i}}) \otimes \bra{0}^{\otimes T-1} \otimes \bra{1}_{0}\right) \nonumber \\
                                                      & = 0.
    \end{align}
    Calculating $\tilde{\D}_{\out}(\ket{\tilde{\eta}_{\phi}}\bra{\tilde{\eta}_{\phi}})$ we find
    \begin{align}
        \D_{\out}(\ket{\tilde{\eta}_{\phi}}\bra{\tilde{\eta}_{\phi}}) = &(\ket{0}\bra{0}_{M_{0}}\otimes\ket{1}\bra{1}^{\otimes T})\ket{\tilde{\eta}_{\phi}}\bra{\tilde{\eta}_{\phi}} (\ket{0}\bra{0}_{M_{0}} \otimes \ket{1}\bra{1}^{\otimes T}) \nonumber \\
                                                                        &- \frac{1}{2} \left\{\ket{0}\bra{0}_{M_{0}} \otimes \ket{1}\bra{1}^{\otimes T}, \ket{\tilde{\eta}_{\phi}}\bra{\tilde{\eta}_{\phi}}\right\},
    \end{align}
    which becomes,
    \begin{align}
        \label{eq:doutloc}
        \D_{\out}(\ket{\tilde{\eta}_{\phi}}\bra{\tilde{\eta}_{\phi}}) = \frac{1}{T}&(\ket{0}\bra{0}_{M_{0}})V(\ket{\phi}\bra{\phi}_{MA})V^{*}(\ket{0}\bra{0}_{M_{0}})\otimes\ket{1}\bra{1}^{\otimes T} \nonumber \\
                                                                   &- \frac{1}{2\sqrt{T}} (\ket{0}\bra{0}_{M_{0}})V\ket{\phi}_{MA} \otimes \ket{1}^{\otimes T}\bra{\tilde{\eta}_{\phi}} \nonumber \\
                                                                   &- \frac{1}{2\sqrt{T}}\ket{\tilde{\eta}_{\phi}}\bra{\phi}_{MA}V^{*}(\ket{0}\bra{0}_{M_{0}}) \otimes \bra{1}^{\otimes T}.
    \end{align}
    Thus we find $\D_{\loc}^{(V, b, a)}(\ket{\tilde{\eta}_{\phi}}\bra{\tilde{\eta}_{\phi}}) = \Gamma_{\out}\tilde{\D}_{\out}(\ket{\tilde{\eta}_{\phi}}\bra{\tilde{\eta}_{\phi}})$ where $\tilde{\D}_{\out}(\ket{\tilde{\eta}_{\phi}}\bra{\tilde{\eta}_{\phi}})$ is given by equation~\eqref{eq:doutloc}.

    Calculating the decay of $\ket{\tilde{\eta}_{\phi}}\bra{\tilde{\eta}_{\phi}}$ under $\D_{\loc}^{(V, b, a)}$ we find,
    \begin{equation}
        |\braket{\ket{\tilde{\eta}_{\phi}}\bra{\tilde{\eta}_{\phi}}, \re\D_{\loc}^{(V, b, a)}(\ket{\tilde{\eta}_{\phi}}\bra{\tilde{\eta}_{\phi}})}| = |\braket{\ket{\tilde{\eta}_{\phi}}\bra{\tilde{\eta}_{\phi}}, \D_{\loc}^{(V, b, a)}(\ket{\tilde{\eta}_{\phi}}\bra{\tilde{\eta}_{\phi}})}| = \Gamma_{\out} |\braket{\ket{\tilde{\eta}_{\phi}}\bra{\tilde{\eta}_{\phi}}, \tilde{\D}_{\out}(\ket{\tilde{\eta}_{\phi}}\bra{\tilde{\eta}_{\phi}})}|.
    \end{equation}
    Applying equation~\eqref{eq:doutloc}, and our definition of $\Gamma_{\out}$ found in equation~\eqref{eq:dissrates} of Definition~\ref{def:QCSATdis}, we find
    \begin{equation}
        |\braket{\ket{\tilde{\eta}_{\phi}}\bra{\tilde{\eta}_{\phi}}, \re\D_{\loc}^{(V, b, a)}(\ket{\tilde{\eta}_{\phi}}\bra{\tilde{\eta}_{\phi}})}| =  \Big| \frac{1}{T}\bra{\phi}V^{*} (\ket{0}\bra{0}_{M_{0}})V\ket{\phi}\bra{\phi}V^{*}(\ket{0}\bra{0}_{M_{0}})V\ket{\phi} - \bra{\phi}V^{*}(\ket{0}\bra{0}_{M_{0}})V\ket{\phi} \Big|,
    \end{equation}
    which simplifies to 
    \begin{equation}
        |\braket{\ket{\tilde{\eta}_{\phi}}\bra{\tilde{\eta}_{\phi}}, \re\D_{\loc}^{(V, a, b)}(\ket{\tilde{\eta}_{\phi}}\bra{\tilde{\eta}_{\phi}})}| = \Big| \frac{1}{T} \Tr\left(\ket{0}\bra{0}_{M_{0}}V\ket{\phi}\bra{\phi}V^{*}\right) -  1\Big| \Tr\left(\ket{0}\bra{0}_{M_{0}}V\ket{\phi}\bra{\phi}V^{*}\right).
    \end{equation}
    Note, the term inside the absolute value is always less than zero since $0 \leq \Tr\left(\ket{0}\bra{0}_{M_{0}}V\ket{\phi}\bra{\phi}V^{*}\right) \leq 1$, thus
    \begin{equation}
        |\braket{\ket{\tilde{\eta}_{\phi}}\bra{\tilde{\eta}_{\phi}}, \re\D_{\loc}^{(V, a, b)}(\ket{\tilde{\eta}_{\phi}}\bra{\tilde{\eta}_{\phi}})}| = \Big(1 - \frac{1}{T} \Tr\left(\ket{0}\bra{0}_{M_{0}}V\ket{\phi}\bra{\phi}V^{*}\right)\Big)\Tr\left(\ket{0}\bra{0}_{M_{0}}V\ket{\phi}\bra{\phi}V^{*}\right),
    \end{equation}
    implying the statement of Lemma~\ref{lm:properdecayloc} holds.
\end{proof}

\section{Subspace Lemmas}
\label{app:G}

Throughout Section~\ref{sec:hardness} we ``project" our analysis through various different subspaces of the full pure state space associated to our Lindbladian dissipators which encode instances of \textsc{Quantum Circuit-SAT}. In this appendix we collect the various different Lemmas characterizing these pure state subspaces and provide their proofs.
\begin{lemma*}[Lemma~\ref{lm:onlyhist}]
    Let $(V, b, a)$ be an instance of \textsc{Quantum Circuit-SAT} where $V=U_{T-1}U_{T-2}\cdots U_{0}$ with $T > 2$. Define subspace $\mathcal{H}_{\prop}$ according to equation~\eqref{eq:Hprop} of Definition~\ref{def:histsubspace} using $V$, then
    \begin{equation}
        \mathcal{H}_{\prop} = U^{(V)}[\mathcal{H}_{MA} \otimes \ket{p}_{C}]
    \end{equation}
    where $\ket{p} = \frac{1}{\sqrt{T}}\sum_{t=0}^{T-1}\ket{t}$ and $U^{(V)}$ is given in equation~\eqref{eq:circrot} of Appendix~\ref{app:A}.
\end{lemma*}
\begin{proof}
    To show $\mathcal{H}_{\prop} = U^{(V)}[\mathcal{H}_{MA} \otimes \ket{p}_{C}]$ we require:
    \begin{enumerate}
        \item For all $\ket{h} \in U^{(V)}[\mathcal{H}_{MA} \otimes \ket{p}_{C}]$, $\ket{\eta_{\phi}} \in \mathcal{H}_{\prop}$
        \item For all $\ket{h} \in \mathcal{H}_{\prop}$, $\ket{h} \in U[\mathcal{H}_{MA} \otimes \ket{p}_{C}]$.
    \end{enumerate}

    Consider the first of these two conditions, for all $\ket{h} \in U[\mathcal{H}_{MA} \otimes \ket{p}_{C}]$ we have $\ket{h}\bra{h} \in \U^{(V)}[\mathcal{B}(\mathcal{H}_{MA}) \otimes \Span{\{\mathbb{I}_{C}, \ket{p}\bra{p}_{C}\}}]$ From Proposition~\ref{prop:propdisker} we have $\ket{h}\bra{h} \in \ker{\D_{\prop}^{(V)}}$ implying $\ket{h} \in \mathcal{H}_{\prop}$, thus we only need to show the second condition in-order to prove the statement.

    Consider the second condition. Let $\ket{h} \in \mathcal{H}_{\prop}$ implying $\ket{h}\bra{h} \in \ker{\D_{\prop}^{(V)}}$. From Proposition~\ref{prop:propdisker} we find
    \begin{equation}
        \label{eq:revstate}
        \sigma = \mathcal{U}^{(V)}(\ket{h}\bra{h}) = \alpha_{1}P_{MA} \otimes \frac{\mathbb{I}_{C}}{T} + \alpha_{2} Q_{MA}\otimes \ket{p}\bra{p}_{C} 
    \end{equation}
    for some operators $P, Q \in \mathcal{B}(\mathcal{H}_{MA})$ where $\Tr(P) = \Tr(Q) = 1$ and $\alpha_{1}, \alpha_{2} \in \mathbb{R}$. Since $\mathcal{U}^{(V)}$ is unitary, $\sigma$ satisfies
    \begin{enumerate}
        \item $\Tr(\sigma) = 1$,
        \item $\sigma \succeq 0$,
        \item $\sigma^{2} = \sigma$.
    \end{enumerate}
    From property 3. we find 
    \begin{equation}
        \sigma^{2} - \sigma = \left(\frac{\alpha_{1}^{2}}{T^{2}} P_{MA}^{2} - \frac{\alpha_{1}}{T} P_{MA}\right) \otimes \mathbb{I}_{C} + \frac{\alpha_{1}\alpha_{2}}{T}\{P_{MA}, Q_{MA}\} \otimes \ket{p}\bra{p}_{C} + \alpha_{2}Q_{MA}^{2} \otimes \ket{p}\bra{p}_{C} - \alpha_{2}Q_{MA} \otimes \ket{p}\bra{p}_{C} = 0.
    \end{equation}
    which implies 
    \begin{equation}
        (\mathbb{I}_{MA} \otimes \bra{p^{\perp}}_{C})(\sigma^{2} - \sigma)(\mathbb{I}_{MA} \otimes \ket{p^{\perp}}_{C}) = \frac{\alpha_{1}^{2}}{T^{2}} P^{2} - \frac{\alpha_{1}}{T} P = 0.
    \end{equation}
    Thus $\frac{\alpha_{1}}{T}P$ is a projector operator. 

    Let $\mathrm{rank}\left(\frac{\alpha_{1}}{T}P\right) > 1$, thus there exists eigenvectors $\ket{v_{1}}$ and $\ket{v_{2}}$ where $\ket{v_{1}} \neq \ket{v_{2}}$, both with eigenvalues of $1$. Calculating the fidelities we find,
    \begin{equation}
        \begin{cases}
            F(\sigma, \ket{v_{1}}_{MA}\otimes\ket{p^{\perp}}_{C}) = \bra{v_{1}}\left(\frac{\alpha_{1}}{T}P\right)\ket{v_{1}} = 1, \\
            F(\sigma, \ket{v_{2}}_{MA}\otimes\ket{p^{\perp}}_{C}) = \bra{v_{2}}\left(\frac{\alpha_{1}}{T}P\right)\ket{v_{2}} = 1,
        \end{cases}
    \end{equation}
    which implies $\ket{v_{1}}\bra{v_{1}}_{MA}\otimes\ket{p^{\perp}}\bra{p^{\perp}}_{C} = \sigma =\ket{v_{2}}\bra{v_{2}}_{MA}\otimes\ket{p^{\perp}}\bra{p^{\perp}}_{C}$ which is a contradiction since $\ket{v_{1}} \neq \ket{v_{2}}$. Thus $\mathrm{rank}\left(\frac{\alpha_{1}}{T}P\right) \ngtr 1$. Assume $\mathrm{rank}(\frac{\alpha_{1}}{T}P) = 1$, thus there exists some $\ket{c} \in \mathcal{H}_{MA}$ such that $\frac{\alpha_{1}}{T} P = \ket{c}\bra{c}$. From property 1. we find,
    \begin{equation}
        \Tr(\sigma) = \mathrm{rank}\left(\frac{\alpha_{1}}{T}P\right)T + \alpha_{2}\Tr(Q) = T + \alpha_{2} = 1,
    \end{equation}
    which implies $\alpha_{2} = 1 - T$.  Let $\{\ket{a}\}$ be an orthnormal basis for $\mathcal{H}_{MA}$, we find 
    \begin{align}
        (\bra{a}_{MA}\otimes\bra{p}_{C})\sigma(\ket{a}_{MA} \otimes \ket{p}_{C}) &= |\braket{a|c}|^{2} + \alpha_{2}\bra{a}Q\ket{a} \leq \sum_{a}\left(|\braket{a|c}|^{2} + \alpha_{2}\bra{a}Q\ket{a}\right) = 1 + \alpha_{2}\Tr(Q) = 1 + \alpha_{2} \nonumber \\
        &\leq 2 - T,
    \end{align}
    which contradicts property 2. for $T > 2$. Thus for $T > 2$, we require $\mathrm{rank}(\frac{\alpha_{1}}{T}P) = 0$ implying $\frac{\alpha_{1}}{T}P = 0$ forcing equation~\eqref{eq:revstate} to become
    \begin{equation}
        \sigma = \mathcal{U}^{(V)}(\ket{h}\bra{h}) = \alpha_{2} Q_{MA}\otimes \ket{p}\bra{p}_{C}.
    \end{equation}
    The above equation along with properties 1., 2., and 3. imply there exists a pure state $\ket{\phi} \in \mathcal{H}_{MA}$ such that $\alpha_{2}Q = \ket{\phi}\bra{\phi}$ thus we find,
    \begin{equation}
        \ket{h}\bra{h} = \mathcal{U}^{(V)}(\sigma) = \mathcal{U}^{(V)}(\ket{\phi}\bra{\phi}_{MA} \otimes \ket{p}\bra{p}_{C}) = \ket{\eta_{\phi}}\bra{\eta_{\phi}}
    \end{equation}
    where $\ket{\eta_{\phi}}$ is a pure history state associated to the quantum circuit $V$ with input state $\ket{\phi}$. Thus $\ket{\eta_{\phi}} \in U[\mathcal{H}_{MA} \otimes \ket{p}]$, and the statement of Lemma~\ref{lm:onlyhist} follows.
\end{proof}
\begin{lemma*}[Lemma~\ref{lm:onlycirc}]
    Let $(V, b, a)$ be an instance of \textsc{Quantum Circuit-SAT} where $V=U_{T-1}U_{T-2}\cdots U_{0}$ with $T > 2$. Define $\mathcal{H}_{\qcirc}$ according to equation~\eqref{eq:Hqcirc} of Definition~\ref{def:circsubspace} using $V$, then
    \begin{equation}
        \mathcal{H}_{\qcirc} = U^{(V)}[\mathcal{H}_{M} \otimes \ket{0}_{A} \otimes \ket{p}_{C}]
    \end{equation}
    where $\ket{p} = \frac{1}{\sqrt{T}}\sum_{t=0}^{T-1}\ket{t}$ and $U^{(V)}$ is given in equation~\eqref{eq:circrot} of Appendix~\ref{app:A}.
\end{lemma*}
\begin{proof}
    To show $\mathcal{H}_{\qcirc} = U^{(V)}[\mathcal{H}_{M}\otimes \ket{0}_{A} \otimes \ket{p}_{C}]$ we require:
    \begin{enumerate}
        \item For all $\ket{h} \in U^{(V)}[\mathcal{H}_{M} \otimes \ket{0}_{A} \otimes \ket{p}_{C}]$, we have $\ket{h} \in \mathcal{H}_{\qcirc}$
        \item For all $\ket{h} \in \mathcal{H}_{\qcirc}$, we have $\ket{h} \in U^{(V)}[\mathcal{H}_{M} \otimes \ket{0}_{A} \otimes \ket{p}_{C}]$.
    \end{enumerate}

    Consider the first of these two conditions, for all $\ket{h} \in U^{(V)}[\mathcal{H}_{M} \otimes \ket{0}_{A} \otimes \ket{p}_{C}]$ we have $\ket{h} \in \mathcal{H}_{\prop}$ through Lemma~\ref{lm:onlyhist}. In-addition,  we have for all $i \in \{0, 1, \cdots, k-1\}$  
    \begin{align}
        (\ket{1}\bra{1}_{A_{i}} \otimes \ket{0}\bra{0}_{C}) \ket{h} &=  (\ket{1}\bra{1}_{A_{i}} \otimes \ket{0}\bra{0}_{C}) U^{(V)}(\ket{\phi}_{M} \otimes \ket{0}_{A} \otimes \ket{p}_{C})  = \frac{1}{\sqrt{T}}\sum_{t=0}^{T-1} \delta_{t,0} (\ket{1}\bra{1}_{A_{i}}) V_{t}(\ket{\phi}_{M} \otimes \ket{0}_{A}) \otimes \ket{0}_{C} \nonumber \\
                                                                    &= \frac{1}{\sqrt{T}}(\ket{1}\bra{1}_{A_{i}}) V_{0}(\ket{\phi}_{M} \otimes \ket{0}_{A}) \otimes \ket{0}_{C} = 0.
    \end{align}
    Since $\D_{\inp}$ is generated by jump operators $\{\ket{1}\bra{1}_{A_{i}} \otimes \ket{0}\bra{0}_{C}\}_{i = 0}^{k-1}$ we have through Lemma~\ref{lm:restrictdfscond} and the above equation that $\D_{\inp}(\ket{h}\bra{h}) = 0$. Thus $\ket{h}\bra{h} \in \ker{\D_{\inp}}$ in-addition to $\ket{h} \in \mathcal{H}_{\prop}$. Therefore $\ket{h} \in \mathcal{H}_{\qcirc}$.

    Consider the second of the two conditions listed above. Let $\ket{h} \in \mathcal{H}_{\qcirc} \subseteq \mathcal{H}_{\prop}$. Lemma~\ref{lm:onlyhist} implies
    \begin{equation}
        \label{eq:histt}
        \ket{h} = U^{(V)}(\ket{\phi}_{MA} \otimes \ket{p}_{C}) = \frac{1}{\sqrt{T}}\sum_{t=0}^{T-1}V_{t}\ket{\phi}_{MA}\otimes\ket{t}_{C}.
    \end{equation}
    Calculating $\braket{\ket{h}\bra{h}, \D_{\inp}(\ket{h}\bra{h})}$ we find,
    \begin{align}
        \label{eq:decayofcirc}
        \braket{\ket{h}\bra{h}, \D_{\inp}(\ket{h}\bra{h})} &= \sum_{i=0}^{k-1}\left(\bra{h}(\ket{1}\bra{1}_{A_{i}} \otimes \ket{0}\bra{0}_{C})\ket{h}\left(\bra{h}(\ket{1}\bra{1}_{A_{i}} \otimes \ket{0}\bra{0}_{C})\ket{h} - 1\right)\right) \nonumber \\
                                                           &=\frac{1}{T}\sum_{i=0}^{k-1}\left(\bra{\phi}_{MA}(\ket{1}\bra{1}_{A_{i}})\ket{\phi}_{MA}\left(\frac{1}{T}\bra{\phi}_{MA}(\ket{1}\bra{1}_{A_{i}})\ket{\phi}_{MA} - 1\right)\right).
    \end{align}
    
    From the definition of $\mathcal{H}_{\qcirc}$ we have $\D_{\inp}(\ket{h}\bra{h})=0$, thus
    \begin{equation}
        \label{eq:decayofcirc1}
        \sum_{i=0}^{k-1}\left(\bra{\phi}_{MA}(\ket{1}\bra{1}_{A_{i}})\ket{\phi}_{MA}\left(\frac{1}{T}\bra{\phi}_{MA}(\ket{1}\bra{1}_{A_{i}})\ket{\phi}_{MA} - 1\right)\right) = 0.
    \end{equation}
    Since $\ket{1}\bra{1}_{A_{i}}$ is a projection operator, $0 \leq \bra{\phi}_{MA}(\ket{1}\bra{1}_{A_{i}})\ket{\phi}_{MA} \leq 1$ implying each of the terms of left-hand side of equation~\eqref{eq:decayofcirc1} must be zero, thus we find for all $i \in \{0, 1, \cdots, k-1\}$,
    \begin{equation}
        \label{eq:decayofcirc2}
        \bra{\phi}_{MA}(\ket{1}\bra{1}_{A_{i}})\ket{\phi}_{MA} = 0.
    \end{equation}
    From expanding $\ket{\phi}_{MA}$ in the single qubit subspace of $A_{i}$ as $\ket{\phi}_{MA} = c_{i,0}\ket{0}_{A_{i}}\otimes\ket{\chi_{0}}_{MA \setminus A_{i}} + c_{i,1}\ket{1}_{A_{i}}\otimes\ket{\chi_{1}}_{MA \setminus A_{i}}$ and applying equation~\eqref{eq:decayofcirc2} we find for all $i \in \{0, 1, \cdots, k-1\}$ that $c_{i,1}=0$, thus for all $i \in \{0, 1, \cdots, k-1\}$,
    \begin{equation}
        \ket{\phi}_{MA} = \ket{0}_{A_{i}}\otimes\ket{\chi_{0}}_{MA \setminus A_{i}},
    \end{equation}
    which can only be true if there exists some quantum state $\ket{\psi}$ on proof register $M$ such that 
    \begin{equation}
        \ket{\phi}_{MA} = \ket{\psi}_{M}\otimes\ket{0}_{A}.
    \end{equation}
    Thus from equation~\eqref{eq:histt} we have for all $\ket{h} \in \mathcal{H}_{\qcirc}$,
    \begin{equation}
        \ket{h} = U^{(V)}(\ket{\psi}_{M} \otimes \ket{0}_{A} \otimes \ket{p}_{C})
    \end{equation}
    implying $\ket{h} \in U^{(V)}[\mathcal{H}_{M} \otimes \ket{0}_{A} \otimes \ket{p}_{C}]$. The statement of Lemma~\ref{lm:onlycirc} follows.
\end{proof}

\section{The Symmetric* $k$-local Hamiltonian Problem}
\label{app:H}

In Section~\ref{sec:LL} we introduce a new local Hamiltonian problem called the \SEPH{$k$} problem which was used to provide a $\mathrm{QMA}$ verification protocol for our \LL{$k$} problem. In this appendix we give roof that this new local Hamiltonian problem is indeed in $\mathrm{QMA}$ by providing its $\mathrm{QMA}$ verification protocol. For reference we revisit the definition of the \SEPH{$k$} problem given below.
\begin{definition*}[\SEPH{$k$} Problem, Definition~\ref{def:SEPH}]
    Let $H$ be a $2k$-local Hamiltonian acting on $2n$ qubits partitioned into $n$ qubit registers $A$ and $B$ with $H$ given by $\{H_{i}\}_{i=1}^{r=\poly(n)}$ where $\mathbb{I} \succeq H_{i} \succeq 0$ and  if $H_{i}$ acts non-trivially on $k'$ qubits in register $A$ then it also acts non-trivially on $k'$ qubits in register $B$, decide
    \begin{itemize}
        \item \emph{(accept-case)} there exists a pure quantum state $\ket{\psi} \in \mathcal{H}_{A}$ such that $(\bra{\psi}_{A} \otimes \bra{\bar{\psi}}_{B}) H(\ket{\psi}_{A} \otimes \ket{\bar{\psi}}_{B}) \leq a$,
        \item \emph{(reject-case)} For all pure quantum states $\ket{\psi} \in \mathcal{H}_{A}$ and $\ket{\phi} \in \mathcal{H}_{B}$, $(\bra{\psi}_{A} \otimes \bra{\phi}_{B}) H (\ket{\psi}_{A} \otimes \ket{\phi}_{B}) \geq b$
    \end{itemize}
    promised one of these to be the case, where $b-a \geq 1 /\poly(n)$.
\end{definition*}
The \SEPH{$k$} problem  is a variant of the \textsc{Separable $k$-Local Hamiltonian} problem studied by Chailloux et al. in \cite{chailloux2012complexity}.  Thus we take inspiration from the $\mathrm{QMA}$ verification protocol provided for the \textsc{Separable $k$-Local Hamiltonian} problem in \cite{chailloux2012complexity}. The $\mathrm{QMA}$ verification protocol queries the verification protocol for the \textsc{Consistency of $k$-Local Density Matrices} problem for which we provide the definition of below.
\begin{definition}[\textsc{Consistency Of $k$-Local Density Matrices} \cite{liu2006consistency}]
    \label{def:localdens}
    Let $\{\rho_{i}\}_{i=1}^{m=\poly(n)}$ be a set of $k$-local density matrices where each $\rho_{i}$ acts on  $C_{i} \subseteq \{0, \cdots, n-1\}$ qubits. Decide, 
    \begin{itemize}
        \item \emph{(accept-case)} There exist an $n$-qubit state $\sigma$ such that for all $i \in [m]$, $D(\rho_{i}, \Tr_{\{0, \cdots, n-1\} \setminus C_{i}}\left[\sigma\right]) = 0$
        \item \emph{(reject-case)}  For all $n$-qubit states $\sigma$ there exists an $i \in [m]$ such that $D(\rho_{i}, \Tr_{\{0, \cdots, n-1\} \setminus C_{i}}\left[\sigma\right]) \geq \beta$,
    \end{itemize}
    promised one of these to be the case, where $\beta \geq 1 / \poly(n)$.
\end{definition}
An \emph{instance of} \textsc{Consistency of $k$-Local Density Matrices} is given by a tuple $(S, \beta)$ where $\beta$ is defined as in Definition~\ref{def:localdens} and the set $S$ consists of the $k$-local density matrices $\rho_{i}$ described in Definition~\ref{def:localdens}. It has been shown by Liu in \cite{liu2006consistency} that $\textsc{Consistency of $k$-Local Density Matrices} \in \mathrm{QMA}$ and thus has a $\mathrm{QMA}$ verification protocol that can be queried non-adaptively in any other $\mathrm{QMA}$ verification protocol. Our verification protocol for \SEPH{$k$} is provided below. 

\begin{algorithm}[hbt!]
    \caption{QMA Verification Protocol for \SEPH{$k$}}
    \label{alg:QMAver}
    \begin{algorithmic}[1]
        \Require An instance of \SEPH{$k$}, $(F, a, b, A, B)$, Merlin sends a classical description of the states $\{\rho_{A_{i}} \}_{i=1}^{m}$ where $\rho_{A_{i}}$ acts on the $k$ qubits in $A$ that $H_{i} \in F$ acts on. Merlin also sends one quantum state $\ket{\psi}$.
        \State $Q_{A} \gets $ Query \textsc{Consistency of $k$-Local Density Matrices} for $\{\rho_{A_{i}}\}_{i=1}^{m}$ and $\beta = (b-a)/8m$
        \If{not $Q_{A}$}
        \State\textbf{return:} reject
        \EndIf
        \State $E \gets \sum_{i=1}^{m}\Tr(H_{i}\rho_{A_{i}} \otimes \bar{\rho}_{A_{i}})$
        \If{$E \leq a$}
        \State\textbf{return:} accept
        \Else
        \State\textbf{return:} reject
        \EndIf
    \end{algorithmic}
\end{algorithm}

Algorithm~\ref{alg:QMAver} is in essence the $\mathrm{QMA}$ verification protocol for \textsc{Separable $k$-Local Hamiltonian} found in \cite{chailloux2012complexity}, therefore our proof of the validity of the protocol follows similarly to that provided in \cite{chailloux2012complexity}. Thus we show the following statement
\begin{proposition*}[Proposition~\ref{prop:SEPHinQMA}]
    For $k \leq \log(n)$, $\SEPH{$k$} \in \mathrm{QMA}$.
\end{proposition*}
\begin{proof}
    Consider the $\mathrm{QMA}$ verification protocal described in Algorithm~\ref{alg:QMAver}. Each state $\rho_{i}$ acts on $2k$ qubits with $k \leq \log(n)$, therefore each $\rho_{i}$ is a $\poly(n) \times \poly(n)$ matrix implying line 4 can be computed in $\poly(n)$ time classical computation. Since \textsc{Conistency of $k$-Local Density Matrice} $\in \mathrm{QMA}$ line 1 compiles a $\poly(n)$ size quantum circuit. Since $\poly(n)$ time classical computation can be mapped to a $\poly(n)$ size quantum circuit, the whole protocol can be compiled into a $\poly(n)$ size quantum circuit and thus is efficient.\\

    \ul{\textbf{Completeness:}} \\

    Let $(F, a, b, A, B)$ be an \emph{(accept-case)} instance of \SEPH{$k$}, thus there exists a pure quantum state $\ket{\psi} \in \mathcal{H}_{A}$ such that 
    \begin{equation}
        \label{eq:acceptcase}
        (\bra{\psi}_{A} \otimes \bra{\bar{\psi}}_{B})H(\ket{\psi}_{A} \otimes \ket{\bar{\psi}}_{B}) \leq a.
    \end{equation}
    Merlin chooses $\rho_{A_{i}} = \Tr_{A \setminus A_{i}}(\ket{\psi}\bra{\psi})$. Since $(\{\rho_{A_{i}}\}_{i=1}^{m}, \beta)$ is an \emph{(accept-case)} instance of \textsc{Consistency of $k$-Local Density Matrices} we have
    \begin{equation}
        \mathrm{Pr}(Q_{A} = \text{accept}) \geq 2 / 3.
    \end{equation}
    Thus the protocol does not reject on line 3 with a probability greater than $2/3$. From equation~\eqref{eq:acceptcase} we find 
    \begin{align}
        a \geq (\bra{\psi}_{A} \otimes \bra{\bar{\psi}}_{B})H(\ket{\psi}_{A} \otimes \ket{\bar{\psi}}_{B}) &= \sum_{i=1}^{m}\Tr\left(H_{i} \ket{\psi}\bra{\psi}_{A}\otimes \ket{\bar{\psi}}\bra{\bar{\psi}}_{B}\right) \nonumber \\
                                                                                                             &= \sum_{i=1}^{m}\Tr\left(H_{i} \Tr_{A \setminus A_{i}}(\ket{\psi}\bra{\psi}_{A})\otimes \Tr_{B \setminus B_{i}}(\ket{\bar{\psi}}\bra{\bar{\psi}}_{B})\right) \nonumber \\
                                                                                                               &= \sum_{i=1}^{m}\Tr\left(H_{i} \rho_{A_{i}} \otimes \bar{\rho}_{A_{i}}\right) = E.
    \end{align}
    Thus if the protocol does not reject on line 3, it will always accept on line 6. Therefore the protocol accepts with a probability at least $2/3$ as required. \\

    \ul{\textbf{Soundness:}} \\

    Let $(F, a, b, A, B)$ be a \emph{(reject-case)} instance of \SEPH{$k$}, thus for all $\ket{\psi} \in \mathcal{H}_{A}$ and $\ket{\chi} \in \mathcal{H}_{B}$ we have
    \begin{equation}
        \label{eq:rejectcase}
        (\bra{\psi}_{A} \otimes \bra{\chi}_{B})H(\ket{\psi}_{A} \otimes \ket{\chi}_{B}) \geq b.
    \end{equation}
    Merlin sends the set $\{\rho_{A_{i}}\}_{i=1}^{m}$, thus the \textsc{Consistency of $k$-Local Density Matrices} query on line 1 either accepts with
    \begin{equation}
        \mathrm{Pr}(Q_{A} = \text{accept}) \leq 1 / 3.
    \end{equation}
    when the set $\{\rho_{A_{i}}\}_{i=1}^{m}$ isn't consistent with some $\sigma \in \mathcal{B}(\mathcal{H}_{A})$ or 
    \begin{equation}
        \mathrm{Pr}(Q_{A} = \text{accept}) \geq 2 /3 .
    \end{equation}
    when the set $\{\rho_{A_{i}}\}_{i=1}^{m}$ does pass the consistency test. When the set does pass the consistency test, then for all $i \in [m]$ there exists a quantum state $\sigma \in \mathcal{B}(\mathcal{H}_{A})$ such that 
    \begin{equation}
        D(\rho_{A_{i}}, \Tr_{A \setminus A_{i}}(\sigma)) \leq \beta = \frac{b - a}{8m}.
    \end{equation}
    By the sub-additivity of trace distance with respect to tensor products we find
    \begin{equation}
        D(\rho_{A_{i}}\otimes \bar{\rho}_{B_{i}}, \Tr_{A \setminus A_{i}}(\sigma) \otimes \Tr_{B \setminus B_{i}}(\bar{\sigma}) ) \leq  D(\rho_{A_{i}}, \Tr_{A \setminus A_{i}}(\sigma)) + D(\bar{\rho}_{B_{i}}, \Tr_{B \setminus B_{i}}(\bar{\sigma})) \leq \frac{b - a}{4m}.
    \end{equation}
    Calculating $E$ in line 4 we find 
    \begin{align}
        E &= \sum_{i=1}^{m}\Tr(H_{i}\rho_{A_{i}}\otimes \bar{\rho}_{A_{i}}) \geq \sum_{i=1}^{m}\left(\Tr(H_{i}\Tr_{A \setminus A_{i}}(\sigma) \otimes \Tr_{B \setminus B_{i}}(\bar{\sigma})) - \frac{b-a}{2m}\right) = \Tr(H\sigma_{A} \otimes \bar{\sigma}_{B}) - \frac{b-a}{2} \nonumber \\
          &\geq  \sum_{i,j}p_{i}p_{j}\Tr(H\ket{\psi_{i}}\bra{\psi_{i}}_{A} \otimes \ket{\bar{\psi}_{j}}\bra{\bar{\psi}_{j}}_{B}) - \frac{b-a}{2}  \geq \min_{i,j}(\bra{\psi_{i}}_{A}\otimes \bra{\bar{\psi_{i}}}_{B})H(\ket{\psi_{i}}_{A}\otimes \ket{\bar{\psi}_{i}}) - \frac{b - a}{2} \nonumber \\
          &\geq  b - \frac{b-a}{2} = \frac{a + b}{2} 
    \end{align} 
    where the first inequality follows from Claim 20 of \cite{chailloux2012complexity} and the last inequality follows from equation~\eqref{eq:rejectcase}. Since we have $b - a \geq 1 /\poly(n)$ we find  
    \begin{equation}
        E \geq a  + \frac{1}{\poly(n)},
    \end{equation}
    therefore in this case, when the set $\{\rho_{A_{i}}\}_{i=1}^{m}$ is consistent with some quantum state in $\sigma \in \mathcal{B}(\mathcal{H}_{A})$ and $Q_{A} = \text{accept}$, the protocol will always reject on line 8. Since in the first case, when the set $\{\rho_{A_{i}}\}_{i=1}^{m}$ is not consistent, the protocol accepts with a probability at most $1/3$, and in the second case, when the set does pass the consistency test, it will always reject, the protocol will accept with a probability at most $1/3$ as required. Therefore $\SEPH{$k$} \in \mathrm{QMA}$.
\end{proof}

\end{document}